\documentclass[10pt]{article}

\usepackage[margin=1.0in]{geometry}
\usepackage{amsthm, amsfonts,amssymb,euscript}
\usepackage{amsmath}
\usepackage{enumerate}
\usepackage{bm}
\usepackage{amssymb}
\usepackage{amsthm}
\usepackage{graphicx}
\usepackage{caption}
\usepackage{subcaption}
\usepackage{amsfonts}
\usepackage{float}
\usepackage{color}
\usepackage{slashed}
\usepackage[affil-it]{authblk}
\usepackage{mathrsfs}
\usepackage{upgreek}
\usepackage{array,makecell}
\usepackage{boldline,multirow}

\usepackage[usenames,dvipsnames,svgnames,table]{xcolor}
\usepackage[colorlinks=true]{hyperref}
\hypersetup{linkcolor=BrickRed, urlcolor=green, citecolor=blue, linktoc=page}

\numberwithin{equation}{section}

\newtheorem*{proposition*}{Proposition}
\newtheorem*{theorem*}{Theorem}
\newtheorem*{conjecture*}{Conjecture}
\newtheorem*{claim*}{Claim}
\newtheorem*{lemma*}{Lemma}
\newtheorem*{corollary*}{Corollary}
\newtheorem{theorem}{Theorem}[section]
\newtheorem{proposition}[theorem]{Proposition}

\newtheorem{lemma}[theorem]{Lemma}
\newtheorem{corollary}[theorem]{Corollary}

\newtheorem*{definition*}{Definition}
\newtheorem{definition}{Definition}[section]
\newtheorem*{assumption*}{\mathcal{A}ssumption}

\newtheorem*{remark*}{Remark}
\newtheorem{remark}{Remark}[section]

\newtheorem{thmx}{Theorem}
\newcommand{\la}{\langle}
\newcommand{\ra}{\rangle}
\newcommand{\R}{\mathbb{R}}
\newcommand{\s}{\mathbb{S}}
\newcommand{\C}{\mathbb{C}}

\newcommand{\Z}{\mathbb{Z}}
\newcommand{\N}{\mathbb{N}}

\newcommand{\snabla}{\slashed{\nabla}}

\newcommand{\chphi}{ \check{\phi}}

\newcommand{\Lbar}{\underline{L}}

\allowdisplaybreaks

\begin{document}

\title{Late-time tails and mode coupling of linear waves on Kerr spacetimes}
\author[1]{Yannis Angelopoulos \thanks {yannis@caltech.edu}}
\author[2]{Stefanos Aretakis\thanks {aretakis@math.toronto.edu}}
\author[3,4]{Dejan Gajic \thanks {D.Gajic@dpmms.cam.ac.uk, d.gajic@ru.nl}}
	\affil[1]{\small The Division of Physics, Mathematics and Astronomy, Caltech,
1200 E California Blvd, Pasadena CA 91125, USA}
	\affil[2]{\small Department of Mathematics, University of Toronto, 40 St George Street, Toronto, ON, Canada}
	\affil[3]{\small Centre for Mathematical Sciences, University of Cambridge, Wilberforce Road, Cambridge CB3 0WB, UK}
	\affil[4]{\small Department of Mathematics, Radboud University, 6525 AJ Nijmegen, The Netherlands}

\date{February 23, 2021}

\maketitle

\begin{abstract}
We provide a rigorous derivation of the precise late-time asymptotics for solutions to the scalar wave equation on subextremal Kerr backgrounds, including the asymptotics for projections to angular frequencies $\ell\geq 1$ and $\ell\geq 2$. The $\ell$-dependent asymptotics on Kerr spacetimes differ significantly from the non-rotating Schwarzschild setting (``Price's law''). The main differences with Schwarzschild are slower decay rates for higher angular frequencies and oscillations along the null generators of the event horizon. We introduce a physical space-based method that resolves the following two main difficulties for establishing $\ell$-dependent asymptotics in the Kerr setting: 1) the coupling of angular modes and 2) a loss of ellipticity in the ergoregion. Our mechanism identifies and exploits the existence of conserved charges along null infinity via a time invertibility theory, which in turn relies on new elliptic estimates in the full black hole exterior. This framework is suitable for resolving the conflicting numerology in Kerr late-time asymptotics that appears in the numerics literature.

\end{abstract}

\tableofcontents

\section{Introduction}
The Kerr spacetimes $(\mathcal{M}_{M,a},g_{M,a})$ constitute a 2-parameter family of solutions to the Einstein vacuum equations
\begin{equation*}
\textnormal{Ric}[g]=0
\end{equation*}
that are expected to describe all possible final states of a wide variety of gravitational collapse scenarios in an astrophysical setting \cite{penrose82}. A preliminary step to describing the intricate dynamical properties of the gravitational radiation that is emitted when black hole exteriors settle down to Kerr solutions is to \emph{verify} that they indeed settle down to Kerr solutions. In the context of the evolution of dynamical black holes arising from ``small perturbations'' of Kerr initial data to the Einstein equations, this is the black hole stability problem. There has been significant recent progress towards addressing the question of linear and nonlinear stability; see for example \cite{Dafermos2016, klainerman17,johns19,anderssonkerr,haf20,shlcosta20} and references therein.

In this paper, we develop mathematical tools necessary for going \emph{beyond} the question of stability, by addressing the \emph{precise} late-time behaviour of gravitational radiation emitted in the evolution of perturbations of Kerr initial data in the context of the model problem of the linear scalar wave equation on a fixed subextremal $(|a|<M)$ Kerr spacetime background $(\mathcal{M}_{M,a},g_{M,a})$:
\begin{equation}
\label{eq:introwaveeq}
\square_{g_{M,a}}\psi=0.
\end{equation}
The evolution of the wave $\psi$ models the evolution of key dynamical quantities for the Einstein equations. Addressing the precise late-time dynamics of gravitational radiation is motivated by the exterior as well as the interior of black holes:
\begin{enumerate}[A)]
\item The decay rates and leading-order coefficients of gravitational radiation measured by observers ``at infinity'' are expected to encode information about the initial perturbation and the final Kerr solution. Deducing information purely from gravitational radiation is important since it can be detected and ana\-lysed experimentally at gravitational wave observatories and forms a key signature of astrophysical black hole processes, see also \cite{extremal-prl}.
\item The precise late-time behaviour of gravitational radiation crossing the event horizon, the boundary between black hole interior and exterior, plays a direct role in uncovering the nature and strength of singularities that may be present in black hole interiors. As such, it is important for addressing the \emph{strong cosmic censorship conjecture}; see the introduction of \cite{dl-scc} for a comprehensive overview of the black hole interior and the strong cosmic censorship conjecture. 
\end{enumerate}

\textbf{In this paper, we determine the precise leading-order late-time behaviour of $\psi$ and its higher angular frequencies $\psi_{\ell\geq 1}$ and $\psi_{\ell \geq 2}$ on Kerr black hole exteriors.} We develop a physical space mechanism that exploits the existence of conserved charges together with a time inversion theory, in order to derive the presence of ``tails'' in the late-time dynamics. Our mechanism deals with new difficulties in deriving late-time asymptotics that are caused by the rotation of the black hole background: a coupling of different angular frequencies and a loss of ellipticity of operators relevant for the time inversion theory. We derive new phenomena that arise due to rotation: slower decay rates at the level of higher angular frequencies and oscillations along the event horizon.

We provide below an outline of the remainder of the introductory section of the paper.

\begin{itemize}
\item In Section \ref{intro:results}, we present the main theorems of the paper.
\item In Section \ref{intro:prevwork}, we discuss previous work on late-time asymptotics for waves on black holes. 
\item In Section \ref{intro:pricelaw}, we sketch an analogue of an angular frequency-dependent Price's law in the Kerr setting.
\item In Section \ref{intro:ideas}, we outline the main new ideas and methods that appear in the proofs of the theorems. 
\end{itemize}

\subsection{Main results and first remarks}
\label{intro:results}

In this section, we present the main results obtained in this paper and provide some additional remarks.

We first introduce briefly the notation appearing in the statements of the theorems below.

\begin{itemize}
\item Let $\Sigma_{\tau}$ denote appropriate spacelike, asymptotically hyperboloidal hypersurfaces in the Kerr manifold $\mathcal{M}_{M,a}$ that intersect the future event horizon $\mathcal{H}^+$. The spacetime region of interest is then the union $\bigcup_{\tau \in [0,\infty)}\Sigma_{\tau}$ which constitutes the main region of interest in $\mathcal{M}_{M,a}$. See Figure \ref{fig:foliationsintro} for a pictorial representation.
\item The hypersurfaces $\Sigma_{\tau}$ are foliated by Boyer--Lindquist spheres $S^2_{\tau,r}$, which we equip with angular coordinates $(\theta,\varphi_*)$. The label $r$ denotes the radial Boyer--Lindquist coordinate, which takes the values $r\in [r_+,\infty)$ along $\Sigma_{\tau}$, with $r=r_+$ at $\mathcal{H}^+$.

Note that the components of the rescaled induced metric $r^{-2}g_{S^2_{\tau,r}}$ approach the components of the metric on the unit round sphere in standard spherical coordinates $(\theta,\varphi_*)$ as $r\to \infty$. In this sense, the Boyer--Lindquist spheres are asymptotically round.

\item The notation $\psi_{\ell}$ and $\psi_{\geq \ell}$ indicates a projection of the function $\psi$ to standard spherical harmonic modes on the spheres $S^2_{\tau,r}$ equipped with the unit round metric, with angular frequencies equal to $\ell$, and greater or equal to $\ell$, respectively. The spherical harmonics here are defined with respect to the angular coordinates $(\theta,\varphi_*)$ and are denoted by $Y_{\ell, m}(\theta,\varphi_*)$. We denote with $\pi_{\ell}$ the operator that projects a function on $S^2_{\tau,r}$ to the spherical harmonics with angular frequency $\ell$.
\end{itemize}

See Section \ref{sec:geom} for a precise introduction of the spacetime geometry, the Boyer--Lindquist coordinates and other notational conventions.

\subsubsection{Main results}
We state here the main theorems proved in the paper.
\begin{figure}[H]
	\begin{center}
\includegraphics[scale=0.5]{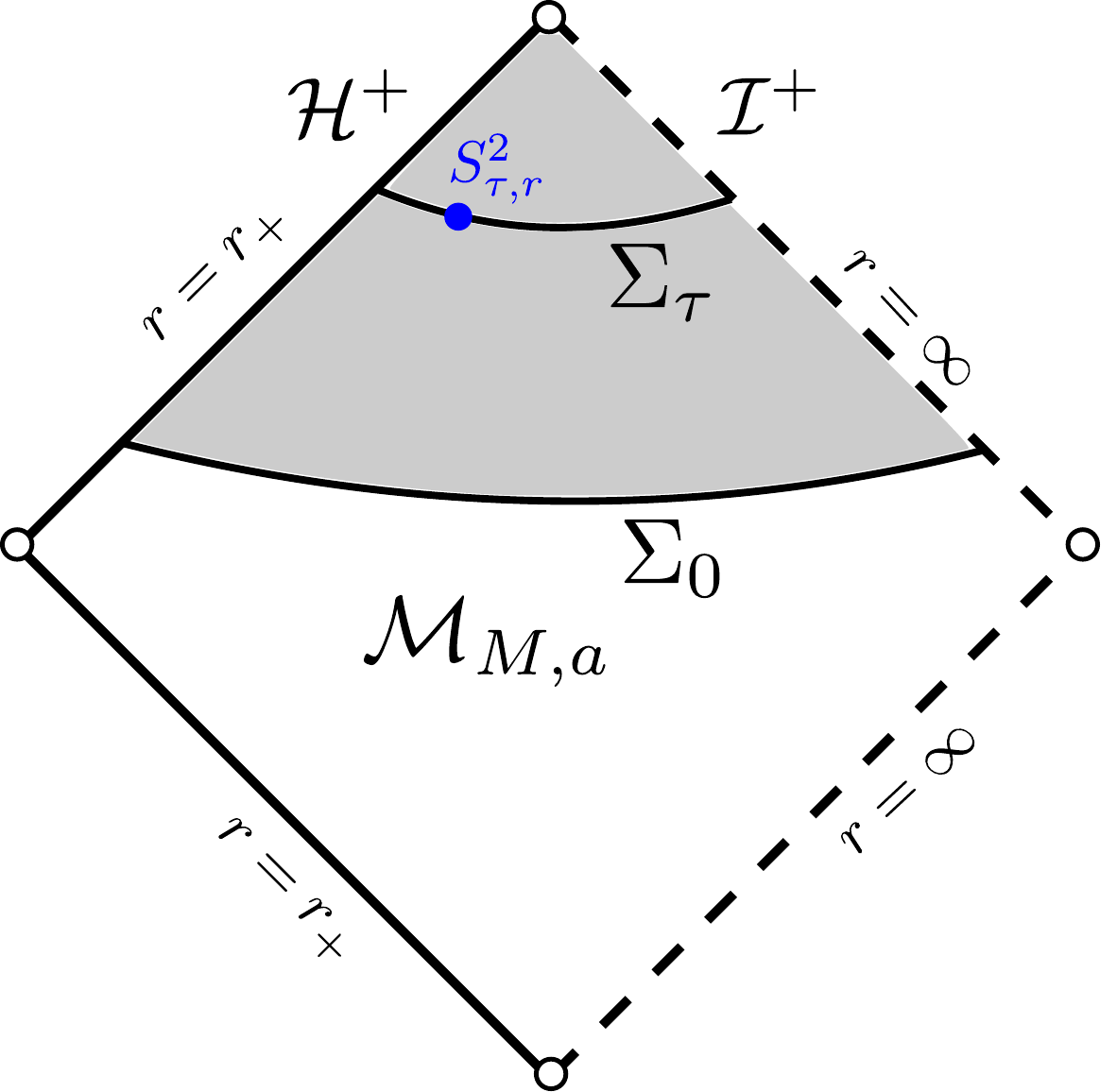}
\end{center}
\vspace{-0.2cm}
\caption{A 2-dimensional representation of the spacetime $\mathcal{M}_{M,a}$, with the hypersurfaces $\Sigma_{\tau}$ and the shaded region depicting $\bigcup_{\tau \in [0,\infty)}\Sigma_{\tau}$. Each point in the picture represents a Boyer--Lindquist sphere $S^2_{\tau,r}$ and the hypersurface $\mathcal{I}^+$, which represents the points $(\tau,\infty,\theta,\varphi_*)$ is depicted at a finite distance.}
	\label{fig:foliationsintro}
\end{figure}

\begin{theorem}
\label{thm:intro1}
Let $\psi$ be a solution arising from smooth and compactly supported initial data for \eqref{eq:introwaveeq}. Then $\psi$ satisfies the following asymptotic behaviour along $\{r=r_0\}$, for each $r_0\geq r_+$:
\begin{align*}
\psi(\tau,r_0,\theta,\varphi_*)=&-8I_0^{(1)}(1+\tau)^{-3}+\mathcal{O}_{r_0}(\tau^{-3-\eta}),\\
r^{-1}\psi_1(\tau,r_0,\theta,\varphi_*)=&-\frac{32}{3}I_1^{(1)}(\theta,\varphi_*)(1+\tau)^{-5}+\mathcal{O}_{r_0}(\tau^{-5-\eta}),\\
\psi_{\geq2}(\tau,r_0,\theta,\varphi_*)=&-\frac{16}{3}\sqrt{\frac{\pi}{5}}a^2I_0^{(1)}Y_{2,0}(\theta)(1+\tau)^{-5}+\mathcal{O}_{r_0}(\tau^{-5-\eta}),\end{align*}
where $\eta>0$, $I_i^{(1)}$ are functions on $\s^2$ and denote the \emph{time-inverted Newman--Penrose charges of $\psi$} (see Section \ref{rm:npcharges}), and $\mathcal{O}_{r_0}(\tau^{-p})$ denotes schematically terms that can be bounded uniformly in $\tau$ by weighted initial data $L^2$-norms multiplied by $(1+\tau)^{-p}$, with constants that may depend on $r_0$.

When $r\to \infty$ the ``radiation field'' $r\psi$ has the following asymptotic behaviour:
\begin{align*}
r\psi(\tau,\infty,\theta,\varphi_*)=&-2I_0^{(1)}(1+\tau)^{-2}+\mathcal{O}(\tau^{-2-\eta}),\\
r\psi_1(\tau,\infty,\theta,\varphi_*)=&-\frac{4}{3}I_1^{(1)}(\theta,\varphi_*)(1+\tau)^{-3}+\mathcal{O}(\tau^{-3-\eta}),\\
r\psi_{\geq 2}(\tau,\infty,\theta,\varphi_*)=&\left[-\frac{2}{5}I_2^{(1)}(\theta,\varphi_*)+\frac{8}{3}\sqrt{\frac{\pi}{5}}a^2I_0^{(1)}Y_{2,0}(\theta)\right](1+\tau)^{-4}+\mathcal{O}(\tau^{-4-\eta}),
\end{align*}
where $\mathcal{O}(\tau^{-p})$ denote schematically terms that can be bounded uniformly in $\tau$ by weighted initial data $L^2$-norms multiplied by $(1+\tau)^{-p}$.
\end{theorem}

A more detailed version of Theorem \ref{thm:intro1} follows directly by combining Propositions \ref{prop:asympl0NP0boundr}, \ref{prop:asympl1NP0boundr} and \ref{prop:asympl2NP0boundr}. In these propositions we moreover obtain global analogues of the expressions in Theorem \ref{thm:intro1} with additional derivatives $\partial_{\tau}^k$ on both sides, for arbitrary $k\in \N_0$. 

\begin{theorem}
\label{thm:intro2}
Let $\psi$ be a solution arising from  initial data for \eqref{eq:introwaveeq} on $\Sigma_0$ that are smooth with respect to the differentiable structure on the conformal compactification of $\Sigma_0$, with non-zero \emph{Newman--Penrose charges} $I_{\ell}$, $\ell=0,1,2$ (see Section \ref{rm:npcharges}). Then the rates and coefficients in the late-time expansions presented in Theorem \ref{thm:intro1} and Corollary \ref{cor:introosc} are modified according to the expressions displayed in Table \ref{table:nonzeronp}.
\end{theorem}

See Section \ref{sec:foliations} for a precise definition of the conformal compactification of $\Sigma_0$. We directly obtain a more detailed version of Theorem \ref{thm:intro1} by combining Corollary \ref{cor:asympl0nearI} and Propositions \ref{prop:asympl0boundr}, \ref{prop:asympl1nearI}, \ref{prop:asympboundrl1}, \ref{prop:mainpropasympl2} and \ref{prop:asympboundrl2}.

\begin{table}[H]
 \small
\begin{center}
  \begin{tabular}{ |cV{4}c  |c | }                                \hline
   Proj  &  $r=r_0\geq r_{+}$         & $\mathcal{I}^{+}$ \\ 
	\Xcline{1-3}{0.04cm} 
   $\psi_{0}$   & $ 4I_{0}\cdot\frac{1}{\tau^2}$    & $2I_0\cdot\frac{1}{\tau}$ \\ \hline
 $\psi_{1}$ &    $\frac{8r_0}{3}I_{1}(\theta,\varphi_{*})\cdot\frac{1}{\tau^4}$       & $\frac{2}{3}I_{1}(\theta, \varphi_{*})\cdot\frac{1}{\tau^2}$       \\ \hline
  $\psi_{\geq 2}$   &    $ \frac{4}{3}\sqrt{\frac{\pi}{5}}a^2I_{0}\cdot Y_{2,0}(\theta)\cdot\frac{1}{\tau^4}$ & $ \left[\frac{2}{15}I_{2}(\theta,\varphi_{*})-\frac{8}{9}\sqrt{\frac{\pi}{5}}a^2I_{0}\cdot Y_{2,0}(\theta)\right]\cdot\frac{1}{\tau^3}$      \\ \hline
  \end{tabular}
  \end{center}
	\vspace{-0.4cm}
 \caption{Late-time asymptotics of $ \psi_0, \psi_1, \psi_{\geq 2}$ for \underline{non-zero} Newman--Penrose charges $I_{\ell}$}
 \label{table:nonzeronp}
\end{table}
\normalsize

\begin{remark}[Logarithmic next-to-leading order terms]
\label{rm:log}
One can apply the arguments in \cite{logasymptotics} to the setting of Kerr spacetimes, with minimal modifications, to obtain the following higher-order extension of the late asymptotics of $r\psi$ at infinity in Theorem \ref{thm:intro1}:
\begin{equation*}
r\psi(\tau,\infty,\theta,\varphi_*)=-2I_0^{(1)}(1+\tau)^{-2}+8MI_0^{(1)}\log(1+\tau)(1+\tau)^{-3}+\mathcal{O}(\tau^{-3}).
\end{equation*}
See also \cite{baskinw} for additional results pertaining to polyhomogeneity in late-time expansions on asymptotically Minkowksi spacetimes.
\end{remark}

\begin{remark}[Initial data regularity]
The initial data in Theorem \ref{thm:intro1} are assumed to be compactly supported. This assumption is made purely to simplify the expressions of the time-inverted Newman--Penrose charges $I_{\ell}^{(1)}$ in terms of initial data for $\psi$; see Section \ref{sec:timinvNP}. The assumption can be weakened by merely assuming the data have sufficient, finite regularity with respect to the differentiable structure on the conformal compactification of the hypersurface $\Sigma_0$, with moreover \underline{vanishing} Newman--Penrose charges $I_{\ell}$. In fact, one can consider even rougher data by modifying the powers of $r$ appearing in the definitions of the Newman--Penrose charges, which will affect the decay rates; see also the upcoming \cite{Kerrburger} for an example of late-time asymptotics arising from rougher initial data that naturally come up in scattering problems.
\end{remark}

\begin{remark}[Blow-up in the black hole interior]
The results in Theorems \ref{thm:intro1} and \ref{thm:intro2} can be applied to justify the assumptions made along the event horizon in \cite{LukSbierski2016} in order to derive blow-up of the $H^1_{\rm loc}$ norm of $\psi$ at the inner horizon, the future causal boundary of the Kerr black hole interior for $a\neq 0$.
\end{remark}

\begin{remark}[Linearized Einstein equations]
\label{rm:teuk}
The methods that are developed in this paper to prove Theorems \ref{thm:intro1} and \ref{thm:intro2} could be extended naturally to the setting of the spin-2 Teukolsky equations, when combined with the integrated energy estimates in \cite{dhr-teukolsky-kerr, shlcosta20}. As the Teukolsky equations dictate the dynamical behaviour of perturbations of Kerr initial data in the context of the linearized Einstein equations, these methods provide a clear strategy for deriving late-time tails in the evolution of metric perturbations.
\end{remark}

\subsubsection{Newman--Penrose charges}
\label{rm:npcharges}

The \textbf{Newman--Penrose charges} $I_{\ell}$  are the values of \emph{asymptotic quantities} that are defined along the spheres foliating future null infinity and are conserved in time. See point 3 of Section \ref{sec:introendec} for an illustration of this conservation property and see Section \ref{sec:defNPconstants} for the precise definition of the Newman--Penrose charges. They were originally introduced in \cite{np2} on Minkowski background spacetimes and more generally in the context of the nonlinear Einstein equations.

In the Minkowski spacetime, the Newman--Penrose charges can be expressed as the following weighted derivatives of the radiation field:
\begin{equation*}
I_{\ell}[\psi]=(2r^2L)^{\ell}(r\psi_{\ell})(\tau,r=\infty,\theta,\varphi_*),
\end{equation*}
and they are independent of $\tau$. Here, $L$ denotes a standard outgoing null vector field.

In Schwarzschild spacetimes, or Kerr spacetimes with mass $M$ and angular parameter $a=0$, the above quantity is still conserved when $\ell=0$, but for $\ell\geq 1$, it has to be modified as follows:
\begin{equation*}
I_{\ell}[\psi]=r^2L(2r^2 w_{\ell}L(\ldots(2r^2 w_1 L(2r^2L(r\psi_{\ell})))\ldots )(\tau,r=\infty,\theta,\varphi_*),
\end{equation*}
with $w_i$ denoting polynomials in $r^{-1}$ of degree $\ell+1-i$, such that $\lim_{r\to \infty}w_i(r)=1$, which depend on the mass $M$. See also \cite{aagprice} for a more precise inductive definition of $I_{\ell}$ on Schwarzschild. Analogous modified expressions can also be shown to hold in more general spherically symmetric, asymptotical flat spacetimes.

In order to construct $I_{\ell}$ on general Kerr spacetimes, one has to deal with the \emph{coupling} of two difficulties:
\begin{enumerate}
\item A modification of the weighted vector fields $r^2L$ due to the presence of mass $M$,
\item The appearance of the spherical harmonic modes $\psi_{\ell-2}$ and $\psi_{\ell+2}$ in the expression for $I_{\ell}$, due to the presence of angular momentum $a\neq 0$ and the use of Boyer--Lindquist spheres.
\end{enumerate}

Note that difficulty 2 is already present in Minkowski when considering a foliation by oblate spheroids that are asymptotically round, rather than exactly round spheres, i.e. by taking $M=0$ in the Kerr metric, which leads to a spacetime that is isometric to Minkowski.\footnote{The eccentricity of the spheroids is then given by $\frac{|a|}{\sqrt{r^2+a^2}}$.} For example, the charge $I_0$ then takes the form:
\begin{equation*}
I_0[\psi]=\left[2r^2L(r\psi_0)-\frac{1}{2}a^2\pi_0(\sin^2\theta T(r\psi))\right](\tau,r=\infty,\theta,\varphi_*),
\end{equation*}
with $\pi_0$ and $\psi_0$ now denoting projections with respect to ellipsoids, rather than round spheres. The presence of the term $\pi_0(\sin^2\theta T(r\psi))$ above may be viewed as an effect of angular mode coupling in Boyer--Lindquist coordinates at infinity. See Section \ref{rm:decoupling} for more details.

In Kerr, $I_0$ takes the same form as in Minkowski with respect to oblate spheroidal coordinates, but $I_{\ell}$ with $\ell \geq 1$ includes modifications depending on $M$. We derive the Newman--Penrose charges $I_{\ell}$ in Kerr, with $\ell=0,1,2$ in Section \ref{sec:defNPconstants} and the same strategy may be applied inductively to find $I_{\ell}$.

The \textbf{time-inverted Newman--Penrose charges}, which are the quantities $I_{\ell}^{(1)}$ appearing in the leading-order terms in Theorem \ref{thm:intro1} can be interpreted as the Newman--Penrose charges of the function:
\begin{equation*}
\partial_{\tau}^{-1}\psi(\tau,r,\theta,\varphi):=-\int_{\tau}^{\infty} \psi(\tau',r,\theta,\varphi)\,d\tau',
\end{equation*}
the time integral of $\psi$. We moreover show that $I_{\ell}^{(1)}$ can be expressed solely in terms of initial data for $\psi$ on $\Sigma_0$, for example, for $\ell=0$:
\begin{equation*}
I_0^{(1)}=\frac{1}{4\pi}M(r_+^2+a^2)\int_{\Sigma_0\cap\mathcal{H}^+}\psi\,d\omega+\frac{1}{4\pi}M\int_{\Sigma_0} \mathbf{n}_0(\psi)\,d\mu_0,
\end{equation*}
with $\mathbf{n}_0$ the normal vector field with respect to $\Sigma_0$ and $d\mu_0$ the natural volume form with respect to the induced metric on $\Sigma_0$. See Section \ref{sec:timinvNP} for the precise integral expressions of $I_{\ell}^{(1)}$.

\subsubsection{Horizon oscillations}
\label{intro:horosc}
As an immediate corollary of the late-time asymptotics presented in Theorem \ref{thm:intro1}, we derive an oscillatory and decaying behaviour of the $\ell=1$ angular mode of $\psi$ when measured along the null generators of the event horizon.
\begin{corollary}
\label{cor:introosc}
Let $\gamma_{\theta,\varphi_{\mathcal{H}^+}}(\tau)$ denote a null generator of the event horizon $\mathcal{H}^+$ emanating from the point $(\theta,\varphi_{\mathcal{H}^+})$ on the 2-sphere $S^2_{0,r_+}=\Sigma_0\cap \mathcal{H}^+$, with time parameter $\tau$. Then:
\begin{align*}
\psi|_{\gamma_{(\theta,\varphi_*)}(\tau)}=&-8I_0^{(1)}(1+\tau)^{-3}+\mathcal{O}(\tau^{-3-\eta}),\\
r^{-1}\psi_1|_{\gamma_{(\theta,\varphi_*)}(\tau)}=&-\frac{32}{3}\sum_{m=-1}^1{I_{1m}^{(1)}}Y_{1,m}(\theta,\varphi_{\mathcal{H}^+})e^{im \upomega_{+}\tau}(1+\tau)^{-5}+\mathcal{O}(\tau^{-5-\eta}),\\
\psi_{\geq 2}|_{\gamma_{(\theta,\varphi_*)}(\tau)}=&-\frac{16}{3}\sqrt{\frac{\pi}{5}}a^2I_0^{(2)}Y_{2,0}(\theta)(1+\tau)^{-5}+\mathcal{O}(\tau^{-5-\eta}),
\end{align*}
where $\upomega_+=\frac{a}{r_+^2+a^2}$ is the \emph{angular velocity of the black hole} and $I_{\ell m}^{(1)}$ denotes the projection of $I_{\ell}^{(1)}$ to the $m$-th azimuthal mode. \end{corollary}

We therefore see that the leading-order behaviour of the $\ell=1$ mode in time \emph{encodes} the angular velocity of the black hole in the form of an oscillatory factor $e^{i m \upomega_+ \tau}$. Corollary \ref{cor:introosc} provides the first rigorous derivation of \emph{horizon oscillations} in Kerr spacetimes, which were originally suggested by Barack--Ori in \cite{LBAO99}. 

The presence of horizon oscillations illustrates the different roles played by the Killing vector fields $T$ and $K$. Here, $T$ denotes a suitably normalized Killing vector field that is timelike for large values of $r$, but spacelike near the event horizon, whereas $K$ denotes the Killing vector field that is tangential to the null generators of the horizon and is timelike close to the horizon, but spacelike far away from the horizon. Hence, $T$ provides a ``natural'' choice of time direction far away from the horizon and $K$ provides a natural choice of time direction close to the horizon.  Note that in the coordinate chart $(\tau,r,\theta,\varphi_*)$ that we have chosen, we can simply express $T=\partial_{\tau}$. The two Killing vector fields can be related as follows:
\begin{equation*}
K=T+\upomega_+ \Phi,
\end{equation*}
where $\Phi$ denotes the Killing vector field that generates the axisymmetry of the Kerr spacetime, and $K=T$ when $a=0$.

Horizon oscillations follow from the fact that when $a\neq 0$ the time derivative \textbf{$T\psi$ decays one power faster than $\psi$ to leading order, whereas the time derivative $K\psi$ \underline{does not}}. In fact, using the above relation between $T$ and $K$, we can see that the non-axisymmetric part of  $K\psi$ must agree precisely with $\upomega_+ \Phi \psi$ to leading order in time, which explains the oscillatory behaviour in the $\ell=1$ mode.

We note that it has been suggested that horizon oscillations will similarly appear in the leading-order late-time asymptotics of general solutions to non-zero integer spin Teukolsky equations, see \cite{barack99grav} and Remark \ref{rm:teuk}. They are therefore expected to play a leading-order role in the general late-time horizon dynamics in the setting of linearized gravity and the Maxwell equations on Kerr spacetimes.

\subsubsection{Angular mode coupling}
\label{rm:decoupling}
The spherical harmonic modes in Theorems \ref{thm:intro1} and \ref{thm:intro2} are coupled in their evolution, in contrast with the Schwarzschild case $(a=0)$ where they evolve independently.

Indeed, when defined with respect to the metric on the unit round sphere on Boyer--Lindquist spheres, the angular mode coupling takes the following form:
\begin{equation}
\label{eq:introcouplingmodes1}
\rho^2\square_{g_{M,a}}\psi_{\ell}= a^2 {T^2}(\pi_{\ell}(\sin^2\theta\psi)-\sin^2\theta \pi_{\ell}\psi),
\end{equation}
where we can schematically write the right-hand side as:
\begin{equation}
\label{eq:introcouplingmodes2}
a^2 {T^2}(\pi_{\ell}(\sin^2\theta\psi)-\sin^2\theta \pi_{\ell}\psi)\approx  a^2( c_- T^2\psi_{\ell-2}+c_0T^2\psi_{\ell}+c_+ T^2\psi_{\ell+2}),
\end{equation}
with $c_0,c_-,c_+$ constants, see Lemma \ref{eq:lprojsin} for a precise expression.

If the different angular modes were uncoupled, we would be able to show that modes supported on \emph{higher} $\ell$ decay \emph{faster}. This is precisely the case when $a=0$, see \cite{aagprice}. In view of \eqref{eq:introcouplingmodes1} and \eqref{eq:introcouplingmodes2}, however, we can see heuristically that the decay of $\psi_{\ell}$ is limited by the decay of $T^2\psi_{\ell-2}$. We would therefore have the following schematic relation for the late-time asymptotics of $\psi_{\ell}$:
\begin{align*}
\psi_{\ell}\sim T^2\psi_{\ell-2}\sim\ldots \sim T^{\ell}\psi_0\quad \textnormal{for even $\ell\geq 2$},\\
\psi_{\ell}\sim T^2\psi_{\ell-2}\sim\ldots \sim T^{\ell-1}\psi_1\quad \textnormal{for odd $\ell\geq 3$}.
\end{align*}
So while the decay rate of $\psi_{\ell}$ is limited by the lower angular modes $T^{\ell}\psi_0$ and $T^{\ell-1}\psi_1$, we can use that $T$-derivatives decay faster, as mentioned in Section \ref{intro:horosc}, together with the above heuristics to predict that higher angular frequencies $\psi_{\ell}$ \emph{do} in fact decay faster than lower angular frequencies, but they decay slower compared to the uncoupled setting. 

In this paper, we give a rigorous validation of the above heuristics in the case of the late-time asymptotics of $\ell=2$. We also discuss the general $\ell$ case in Section \ref{intro:pricelaw}.

We conclude that while there is no \emph{a priori} natural choice of 2-spheres to foliate $\mathcal{M}_{M,a}$, Theorems \ref{thm:intro1} and \ref{thm:intro2} imply that the choice of Boyer--Lindquist spheres is natural from the point of view of late-time asymptotics, as it has the favourable property that it results in a ``late-time decoupling'' of the corresponding spherical harmonic modes $\psi_{\ell}$, i.e. the terms in the expansion of $\psi$ in $\tau^{-1}$ at any fixed $r$ up to order $N\in \N$ will be supported only on spherical harmonic modes with angular frequencies $\ell\leq L(N)<\infty$, where we expect $L(N)=\max\{N-3,0\}$.

\subsection{Previous work}
\label{intro:prevwork}
In this section we give an overview of some relevant previous results in the literature on late-time tails and sharp decay estimates for waves on black hole spacetimes.

\subsubsection{Physics literature}

The first discussion on inverse polynomial late-time tails for wave equations on black hole background appeared in a paper of Price \cite{Price1972}, who provided a heuristic argument for the presence of a $\tau^{-2\ell-3}$ tail in Schwarzschild. Late-time tails on spherically symmetric backgrounds have since been discussed frequently in the physics literature, using both heuristic and numerical arguments; see the introduction of \cite{paper4} for an overview.

The first numerical discussion on late-time tails in the context of subextremal Kerr spacetimes appeared in \cite{krivan97}, which was followed by heuristic analyses in \cite{hod99} and \cite{LBAO99}. The latter work suggested the power laws for generic, compactly supported data that are stated in Section \ref{intro:pricelaw1}. Subsequent numerical work on late-time tails focused on the case of initial data supported on a single harmonic mode and has produced conflicting numerology, caused by the difficulty of characterising the evolutionary coupling and excitation of spherical harmonic modes (numerically). We refer to \cite{zengi14,burko14} for the latest  numerical results and corrected predictions of late-time tails and also to the references therein for a more complete history of the problem. As discussed in Section \ref{intro:pricelaw2}, the methods developed in the present paper allow for a final resolution of the differences in the numerology of late-time tails suggested in the literature, by relating the numerology to the vanishing and non-vanishing of time-inverted Newman--Penrose charges, which we expect to align with the numerical results of \cite{zengi14,burko14}.

\subsubsection{Mathematics literature}
The first mathematically rigorous derivation of the leading-order late-time behaviour of waves on black hole background was obtained in \cite{paper1, paper2} for a variety of spherically symmetric backgrounds. The presence of $\ell$-dependent late-time tails on Schwarzschild and more general subextremal Reissner--Nordstr\"om backgrounds is derived in \cite{aagprice}. These results appeal to important ideas and results that appeared in previous literature on (integrated) energy decay estimates, most notably \cite{redshift,newmethod}, see also \cite{lecturesMD}. The methods of \cite{paper1, paper2} have also been adapted to the setting of the Dirac equation on Schwarzschild backgrounds \cite{ma2}.

A proof of the existence of $\tau^{-3}$ tails on Kerr spacetimes was recently established in \cite{hintzprice}, using methods based in Fourier space.

We refer also to \cite{tataru3, metal,hintzprice} for sharp decay estimates for $\psi$ in a linear setting and \cite{MDIR05} in a spherically symmetric nonlinear setting. See the also the results in \cite{dssprice,other1}, investigating the $\ell$-dependence in decay estimates and \cite{ma1} for some sharp decay estimates in the context of the Maxwell equations and a characterisation of the Newman--Penrose charges on Schwarzschild in the Maxwell setting.

Finally, the results in the present paper appeal to the energy boundedness and integrated energy decay estimates that have already been established in Kerr spacetimes in \cite{part3}, which capture the key geometric obstructions to wave decay on Kerr caused by the presence of trapped null geodesics and an ergoregion. We refer to the introduction of \cite{part3} for a complete overview of the history of integrated energy decay estimates on Kerr spacetimes.

\subsubsection{Extremal black hole spacetimes}
The leading-order behaviour of waves on extremal black hole backgrounds features additional interesting geometric phenomena. The first mathematically rigorous proof of the existence and form of late-time tails on extremal Reissner--Nordstr\"om spacetimes was obtained in \cite{paper4} and follows previous numerical and heuristic results, see the introduction of \cite{paper4} for a comprehensive overview. The late-time tails in this setting involve an additional conserved charge that occurs along the event horizon of extremal black hole spacetimes and was first discovered in \cite{aretakis1}. It is connected to the absence of the redshift effect along extremal black hole horizons and may be related to the Newman--Penrose charges, see \cite{hm2012,bizon2012}. See also \cite{extremal-prl} and the subsequent numerical work \cite{burko19,burko21} for discussions on the \emph{signature} of this conserved horizon charge at infinity.

While decay estimates have also been obtained in the setting of rotating extremal black hole spacetimes, for the special case of axisymmetric solutions \cite{aretakis3}, the late-time properties of non-axisymmetric linear waves on extremal Kerr ($|a|=M$) remain open. In this setting, only mode stability has been established in a mathematically rigorous setting in \cite{costa20}, but heuristics have also been provided on the rates of late-time tails for fixed azimuthal modes in \cite{zimmerman1}.

\subsection{Analogues of Price's law in Kerr}
\label{intro:pricelaw}

\subsubsection{Generic initial data}
\label{intro:pricelaw1}
One can apply the methods developed in this paper to extend the results of Theorem \ref{thm:intro1} (and Theorem \ref{thm:intro2}) beyond $\ell=2$ and obtain the following late-time asymptotics for higher spherical harmonic modes:
\begin{align*}
\psi_{\ell=2k}(\tau,r_0,\theta,\varphi_*)=&\:c_{\ell,0} a^kI_0^{(1)}Y_{\ell, 0}(\theta)(1+\tau)^{-\ell-3}+\mathcal{O}_{r_0}(\tau^{-\ell-3-\eta}),\\
r^{-1}\psi_{\ell=2k+1}(\tau,r_0,\theta,\varphi_*)=&\:\sum_{m=-1}^1c_{\ell,m} a^kI_{1m}^{(1)}Y_{\ell, m}(\theta,\varphi_*)(1+\tau)^{-\ell-4}+\mathcal{O}_{r_0}(\tau^{-\ell-4-\eta}),\\
r\psi_{\ell=2k}(\tau,\infty,\theta,\varphi_*)=&\:\sum_{k=0}^{\frac{\ell}{2}}\sum_{m=-2k}^{2k} c_{\ell,k,m} a^{\ell-2k}I_{2k\,m}^{(1)}Y_{\ell, m}(\theta,\varphi_*)(1+\tau)^{-\ell-2}+\mathcal{O}(\tau^{-\ell-2-\eta}),\\
r\psi_{\ell=2k+1}(\tau,\infty,\theta,\varphi_*)=&\:\sum_{k=0}^{\frac{\ell-1}{2}}\sum_{m=-(2k+1)}^{2k+1} c_{\ell,k,m} a^{\ell-(2k+1)}I_{2k+1\,m}^{(1)}Y_{\ell, m}(\theta,\varphi_*)(1+\tau)^{-\ell-2}+\mathcal{O}(\tau^{-\ell-2-\eta}),
\end{align*}
where $c_{\ell,m}$ and $c_{\ell,k,m}$ are dimensionless constants.  We do not provide a proof of the above expressions in the present paper, but we refer to the heuristics in Section \ref{rm:decoupling} and leave them for the interested reader to verify, equipped with the techniques and estimates that are developed in the paper.

The above behaviour suggests in particular the following version of ``Price's law'' at fixed radius for $a\neq 0$: 
\begin{align*}
\psi_{\ell}\sim & \tau^{-\ell-3}\quad  \textnormal{when $\ell$ is odd},\\
\psi_{\ell}\sim & \tau^{-\ell-4}\quad \textnormal{when $\ell$ is even},
\end{align*}
 which differs from Price's law when $a=0$: $\psi_{\ell}\sim \tau^{-2\ell-3}$, see \cite{aagprice}. We note that the above decay rates are consistent with the heuristics in \cite{LBAO99}.

\subsubsection{Initial data restricted to higher angular frequencies and mode excitation}
\label{intro:pricelaw2}
Rather than considering generic initial data supported on all angular frequencies, one can further investigate the mode-coupling mechanism by considering restricted data, supported on angular frequencies $ \geq \ell'$ and then investigate the late-time tails of $\psi_{\ell}$ for each $\ell\in \N_0$. This may be viewed as a study of ``mode excitation'' to leading order in time. 

For $|\ell-\ell'|>2$, it follows easily that $I_{\ell}^{(1)}=0$. In this case, one would have to apply the time inversion procedure $k$ times, for suitable values of $k$, to obtain a non-vanishing \emph{higher-order} time-inverted charge $I_{\ell}^{(k)}[\psi]:=I_{\ell}[\partial_{\tau}^{-k}\psi]$, which would then be relevant for the late-time behaviour. It is straightforward to show via the equation for the time inverse \eqref{eq:timeinvimportant} that if the (smooth and compactly supported) initial data are supported on the angular frequency $\ell'$, with $\ell'\notin \{0,1\}$, the first (generically) non-zero time-inverted charge for $\ell\neq \ell'$ would be 
\begin{equation*}
I_{
\ell}^{(|\ell'-\ell|-1)}[\psi]=I_{\ell}[\partial_{\tau}^{-|\ell'-\ell|+1}\psi],
\end{equation*}
 if $\ell$ and $\ell'$ are both either even or odd (note that $I_{\ell}^{(k)}$ vanishes for all $k$ if $\ell$ and $\ell'$ have opposite parity). Moreover, $I_{\ell'}^{(1)}=I_{\ell'}[\partial_{\tau}^{-1}\psi]\neq 0$, generically. See also Section 9.2 of \cite{paper2} for a discussion on higher-order time inversions in the spherically symmetric context. 

In light of the non-vanishing of the above higher-order time inverted charges, we expect that, analogously to what is described in Section \ref{intro:pricelaw}, one could apply the methods in the present paper to obtain schematically the following late-time tails for initial data supported on angular frequency $\ell'$: let $j=0$ when $\ell'$ is even and $j=1$ when $\ell'$ is odd, then \begin{align}
\label{eq:pricelawharmonicdata1}
\psi_{\ell}|_{r=r_0}\sim & \sum_{\substack{0\leq k\leq \lfloor \frac{\max\{\ell,\ell'-2\}}{2}\rfloor }}I_{2k+j}^{(|\ell'-2k|-1)}\tau^{-\ell-\ell'-1}\quad \textnormal{when}\,\, \ell'\notin \{0,1\},\\
\label{eq:pricelawharmonicdata2}
\psi_{\ell}|_{r=r_0}\sim & I_{\ell'}^{(1)}\tau^{-\ell-\ell'-3}\quad \textnormal{when}\,\,\ell'\in \{0,1\},\\
\label{eq:pricelawharmonicdata3}
r\psi_{\ell}|_{\mathcal{I}^+}\sim & I_{\ell}^{(\ell'-\ell-1)}\tau^{-\ell'} \quad \textnormal{when}\,\,\ell\leq \ell'-2,\\
\label{eq:pricelawharmonicdata4}
 r\psi_{\ell}|_{\mathcal{I}^+}\sim & \sum_{-1\leq k\leq \min\{1,\frac{\ell-\ell'}{2}\}}^1I_{\ell'+2k}^{(1)}\tau^{-\ell-2} \quad \textnormal{when}\,\,\ell\geq \ell',
\end{align}
with $\ell'$ and $\ell$ both either even or odd. Note in particular that \eqref{eq:pricelawharmonicdata2}--\eqref{eq:pricelawharmonicdata4} follow directly from Theorem \ref{thm:intro1} and the expressions in Section \ref{intro:pricelaw1}, whereas \eqref{eq:pricelawharmonicdata1} requires a minor further extension of the methods in the present paper.

The above power laws agree with an extrapolation of the numerical results of \cite{zengi14}, but deviate from earlier suggestions in the literature. See also Section \ref{intro:prevwork}.

\subsection{Main ideas and techniques}
\label{intro:ideas}

In this section, we introduce the main new ideas and techniques involved in proving Theorems \ref{thm:intro1} and \ref{thm:intro2}.

We make use of the following additional notation in this section:
\begin{itemize}
\item We cover $\Sigma_{0}$ with the coordinates $(\uprho,\theta,\varphi_*)$, where $\uprho=r|_{\Sigma_0}$ and employ the chart $(\tau,\uprho,\theta,\varphi_*)$ to cover the spacetime region $\bigcup_{\tau\in[0,\infty)} \Sigma_{\tau}$. Furthermore, we denote the standard volume form on the unit round 2-sphere $\s^2$ by $d\omega$ and the natural volume form corresponding to the induced metric on $\Sigma_{\tau}$ by $d\mu_{\tau}$.
\item We use the notation $T=\partial_{\tau}$ and $X=\partial_{\uprho}$, and we use $L$ and $\underline{L}$ to denote the principal outgoing and ingoing null vector fields in Kerr. We also use $\snabla_{\s^2}$ to denote the covariant derivative on $\s^2$.
\item We the following rescaling of $\psi$:
\begin{equation*}
\phi=\sqrt{r^2+a^2}\psi,
\end{equation*}
where the limit $\phi(\tau,\infty,\theta,\varphi_*)$ is known as the \emph{Friedlander radiation field}. This is a natural quantity to consider for scattering problems; see for example \cite{linearscattering}.
\item We use $\mathbb{T}$ to denote the stress energy momentum tensor corresponding to \eqref{eq:introwaveeq}:
\begin{equation*}
\mathbb{T}(V,W)=V\psi W\psi-\frac{1}{2}g_{M,a}(V,W)(g^{-1}_{M,a})^{\alpha \beta}\partial_{\alpha}\psi\partial_{\beta}\psi.
\end{equation*}
\end{itemize}

We refer to Section \ref{sec:geom} for a more precise introduction to all the above notation.

\subsubsection{Step 1: Energy decay via new hierarchies of $r^p$-weighted estimates}
\label{sec:introendec}
The first step towards determining the leading-order behaviour of $\psi$ is to obtain time decay estimates for energy norms along the leaves of a suitable spacetime foliation.

We consider a foliation by asymptotically hyperboloidal hypersurfaces $\Sigma_{\tau}$ intersecting the event horizon, see Figure \ref{fig:foliationsintro}. We obtain energy decay by using extensions of the original Dafermos--Rodnianski $r^p$-weighted energy method \cite{newmethod}. See \cite{linearscattering} for an application of the Dafermos--Rodnianski $r^p$-weighted energy method in Kerr and \cite{moschidis1} for the case of more general spacetimes.

\begin{enumerate}
\item (\textbf{Dafermos--Rodnianski hierarchy})
Consider first the Dafermos--Rodnianski (D--R) hierarchy, which takes the schematic form
\begin{equation}
\label{eq:introdafrod}
\int_{\tau_1}^{\tau_2}\int_{\Sigma_{\tau}} r^{p-1}(L\phi)^2\,d\omega d\uprho d\tau\lesssim  \int_{\Sigma_{\tau_1}} r^{p}(L\phi)^2\,d\omega d\uprho+\ldots,
\end{equation}
for $0<p\leq 2$, where we have omitted additional angular and $T$-derivatives appearing on the left-hand side and $\ldots$ denotes terms arising from an application of integrated energy (Morawetz) and energy boundedness estimates, which we ignore in the current discussion. 

Taking $p=1$ and applying the mean-value theorem in $\tau$ gives:
\begin{equation*}
\int_{\Sigma_{\tau}} (L\phi)^2\,d\omega d\uprho \lesssim (1+\tau)^{-1}\int_{\Sigma_{\tau_0}} r(L\phi)^2\,d\omega d\uprho+\ldots.
\end{equation*}
Taking $p=2$ and applying the mean value theorem once again (along an appropriate sequence of times) allows us to estimate the right-hand side further and obtain the following uniform energy decay estimates:
\begin{equation*}
\int_{\Sigma_{\tau}} (L\phi)^2\,d\omega d\uprho \lesssim (1+\tau)^{-2}\int_{\Sigma_{\tau_0}} r^2(L\phi)^2\,d\omega d\uprho+\ldots.
\end{equation*}
\item (\textbf{Extending the D--R hierarchy via $r^2L$ derivatives})
The D--R hierarhcy is limited by the restriction $p\leq 2$. To be able to get a decay rate of $(1+\tau)^{-3}$ and higher, we need control over the left-hand side of \eqref{eq:introdafrod} for $p\geq 3$. By an application of a standard Hardy inequality, we can show that:
\begin{equation*}
\begin{split}
\int_{\tau_1}^{\tau_2}\int_{\Sigma_{\tau}} r^{2}(L\phi)^2\,d\omega d\uprho d\tau= &\:\int_{\tau_1}^{\tau_2}\int_{\Sigma_{\tau}} r^{-2}(r^2L\phi)^2\,d\omega d\uprho d\tau\\
\lesssim &\: \int_{\tau_1}^{\tau_2}\int_{\Sigma_{\tau}}  (L(r^2L\phi))^2\,d\omega d\uprho d\tau+\ldots.
\end{split}
\end{equation*}
In the case where $\psi$ is supported on $\ell\geq 1$, we can further estimate: for $0\leq p\leq 2$
\begin{equation*}
\int_{\tau_1}^{\tau_2}\int_{\Sigma_{\tau}} r^{p-1}(L(r^2L\phi_{\geq 1}))^2\,d\omega d\uprho d\tau\lesssim \int_{\Sigma_{\tau_1}} r^{p}(L(r^2L\phi_{\geq 1})^2\,d\omega d\uprho+\boxed{a^2\int_{\tau_1}^{\tau_2}\int_{\Sigma_{\tau}} r^{p-3}(r^2LT^2\phi)^2\,d\omega d\uprho d\tau}\ldots.
\end{equation*}
In fact, more generally, for all $n\leq \ell$ and $0\leq p\leq 2$:
\begin{equation}
\label{eq:introcommhier}
\int_{\tau_1}^{\tau_2}\int_{\Sigma_{\tau}} r^{p-1}(L(r^2L)^{n}\phi_{\geq \ell})^2\,d\omega d\uprho d\tau\lesssim \int_{\Sigma_{\tau_1}} r^{p}(L(r^2L)^n\phi_{\geq \ell})^2\,d\omega d\uprho+\boxed{a^2\int_{\tau_1}^{\tau_2}\int_{\Sigma_{\tau}} r^{p-3}((r^2L)^{n}T^2\phi_{\geq \ell-2})^2\,d\omega d\uprho d\tau}\ldots.
\end{equation}

When $a=0$, we can therefore continue applying the mean-value theorem together with a Hardy inequality and apply \eqref{eq:introcommhier} in order to obtain (at least) $(1+\tau)^{-2-2\ell}$ energy decay for $\phi_{\geq \ell}$. This is done in \cite{aagprice}.

However, when $a\neq 0$, we have to address the boxed term in \eqref{eq:introcommhier}, which reflects the fact that we cannot treat $\phi_{\geq \ell}$ independently from $\phi_{\leq \ell-1}$ (recall that the spherical harmonic modes are coupled).

For example, in the case $\ell=2$, we can rewrite the boxed term in \eqref{eq:introcommhier} as follows:
\begin{equation*}
a^2\int_{\tau_1}^{\tau_2}\int_{\Sigma_{\tau}} r^{p-3}((r^2L)^{2}T^2\phi)^2\,d\omega d\uprho d\tau\lesssim a^2\sum_{k=0}^1\int_{\tau_1}^{\tau_2}\int_{\Sigma_{\tau}} r^{p+3}(L(rL)^kT^2\phi)^2\,d\omega d\uprho.
\end{equation*}

By applying the wave equation \eqref{eq:introwaveeq} we can exchange $T$ derivatives for $L$ and angular derivatives to estimate further:
\begin{equation*}
a^2\sum_{k=0}^1\int_{\tau_1}^{\tau_2}\int_{\Sigma_{\tau}} r^{p+3}(L(rL)^kT^2\phi)^2\,d\omega d\uprho \lesssim a^2\sum_{0\leq k_1+k_2\leq 3}\int_{\tau_1}^{\tau_2}\int_{\Sigma_{\tau}} r^{p-1}|\snabla_{\s^2}^{k_1}L(rL)^{k_2}\phi|^2\,d\omega d\uprho.
\end{equation*}
When then show that the term above can be controlled after showing the estimate \eqref{eq:introdafrod} also holds when $\phi$ is replaced with $\snabla_{\s^2}^{k_1}(rL)^{k_2}\phi$. That is to say, \emph{the (extended) $r^p$-weighted hierarchies remain valid after commuting arbitrarily many times with $rL$ and $\snabla_{\s^2}$}. This favourable commutation property also plays an important role in \cite{moschidis1,volker1}. We stress moreover that this property contrasts strongly with commutation with $r^2L$, which, as outlined above, is more delicate as it is only valid after projecting to suitably higher angular frequencies.

We apply the above strategy for $\ell=1,2,3$ to obtain $(1+\tau)^{-2-2\ell}$ energy decay for $\phi_{\geq \ell}$ and moreover\\ $(1+\tau)^{-2-2\ell-2k}$ energy decay for $T^k\phi$ by exploiting the above exchange of $T$-derivatives for $rL$-derivatives and angular derivatives to extend the hierarchy even further.

See Section \ref{sec:rpest} for a precise discussion of the various $r^p$-weighted hierarchies and Section \ref{sec:edecay} for the arguments that use these hierarchies to obtain energy decay.

\item (\textbf{Maximal length hierarchies via conserved charges at infinity})
In order to obtain sharp energy decay rates, we complement the extended hierarchies above with an additional extension that relies on the existence of conservation charges at future null infinity.

We consider first $\phi_0$. In order to go beyond $p=2$ in the Dafermos--Rodnianski hierarchy \eqref{eq:introdafrod}, we replace $L\phi_0$ in \eqref{eq:introdafrod} with $P_0$, which is defined as follows:
\begin{equation*}
P_0:=L\phi_0-\frac{1}{4}\frac{\Delta}{(r^2+a^2)^2}a^2\pi_0(\sin^2\theta T\phi).
\end{equation*}
Since $P_0$ satisfies schematically
\begin{equation*}
\underline{L}P_0=O(r^{-3})(\phi_0+a^2rL\phi_0+a^2T\phi_0+a^2r L\phi_2)
\end{equation*}
it follows that $\lim_{r\to \infty}r^2{P_0}|_{\Sigma_{\tau}}$ is conserved in $\tau$ and an analogue of \eqref{eq:introdafrod} holds for the extended range $0<p<3$, with $L\phi_0$ replaced by $P_0$. We use the constant $I_0$ to denote these limits and refer to this quantity as the Newman--Penrose charge for $\ell=0$; see Section \ref{rm:npcharges}.

We can similarly replace $(r^2L)\phi_1$ and $(r^2L)^2\phi_2$ by the quantities $P_1$ and $P_2$ (defined precisely in Section \ref{sec:defNPconstants}) to extend \eqref{eq:introcommhier}  with $n=\ell\in\{0,1,2\}$, to $p<3$. See also  \cite{paper1, paper2, aagprice} for a similar extension of the hierarchy of $r^p$-weighted energy estimates in the $a=0$ case. In this case, the corresponding limits $\lim_{r\to \infty}r^2{P_{\ell}}|_{\Sigma_{\tau}}$, which are conserved in $\tau$, are functions of $(\theta,\varphi_*)$ and are denoted by $I_{\ell}(\theta,\varphi_*)$. We refer to them as the Newman--Penrose charges for $\ell=1$ and $\ell=2$.

We therefore conclude the energy decay rates $(1+\tau)^{-3-2\ell-2k}$ for $T^k\phi_{\ell}$, with $\ell=0,1,2$, which we can combine with the already established energy decay rate $(1+\tau)^{-9-2k}$ for $\phi_{\geq 3}$.

The application of the quantities $P_{\ell}$ to $r^p$-weighted energy estimates are derived in Section \ref{sec:rpestPi}.
\end{enumerate}

By a straightforward application of the fundamental theorem of calculus in $\uprho$ (see Appendix \ref{sec:apppoint}) together with standard Sobolev estimates on $\s^2$, we can convert the energy decay estimates above into the following pointwise estimates:
\begin{align*}
|rT^k\psi_{0}|\lesssim &\: (1+\tau)^{-1+\epsilon},\\
|rT^k\psi_{1}|\lesssim &\: (1+\tau)^{-2-k+\epsilon},\\
|rT^k\psi_{2}|\lesssim  &\: (1+\tau)^{-3-k+\epsilon},\\
|rT^k\psi_{\geq 3}|\lesssim  &\: (1+\tau)^{-\frac{7}{2}-k+\epsilon}.
\end{align*}

\textbf{Additional difficulties:} In the above outline, we have suppressed the role of the geometric phenomena of \emph{redshift}, \emph{trapped null geodesics} and \emph{ergoregion} towards establishing energy boundedness and decay in Kerr spacetimes. While these play an important role, their effect on integrated energy estimates has already been dealt with in \cite{part3} and we appeal to these estimates in the present paper. We refer to the introduction of \cite{part3} for a comprehensive discussion.

 Note however that the estimates in \cite{part3} do not distinguish between different spherical harmonic modes $\psi_{\ell}$, so we need to modify them slightly to obtain more refined integrated estimates for $\psi_{\geq \ell}$ to be able to carry out the arguments described above. We give an overview of the necessary energy boundedness and integrated energy decay estimates, including the refined integrated estimates for $\psi_{\geq \ell}$, in Section \ref{sec:prelimen}.

\subsubsection{Step 2: Improved pointwise decay via a hierarchy of elliptic estimates}
\label{sec:introelliptic}
We subsequently improve the pointwise decay rates for $r\psi$ above further by considering the $r$-rescalings: $\psi_0,\psi_2,\psi_{\geq 3}$ and  and $r^{-1}\psi_1$, and we apply in addition a hierarchy of elliptic estimates.

We first rewrite the wave equation \eqref{eq:introwaveeq} as follows;
\begin{equation*}
\mathcal{L}\psi=F[T\psi],
\end{equation*}
with $\mathcal{L}$ defined as the differential operator:
\begin{equation*}
\mathcal{L}\psi=X(\Delta X\psi)+ 2a X\Phi \psi+\slashed{\Delta}_{\s^2}\psi
\end{equation*}
and $F[T\psi]$ containing only $T$-derivatives of $\psi$, see \eqref{eq:inhomo} for the precise expression. In this section, we treat $F[T\psi]$ as an inhomogeneity. 

The operator $\mathcal{L}$ is elliptic when $T$ is a timelike vector field. This is the case outside the ergoregion, i.e. when  $r^2-2Mr+a^2\cos^2\theta> 0$. The loss of ellipticity inside the ergoregion prevents the use of a standard elliptic estimate to control $\psi$ in terms of $F$, which would involve integrating by parts the equation $(\mathcal{L}\psi)^2=F^2$ on $\Sigma_{\tau}$.\footnote{In the $a=0$ case, where ellipticity only fails at the event horizon, this strategy does work and results in elliptic estimates that degenerate at the event horizon. See for example \cite{paper1,aagprice}.}

We establish nevertheless for $k\in (\frac{1}{2},\ell+\frac{1}{2})$:
\begin{equation}
\label{eq:introellhier}
\sum_{n=0}^{\ell} \int_{\Sigma_{\tau}}r^{-2k}\left[ (X(rX)^n \psi_{\ell})^2+r^{-2} |\snabla_{\s^2}(rX)^n \psi_{\ell}|^2\right]\,r^2d\omega d\uprho\lesssim \sum_{n=0}^{\ell}\int_{\Sigma_{\tau}} r^{-2k}( (rX)^{n}({\pi}_{\ell}F))^2\,d\omega d\uprho
\end{equation}
by applying the following estimates:
\begin{enumerate}[A.)]
\item A ``spacelike redshift multiplier estimate", i.e. when $\ell=0$, we integrate by parts on $\Sigma_{\tau}\cap\{r\leq R_0\}$ the product:
\begin{equation*}
X\psi_0\cdot \mathcal{L}\psi_0
\end{equation*}
and more generally,
\begin{equation*}
X^{\ell+1}\psi_{\ell}\cdot X^{\ell}(\mathcal{L}\psi_{\ell}).
\end{equation*}
In contrast with the standard redshift estimates \cite{redshift,lecturesMD}, which involve integrating by parts in spacetime with timelike vector field multipliers, and which are only valid in a neighbourhood of the event horizon, the above spacelike redshift estimates hold with arbitrarily large radial coordinate $r$.
\item An elliptic estimate in the region $r\geq R_0$, with $R_0\gg M$, which follows from integrating
\begin{equation*}
r^{-2k}(X^{\ell}\mathcal{L}\psi)^2=r^{-2k}(X^{\ell}F)^2
\end{equation*}
 where we apply A.) to estimate the boundary terms on $r=R_0$. The cross terms that arise when expanding the square on the left-hand side above are estimated via a Poincar\'e inequality and the range of allowed values for $k$ depends therefore on the angular frequency $\ell$.
 \item  Another spacelike redshift multiplier estimate to control lower-order derivative terms, by integrating
 \begin{equation*}
X^{n+1}\psi_{\ell}\cdot X^{n}(\mathcal{L}\psi_{\ell}),
\end{equation*}
with $n\leq \ell-1$ and applying an induction argument, using A.) and B.).
\end{enumerate}

Filling in the expression for $F[T\psi]$, we therefore obtain the following \emph{elliptic hierarchy} of energy estimates:
\begin{equation}
\label{eq:introellhier}
\begin{split}
\sum_{n=0}^{\ell} \int_{\Sigma_{\tau}}r^{-2k}\mathbb{T}[(rX)^n\psi_{\ell}](T,\mathbf{n}_{\tau})\,d\mu_{\tau}\lesssim &\:\sum_{n=0}^{\ell} \int_{\Sigma_{\tau}}r^{-2k+2}\mathbb{T}[(rX)^nT\psi_{\ell}](T,\mathbf{n}_{\tau})\,d\mu_{\tau}\\
&+a^2\int_{\Sigma_{\tau}}r^{-2k}\left\{\mathbb{T}[(rX)^nT^2\psi_{\ell-2}](T,\mathbf{n}_{\tau})+\mathbb{T}[(rX)^nT\psi_{\ell+2}](T,\mathbf{n}_{\tau})\right\}\,d\mu_{\tau},
\end{split}
\end{equation}
for $k\in (\frac{1}{2},\ell+\frac{1}{2})$.

The elliptic hierarchy allows us to exchange negative $r$ weights for additional $T$ derivatives, which decay faster in time. One may compare this with extra time decay obtained via an extended hierarchy $r^p$-weighted energy estimates, as described in Section \ref{sec:introendec}, where additional time decay follows instead by controlling an additional \underline{time integral} on the left-hand side and applying the mean-value theorem.

In the $a=0$ case, we can apply \eqref{eq:introellhier} starting from $k=\frac{1}{2}-\delta$ to $k=\ell+\frac{1}{2}-\delta$ to obtain a hierarchy of estimates of length $\ell$ and therefore estimate:
\begin{equation*}
\begin{split}
\sum_{n=0}^{\ell} \int_{\Sigma_{\tau}}r^{-2\ell-1-\delta}\mathbb{T}[(rX)^n\psi_{\ell}](T,\mathbf{n}_{\tau})\,d\mu_{\tau}\lesssim &\:\sum_{n=0}^{\ell} \int_{\Sigma_{\tau}}r^{1-\delta}\mathbb{T}[(rX)^nT^{\ell}\psi_{\ell}](T,\mathbf{n}_{\tau})\,d\mu_{\tau},
\end{split}
\end{equation*}
and obtain the sharp decay rate of the weighted energy on the left-hand side, together with the sharp decay rate of $r^{-\ell}\psi_{\ell}$.

When $a\neq 0$, the above strategy can be repeated for $\ell=0$ and $\ell=1$, but it fails for $\ell=2$, due to the presence of the terms involving $\psi_0$ on the right-hand side of \eqref{eq:introellhier} when $\ell=0$, which limits the decay rate and hence the length of the elliptic hierarchy. We refer to Section \ref{sec:elliptic} for a precise discussion of the above arguments.

When $a\neq 0$, we therefore conclude pointwise decay estimates with the following weights:
\begin{align*}
|T^k\psi_{0}|\lesssim &\: (1+\tau)^{-2+\epsilon},\\
|r^{-1}T^k\psi_{1}|\lesssim &\: (1+\tau)^{-4-k+\epsilon},\\
|T^k\psi_{2}|\lesssim  &\: (1+\tau)^{-4-k+\epsilon},\\
|T^k\psi_{\geq 3}|\lesssim  &\: (1+\tau)^{-\frac{9}{2}-k+\epsilon}.
\end{align*}
Note in particular that the decay rate for the $\ell=1$ and $\ell=2$ mode is the same, in contrast with the $a=0$ case.

\subsubsection{Step 3: Elliptic theory of time inversion}
The pointwise decay rates stated in Sections \ref{sec:introendec} and \ref{sec:introelliptic} are (almost) sharp in the case of initial data for which the conserved Newman--Penrose charges $I_{\ell}$ are non-zero, with $\ell=0,1,2$. For initial data that decay faster as $r\to \infty$, such that $I_{\ell}=0$, it is possible to gain one power in the decay rate.

The idea is the following: we can show that $\psi$ can be thought of as being a time derivative of another solution to \eqref{eq:introwaveeq}, i.e.\
\begin{equation*}
T\widetilde{\psi}=\psi,
\end{equation*}
for some suitably regular function $\widetilde{\psi}$ such that $\widetilde{\psi}_{\Sigma_0}\to 0$ as $r\to \infty$ and $\square_{g_{M,a}}\widetilde{\psi}=0$. We denote $T^{-1}\psi:=\widetilde{\psi}$ and refer to it as the \emph{time inverse} of $\psi$.

When the charges $I_{\ell}[\psi]$ are zero, for example, when the initial data for $\psi$ is compactly supported, it turns out that the charges corresponding to $\widetilde{\psi}$ are finite, i.e.\
\begin{equation*}
I_{\ell}[\widetilde{\psi}]<\infty
\end{equation*}
or $\ell=0,1,2$ and, generically, $I_{\ell}[\widetilde{\psi}]\neq 0$. Furthermore, we can express $I_{\ell}[\widetilde{\psi}]$ purely in terms of initial data for $\psi$, see Section \ref{rm:npcharges} and Section \ref{sec:timinvNP}.

We construct the initial data leading to $T^{-1}\psi$ by solving the equation:
\begin{equation*}
\mathcal{L} \widetilde{\psi}=F[\psi].
\end{equation*}
As explained in Section \ref{sec:introelliptic}, the operator $\mathcal{L}$ is elliptic outside the ergoregion. Constructing $\mathcal{L}^{-1}$ requires the use of spacelike redshift multipliers and is reminiscent of the strategy for constructing resolvent operators in \cite{warn15}. Note that in the $a=0$ setting, the spherical modes $\psi_{\ell}$ can be treated independently, so obtaining $\mathcal{L}^{-1}$ for a single mode amounts to solving a standard ODE, see \cite{aagprice}. When $a\neq 0$, however, it is neccessary to construct $\mathcal{L}^{-1}$ acting on the full solution \emph{before} projecting to a fixed spherical harmonic mode.

In Section \ref{sec:timeinv}, we construct $\mathcal{L}^{-1}$ and established $r$-decay properties $T^{-1}\psi|_{\Sigma_0}$ required to apply the time decay arguments sketched in the sections when $\psi$ is replaced with $T^{-1}\psi$.

We then obtain in particular the decay estimates:
\begin{align*}
|\psi_{0}|=|T\widetilde{\psi}_{0}|\lesssim &\: (1+\tau)^{-3+\epsilon},\quad |r\psi_{0}|=|r T\widetilde{\psi}_{0}|\lesssim  (1+\tau)^{-2+\epsilon}\\
|r^{-1}\psi_{1}|=|r^{-1}T \widetilde{\psi}_{1}|\lesssim &\: (1+\tau)^{-4-k+\epsilon},\quad |r\psi_{1}|=|rT \widetilde{\psi}_{1}|\lesssim  (1+\tau)^{-3-k+\epsilon}\\
|\psi_{2}|=|T\widetilde{\psi}_{2}|\lesssim  &\: (1+\tau)^{-5-k+\epsilon},\quad |r\psi_{2}|=|rT\widetilde{\psi}_{2}|\lesssim (1+\tau)^{-4-k+\epsilon}\\
|\psi_{\geq 3}|=|T\widetilde{\psi}_{\geq 3}|\lesssim  &\: (1+\tau)^{-\frac{11}{2}+\epsilon},\quad |r\psi_{\geq 3}|=|rT \widetilde{\psi}_{\geq 3}|\lesssim  (1+\tau)^{-\frac{9}{2}+\epsilon},
\end{align*}
which are sharp up to arbitrarily small $\epsilon$.

\subsubsection{Step 4: Late-time asymptotics and inverse polynomial tails}
In order to derive the precise leading order terms of $\psi$ and $\psi_{\geq \ell}$ with $\ell=1,2$, we need to go beyond the \emph{upper bound estimates} outlined in the previous sections.

We extend the method developed in \cite{paper2} in a spherically symmetric setting to the setting of Kerr spacetimes. The strategy is to derive the leading-order asymptotics for $T^k\widetilde{\psi}$ (with arbitrary $k$) by: 1) extending the conservation and non-vanishing property of $r^2P_{\ell}$ from $r=\ $ to a suitable far-away spacetime region to obtain the late-time asymptotics for $r \psi_{\ell}$ there, and 2) propagating these asymptotics to the rest of the spacetime by using that
\begin{equation*}
X^{1+\ell}\widetilde{\psi}_{\ell}\quad \textnormal{and} \quad T X^{\ell}\widetilde{\psi}_{\ell} 
\end{equation*}
both decay \emph{faster} in time than $X^{\ell}\widetilde{\psi}_{\ell}$ when $\ell=0,1$, $X\widetilde{\psi}_{2}$ and $T\widetilde{\psi}_{2}$ decay faster than $\widetilde{\psi}_2$, so we can simply integrate (multiply times) in the $X$ direction. This allows us to prove Theorem \ref{thm:intro2}. Theorem \ref{thm:intro1} then follows simply from the observation that $\psi=T\widetilde{\psi}$.

 We refer to a more expansive outline of a similar strategy in the spherically symmetric setting in \cite{paper2, aagprice}. The key new difficulty in Kerr in this step is the coupling between spherical harmonic modes $\psi_{\ell}$. See Section \ref{rm:decoupling} for an outline of this phenomenon. More concretely, obtaining late-time asymptotics in the far-away region for $\psi_2$ requires deriving first the late-time asymptotics of $\psi_0$ and plugging those in suitably in the derivation for $\psi_2$. A precise derivation of the late-time asymptotics for the $\ell=0,1,2$ modes and the effects of mode coupling can be found in Sections \ref{sec:latetimeasympl0}--\ref{sec:latetimeasympl2}.

\subsection{Outline of the paper}
We provide below an outline of the remaining sections in the paper.
\begin{itemize}
\item In Section \ref{sec:geom} we introduce the Kerr geometry and the relevant vector fields and spacetime foliations. We moreover state a systematic method for integrating by parts, which we apply to derive all the energy estimates in the paper. Finally, we state and derive some key properties regarding spherical harmonics.
\item In Section \ref{sec:equations} we provide various different forms of the wave equation \ref{eq:introwaveeq} and state the preliminary energy boundedness and integrated energy decay estimates, which we utilise in the rest of the paper.
\item In Section \ref{sec:defNPconstants} we define the Newman--Penrose charges for the $\ell=0,1,2$ modes and derive their conservation properties.
\item In Section \ref{sec:rpest} we derive the main hierarchies of $r^p$-weighted energy estimates.
\item In Section \ref{sec:edecay} we use the $r^p$-weighted energy estimates from Section \ref{sec:rpest} to derive energy decay estimates.
\item In Section \ref{sec:elliptic}, we derive a hierarchy of weighted elliptic estimates.
\item In Section \ref{sec:poinwdecay}, we convert the (weighted) energy decay estimates into pointwise decay estimates.
\item In Section \ref{sec:timeinv}, we construct the time inverse $T^{-1}\psi$ and establish its decay properties along $\Sigma_0$. We moreover define the time-inverted Newman--Penrose charges and derive explicit expressions for them in terms of initial data for $\psi$ on $\Sigma_0$.
\item Finally, in Sections \ref{sec:latetimeasympl0}--\ref{sec:latetimeasympl2}, we use the pointwise estimates from Section \ref{sec:poinwdecay}, applied moreover to both $\psi$ and the time inverse $T^{-1}\psi$ constructed in Section \ref{sec:timeinv}, to derive the precise late-time asymptotics for $\psi_0$, $\psi_1$ and $\psi_{2}$ and hence determine the precise form of the corresponding late-time tails.
\end{itemize}

\subsection{Acknowledgments}
The second author (S.A.) acknowledges support through the NSERC grant 502581 and the Ontario Early Researcher Award.

\section{Preliminaries: geometry}
\label{sec:geom}
In this section, we introduce the Kerr family of spacetimes and the spacetime foliations of interest. We moreover derive a systematic method of integrating by parts to derive energy estimates, and we prove some relevant properties of spherical harmonic decompositions.
\subsection{The Kerr metric}
We introduce in this section the $2$-parameter family of Kerr black hole exteriors $(\mathcal{M}_{M,a},g_{M,a})$ in \emph{ingoing Kerr coordinates}. 

Let $\mathcal{M}_{M,a}=\R_v\times [r_+,\infty)_r\times (\s^2)_{\theta,\varphi_*}$ be a manifold-with-boundary equipped with a coordinate chart $(v,r,\theta,\varphi_*)$ that is global, excluding the standard degeneration of spherical coordinates $(\theta,\varphi_*)$ on $\s^2$. Let $g_{M,a}$ denote the Lorentzian metric:
\begin{equation}
\label{eq:kerrmetric}
\begin{split}
g_{M,a}=&-\rho^{-2}\left(\Delta-a^2\sin^2\theta\right)dv^2+2dvdr-4Ma r \rho^{-2}  \sin^2\theta dv d\varphi_*-2a\sin^2\theta drd\varphi_*+\rho^2d\theta^2\\
&+\rho^{-2}((r^2+a^2)^2-a^2\Delta \sin^2\theta)\sin^2\theta d\varphi_*^2,
\end{split}
\end{equation}
where
\begin{align*}
\Delta=&\:r^2-2Mr+a^2=(r_+-r)(r-r_-),\\
r_{+,-}=&\:M\pm \sqrt{M^2-a^2},\\
\rho^2=&\: r^2+a^2\cos^2\theta.
\end{align*}
We will restrict our considerations to sub-extremal Kerr spacetimes by assuming that $|a|<M$, with $M>0$.\footnote{Extremal Kerr spacetimes correspond to the subclass of spacetimes satisfying $|a|=M$.}

We denote the boundary of $\mathcal{M}_{M,a}$ with $\mathcal{H}^+:=\{r=r_+\}$ and will refer to $\mathcal{H}^+$ as the \emph{(future) event horizon} of the spacetime. The level sets $S^2_{v',r'}=\{v=v',r=r'\}$ are 2-surfaces diffeomorphic to $\s^2$, which are known as the \emph{Boyer--Lindquist spheres}. Note that it is possible to extend the spacetimes $(\mathcal{M}_{M,a},g_{M,a})$ smoothly by attaching a \emph{black hole region} where $r_-<r<r_+$. This extension will not be necessary in the present paper.

On the manifold $\mathcal{M}_{M,a}\setminus \mathcal{H}^+$, one can alternatively consider the \emph{standard Boyer--Lindquist coordinates} $(t,r,\theta,\varphi)$, where
\begin{align*}
t=&\:v-r_*,\\
\varphi=&\:\varphi_*+\int_{r}^{\infty} \frac{a}{\Delta}\,dr' \mod 2\pi,\\
\end{align*}
where $r_*: (r_+,\infty)\to \R$ is a solution to
\begin{equation*}
\frac{dr_*}{dr}=\frac{r^2+a^2}{\Delta},
\end{equation*}
defined uniquely up to a constant.

We introduce also the function $u: \mathcal{M}_{M,a}\setminus \mathcal{H}^+\to \R$, with
\begin{equation*}
u=v-2r_*=t-r_*.
\end{equation*}

It will be convenient to state moreover the inverse metric $g^{-1}_{M,a}$:

\begin{equation}
\label{eq:invmetric}
\begin{split}
g^{-1}_{M,a}=&\:a^2\rho^{-2}\sin^2\theta \partial_v\otimes \partial_v+\rho^{-2}(r^2+a^2)[ \partial_v\otimes \partial_r+ \partial_r\otimes \partial_v]+\Delta \rho^{-2}\partial_r\otimes \partial_r\\
&+a\rho^{-2} [(\partial_v+\partial_r)\otimes \partial_{\varphi_*}+\partial_{\varphi_*}\otimes (\partial_v+\partial_r)]+\rho^{-2}[\partial_{\theta}\otimes \partial_{\theta}+\sin^{-2}\theta \partial_{\varphi_*}\otimes\partial_{\varphi_*}].
\end{split}
\end{equation}

\subsection{Vector fields}
\label{sec:vf}
The following vector fields are Killing vector fields with respect to $g_{M,a}$:
\begin{align*}
T=&\:\partial_v,\\
\Phi=&\:\partial_{\varphi_*},\\
K=&\: T+ \upomega_+\Phi,
\end{align*}
where $\upomega_+=\frac{a}{r_+^2+a^2}$ may be interpreted as the ``angular velocity of the Kerr black hole''. Furthermore, $K$ is tangent to $\mathcal{H}^+$.

It can easily be verified that the vector field $T$ is timelike for all $(v,r,\theta,\varphi_*)$ satisfying the condition:
\begin{equation*}
\Delta-a^2\sin^2\theta>0.
\end{equation*}
The subset $\{\Delta-a^2\sin^2\theta< 0\}$, where $T$ fails to be causal is called the \emph{ergoregion}. Note that
\begin{equation*}
\{\Delta-a^2\sin^2\theta< 0\} \subset \{r<2M\}.
\end{equation*}
 
We will introduce the following additional vector fields, which will play an important role in later analysis:
\begin{align*}
Y=&\:\partial_r,\\
\underline{L}=&\:-\frac{\Delta}{2(r^2+a^2)}Y,\\
L=&\: T+\frac{a}{r^2+a^2}\Phi-\underline{L}.
\end{align*}
The vector fields $L$ and $\underline{L}$ are null and define the \emph{principal null directions}.

Note moreover the following identities
\begin{align*}
L(v)=&\:1,\\
\underline{L}(u)=&\: 1,\\
g(L,\underline{L})=&-\frac{\Delta \rho^2}{2(r^2+a^2)^2},\\
[\underline{L},L]=&\:\frac{a r \Delta}{(r^2+a^2)^{3}} \Phi,\\
\underline{L}(\varphi_*)=&\:0,\\
L(\varphi_*)=&\:\frac{a}{r^2+a^2}.
\end{align*}
We introduce moreover the \emph{horizon azimuthal angle}
\begin{equation}
\label{eq:phihor}
{\varphi}_{\mathcal{H}^+}=\varphi_*-\upomega_+v,
\end{equation}
which satisfies $K({\varphi}_{\mathcal{H}^+})=0$, i.e.\ it is constant along the null generators of the future event horizon $\mathcal{H}^+$.

Finally, we introduce some notation for angular derivatives. Consider $\s^2$ equipped with standard spherical coordinates $(\theta,\varphi_*)$. We denote with $\snabla_{\s^2}$ the covariant derivative (Levi--Civita connection) on the unit round sphere $\s^2$. The corresponding Laplacian $\slashed{\Delta}_{\s^2}$ takes the following form in $(\theta,\varphi_*)$ coordinates:
\begin{equation*}
\slashed{\Delta}_{\s^2}(\cdot)=\frac{1}{\sin  \theta}\partial_{\theta}(\sin \theta \partial_{\theta}(\cdot))+\frac{1}{\sin^2\theta}\partial_{\varphi_*}^2(\cdot).
\end{equation*}

\subsection{Foliations and conformal coordinates}
\label{sec:foliations}
It will be convenient to keep track of decay with respect to the coordinate $r$ via the following notation: let $f :[r_+,\infty)_r\to \R$ and $k\in \Z$. We write $f=O(r^{-k})$ when there exists a constant $C>0$, depending only on $M$ and $a$ such that $|f(r)|\leq Cr^{-k}$ for all $r\in [r_+,\infty)$. We write $f=O_N(r^{-k})$ when moreover $\frac{d^n f}{dx^n} =O(r^{-k-n})$ for all $0\leq n\leq N$, where $f$ is assumed to be suitably regular. Finally, we write $f= O_{\infty}(r^{-k})$ if $f=O_N(r^{-k})$ for all $N\in \N_0$.

Consider the following \emph{time function} on $\mathcal{M}_{M,a}$:
\begin{equation*}
\tau=v-\int_{r_+}^r h(r')\,dr'-v_0,
\end{equation*}
where $v_0\in \R_{>0}$ and $h: [r_+,\infty)\to \R$ is a smooth, non-negative function, such that
\begin{align}
\label{eq:condh1}
\frac{2(r^2+a^2)}{\Delta}-h(r)=&\:O(r^{-2 }),\\
\label{eq:condh2}
h(r)\left(\frac{2(r^2+a^2)}{\Delta}-h(r)\right)> &\:a^2\Delta^{-1}\sin^2\theta.
\end{align}
By construction, $T(\tau)=1$.

Let $\hat{h}:[0,\frac{1}{r_+}]_x\to \R$, with 
\begin{equation*}
\hat{h}(x)=2(x^{-2}+a^2)-h(x^{-1})(x^{-2}-2Mx^{-1}+a^2). 
\end{equation*}

Then $\hat{h}(x)=O(x^0)$ by \eqref{eq:condh1}, where we take $O(x^k)$, $O_N(x^k)$ and $O_{\infty}(x^k)$ to have the same meaning as above with $x$ replacing $r$.

We will make the following additional assumptions for the sake of convenience:
\begin{align}
\label{eq:condh3}
\hat{h}\in &\:C^{\infty}([0,r_+^{-1}]),\\
\label{eq:condh4}
\hat{h}(0)=:&\:h_0> 0.
\end{align}
Smoothness of $\hat{h}$ is not a necessary assumption and the main results below are still valid with a lower regularity assumption on $\hat{h}$, in particular they hold also when $\hat{h}\in C^{N_0}([0,r_+^{-1}])$, with $N_0\in \N_0$ chosen suitably large.

Define the following 1-parameter family of hypersurfaces:
\begin{equation*}
\Sigma_{\tau'}:=\{\tau=\tau'\}.
\end{equation*}
We denote $\Sigma:=\Sigma_0$ and $\mathcal{R}=J^+(\Sigma)\cap J^-(\mathcal{H}^+)$.

Let $R>2M$ be an arbitrarily large radius. Then we introduce the notation:
\begin{equation*}
D_{\tau_1}^{\tau_2}:=\bigcup_{\tau_1\leq \tau\leq \tau_2} \Sigma_{\tau} \cap \{r\geq R\}.
\end{equation*}

We moreover denote $N_{\tau}=\Sigma_{\tau}\cap\{r\geq R\}$.

Let $\uprho=r|_{\Sigma}$. Then $(\uprho,\theta,\varphi_*)$ defines a coordinate chart on $\Sigma$ and $(\tau,\uprho,\theta,\varphi_*)$ defines a coordinate chart on $\mathcal{R}$. We moreover introduce the notation:
\begin{equation*}
X:=(\partial_{\uprho})_{\Sigma_{\tau}}=Y+h T.
\end{equation*}

Note that we can express the metric $g_{M,a}$ as follows in $(\tau,\uprho,\theta,\varphi_*)$ coordinates:
\begin{equation*}
\begin{split}
g_{M,a}=&-\rho^{-2}\left(\Delta-a^2\sin^2\theta\right)d\tau^2-2\left(\rho^{-2}h(\Delta-a^2\sin^2\theta)-1\right)d\tau d\uprho+\left(2\rho^2-h(\Delta-a^2\sin^2\theta)\right)\rho^{-2}hd\uprho^2\\
&-4Ma r \rho^{-2}  \sin^2\theta d\tau d\varphi_*-2(a+2Mar \rho^{-2} h)\sin^2\theta d\uprho d\varphi_*+\rho^2d\theta^2+\rho^{-2}((r^2+a^2)^2-a^2\Delta \sin^2\theta)\sin^2\theta d\varphi_*^2.
\end{split}
\end{equation*}

We further define the function
\begin{equation*}
s=v-2r_*+\int_{r_+}^r h(r')\,dr'+v_0=\tau-2r_*+2\int_{r_+}^r h(r')\,dr'+2v_0.
\end{equation*}
and consider the corresponding level sets
\begin{equation*}
I_{s'}:=\{s=s'\}\cap \mathcal{R}.
\end{equation*}
We can equip the level sets $I_s$ with the coordinate chart $(\tau,\theta,\varphi_*)$, and we express:
\begin{equation*}
\begin{split}
(\partial_{\tau})_{I_s}=&\:T-\frac{1}{2}\frac{1}{h-\frac{r^2+a^2}{\Delta}}X,\\
=&\:-\frac{1}{2}\frac{\frac{2(r^2+a^2)}{\Delta}-h}{h-\frac{r^2+a^2}{\Delta}}T-\frac{1}{2}\frac{\Delta}{r^2+a^2+(\Delta h-2(r^2+a^2))}Y.
\end{split}
\end{equation*}

See Figure \ref{fig:foliations} below for the relevant pictorial representations.

\begin{figure}[H]
	\begin{center}
\includegraphics[scale=0.5]{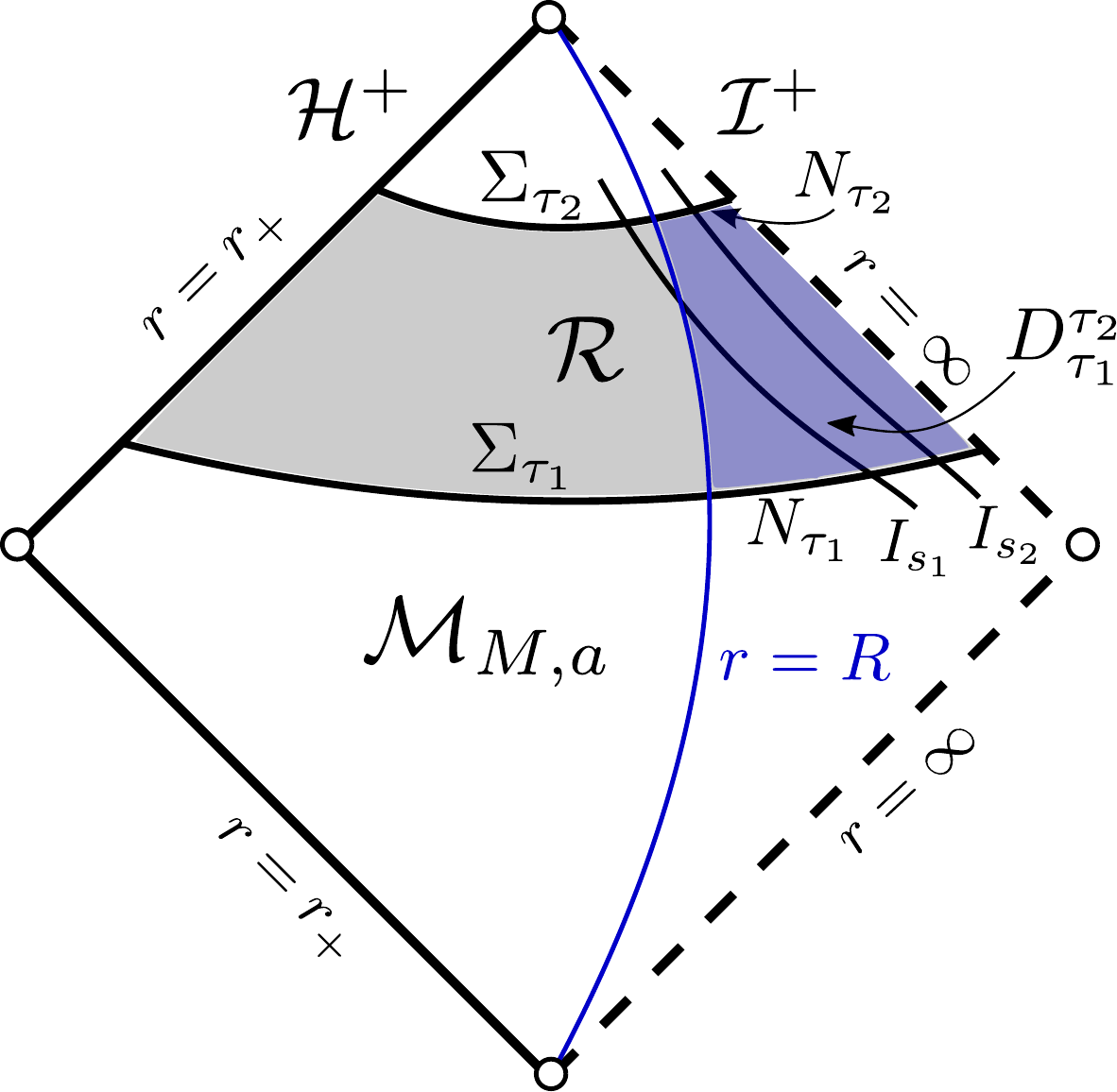}
\end{center}
\vspace{-0.2cm}
\caption{A 2-dimensional representation of the spacetime $\mathcal{M}_{M,a}$ and the hypersurfaces $\Sigma_{\tau_i}$, $I_{s_i}$, with $i=1,2$ and $\tau_1<\tau_2$, $s_1<s_2$. Each point in the picture represents a Boyer--Lindquist sphere $S^2_{v,r}$ and the hypersurface $\mathcal{I}^+$ is depicted at a finite distance.}
	\label{fig:foliations}
\end{figure}

The lemma below establishes the key causual properties of the hypersurfaces $\Sigma_{\tau}$, which are also represented pictorially in Figure \ref{fig:foliations}.

\begin{lemma}
\label{lm:propSigmat}
The 1-parameter family $\{\Sigma_{\tau}\}_{\tau\geq 0}$ and the level sets $I_s$ satisfy the following properties:
\begin{enumerate}
\item $\Sigma_{\tau}$ and $I_{s}$ (if non-empty) are spacelike for all $\tau,s\geq 0$,\\
\item $\Sigma_{\tau}$ is isometric to $\Sigma_{\tau'}$ for all $\tau,\tau'\geq 0$,\\
\item $\Sigma_{\tau}\cap \mathcal{H}^+=\{v=\tau+v_0\}\cap \mathcal{H}^+=S^2_{\tau+v_0,r_+}$,\\
\item For $v_0$ sufficiently large, depending on $h$, there exists $u_0>0$ such that $\Sigma_0 \subset J^+(\{u=u_0\})$.
\end{enumerate}
\end{lemma}
\begin{proof}
Property 3. follows directly from the definition of $\tau$. Property 2. follows from the fact that each $\Sigma_{\tau}$ can be obtained from $\Sigma_0$ along the flow corresponding to the Killing vector field $T$.  

By \eqref{eq:invmetric}, it follows that
\begin{equation*}
g^{-1}(d\tau,d\tau)=g^{-1}(dv-hdr,dv-hdr)=a^2\rho^{-2}\sin^2\theta-\Delta \rho^{-2}h\left(\frac{2(r^2+a^2)}{\Delta}-h\right),
\end{equation*}
so by \eqref{eq:condh2}, Property 1. must also hold for $\Sigma_{\tau}$. The same conclusion follows for $I_s$:
\begin{equation*}
g^{-1}(ds,ds)=g^{-1}(dv+(h-2(r^2+a^2)\Delta^{-1})dr,dv+(h-2(r^2+a^2)\Delta^{-1})dr)=a^2\rho^{-2}\sin^2\theta-\Delta \rho^{-2}h\left(\frac{2(r^2+a^2)}{\Delta}-h\right)<0.
\end{equation*}

In order to infer that property 4. holds, we first observe that
\begin{equation*}
u|_{\Sigma_0}(r)=v|_{\Sigma_0}(r)-2r_*(r)=v_0+\int_{r_+}^r h(r')\,dr'-2r_*(r).
\end{equation*}
so $u|_{\Sigma_0}(2r_+)=v_0-\int_{r_+}^{2r_+}h(r')\,dr'-2r_*(2r_+)>\frac{1}{2}v_0$ for $v_0$ suitably large. Then for $r\geq 2r_+$
\begin{equation*}
u|_{\Sigma_0}(r)=u|_{\Sigma_0}(2r_+)+\int_{2r_+}^r (h-\frac{2(r'^2+a^2)}{\Delta})(r')\,dr'\to u_0\in \R\quad \textnormal{as $r\to \infty$},
\end{equation*}
as the integral above is well-defined by the asymptotics in \eqref{eq:condh1}. For $v_0$ suitably large, we can additionally ensure positivity of $u_0$.
\end{proof}

Let $\mathbf{n}_{\tau}$ denote the future-directed unit normal vector field with respect to $\Sigma_{\tau}$. We may extend $\mathbf{n}_{\tau}$ as a vector field on $\mathcal{R}$ and introduce the vector field $N$ as follows via a smooth cut-off function:
\begin{align}
N=&\:\mathbf{n}_{\tau} && \textnormal{when $r\leq 3M$},\\
N=&\: T && \textnormal{when $r\geq 4M$}.
\end{align}
such that $N$ is timelike everywhere.\footnote{We used here that $g(T,T)<0$ for $r>2M$, i.e. outside the ergoregion.} It follows moreover that exists a constant $c=c(M,a)>0$, such that $g(N,N)\leq -c$.

From Lemma \ref{lm:propSigmat} it follows that $\{\Sigma_{\tau}\}_{\tau\geq 0}$ defines a foliation of $\mathcal{R}$ by isometric spacelike hypersurfaces. We can express $d\mu$, the natural volume form with respect to $g_{M,a}$, as follows:
\begin{equation*}
d\mu:=\sqrt{\det g_{M,a}}dv dr d\theta d\varphi_*=\rho^2 d\omega dv dr= \rho^2 d\omega d\tau d\uprho,
\end{equation*}
with $d\omega=\sin \theta d\theta d\varphi_*$ the natural volume form on $\s^2$.

Various estimates on the induced volume forms on $\Sigma_{\tau}$ and $I_s$ that are relevant in the context of divergence theorems are contained in Section 2.1.2 of \cite{dhr-teukolsky-kerr}. For the sake of completeness and notational convention, we prove these estimates in Lemma \ref{lm:volumeforms} below.

\begin{lemma}
\label{lm:volumeforms}
The volume forms $d\mu_{\tau}:=\sqrt{\det g|_{\Sigma_{\tau}}}d\uprho d\theta d\varphi_*$ and $d\mu_{I_s}:=\sqrt{\det g|_{I_s}}d\tau d\theta d\varphi_*$ satisfy the following properties:
\begin{enumerate}[\rm (i)]
\item We can express 
\begin{align*}
d\mu_{\tau}=&\:m_1(\uprho,\theta) rd\omega d\uprho,\\
d\mu_{I_s}=&\:m_2(\uprho,\theta) rd\omega d\tau && \textnormal{in $\{r\geq r_+\}$},
\end{align*}
with $m_i: [r_+,\infty)\times (0,\pi)\to \R$, $i=1,2$, smooth functions satisfying:
\begin{equation*}
c_h\leq m_i\leq C_h,
\end{equation*}
for positive constants $c_h,C_h$, depending on $M$, $a$ and the choice of $h$.
\item
We can moreover express:
\begin{align*}
-g\left(L,\mathbf{n}_{\tau}\right)\sqrt{\det g|_{\Sigma_{\tau}}}=&\:m_{L,1}(\uprho,\theta) \sin \theta,\\
-g\left(\frac{r^2+a^2}{\Delta}\underline{L},\mathbf{n}_{\tau}\right)\sqrt{\det g|_{\Sigma_{\tau}}}=&\:m_{\underline{L},1}(\uprho,\theta) r^2\sin \theta,\\
-g\left(L,\mathbf{n}_{I_s}\right)\sqrt{\det g|_{I_s}}=&\:m_{L,2}(\uprho,\theta) r^2\sin \theta,\\
-g\left(\frac{r^2+a^2}{\Delta}\underline{L},\mathbf{n}_{I_s}\right)\sqrt{\det g|_{I_s}}=&\:m_{\underline{L},2}(\uprho,\theta) \sin \theta,
\end{align*}
with
$m_{L,i}, m_{\underline{L},i}: [r_+,\infty)\times (0,\pi)\to \R$, $i=1,2$, smooth functions satisfying:
\begin{equation*}
c_h\leq m_{L,i}, m_{\underline{L},i}\leq C_h,
\end{equation*}
for $c_h,C_h>0$.
\item The following divergence identities hold:
\begin{align*}
\nabla_{\alpha}\left(\rho^{-2}\frac{r^2+a^2}{\Delta}L^{\alpha}\right)=&\:0,\\
\nabla_{\alpha}\left(\rho^{-2}\frac{r^2+a^2}{\Delta}\underline{L}^{\alpha}\right)=&\:0.
\end{align*}
\end{enumerate}
\end{lemma}
\begin{proof}
Property (i) results from the observation that $g(X,X)\sim r^{-2}$ and $g((\partial_{\tau})_{I_s},(\partial_{\tau})_{I_s})\sim r^{-2}$. 

In order to obtain property (ii), note first that we can express:
\begin{equation*}
g^{-1}=g^{-1}(d\tau^{\sharp},d\tau^{\sharp})d\tau^{\sharp}\otimes d\tau^{\sharp}+g^{-1}|_{\Sigma_{\tau}},
\end{equation*}
so we have that with respect to $(\tau,\uprho,\theta,\varphi_*)$ coordinates:
\begin{equation*}
\det g=\frac{\det g|_{\Sigma_{\tau}}}{g^{-1}(d\tau^{\sharp},d\tau^{\sharp})}
\end{equation*}
and we must therefore have that
\begin{equation}
\label{eq:normalder}
-\sqrt{\det g|_{\Sigma_{\tau}}}\mathbf{n}_{{\tau}}=\rho^2\sin \theta d\tau^{\sharp}.
\end{equation}
Furthermore,
\begin{align*}
\rho^2 d\tau^{\sharp}=&\:[a^2\sin^2\theta-(r^2+a^2)h]T+[r^2+a^2-\Delta h]Y+a(1-h)\Phi,\\
\rho^2 ds^{\sharp}=&\:[a^2\sin^2\theta-(r^2+a^2)\Delta^{-1}\hat{h}(r^{-1})]T+[r^2+a^2-\hat{h}(r^{-1})]Y+a(1-h)\Phi,
\end{align*}
where $\hat{h}(r^{-1})=2(r^2+a^2)-\Delta h(r)$, so after a straightforward computation, we obtain
\begin{align*}
g\left(\frac{r^2+a^2}{\Delta}\underline{L},\rho^2d\tau^{\sharp}\right)=&\: \frac{1}{2}(r^2+a^2)h+(1-2h)a^2\sin^2\theta=r^2+O(r),\\
g\left(L,\rho^2d\tau^{\sharp}\right)=&\:\frac{1}{2}\hat{h}(1/r)+O(r^{-1}).\end{align*}
Similar expressions involving $(ds)^{\sharp}$ follow after replacing $h(r)$ with $\Delta^{-1}\hat{h}(r^{-1})$ and vice versa. Property (ii) then follows from the assumption that $\hat{h}(0)=h_0>0$. 

Finally, we obtain property (iii) via a straightforward computation.
\end{proof}

We introduce the \emph{conformal radial coordinate} $x$, defined as follows in $\mathcal{R}$:
\begin{equation*}
x=\frac{1}{\uprho}.
\end{equation*}
We now define the conformally rescaled metric $\hat{g}_{M,a}$ in the region $\mathcal{R}=[0,\infty)_{\tau}\times (0,r_+^{-1}]_x\times \s^2$:
\begin{equation*}
\hat{g}_{M,a}:=r^{-2}{g}_{M,a}=x^2{g}_{M,a}.
\end{equation*}
Note that
\begin{equation*}
\begin{split}
\hat{g}_{M,a}=&-(x^2+O_{\infty}(x^3))d\tau^2+2(1+O_{\infty}(x))d\tau dx+O_{\infty}(x^0)dx^2+O_{\infty}(x^5)\sin^2\theta d\tau d\varphi_*-2(a^2+O_{\infty}(x))dxd\varphi_*\\
&+(1+O_{\infty}(x^2))d\theta^2+(1+O_{\infty}(x))\sin^2\theta d\varphi_*^2,
\end{split}
\end{equation*}
and the above metric components are smooth functions with respect to $(\tau,x,\theta,\varphi_*)$. Furthermore, $\hat{g}_{M,a}$ can be smoothly extended to the manifold-with-boundary:
\begin{equation*}
\widehat{\mathcal{R}}=[0,\infty)_{\tau}\times [0,r_+^{-1}]_x\times \s^2.
\end{equation*}
We refer to the boundary
\begin{equation*}
\mathcal{I}^+=[0,\infty)_{\tau}\times \{0\}_x\times \s^2
\end{equation*}
as \emph{future null infinity}. It follows immediately that $\mathcal{I}^+$ is a null hypersurface with respect to $\hat{g}_{M,a}$. Furthermore, it is straightforward to see that
\begin{equation*}
x|_{I_s}(\tau)\to 0
\end{equation*}
as $s\to \infty$.

We moreover denote with
\begin{equation*}
\widehat{\Sigma}_{\tau}:=\{\tau\}_{\tau}\times  [0,r_+^{-1}]_x\times \s^2
\end{equation*}
the extension of $\Sigma_{\tau}$ to $\widehat{\mathcal{R}}$. 

\subsection{Additional notation}
We will denote with $H^k(\Sigma_{\tau})$, with $k\in\N_0$, the standard Sobelev spaces with respect to the natural volume form corresponding to the induced metric $g_{M,a}|_{\Sigma_{\tau}}$. Similarly, will denote with $H^k(\widehat{\Sigma}_{\tau})$, with $k\in\N_0$, the standard Sobelev spaces with respect to the natural volume form corresponding to the induced conformal metric $\hat{g}_{M,a}|_{\Sigma_{\tau}}$.

Similarly, we use $H^k(\s^2)$ to denote Sobolev spaces of functions on $\s^2$ with respect to the standard volume form.

We will frequently use the notations $c$ and $C$ to indicate constants appearing at the right-hand side of an inequality. When the notation $C$ or $c$ appears in an inequality, we will make use of the following ``algebra of constants'':
\begin{equation*}
C+C=C C=C,\quad c+c=c c=c.
\end{equation*}
in order to avoid the introduction of additional notation to denote different constants in an estimate.

We will use the notation $f \sim g$, with $f,g$ two non-negative definite expressions to mean:
\begin{equation*}
c g \leq f\leq C g,
\end{equation*}
with constants $0<c<C$ that depend only on $a$, $M$ and the function $h$ that determines the foliation $\{\Sigma_{\tau}\}$.

\subsection{Systematic integration by parts}
We appeal to the properties of $\Sigma_{\tau}$ and $I_s$ established in Lemma \ref{lm:volumeforms} to derive a systematic method of integrating by parts; see also Remark 5.1 of \cite{dhr-teukolsky-kerr}.
\begin{lemma}
\label{lm:intbyparts}
Let $\mathcal{F}_{\underline{L}},\mathcal{F}_L$ be smooth functions on $\mathcal{M}_{M,a}$ and let $\mathcal{F}_{\slashed{\nabla}}$ be a smooth vector field on $S^2_{v,r}$ extended to $\mathcal{M}_{M,a}$, such that the following identity holds:
\begin{equation}
\label{eq:schematicid}
\underline{L}(\mathcal{F}_{\underline{L}})+L(\mathcal{F}_L)+\slashed{\textnormal{div}}_{\s^2} \mathcal{F}_{\slashed{\nabla}}+\Phi(F_{\Phi})+\mathcal{J}=0,
\end{equation}
where $\slashed{\textnormal{div}}_{\s^2}$ denotes the divergence operator with respect to the round metric on $\s^2$.
\begin{enumerate}[\rm (i)]
\item Then
\begin{equation*}
\textnormal{div}\, \left[\left(\rho^{-2} \frac{r^2+a^2}{\Delta}\mathcal{F}_{\underline{L}}\right) \underline{L}+ \left(\rho^{-2} \frac{r^2+a^2}{\Delta}\mathcal{F}_L\right)L \right]+\rho^{-2} \frac{r^2+a^2}{\Delta}\mathcal{J}+\rho^{-2} \frac{r^2+a^2}{\Delta}(\slashed{\textnormal{div}}_{\s^2} \mathcal{F}_{\slashed{\nabla}}+\Phi(\mathcal{F}_{\Phi}))=0.
\end{equation*}
The following integral identity holds:
\begin{equation*}
\begin{split}
\int_{\Sigma_{\tau_2}\cap\{s\leq S\}}& \rho^{-2}\mathcal{F}_{\underline{L}} \sqrt{\det g_{\Sigma_{\tau}}} g\left(\frac{r^2+a^2}{\Delta}\underline{L},-\mathbf{n}_{\Sigma_{\tau}}\right)+\rho^{-2}\mathcal{F}_L\sqrt{\det g_{\Sigma_{\tau}}} g\left(\frac{r^2+a^2}{\Delta}L,-\mathbf{n}_{\Sigma_{\tau}}\right)\,d\uprho d\theta d\varphi_*\\
&+\int_{\mathcal{H}^+\cap\{\tau_1\leq \tau \leq \tau_2\}} \frac{1}{2}\mathcal{F}_{\underline{L}}\,d\omega d\tau\\
&+\int_{I_S\cap\{\tau_1\leq \tau \leq \tau_2\}} \rho^{-2}\mathcal{F}_{\underline{L}} \sqrt{\det g_{I_S}} g\left(\frac{r^2+a^2}{\Delta}\underline{L},-\mathbf{n}_{I_s}\right)+\rho^{-2}\mathcal{F}_L\sqrt{\det g_{I_S}}g\left(\frac{r^2+a^2}{\Delta}L,-\mathbf{n}_{I_S}\right)\,d\tau d\theta d\varphi_*  \\
&+\int_{D^{\tau_2}_{\tau_1}\cap\{s\leq S\}} \rho^{-2} \frac{r^2+a^2}{\Delta}\mathcal{J}\,d\mu\\
=&\:\int_{\Sigma_{\tau_1}\cap\{s\leq S\}}\rho^{-2} \mathcal{F}_{\underline{L}} \sqrt{\det g_{\Sigma_{\tau}}} g\left(\frac{r^2+a^2}{\Delta}\underline{L},-\mathbf{n}_{\Sigma_{\tau}}\right)+\rho^{-2}\mathcal{F}_L \sqrt{\det g_{\Sigma_{\tau}}} g\left(\frac{r^2+a^2}{\Delta}L,-\mathbf{n}_{\Sigma_{\tau}}\right)\,d\uprho d\theta d\varphi_*.
\end{split}
\end{equation*}
\item If $\mathcal{F}_{\underline{L}},\mathcal{F}_L, \mathcal{F}_{\slashed{\nabla}}$  are non-negative definite, then we have that:
\begin{equation*}
\begin{split}
\int_{\Sigma_{\tau_2}\cap\{s\leq S\}}& \left(\mathcal{F}_{\underline{L}}+\Delta^{-1}\mathcal{F}_L\right)\,d\omega d\uprho+\int_{I_S\cap\{\tau_1\leq \tau \leq \tau_2\}}\left(\mathcal{F}_L+r^{-2}\mathcal{F}_{\underline{L}}\right)\,d\omega d\tau +\int_{\tau_1}^{\tau_2}\left[\int_{\Sigma_{\tau}\cap\{s\leq S\}} \frac{r^2+a^2}{\Delta}\mathcal{J}\,d\omega d\uprho\right]\,d\tau\\
&+\int_{\mathcal{H}^+\cap\{\tau_1\leq \tau \leq \tau_2\}} \mathcal{F}_{\underline{L}}\,d\omega d\tau\\
\sim &\:\int_{\Sigma_{\tau_1}\cap\{s\leq S\}} \left(\mathcal{F}_{\underline{L}}+\Delta^{-1}\mathcal{F}_L\right)\,d\omega d\uprho.
\end{split}
\end{equation*}
\end{enumerate}
\end{lemma}
\begin{proof}
Part (i) follows directly from Stokes' Theorem combined with (iii) of Lemma \ref{lm:volumeforms}. Part (ii) follows by applying additionally the remaining identities in Lemma \ref{lm:volumeforms}.
\end{proof}

\subsection{Spherical harmonic modes}
Let $\ell\in \N_0$ and consider the following projection operators
\begin{align*}
\pi_{\ell}:& L^2(\s^2)\to L^2(\s^2),\\
\pi_{\ell}f=&f_{\ell}:=\sum_{m=-\ell}^{\ell} f_{\ell, m} Y_{\ell, m}(\theta,\varphi_*),
\end{align*}
with $f_{\ell, m} \in \C$ and $Y_{\ell, m}(\theta,\varphi_*)$, $m=-\ell,\ldots,\ell$ spherical harmonics with angular momentum $\ell$, with respect to the polar angle $\theta$ and the azimuthal angle $\varphi_*$.

Note that
\begin{equation*}
\sum_{\ell=0}^{\infty}\int_{\s^2} f_{\ell}^2\,d\omega=\int_{\s^2}f^2\,d\omega.
\end{equation*}

The operator $\pi_{\ell}$ is well-defined on the function space $C^{\infty}(\Sigma_{\tau})$, where we interpret $\pi_{\ell}$ as acting on the restrictions of functions in $C^{\infty}(\Sigma_{\tau})$ to functions on the Boyer--Lindquist spheres foliating $\Sigma_{\tau}$, which we cover with angular coordinates $(\theta,\varphi_*)$. Since $\pi_{\ell}$ is a bounded linear operator with respect to $||\cdot||_{L^2(\Sigma_{\tau})}$, the following extension is also well-defined:
\begin{equation*}
\pi_{\ell}: L^2(\Sigma_{\tau})\to L^2(\Sigma_{\tau}).
\end{equation*}

We will introduce the following additional notation: let $L\in \N_0$ and $f\in  L^2(\Sigma_{\tau})$, then
\begin{align*}
f_{\leq L}:=&\:\sum_{\ell=0}^L \pi_{\ell}f,\\
f_{\geq L+1}:=&\:f-f_{\leq L}.
\end{align*}

We will moreover need to investigate how the projection operators $\pi_{\ell}$ act on product functions of the form $\sin^2\theta f$, with $f\in L^2(\s^2)$.
\begin{lemma}
\label{eq:lprojsin}
Let $f\in L^2(\s^2)$. Then there exist numerical constants $N_{\ell, m,-2}, N_{\ell, m,0},  N_{\ell, m,+2}$, such that
\begin{equation}
\label{eq:lproj}
\pi_{\ell}(\sin^2\theta f)=\sum_{m=-\ell}^{\ell}( N_{\ell, m,-2}f_{\ell-2, m}+N_{\ell,m, 0}f_{\ell, m}+ N_{\ell, m,+2}f_{\ell+2, m})Y_{\ell, m}(\theta,\varphi_*),
\end{equation}
with $N_{\ell, m,0}$ and  $N_{\ell, m,+2}$ non-vanishing and $N_{\ell, m,-2}$ non-vanishing if and only if $|m|\leq \ell-2$.

In particular,
\begin{equation}
\label{eq:l2projl0}
\pi_2(\sin^2\theta \pi_0(f))=-\frac{4}{3}\sqrt{\frac{\pi}{5}}\pi_0(f)Y_{2,0}(\theta).
\end{equation}
\end{lemma}
\begin{proof}
We can express:
\begin{equation*}
\pi_{\ell}(\sin^2\theta f)=\sum_{m=-\ell}^{\ell} \left[\sum_{\ell'=|m|}^{\infty} f_{\ell ', m}\int_{\s^2}\sin^2\theta Y_{\ell', m}\overline{Y}_{\ell, m}\,\sin \theta d\theta d\varphi\right]Y_{\ell, m}.
\end{equation*}
Note that for given $\ell$ and $\ell'$, the integral
\begin{equation*}
\int_{\s^2}\sin^2\theta Y_{\ell', m}\overline{Y}_{\ell, m}\,\sin \theta d\theta d\varphi=0
\end{equation*}
if and only if
\begin{equation}
\label{eq:intlegendre}
\int_{-1}^1(1-x^2)P^m_{\ell'}(x)P^m_{\ell}(x)\,dx=0,
\end{equation}
where $P^m_{\ell}$ and $P^m_{\ell'}$ denote associated Legendre polynomials. By the definition of Legendre polynomials, it follows immediately that the above integral vanishes if $\ell'-\ell$ is odd and is non-vanishing when $\ell=\ell'$. It remains to consider the cases when $\ell'-\ell$ is even and non-zero. Without loss of generality, we can assume that $0 \leq m \leq \ell'$ and $\ell\geq \ell'+2$.

We use the following standard recursive relations between associated Legendre polynomials (see for example Chapter 8 of \cite{abram70}): let $l\neq 0$, then
\begin{align*}
\sqrt{1-x^2}P_{l}^m=&-(2l+1)^{-1}[(l-m+1)(l-m+2)P^{m-1}_{l+1}-(l+m-1)(l+m)P^{m-1}_{l-1}]\quad \textnormal{for  $m\geq 1$},\\
\sqrt{1-x^2}P_{l}^m=&-(2l+1)^{-1}[P^{m+1}_{l+1}-P^{m+1}_{l-1}]\quad \textnormal{for $m\leq l-2$}.
\end{align*}
to write 
\begin{equation*}
\begin{split}
(2\ell+1)(2\ell'+1)(1-x^2)P^m_{\ell'}P^m_{\ell}=&\: \left[(\ell-m+1)(\ell-m+2)P^{m-1}_{\ell+1}-(\ell+m-1)(\ell+m)P^{m-1}_{\ell-1}\right]\\
\cdot &\left[(\ell'-m+1)(\ell'-m+2)P^{m-1}_{\ell'+1}-(\ell'+m-1)(\ell'+m)P^{m-1}_{\ell'-1}\right],
\end{split}
\end{equation*}
when $\ell'\neq 0$ and $m\geq 1$. For $\ell'\neq 0$ and $m=0$, we obtain instead
\begin{equation*}
(2\ell+1)(2\ell'+1)(1-x^2)P^0_{\ell'}P^0_{\ell}=[P^{1}_{\ell+1}-P^{1}_{\ell-1}][P^{1}_{\ell'+1}-P^{1}_{\ell'-1}].
\end{equation*}
The integral \eqref{eq:intlegendre} corresponding to the above cases is non-vanishing if and only if $\ell=\ell'$ or $|\ell-\ell'|=2$.

It only remains to consider the case: $\ell'=m=0$ and $\ell\geq 2$, for which it follows by integrating by parts that
\begin{equation*}
\begin{split}
\int_{-1}^1(1-x^2)P^0_{\ell}P^0_{0}\,dx=2^{-\ell}(\ell!)^{-1}\int_{-1}^1(1-x^2)\frac{d^{\ell}}{dx^{\ell}}((x^2-1)^{\ell})\,dx
\end{split}
\end{equation*}
is non-vanishing if and only if $\ell=2$. The expression \eqref{eq:lproj} then follows and \eqref{eq:l2projl0} can easily be computed explicitly by keeping track of the appropriate normalization constants.
\end{proof}

In the lemma below, we state various Poincar\'e-type inequalities on $\s^2$.

\begin{lemma}
Let $f\in H^2(\s^2)$. Then:
\begin{align}
\label{eq:poincare1}
\int_{\s^2} |\snabla_{\s^2}f_{\geq \ell}|^2\,d\omega\geq &\:\ell(\ell+1)\int_{\s^2} f_{\geq \ell}^2\,d\omega,\\
\label{eq:poincare2}
\int_{\s^2} |\snabla_{\s^2}f_{ \ell}|^2\,d\omega= &\:\ell(\ell+1)\int_{\s^2} f_{\ell}^2\,d\omega,\\
\label{eq:poincare3}
\int_{\s^2} f_{\ell}^2+|\snabla_{\s^2}f_{ \ell}|^2+|\snabla_{\s^2}^2f_{ \ell}|^2\,d\omega\leq &\:(1+\ell(\ell+1)+\ell^2(\ell+1)^2)\int_{\s^2} f_{\ell}^2\,d\omega.
\end{align}
Furthermore, let $N\in \N$ and assume that $f\in H^{\max\{N,2\}}(\s^2)$. Then there exists constants $0<c<C$, depending on $N$, such that
\begin{equation}
\label{eq:angdercontrol}
c \int_{\s^2} |\snabla_{\s^2}^s \slashed{\Delta}_{\s^2}^{\frac{N-s}{2}}f|^2\,d\omega\leq \sum_{n=0}^N\int_{\s^2} |\snabla_{\s^2}^nf|^2\,d\omega \leq C\int_{\s^2} |\snabla_{\s^2}^s \slashed{\Delta}_{\s^2}^{\frac{N-s}{2}}f|^2\,d\omega,
\end{equation}
where $s=0$ if $N$ is even and $s=1$ if $N$ is odd.
\end{lemma}
\begin{proof}
The inequalities \eqref{eq:poincare1}--\eqref{eq:poincare3} follow by decomposing $f_{\geq \ell}=\sum_{\ell'=\ell}^{\infty} f_{\ell'}$, using that $\slashed{\Delta}_{\s^2}f_{\ell}=-\ell(\ell+1)f_{\ell}$ and integrating by parts on $\s^2$.

We obtain \eqref{eq:angdercontrol} by integrating by parts with respect to covariant derivation on $\s^2$.
\end{proof}

\subsection{Hardy inequalities along hyperboloidal hypersurfaces}
Throughout this paper we will frequently appeal to Hardy inequalities, in the form of the estimates in Lemma \ref{lm:hardy} below, to estimate lower-order weighted derivatives in terms of higher-order weighted derivatives. 
\begin{lemma}
\label{lm:hardy}
Let $0\leq a<b$ and let $f\in C^1([a,b])$, then for $p\in \R\setminus \{-1\}$:
\begin{align}
\label{eq:hardy1d}
\int_{a}^b x^{p}f^2(x)\,dx\leq&\: 4(p+1)^{-2} \int_{a}^b x^{p+2}\left(\frac{df}{dx}\right)^2(x)\,dx+2(p+1)^{-1}\left[(x^{p+1}f^2)(b)-(x^{p+1}f^2)(a)\right].
\end{align}
Let $r_+\leq r_1<r_2$ and let $f\in C^1(\Sigma_{\tau})$, then for $p\in \R\setminus \{-1\}$:
\begin{equation}
\label{eq:hardyX}
\begin{split}
\int_{\Sigma_{\tau}\cap\{r_1\leq \uprho\leq r_2\}} r^{p} f^2\,d\omega d\uprho\leq&\: 4(p+1)^{-2} \int_{\Sigma_{\tau}\cap\{r_1\leq \uprho\leq r_2\}} r^{p+2} (Xf)^2\,d\omega d\uprho\\
&+2(p+1)^{-1}\left[\int_{S^2_{\tau,r_2}} r_2^{p+1}f^2\,d\omega-\int_{S^2_{\tau,r_1}} r_1^{p+1}f^2\,d\omega\right].
\end{split}
\end{equation}

And for $r_+< r_1<r_2$ and $f\in C^1(\mathcal{R})$, there exists a constant $C=C(M,a,r_1)>0$ such that
\begin{equation}
\label{eq:hardyL}
\begin{split}
\int_{\Sigma_{\tau}\cap\{r_1\leq \uprho\leq r_2\}} r^{p} f^2\,d\omega d\uprho\leq&\: C(p+1)^{-2} \int_{\Sigma_{\tau}\cap\{r_1\leq \uprho\leq r_2\}} r^{p+2} (Lf)^2+r^{p-2}(Tf)^2+r^{p-2}(\Phi f)^2\,d\omega d\uprho\\
&+2(p+1)^{-1}\left[\int_{S^2_{\tau,r_2}} r_2^{p+1}f^2\,d\omega-\int_{S^2_{\tau,r_1}} r_1^{p+1}f^2\,d\omega\right].
\end{split}
\end{equation}
\end{lemma}
\begin{proof}
The estimate \eqref{eq:hardy1d} follows from integrating $\frac{d}{dx}\left(x^{p+1} f^2)\right)$ and applying the fundamental theorem of calculus and then immediately implies \eqref{eq:hardyX}. We conclude  \eqref{eq:hardyL} from  \eqref{eq:hardyX} combined with the following relation between $L$ and $X$:
\begin{equation*}
L=\frac{\Delta}{2(r^2+a^2)}X+O(r^{-2})T+O(r^{-2})\Phi.
\end{equation*}
\end{proof}

\section{Preliminaries: wave equation}
\label{sec:equations}
We recall and derive in this section some preliminary estimates for the geometric wave equation on Kerr spacetimes that will form important ingredients for the estimates in subsequent sections.
\subsection{Main equations}
\label{sec:maineq}
In the lemma below we represent the wave equation $\square_{g_{M,a}}\psi=0$ in terms of the vector fields $T,Y,\Phi,X,L$ and $\Lbar$. 
\begin{lemma}
\label{lm:maineqs}
Let $\psi\in C^{\infty}(\mathcal{R}\to \C)$ be a solution to
\begin{equation}
\label{eq:waveeq}
\square_{g_{M,a}}\psi:=\frac{1}{\sqrt{\det g_{M,a}}}\partial_{\alpha}\left(\sqrt{\det g_{M,a}}(g_{M,a}^{-1})^{\alpha \beta}\partial_{\beta}\psi \right)=0.
\end{equation}
Then $\psi$ satisfies the following equation:
\begin{equation}
\label{eq:waveeq1}
0=\:a^2\sin^2\theta T^2\psi+2(r^2+a^2)TY\psi+Y(\Delta Y\psi)+2a T\Phi\psi+2aY\Phi \psi+2rT\psi+\slashed{\Delta}_{\s^2}\psi.
\end{equation}
We can reformulate \eqref{eq:waveeq1} as follows:
\begin{equation}
\label{eq:waveeq2}
\begin{split}
X(\Delta X\psi)+2aX\Phi \psi+\slashed{\Delta}_{\s^2}\psi=&\:2[h\Delta-(r^2+a^2)] XT\psi+[(\Delta h)'-2r]T\psi\\
&+[2h(r^2+a^2)-h^2\Delta-a^2\sin^2\theta] T^2\psi+2a(h-1) T\Phi\psi.
\end{split}
\end{equation}
Furthermore, the rescaled quantity $\phi:=\sqrt{r^2+a^2}\psi$ satisfies the equations:
\begin{align}
\label{eq:phieq1}
0=&\:a^2\sin^2\theta T^2\phi+2a T\Phi\phi+2aY\Phi \phi+2(r^2+a^2)TY\phi+(r^2+a^2)Y\left(\frac{\Delta}{r^2+a^2}Y\phi\right)+\slashed{\Delta}_{\s^2}\phi\\\nonumber
&-2a\frac{r}{r^2+a^2} \Phi\phi-\sqrt{r^2+a^2}\frac{d}{dr}(r(r^2+a^2)^{-3/2}\Delta) \phi,\\ 
\label{eq:phieq1}
4\underline{L}L\phi=&\:a^2\sin^2\theta \frac{\Delta}{(r^2+a^2)^2}T^2\phi+2a  \frac{\Delta}{(r^2+a^2)^2}T\Phi\phi+  \frac{\Delta}{(r^2+a^2)^2} \slashed{\Delta}_{\s^2}\phi+2a  \frac{r\Delta}{(r^2+a^2)^3} \Phi \phi\\ \nonumber
&-\frac{1}{2r}\frac{d}{dr}\left(r^2\Delta^2(r^2+a^2)^{-3}\right)\phi.
\end{align}
\end{lemma}
\begin{proof}
The above identities follow from straightforward computations that use the expressions for $g_{M,a}$ and $g_{M,a}^{-1}$ in \eqref{eq:kerrmetric} and \eqref{eq:invmetric}, respectively, and the definitions of the vector fields $T,\Phi, Y, X,L,\underline{L}$ in Section \ref{sec:vf}.
\end{proof}

By rewriting \eqref{eq:waveeq2} in terms of the conformal coordinates $(\tau,x,\theta,\varphi_*)$, we obtain an equation for $\phi$ (defined in the statement of Lemma \ref{lm:maineqs}) with smooth coefficients with respect to the differentiable structure on $\widehat{\mathcal{R}}$:

\begin{corollary}
Let $\psi\in C^{\infty}(\mathcal{R} \to \C)$ be a solution to \eqref{eq:waveeq}. Then $\phi$ satisfies the following equation in $(\tau,x,\theta,\varphi_*)$ coordinates:
\begin{equation}
\label{eq:confeqphi}
\begin{split}
0=&\:(1+a^2x^2)\partial_x\left(\frac{\Delta}{(r^2+a^2)}x^2\partial_x\phi\right)+2[1+a^2x^2-x^2\hat{h}(x)]\partial_{\tau}\partial_x\phi+\left[-\frac{2(r^2+a^2)-\hat{h}}{\Delta} \hat{h}+a^2\sin^2\theta\right]\partial_{\tau}^2\phi+\slashed{\Delta}_{\s^2}\phi\\
&+2a\left[1-\frac{2(r^2+a^2)-\hat{h}}{\Delta}\right] \partial_{\tau}\partial_{\varphi_*}\phi-2ax^2\partial_x\partial_{\varphi_*}\phi-(1+a^2x^2)\frac{d}{dx}\left(\frac{\hat{h}x^2}{1+a^2x^2}\right)\partial_{\tau}\phi-\frac{2a x}{1+a^2x^2}\partial_{\varphi_*}\phi\\
&+(1+a^2x^2)^{\frac{1}{2}}\frac{d}{dx}\left(\frac{(1-2Mx+a^2)}{(1+a^2x^2)^{\frac{3}{2}}}\right)\phi.
\end{split}
\end{equation}
\end{corollary}

\begin{remark}
When viewed as a function on the extended manifold $\widehat{\mathcal{R}}$, the restriction $\phi|_{\mathcal{I}^+}$ is known as \emph{the (Friedlander) radiation field}.
\end{remark}

We introduce the following higher-order quantities in $\{r>r_+\}$:
\begin{equation*}
\phi^{(n)}=(2(r^2+a^2)^2\Delta^{-1}L)^n\phi.
\end{equation*}
Observe that we can alternatively express $\phi^{(n)}=(-1)^n(\partial_x)^n\phi$ in $(u,x,\theta,\varphi_*)$ coordinates, so $\phi^{(n)}$ may be considered natural higher-order analogues of $\phi$ from the point of view of the conformal spacetime $(\widehat{\mathcal{R}},\hat{g}_{M,a})$.

\begin{proposition}
\label{eq:ho}
Let $\psi\in C^{\infty}(\mathcal{R}\to \C)$ be a solution to \eqref{eq:waveeq}. Then $\phi^{(n)}$ satisfies the following equation
\begin{equation}
\label{eq:commeq}
\begin{split}
4\underline{L}L\phi^{(n)}= &\:a^2\sin^2\theta \frac{\Delta}{(r^2+a^2)^2}T^2 \phi^{(n)}+2a  \frac{\Delta}{(r^2+a^2)^2}T\Phi\phi^{(n)}+2(1+2n)a  \frac{r\Delta}{(r^2+a^2)^3} \Phi \phi^{(n)}\\
&+\Delta (r^2+a^2)^{-2} \slashed{\Delta}_{\s^2}\phi^{(n)}+\Delta (r^2+a^2)^{-2}[n(n+1)+O_{\infty}(r^{-1})] \phi^{(n)}-[4nr^{-1}+O_{\infty}(r^{-2})] L \phi^{(n)}\\
&+a \sum_{k=0}^{n-1}O_{\infty}(r^{-2})\Phi\phi^{(k)}+n\sum_{k=0}^{n-1}O_{\infty}(r^{-2})\phi^{(k)}.
\end{split}
\end{equation}
\end{proposition}
\begin{proof}
We will prove \eqref{eq:commeq} by induction. First, note that the $n=0$ case follows from \eqref{eq:phieq1}. Now suppose \eqref{eq:commeq} holds for some $n\in \N$. Then
\begin{equation*}
\begin{split}
4\underline{L}L\phi^{(n+1)}=&\:4L\left(\underline{L}\left(2\Delta^{-1} (r^2+a^2)^{2} L\phi^{(n)}\right)\right)+4[\underline{L},L] \phi^{(n+1)}\\
=&\:8L\left(\Delta^{-1} (r^2+a^2)^{2} \underline{L} L\phi^{(n)}\right)-4L\left(\frac{d}{dr}(\Delta^{-1} (r^2+a^2)^{2})\Delta (r^2+a^2)^{-1} L\phi^{(n)}\right)\\
&+4\frac{a r \Delta}{(r^2+a^2)^{3}} \Phi \phi^{(n+1)}.
\end{split}
\end{equation*}
By applying \eqref{eq:commeq}, we obtain
\begin{equation*}
\begin{split}
8L\left(\Delta^{-1} (r^2+a^2)^{2} \underline{L} L\phi^{(n)}\right)=&\:a^2\sin^2\theta \frac{\Delta}{(r^2+a^2)^2}T^2 \phi^{(n+1)}+2a  \frac{\Delta}{(r^2+a^2)^2}T\Phi\phi^{(n+1)}+\Delta (r^2+a^2)^{-2} \slashed{\Delta}_{\s^2}\phi^{(n+1)}\\
&+4(1-2n)a\frac{r\Delta}{(r^2+a^2)^2} \Phi \phi^{(n+1)}+\Delta (r^2+a^2)^{-2}[n(n+1)+O_{\infty}(r^{-1})] \phi^{(n+1)}\\
&-[4nr^{-1}+O_{\infty}(r^{-2})] L \phi^{(n+1)}\\
&+a \sum_{k=0}^{n}O_{\infty}(r^{-2})\Phi\phi^{(k)}+(n+1)\sum_{k=0}^{n}O_{\infty}(r^{-2})\phi^{(k)}.
\end{split}
\end{equation*}
Furthermore,
\begin{equation*}
\begin{split}
-4L\left(\frac{d}{dr}(\Delta^{-1} (r^2+a^2)^{2})\Delta (r^2+a^2)^{-1} L\phi^{(n)}\right)=&\:[-4 r^{-1}+O_{\infty}(r^{-2})] L\phi^{(n+1)}+[2r^{-2}+O_{\infty}(r^{-3}) ]\phi^{(n)}
\end{split}
\end{equation*}
By combining the above equations, we arrive at \eqref{eq:commeq} with $n$ replaced by $n+1$, which concludes the induction argument.
\end{proof}

\subsection{Propagation of regularity and limits at null infinity}
It will be convenient to appeal to a propagation of regularity result for the wave equation \eqref{eq:waveeq} with respect to the differentiable structure on both $\mathcal{R}$ and $\widehat{\mathcal{R}}$.
\begin{proposition}
Let $k\in \N_0$. 
\begin{enumerate}[\rm (i)]
\item
Let $(\Psi,\Psi')\in H^{k+1}_{\rm loc}(\Sigma_0)\times H^{k}_{\rm loc}(\Sigma_0)$. Then there exists a solution $\psi$ to \eqref{eq:waveeq} (in the distributional sense), such that for all $T\geq 0$:
\begin{align*}
\psi\in &\:C^0([0,T], H^{k+1}_{\rm loc}(\Sigma_0)) \cap C^1([0,T], H^{k}_{\rm loc}(\Sigma_0)),\\
\psi|_{\Sigma_0}=&\:\Psi,\\
T\psi|_{\Sigma_0}=&\:\Psi'.
\end{align*}
\item Assume moreover that $(\sqrt{r^2+a^2}\Psi,\sqrt{r^2+a^2}\Psi')\in H^{k+1}(\widehat{\Sigma}_0)\times H^{k}(\widehat{\Sigma}_0)$. Then, for all $\tau\geq 0$,
\begin{equation}
\label{eq:reghoderphi}
(\phi|_{\Sigma_{\tau}},T\phi|_{\Sigma_{\tau}})\in  H^{k+1}(\widehat{\Sigma}_{\tau})\times H^{k}(\widehat{\Sigma}_{\tau}).
\end{equation}
\item Let $(\Psi,\Psi')\in C^{\infty}(\Sigma_0)\times C^{\infty}(\Sigma_0)$. Assume moreover that
\begin{equation*}
(\sqrt{r^2+a^2}\Psi,\sqrt{r^2+a^2}\Psi')\in H^{N+1}(\widehat{\Sigma}_0)\times H^{N}(\widehat{\Sigma}_0)
\end{equation*}
for some $N\in \N_0$. Then for all $\tau\geq 0$ and $k,m,n\in \N_0$, such that $k+m+n\leq N$,
\begin{equation*}
\int_{\s^2}|\snabla_{\s^2}^{k} \partial_{\tau}^m\phi^{(n)}|^2(\tau,\uprho,\theta,\varphi)\,d\omega
\end{equation*}
attains a finite limit as $\uprho\to \infty$.
\end{enumerate}
\end{proposition}
\begin{proof}
Part (i) follows from standard result for linear wave equations, see for example \cite{sogge}. In order to obtain part (ii), we consider \eqref{eq:confeqphi} and observe that the coefficients in the equation are smooth functions of $(x,\theta)$ (note in particular the smoothness in $r$.) Therefore, \eqref{eq:reghoderphi} follows from standard (finite-in-time) higher-order energy estimates with respect to $(\tau,x,\theta,\varphi_*)$ coordinates.

By the fundamental theorem of calculus it follows that for $0<x_0\leq r_+^{-1}$ and $x\leq x_0$
\begin{equation*}
\begin{split}
\int_{\s^2}&|\snabla_{\s^2}^{k_1}\partial_\tau^{k_2}\partial_{x}^{k_3}\phi|^2(\tau,x,\theta,\varphi)\,d\omega=\int_{\s^2}|\snabla_{\s^2}^{k_1}\partial_\tau^{k_2}\partial_{x}^{k_3}\phi|^2(\tau,x_0,\theta,\varphi)\,d\omega\\
&-2\int_x^{x_0}\snabla_{\s^2}^{k_1}\partial_\tau^{k_2}\partial_{x}^{k_3}\phi \cdot \snabla_{\s^2}^{k_1}\partial_\tau^{k_2}\partial_{x}^{k_3+1}\phi (\tau,x',\theta,\varphi)\,d\omega dx'.
\end{split}
\end{equation*}
Observe that the first integral on the right-hand side is finite by smoothness of $\psi$ with respect to the differentiable structure on $\mathcal{R}$, which in turn follows from the smoothness of $\Psi$ and $\Psi'$, together with a standard propagation of regularity argument; see again \cite{sogge}.

  Then, after applying Cauchy--Schwarz to the second term on the right-hand side of the equation above, together with \eqref{eq:reghoderphi} and \eqref{eq:confeqphi}, it follows that the limit of the left-hand side at $x=0$ is well-defined provided $k_1+k_2\leq N$. Since $\snabla_{\s^2}^{k}\partial_{\tau}^{m}\phi^{(n)}$, with $k+m+n\leq N$, can be written as a linear combination of $\snabla_{\s^2}^{k_1}\partial_\tau^{k_2}\partial_{x}^{k_3}\phi$ with $k_1+k_2+k_3\leq N$, we obtain part (iii).
\end{proof}

Throughout the remainder of the paper we will assume for the sake of convenience that the initial data for \eqref{eq:waveeq} satisfy:
\begin{equation*}
(\sqrt{r^2+a^2}\Psi,\sqrt{r^2+a^2}\Psi')\in C^{\infty}(\widehat{\Sigma}_0)\times C^{\infty}(\widehat{\Sigma}_0),
\end{equation*}
though this assumption can be significantly weakened by a standard density argument with respect to the (weighted) energy norms in the relevant estimates.
\subsection{Energy boundedness and integrated local energy decay}
\label{sec:prelimen}
Two key ingredients towards deriving the results in the present article are \emph{energy boundedness} and \emph{integrated local energy decay}. In the context of subextremal Kerr spacetimes with $|a|<M$, these were both derived in \cite{part3}. We summarize the relevant results as Theorem \ref{thm:morawetz} below.

\begin{thmx}[Dafermos--Rodnianski--Shlapentokh-Rothman, \cite{part3}]
\label{thm:morawetz}
 Let $\psi$ be a solution to \eqref{eq:waveeq} arising from suitably regular initial data with respect to $\widehat{\mathcal{R}}$. Then the following estimates hold:
 \begin{enumerate}[\rm (i)]
 \item There exists a constant $C=C(M,a)>0$, such that for any $\tau_2\geq \tau_1\geq 0$:
 \begin{equation}
 \label{eq:ebound}
 \int_{\Sigma_{\tau_2}}  \mathbf{J}^{N}[\psi]\cdot \mathbf{n}_{\tau_2}\, r^2d\omega d \uprho\leq C \int_{\Sigma_{\tau_1}}  \mathbf{J}^{N}[\psi]\cdot \mathbf{n}_{\tau_1}\, r^2d\omega d \uprho.
 \end{equation}
 \item
Let $n\in \N_0$. For $R_0>0$ suitably large and $R_1>0$, there exists a constant $C=C(n,R_0,R_1,M,a)>0$, such that for any $\tau_2\geq \tau_1\geq 0$:
 \begin{equation}
 \label{eq:morawetz1}
\sum_{0\leq k_1+k_2+k_3\leq n}\int_{\tau_1}^{\tau_2}\left[ \int_{\Sigma_{\tau} \cap \{R_0\leq r \leq R_1\}} |\snabla_{\s^2}^{k_1} T^{k_2}Y^{k_3}\psi|^2 d\omega d\uprho\right]\,d\tau\leq C \sum_{m\leq n}\int_{\Sigma_{\tau_1}}  \mathbf{J}^{N}[T^m\psi]\cdot \mathbf{n}_{\tau_1}\, r^2d\omega d \uprho.
 \end{equation}
 Furthermore, for all $\delta>0$, there exists $C=C(n,M,a)>0$, such that
 \begin{equation}
\label{eq:eboundradfield2}
\int_{D^{\tau_2}_{\tau_1}} r^{-1-\delta}\left[ (\Lbar \phi)^2+(L \phi)^2+r^{-2}\phi^2+r^{-2 }|\snabla_{\s^2}\phi|^2\right]\,d\omega d\rho d\tau \leq  C \int_{\Sigma_{\tau_1}}  \mathbf{J}^{N}[\psi]\cdot \mathbf{n}_{\tau_1}\, r^2d\omega d \uprho.
\end{equation}
 \item Let $n\in \N_0$. For $R_0>0$ arbitrarily large, there exists a constant $C=C(n,R_0,M,a)>0$, such that for any $\tau_2\geq \tau_1\geq 0$:
 \begin{equation}
\label{eq:morawetz2}
\sum_{0\leq k_1+k_2+k_3\leq n}\int_{\tau_1}^{\tau_2}\left[ \int_{\Sigma_{\tau} \cap \{r\leq R_0\}} |\snabla_{\s^2}^{k_1} T^{k_2}Y^{k_3}\psi|^2 d\omega d\uprho\right]\,d\tau\leq C \sum_{k_1+k_2\leq n+1}\int_{\Sigma_0}  \mathbf{J}^{N}[T^{k_1}\Phi^{k_2}\psi]\cdot \mathbf{n}_{\tau}\, r^2d\omega d \uprho.
 \end{equation}
 \end{enumerate}
\end{thmx}

We note moreover that by definition of $\Sigma_{\tau}$ in Section \ref{sec:foliations}, and the choice $h_0\neq 0$ in particular, we can estimate
\begin{equation*}
 \mathbf{J}^{N}[\psi]\cdot \mathbf{n}_{\tau}\sim (X\psi)^2+r^{-2}(T\psi)^2+r^{-2}|\snabla _{\s^2}\psi|^2 .
\end{equation*}
It follows immediately by combining \eqref{eq:ebound} and \eqref{eq:morawetz1} and applying \eqref{eq:hardyX} with $p=0$ that for all $0\leq \tau_1\leq \tau_2$:
\begin{equation}
\label{eq:eboundradfield}
\int_{N_{\tau_2}} (L\phi)^2+r^{-2}(\underline{L} \phi)^2+r^{-2}|\snabla_{\s^2}\phi|^2\,d\omega d\rho+\int_{I_{s}\cap D^{\tau_2}_{\tau_1}} (\Lbar \phi)^2+r^{-2}(L \phi)^2+r^{-2}|\snabla_{\s^2}\phi|^2\,d\omega d\rho \leq  C \int_{\Sigma_{\tau_1}}  \mathbf{J}^{N}[\psi]\cdot \mathbf{n}_{\tau_1}\, r^2d\omega d \uprho.
\end{equation}

In order to obtain more refined energy decay estimates for the projections $\psi_{\geq \ell}$, we need additional integrated local energy decay estimates for $\psi_{\geq \ell}$.

\begin{proposition}
 Let $\psi$ be a solution to \eqref{eq:waveeq} arising from suitably regular initial data with respect to $\widehat{\mathcal{R}}$. Then the following estimates hold for $\ell\in \N$:
 \begin{enumerate}[\rm (i)]
 \item
For $R_0>2M$ suitably large, there exists a constant $C=C(M,a,R _0)>0$, such that for any $\tau_2\geq \tau_1\geq 0$:
 \begin{equation}
 \begin{split}
 \label{eq:morawetzell1}
\int_{\Sigma_{\tau_2}}& J^N[\psi_{\geq \ell}]\cdot  \mathbf{n}_{\tau_2}\, r^2d\omega d \uprho\\
&+\int_{\tau_1}^{\infty}\left[ \int_{\Sigma_{\tau}\cap \{r\geq R_0\}}  M^{\delta}r^{1-\delta}[  (T\psi_{\geq \ell})^2+(Y\psi_{\geq \ell})^2]+r^{-1}  |\snabla_{\s^2}\psi_{\geq \ell}|^2+M^{\delta}r^{-1-\delta}\psi_{\geq \ell}^2\,d\omega d\uprho\right]d\tau \\
\leq&\: Ca^2\int_{\tau_1}^{\infty}\left[ \int_{\Sigma_{\tau}}  M^{-\delta}r^{-1+\delta} \left[(T^2\psi_{\geq \max \{\ell-2,0\}})^2+M^2(T^3\psi_{\geq \max \{\ell-2,0\}})^2+(T^2\Phi\psi_{\geq \max \{\ell-2,0\}})^2\right]\,d\omega d \uprho\right] d\tau\\
&+C\int_{\Sigma_{\tau_1}} \left[\mathbf{J}^{N}[\psi_{\geq\ell}]\cdot \mathbf{n}_{\tau_1}+a^2\mathbf{J}^{N}[T^2\psi_{\geq \max \{0,\ell-2\}}]\cdot \mathbf{n}_{\tau_1}\right]\, r^2d\omega d \uprho.
\end{split}
 \end{equation}
 \item There exists a constant $C=C(M,a)>0$, such that for any $\tau_2\geq \tau_1\geq 0$:
 \begin{equation}
 \begin{split}
\label {eq:morawetzell2}
\int_{\tau_1}^{\infty}&\left[ \int_{\Sigma_{\tau}}  M^{\delta}r^{1-\delta}[ (T\psi_{\geq \ell})^2+(Y\psi_{\geq \ell})^2]+r^{-1} |\snabla_{\s^2}\psi_{\geq \ell}|^2+M^{\delta}r^{-1-\delta}\psi_{\geq \ell}^2\,d\omega d\uprho\right]d\tau \\
\leq&\: Ca^2 \sum_{k_1+k_2\leq 2} \int_{\tau_1}^{\tau_2}\left[ \int_{\Sigma_{\tau}}  M^{-\delta}r^{-1+\delta}(T^{2+k_1}\Phi^{k_2}\psi_{\geq \max \{\ell-2,0\}})^2\,d\omega d \uprho\right] d\tau\\
&+C\sum_{k_1+k_2\leq 1}\int_{\Sigma_{\tau_1}} M^{k_1} \mathbf{J}^{N}[T^{k_1}\Phi^{k_2}\psi_{\geq\ell}]\cdot \mathbf{n}_{\tau_1}+a^2M^{k_1}\mathbf{J}^{N}[T^{2+k_1}\Phi^{k_2}\psi_{\geq \max \{0,\ell-2\}}]\cdot \mathbf{n}_{\tau_1}\, r^2d\omega d \uprho.
\end{split}
 \end{equation}
 \item Let $n\in \N_0$. For $R_0>2M$ suitably large and $R_1>R_0$, there exists a constant $C=C(M,a,n,R_0,R_1)>0$, such that for any $\tau_1\geq 0$:
  \begin{equation}
 \label{eq:morawetzell3}
  \begin{split}
&\sum_{0\leq k_1+k_2+k_3\leq n}\int_{\tau_1}^{\infty}\left[ \int_{\Sigma_{\tau} \cap \{R_0\leq r \leq R_1\}}|\snabla_{\s^2}^{k_1} T^{k_2}Y^{k_3}\psi_{\geq \ell}|^2 d\omega d\uprho\right]\,d\tau\\
\leq&\: C\sum_{k=0}^n  \left[\int_{\Sigma_{\tau_1}}J^N[T^{k}\psi_{\geq \ell}]\cdot \mathbf{n}_{\tau_1}\, r^2d\omega d \uprho+\sum_{\substack{1\leq l\leq \lceil \ell/2\rceil \\ 0\leq l_1+l_2\leq l}} a^2 \int_{\Sigma_{\tau_1}}J^N[T^{k+2l}T^{l_1}\Phi^{l_2}\psi_{\geq \max\{\ell-2l,0\}}\cdot \mathbf{n}_{\tau_1}\, r^2d\omega d \uprho\right].
\end{split}
 \end{equation}
 \end{enumerate}
\end{proposition}
\begin{proof}
Let $\psi$ be a solution to \eqref{eq:waveeq} arising from suitably regular initial data with respect to $\widehat{\mathcal{R}}$. Then $\psi_{\geq \ell}$ satisfies the following inhomogeneous equation:
\begin{equation*}
\rho^2\square_g\psi_{\geq \ell}=a^2T^2[\pi_{\geq \ell},\sin^2\theta]\psi=:F_{\geq \ell}.
\end{equation*}
We can generalize the Morawetz estimate in Proposition 9.1.1 of \cite{part3} to inhomogeneous wave equations where the inhomogeneity $F_{\geq \ell}$ is given by the above expression and we consider a hyperboloidal foliation by $\Sigma_{\tau}$ rather than an asymptotically flat foliation; see also Proposition 9.8.1 of \cite{part3}. We obtain the following estimate after assuming without loss of generality that $M=1$:
\begin{equation*}
\begin{split}
\int_{\tau_1}^{\infty}&\left[ \int_{\Sigma_{\tau}}  r^{1-\delta}[\zeta (T\psi_{\geq \ell})^2+(Y\psi_{\geq \ell})^2]+r^{-1}\zeta |\snabla_{\s^2}\psi_{\geq \ell}|^2+r^{-1-\delta}\psi_{\geq \ell}^2\,d\omega d\uprho\right]d\tau\\
\leq &\: C\int_{\tau_1}^{\infty}\left[ \int_{\Sigma_{\tau}} O_{\infty}(r^0) F_{\geq \ell}\cdot (T\psi_{\geq \ell} +\Phi \psi_{\geq \ell} +Y\psi_{\geq \ell} +r^{-1}\psi_{\geq \ell} )\,d\omega d\uprho\right]d\tau+C\int_{\Sigma_{\tau_1}}  \mathbf{J}^{N}[\psi_{\geq \ell}]\cdot \mathbf{n}_{0}\, r^2d\omega d \uprho,
\end{split}
\end{equation*}
with $\delta>0$ and $\zeta$ a smooth cut-off function that vanishes in an interval $[-s_{1,3M}+3M,3M+s_{2,3M}]$, with $-s_{1,3M}+3M>r_+$ and $s_{2,3M}<\infty$ depending on $a$ and $M$.

After integrating by parts in $T$ and taking $\epsilon>0$ arbitrarily small, we can further estimate:
\begin{equation*}
\begin{split}
\int_{\tau_1}^{\infty}&\left[ \int_{\Sigma_{\tau}} O_{\infty}(r^0) F_{\geq \ell}\cdot (T\psi_{\geq \ell} +\Phi \psi_{\geq \ell} +Y\psi +\psi_{\geq \ell} )\,d\omega d\uprho\right]d\tau\\
 \leq&\: \epsilon \int_{\tau_1}^{\infty}\left[ \int_{\Sigma_{\tau}}  r^{1-\delta}[\zeta (T\psi_{\geq \ell})^2+(Y\psi_{\geq \ell})^2]+r^{-1}\zeta |\snabla_{\s^2}\psi_{\geq \ell}|^2+r^{-1-\delta}\psi_{\geq \ell}^2\,d\omega d\uprho\right]d\tau\\
 &+C\epsilon^{-1}\int_{\tau_1}^{\infty}\left[ \int_{\Sigma_{\tau}} r^{-1+\delta} F_{\geq \ell}^2\,d\omega d\uprho\right]d\tau\\
 &+C\epsilon^{-1}\int_{\tau_1}^{\infty}\left[ \int_{\Sigma_{\tau}\cap \{-s_{3,3M}+3M\leq r\leq 3M+s_{4,3M}\}} O_{\infty}(r^0) (TF_{\geq \ell}+\Phi F_{\geq \ell})\cdot \psi_{\geq \ell}\,d\omega d\uprho\right]d\tau\\
 &+C\epsilon^{-1}\int_{\Sigma_{\tau_1}\cap \{-s_{3,3M}+3M\leq r\leq 3M+s_{4,3M}\}}|F_{\geq \ell}||\psi_{\geq \ell}|\,d\omega d\uprho,
\end{split}
\end{equation*}
with suitable $s_{1,3M}<s_{3,3M}<(3M-r_+)$ and $s_{2,3M}<s_{4,3M}<\infty$. We have that
\begin{equation*}
\int_{\s^2}F_{\geq \ell}^2 \,d\omega \leq C \int_{\s^2}(T^2\psi_{\geq \ell-2})^2\,d\omega.
\end{equation*}

We obtain \eqref{eq:morawetzell1} by combining the above estimates. The estimate \eqref{eq:morawetzell2} follows by additionally applying the above estimates with $\psi$ replaced by $T\psi$ or $\Phi\psi$ (using that $T$ and $\Phi$ are Killing vector fields), to remove the degenerate factor $\zeta$ on the left-hand side. The estimate \eqref{eq:morawetzell3} follows from standard higher-order elliptic estimates in the far-away region $2M<R_0<r<R_1$, where $T$ is timelike, together with a repeated application of \eqref{eq:morawetzell1}, using moreover that all estimates apply to $T^k\psi$ replacing $\psi$.
\end{proof}

\section{Newman--Penrose charges}
\label{sec:defNPconstants}

In this section we derive conservation laws for weighted derivatives of $\phi$ along $\mathcal{I}^+$ by constructing Newman--Penrose charges in Kerr spacetimes. These Newman--Penrose charges will play a key role in the analysis in the remainder of the paper, both when deriving sharp decay estimates and precise late-time asymptotics for solutions to \eqref{eq:waveeq}.

We first introduce the following \emph{renormalized derivatives} of $\phi$:
\begin{align}
\label{eq:checkphi1}
\chphi^{(1)}:=&\:\left[1+(\alpha+\alpha_{\Phi}\Phi ) r^{-1}+(\beta+\beta_{\Phi}\Phi+\beta_{\Phi^2}\Phi^2 ) r^{-2}\right]\frac{2(r^2+a^2)^2}{\Delta}L\phi,\\
\label{eq:checkphi2}
\chphi^{(2)}:=&\:\left[1+(\gamma+\gamma_{\Phi}\Phi ) r^{-1}\right]\frac{2(r^2+a^2)^2}{\Delta}L\chphi^{(1)},
\end{align}
with 
\begin{align*}
\alpha=&\:-M, && \beta= -\frac{1}{3}M^2-\frac{1}{6}a^2, && \gamma=-2M,\\
\alpha_{\Phi}=&\: a, && \beta_{\Phi}= 0,\: \beta_{\Phi^2}=\frac{1}{3}a^2 && \gamma_{\Phi}=0.
\end{align*}

It will be useful to consider a decompositions into \emph{azimuthal modes} $\psi=\sum_{m\in \Z}(\psi)_m$, where $(\psi)_m$ can be expressed via the following spherical harmonic decomposition:
\begin{equation*}
(\psi)_m=\sum_{\ell=|m|}^{\infty} \psi_{\ell,m}Y_{\ell,m}.
\end{equation*}
Note that $\Phi (\psi)_m=im (\psi)_m$ and the azimuthal modes $(\psi)_m$ are \emph{decoupled}, i.e.\ they independently satisfy \eqref{eq:waveeq}: $\square_{g_{M,a}}(\psi)_m=0$, so it is possible to consider each $m$-th azimuthal mode independently.

We can then split
\begin{equation*}
\phi_1=\pi_1(\phi)=\sum_{m=-1}^1 (\phi_1)_m\quad \textnormal{and}\quad  \phi_2=\pi_2(\phi)=\sum_{m=-2}^2 (\phi_2)_m.
\end{equation*}
with $\Phi(\phi_j)_m=im (\phi_j)_m$ for $j=1,2$. When restricted to fixed azimuthal modes, \eqref{eq:checkphi1} and \eqref{eq:checkphi2} take the form:
\begin{align*}
(\chphi^{(1)})_m:=&\:\left[1+(\alpha+im \alpha_{\Phi}) r^{-1}+(\beta+im \beta_{\Phi}-m^2\beta_{\Phi^2} ) r^{-2}\right]\frac{2(r^2+a^2)^2}{\Delta}L(\phi)_m,\\
(\chphi^{(2)})_m:=&\:\left[1+(\gamma+im \gamma_{\Phi} ) r^{-1}\right]\frac{2(r^2+a^2)^2}{\Delta}L(\chphi^{(1)})_m.
\end{align*}
In a slight abuse of notation, when restricting to fixed azimuthal modes with azimuthal number $m$, we will therefore also use $\alpha$, $\beta$ and $\gamma$ to denote also the complex numbers $\alpha+im \alpha_{\Phi}$, $\beta+im \beta_{\Phi}+m^2\beta_{\Phi^2}$ and $\gamma+im \gamma_{\Phi}$, respectively.

\begin{proposition}
\label{prop:maineqchph}
Let $\psi\in C^{\infty}(\mathcal{R}\to \C)$ denote a solution to \eqref{eq:waveeq} that is supported on a fixed azimuthal mode with azimuthal number $m$, i.e. it satisfies $\Phi \psi =im \psi $. Then:
\begin{equation}
\label{eq:maineqchph1l1}
\begin{split}
4\frac{(r^2+a^2)^2}{\Delta}\underline{L}L\chphi^{(1)}_1=&\:4\frac{(r^2+a^2)^2}{\Delta}\underline{L}L\pi_1(\chphi^{(1)})=\left[ a^2\pi_1(\sin^2\theta T^2\chphi^{(1)})+2a T\Phi\chphi^{(1)}_1\right]\\
&+[-\alpha-\alpha_{\Phi}\Phi+O_{\infty}(r^{-1})] [a^2\pi_1(\sin^2\theta T^2\phi)+2a T\Phi\phi_1]\\
&+[-4r+O_{\infty} (1)]L\chphi^{(1)}_1+O_{\infty}(r^{-1})\chphi^{(1)}_1+O_{\infty}(r^{-1})\Phi\chphi^{(1)}_1 +O_{\infty}(r^{-1})\phi_1
\end{split}
\end{equation}
and
\begin{equation}
\label{eq:maineqchph2l2}
\begin{split}
4\frac{(r^2+a^2)^2}{\Delta}\underline{L}L\chphi^{(2)}_2=&\:\left[ a^2\pi_2(\sin^2\theta T^2\chphi^{(2)})+2a T\Phi\chphi^{(2)}_2\right]\\
&+[M-\gamma-a \Phi+O_{\infty}(r^{-1})][a^2\pi_2(\sin^2\theta T^2\chphi^{(1)})+2a T\Phi\chphi^{(1)}_2]\\
&+[2\beta+2\beta_{\Phi}\Phi+2\beta_{\Phi^2}\Phi^2+\gamma(a\Phi-M)+O_{\infty}(r^{-1})][a^2\pi_2(\sin^2\theta T^2\phi)+2aT\Phi\phi_2]\\
&+[-8r+O_{\infty}(r^0)]L\chphi^{(2)}_2+O_{\infty}(r^{-1}) \chphi^{(2)}_2+O_{\infty}(r^{-1}) \Phi \chphi^{(2)}_2\\
&+O_{\infty}(r^{-1}) \Phi\chphi^{(1)}_2+O_{\infty}(r^{-1}) \chphi^{(1)}_2+O_{\infty}(r^{-1}) \Phi\phi_2+O_{\infty}(r^{-1}) \phi_2,
\end{split}
\end{equation}
where we allow the constants in the terms in $O_{\infty}(r^{-k})$, $k=0,1$, to depend also on $m$.
\end{proposition}
\begin{proof}
We have that
\begin{equation*}
\chphi^{(1)}=(1+\alpha r^{-1}+\beta r^{-2})\frac{2(r^2+a^2)^2}{\Delta}L\phi,
\end{equation*}
and therefore
\begin{equation*}
\begin{split}
4\underline{L}\chphi^{(1)}=&\:\overbrace{2\underline{L}\left((1+\alpha r^{-1}+\beta r^{-2})\frac{2(r^2+a^2)^2}{\Delta}\right)(1+\alpha r^{-1}+\beta r^{-2})^{-1}\frac{\Delta}{(r^2+a^2)^2} \chphi^{(1)}}^{=:T_1}\\
&+\overbrace{(1+\alpha r^{-1}+\beta r^{-2})\frac{2(r^2+a^2)^2}{\Delta}4\Lbar L\phi}^{=:T_2}.
\end{split}
\end{equation*}
We can write
\begin{equation*}
T_1= \frac{\Delta}{(r^2+a^2)^2}[-4r+2(\alpha+2M)+O_{\infty}(r^{-1})]\chphi^{(1)},
\end{equation*}
where we allow the terms in $O_{\infty}(r^{-1})$ to depend on $\alpha$ and $\beta$.

Hence,
\begin{equation*}
L(T_1)=\frac{\Delta}{(r^2+a^2)^2}[-4r+2(\alpha+2M)+O_{\infty}(r^{-1})]L\chphi^{(1)}+\frac{\Delta}{(r^2+a^2)^2} \left[2-(2\alpha+12M)r^{-1}+O_{\infty}(r^{-2})\right] \chphi^{(1)}.
\end{equation*}
Furthermore, we can rewrite \eqref{eq:phieq1} as follows:
\begin{equation*}
\begin{split}
4\underline{L}L\phi=&\:\frac{\Delta}{(r^2+a^2)^2}\left[ a^2\sin^2\theta T^2\phi+2a T\Phi\phi+  \slashed{\Delta}_{\s^2}\phi+2a  \frac{r}{r^2+a^2} \Phi \phi\right]\\ \nonumber
&+\frac{\Delta}{(r^2+a^2)^2}[-2Mr^{-1}-a^2r^{-2}+O_{\infty} (r^{-3})]\phi.
\end{split}
\end{equation*}
Therefore,
\begin{equation*}
T_2=2(1+\alpha r^{-1}+\beta r^{-2})\left[a^2\sin^2\theta T^2\phi+2a T\Phi\phi+  \slashed{\Delta}_{\s^2}\phi+2a  \frac{r}{r^2+a^2}\Phi\phi+(-2Mr^{-1}-a^2r^{-2}+O_{\infty} (r^{-3}))\phi\right].
\end{equation*}
We then obtain
\begin{equation*}
\begin{split}
L(T_2)=&\:\frac{\Delta}{(r^2+a^2)^2}\left[ a^2\sin^2\theta T^2\chphi^{(1)}+2a T\Phi\chphi^{(1)}+2a  \frac{r}{r^2+a^2} \Phi \chphi^{(1)}+ \slashed{\Delta}_{\s^2}\chphi^{(1)}\right]\\
&+\frac{\Delta}{(r^2+a^2)^2}[-2Mr^{-1}-a^2r^{-2}+O_{\infty} (r^{-3})]\chphi^{(1)}\\
&+\frac{\Delta}{(r^2+a^2)^2}[-\alpha-2\beta r^{-1}+O_{\infty}(r^{-2})] [a^2\sin^2\theta T^2\phi+2a T\Phi\phi+  \slashed{\Delta}_{\s^2}\phi]\\
&+\frac{\Delta}{(r^2+a^2)^2}[-2a-4a\alpha r^{-1}+O_{\infty}(r^{-2})]\Phi \phi+\frac{\Delta}{(r^2+a^2)^2}\left[2M+(2a^2+4 \alpha M)r^{-1}+O_{\infty}(r^{-2})\right]\phi
\end{split}
\end{equation*}
Combining the expressions for $L(T_1)$ and $L(T_2)$ and commuting $L$ and $\Lbar$, we obtain:
\begin{equation}
\label{eq:maineqchph1}
\begin{split}
4\frac{(r^2+a^2)^2}{\Delta}\underline{L}L\chphi^{(1)}=&\:\left[ a^2\sin^2\theta T^2\chphi^{(1)}+2a T\Phi\chphi^{(1)}+6a  \frac{r}{r^2+a^2} \Phi \chphi^{(1)}\right]\\
&+[-4r+2(\alpha+2M)+O_{\infty}(r^{-1})]L\chphi^{(1)}+\left[2+\slashed{\Delta}_{\s^2}\right]\chphi^{(1)}\\
&+[(-2\alpha-14M)r^{-1}+O_{\infty}(r^{-2})]\chphi^{(1)}\\
&+[-\alpha-2\beta r^{-1}+O_{\infty}(r^{-2})] [a^2\sin^2\theta T^2\phi+2a T\Phi\phi+  \slashed{\Delta}_{\s^2}\phi]\\
&+[-2a-4a\alpha r^{-1}+O_{\infty}(r^{-2})]\Phi \phi+\left[2M+(2a^2+4 \alpha M)r^{-1}+O_{\infty}(r^{-2})\right]\phi
\end{split}
\end{equation}
In order to guarantee that the right-hand side above either contains a $T$-derivative or vanishes when we act with $\pi_1$ on both sides and take $r\to \infty$, we use that
\begin{equation*}
\alpha=-M+i a m.
\end{equation*}

Now let
\begin{equation*}
\chphi^{(2)}=(1+\gamma r^{-1})\frac{2(r^2+a^2)^2}{\Delta}L\chphi^{(1)},
\end{equation*}
so
\begin{equation*}
4\underline{L}\chphi^{(2)}=\overbrace{2\underline{L}\left((1+\gamma r^{-1})\frac{2(r^2+a^2)^2}{\Delta}\right)(1+\gamma r^{-1})^{-1}\frac{\Delta}{(r^2+a^2)^2} \chphi^{(2)}}^{=:T_3}+\overbrace{(1+\gamma r^{-1})\frac{2(r^2+a^2)^2}{\Delta}4\Lbar L\chphi^{(1)}}^{=:T_4}.
\end{equation*}
Note that
\begin{equation*}
L(T_3)=\frac{\Delta}{(r^2+a^2)^2}[-4r+O_{\infty}(r^0)]L\chphi^{(2)}+\frac{\Delta}{(r^2+a^2)^2} \left[2+O_{\infty}(r^{-1})\right] \chphi^{(2)},
\end{equation*}
where the terms in $O_{\infty}(r^{-1})$ depend on $\gamma$.

By \eqref{eq:maineqchph1} with $\alpha=-M+iam$, we have that:
\begin{equation*}
\begin{split}
T_4=&\:2(1+\gamma r^{-1})[a^2\sin^2\theta T^2\check{\phi}^{(1)}+2aT\Phi \check{\phi}^{(1)}+6a r(r^2+a^2)^{-1}\Phi  \check{\phi}^{(1)}]\\
&+\Delta (r^2+a^2)^{-2}[-4r+2(M+iam)+O_{\infty}(r^{-1})]\check{\phi}^{(2)}+2(1+\gamma r^{-1})[2+\slashed{\Delta}_{\s^2}]\check{\phi}^{(1)}\\
&+2(1+\gamma r^{-1})[(-12M-2iam)r^{-1}+O_{\infty}(r^{-2})]\check{\phi}^{(1)}\\
&+2(1+\gamma r^{-1})[M-iam-2\beta r^{-1}+O_{\infty}(r^{-2})][a^2\sin^2\theta T^2\phi+2aT\Phi\phi+\slashed{\Delta}_{\s^2}\phi] \\
&+2(1+\gamma r^{-1})[-2a+(4Ma-4ima^2)r^{-1}+O_{\infty}(r^{-2})]\Phi\phi\\
&+2(1+\gamma r^{-1})[2M+(2a^2-4M^2+4iamM)r^{-1}+O_{\infty}(r^{-2})]\phi
\end{split}
\end{equation*}
and hence,
\begin{equation*}
\begin{split}
\frac{(r^2+a^2)^2}{\Delta} L(T_4)=&\:[a^2\sin^2\theta T^2\check{\phi}^{(2)}+2aT\Phi \check{\phi}^{(2)}+6a r(r^2+a^2)^{-1}\Phi  \check{\phi}^{(2)}]\\
&+ [-4r+O_{\infty}(r^0)]L\check{\phi}^{(2)}+[2+\slashed{\Delta}_{\s^2}+O_{\infty}(r^{-1})]\check{\phi}^{(2)}\\
&-[\gamma +O_{\infty}(r^{-1})][a^2\sin^2\theta T^2\check{\phi}^{(1)}+2aT\Phi \check{\phi}^{(1)}]-[6a+O_{\infty}(r^{-1}]\Phi \check{\phi}^{(1)}\\
&+[2+O_{\infty}(r^{-1})]\check{\phi}^{(2)}-[\gamma+O_{\infty}(r^{-1})](2+\slashed{\Delta}_{\s^2})\check{\phi}^{(1)}+[12M+2iam+O_{\infty}(r^{-1})]\check{\phi}^{(1)}\\
&+[M-iam+O_{\infty}(r^{-1})](a^2\sin^2\theta T^2\check{\phi}^{(1)}+2aT\Phi\check{\phi}^{(1)}+\slashed{\Delta}_{\s^2}\check{\phi}^{(1)})\\
&+[-2a+O_{\infty}(r^{-1})]\Phi \check{\phi}^{(1)}+[2M+O_{\infty}(r^{-1})]\check{\phi}^{(1)}\\
&+[-\gamma(M-iam)+2\beta+O_{\infty}(r^{-1})](a^2\sin^2\theta T^2\phi+2aT\Phi\phi+\slashed{\Delta}_{\s^2}\phi)\\
&+[2a \gamma-4Ma+4ima^2+O_{\infty}(r^{-1})]\Phi \phi\\
&+[-2M\gamma-(2a^2-4M^2+4iamM)+O_{\infty}(r^{-1})]\phi.
\end{split}
\end{equation*}
Combining the above expression with the expression for $L(T_3)$, using that $\Phi \phi=im\phi$ (and commuting $L$ and $\underline{L}$), we obtain
\begin{equation}
\label{eq:maineqchph2}
\begin{split}
4\frac{(r^2+a^2)^2}{\Delta}\underline{L}L\chphi^{(2)}=&\:\left[ a^2\sin^2\theta T^2\chphi^{(2)}+2a T\Phi\chphi^{(2)}+10a  \frac{r}{r^2+a^2} \Phi \chphi^{(2)}\right]\\
&+[-8r+O_{\infty}(r^0)]L\chphi^{(2)}+\left[6+\slashed{\Delta}_{\s^2}\right]\chphi^{(2)}+O_{\infty}(r^{-1}) \chphi^{(2)}\\
&+[M-aim-\gamma+O_{\infty}(r^{-1})][a^2\sin^2\theta T^2\chphi^{(1)}+2a T\Phi\chphi^{(1)}]\\
&+[-\gamma(2+\slashed{\Delta}_{\s^2})+14M-6iam+(M-iam)\slashed{\Delta}_{\s^2}+O_{\infty}(r^{-1})] \check{\phi}^{(1)}\\
&+[2\beta-M\gamma+iam\gamma+O_{\infty}(r^{-1})](a^2\sin^2\theta T^2\phi+2aT\Phi\phi)\\
&+[(-M\gamma+iam\gamma+2\beta)\slashed{\Delta}_{\s^2}-2M\gamma-2a^2+4M^2-4iamM+2iam\gamma-4iamM\\
&-4m^2a^2+O_{\infty}(r^{-1})]\phi
\end{split}
\end{equation}

In order to guarantee that the terms involving $\chphi^{(2)}$ and $\chphi^{(1)}$ on the right-hand side above either contains a $T$-derivative or vanishes when we act with $\pi_2$ on both sides and take $r\to \infty$, we use that:
\begin{equation*}
\gamma=-2M.
\end{equation*}
With the above value of $\gamma$, it follows that if we moreover use that \begin{equation*}
\beta=-\frac{1}{3}M^2-\frac{1}{6}a^2-\frac{1}{3}a^2 m^2,
\end{equation*}
then all the terms involving $\phi$ on the right-hand side above either contain a $T$-derivative or vanish when we act with $\pi_2$ on both sides and take $r\to \infty$.
\end{proof}

We define the following special linear combinations of derivatives of $\check{\phi}^{(j)}$, $j=0,1,2$:
\begin{align}
\label{eq:defP0}
\frac{(r^2+a^2)^2}{\Delta}P_0:=&\:\frac{(r^2+a^2)^2}{\Delta}L\phi_0-\frac{1}{4}a^2\pi_0(\sin^2\theta T\phi),\\
\label{eq:defP1}
\frac{(r^2+a^2)^2}{\Delta}P_1:=&\:\frac{(r^2+a^2)^2}{\Delta}L\chphi_1^{(1)}-\frac{1}{4}\left[ a^2\pi_1(\sin^2\theta T\chphi^{(1)})+2a \Phi\chphi^{(1)}_1\right]\\ \nonumber
&-\frac{1}{4}[-\alpha-\alpha_{\Phi}\Phi] [a^2\pi_1(\sin^2\theta T\phi)+2a \Phi\phi_1],\\
\label{eq:defP2}
\frac{(r^2+a^2)^2}{\Delta}P_2:=&\:\frac{(r^2+a^2)^2}{\Delta}L\chphi_2^{(2)}-\frac{1}{4}\left[ a^2\pi_2(\sin^2\theta T\chphi^{(2)})+2a \Phi\chphi^{(2)}_2\right]\\ \nonumber
&-\frac{1}{4}[M-\gamma-a \Phi][a^2\pi_2(\sin^2\theta T\chphi^{(1)})+2a \Phi\chphi^{(1)}_2]\\ \nonumber
&-\frac{1}{4}[2\beta+2\beta_{\Phi^2}\Phi^2+\gamma(a\Phi-M)][a^2\pi_2(\sin^2\theta T\phi)+2a\Phi\phi_2].
\end{align}

\begin{corollary}
Let $\psi\in C^{\infty}(\mathcal{R}\to \C)$ denote a solution to \eqref{eq:waveeq} that is supported on a fixed azimuthal mode with azimuthal number $m$, i.e. it satisfies $\Phi \psi =im \psi $. Then:
\begin{align}
\label{eq:maineqP0}
4\underline{L}P_0=&\:O_{\infty}(r^{-3})\phi_0+a^2O_{\infty}(r^{-3})T\pi_0(\sin^2\theta \phi )+a^2O_{\infty}(r^{-2})LT\pi_0(\sin^2\theta \phi ),\\
\label{eq:maineqP1}
4\underline{L}P_1=&\:[-4r^{-1}+O_{\infty}(r^{-2})]P_1+O_{\infty}(r^{-3})[\chphi^{(1)}_1+\phi_1]+a^2O_{\infty}(r^{-3})[T\pi_1(\sin^2\theta \chphi^{(1)}) +T\pi_1(\sin^2\theta \phi )]\\ \nonumber
&+a^2O_{\infty}(r^{-2})[LT\pi_1(\sin^2\theta \chphi^{(1)}) +LT\pi_1(\sin^2\theta \phi )],\\
\label{eq:maineqP2}
4\underline{L}P_2=&\:[-8r^{-1}+O_{\infty}(r^{-2})]P_2+O_{\infty}(r^{-3})[\chphi^{(2)}_2+\chphi^{(1)}_2+\phi_2]\\ \nonumber
&+a^2O_{\infty}(r^{-3})[T\pi_2(\sin^2\theta \chphi^{(2)})+T\pi_2(\sin^2\theta \chphi^{(1)}) +T\pi_2(\sin^2\theta \phi )]\\ \nonumber
&+a^2O_{\infty}(r^{-2})[LT\pi_2(\sin^2\theta \chphi^{(2)})+LT\pi_2(\sin^2\theta \chphi^{(1)}) +LT\pi_2(\sin^2\theta \phi )],
\end{align}
where we allow the constants in the terms in $O_{\infty}(r^{-k})$, $k=0,1$, to depend also on $m$.

Furthermore, if $\phi \in C^{4}(\widehat{\mathcal{R}})$, then
\begin{align}
\label{eq:consvP0}
\lim_{\uprho \to \infty} \underline{L}(r^2 P_0)(\tau,\uprho)=&\:0,\\
\label{eq:consvP1}
\lim_{\uprho \to \infty}  \underline{L}(r^2 P_1)(\tau,\uprho,\theta,\varphi_*)=&\:0 ,\\
\label{eq:consvP2}
\lim_{\uprho \to \infty} \underline{L}(r^2 P_2)(\tau,\uprho,\theta,\varphi_*)=&\:0.
\end{align}
\end{corollary}
\begin{proof}
	By the rewriting the equations in Proposition \ref{prop:maineqchph} in terms of the quantities $P_i$, $i=0,1,2$, using the definitions \eqref{eq:defP0}--\eqref{eq:defP2}, we obtain \eqref{eq:maineqP0}--\eqref{eq:maineqP2}.
	
	Now suppose $\phi \in C^{4}(\widehat{\mathcal{R}})$, then the limits on the left-hand sides of \eqref{eq:consvP0}--\eqref{eq:consvP2} are well-defined. By applying \eqref{eq:maineqP0}--\eqref{eq:maineqP2}, we moreover show that these limits must vanish.
\end{proof}

\begin{definition}
Let $\psi$ be a solution to such that  $\phi=\sqrt{r^2+a^2}\psi \in C^{4}(\widehat{\mathcal{R}})$. Then we define the \underline{\emph{Newman--Penrose}} \underline{\emph{charges}} as the following limits:
\begin{align*}
I_0[\psi]=&\:\lim_{\uprho \to \infty} 2r^2 P_0(\tau,\uprho),\\
I_1[\psi](\theta,\varphi_*)=&\:\lim_{\uprho \to \infty} 2r^2 P_1(\tau,\uprho,\theta,\varphi_*),\\
I_2[\psi](\theta,\varphi_*)=&\:\lim_{\uprho \to \infty} 2r^2 P_2(\tau,\uprho,\theta,\varphi_*),
\end{align*}
which are well-defined and \underline{conserved in $\tau$} by \eqref{eq:consvP0}--\eqref{eq:consvP2}.\end{definition}

\section{Hierarchies of $r^p$-weighted energy estimates}
\label{sec:rpest}
In this section, we derive the key weighted hierarchies of estimates that are involved in establishing sharp energy decay estimates.

We will use $R$ to denote the area radius that appears in the definitions of $D^{\tau_2}_{\tau_1}$ and $N_{\tau}$ in Section \ref{sec:foliations}. We will need to take $R>r_+$ appropriately large for the estimates in the sections below to hold.

\subsection{Hierarchies of $r^p$-weighted energy estimates for $\phi^{(n)}$}
\label{sec:rpestmain}
This section is concerned with establishing hierarchies of weighted energy estimates for the quantities $\phi^{(n)}$, which were introduced in Section \ref{sec:maineq}. In Proposition \ref{prop:rpest} below, we show that we can derive weighted energy estimates for a larger range of $n\in \N_0$, provided we restrict to $\phi^{(n)}_{\geq \ell}$ with $\ell$ appropriately large.
\begin{proposition}
\label{prop:rpest}
Let $n\in \N_0$ , $\ell\in \N_0$, such that $\ell \geq n$ and $-4n<p\leq 2$. For $R>r_+$ suitably large, there exists a constant $C=C(M,a,n,\ell,R,p)>0$, such that
\begin{equation}
\begin{split}
\label{eq:rpest}
\int_{N_{\tau_2}}& r^p(L\phi^{(n)}_{\geq \ell})^2+r^{p-4}[|\snabla_{\s^2}\phi^{(n)}_{\geq \ell}|^2-n(n+1)(\phi_{\geq n}^{(n)})^2]\,d\omega d\uprho\\
&+\int_{\tau_1}^{\tau_2}\left[\int_{N_{\tau}} r^{p-1}(L\phi^{(n)}_{\geq \ell})^2+(2-p)r^{p-3}\left(|\snabla_{\s^2}\phi^{(n)}_{\geq \ell}|^2-n(n+1)(\phi_{\geq \ell}^{(n)})^2+\delta_{n0}\delta_{\ell0}a^2\sin^2\theta (T\phi)^2\right)\,d\omega d\uprho\right]d\tau\\
\leq &\: C\int_{N_{\tau_1}} r^p(L\phi^{(n)}_{\geq \ell})^2+r^{p-4}\left( |\snabla_{\s^2}\phi^{(n)}_{\geq \ell}|^2-n(n+1)(\phi_{\geq \ell}^{(n)})^2\right)\,d\omega d\uprho\\
&+C(n+\ell)\int_{\tau_1}^{\tau_2}\left[\int_{N_{\tau}} r^{p-3}a^4(T^2\phi^{(n)}_{\geq \max\{\ell-2,0\}})^2+r^{p-3}a^2(T\Phi \phi^{(n)}_{\geq \ell})^2\,d\omega d\uprho\right]d\tau\\
&+C n \int_{\tau_1}^{\tau_2}\left[\int_{N_{\tau}} r^{p-1}(L \Phi\phi_{\geq \ell}^{(n-1)})^2+ r^{p-3}(LT\phi_{\geq \ell}^{(n-1)})^2\,d\omega d\uprho\right]d\tau\\
&+C n \sum_{m=0}^{n-1}\int_{\tau_1}^{\tau_2}\left[\int_{N_{\tau}} a^2r^{p-3}(\Phi\phi_{\geq \ell}^{(m)})^2+r^{p-3}(\phi_{\geq \ell}^{(m)})^2\,d\omega d\uprho\right]d\tau\\
&+ C\sum_{m=0}^n  \left[\int_{\Sigma_{\tau_1}}J^N[T^{m}\psi_{\geq \ell}]\cdot \mathbf{n}_{\tau_1}\, r^2d\omega d \uprho+a^2(1-\delta_{n0}\delta_{\ell 0})\sum_{\substack{1\leq l\leq \lceil \ell/2\rceil \\ 0\leq l_1+l_2\leq l}} \int_{\Sigma_{\tau_1}}J^N[T^{m+2l}T^{l_1}\Phi^{l_2}\psi_{\geq \max\{\ell-2k,0\}}\cdot \mathbf{n}_{\tau_1}\, r^2d\omega d \uprho\right]
\end{split}
\end{equation}
Furthermore, the estimate \eqref{eq:rpest} also holds with the last line replaced by
\begin{equation}
\label{eq:rpestvar}
C\sum_{m=0}^n  \int_{\Sigma_{\tau_1}}J^N[T^{m}\psi]\cdot \mathbf{n}_{\tau_1}\, r^2d\omega d \uprho.
\end{equation}
\end{proposition}
\begin{proof}
We apply $\pi_{\geq \ell}$ and then multiply both sides of \eqref{eq:commeq} by $-\frac{1}{2}r^{p-2}\frac{(r^2+a^2)^2}{\Delta}L\phi^{(n)}_{\geq \ell}$ to obtain:
\begin{equation*}
\begin{split}
0=&-\frac{1}{2}r^{p-2}L\phi^{(n)}_{\geq  \ell}\Bigg[-4\frac{(r^2+a^2)^2}{\Delta}\underline{L}L\phi^{(n)}_{\geq  \ell}+a^2 T^2 \pi_{\geq  \ell}(\sin^2\theta \phi^{(n)})+2a T\Phi\phi^{(n)}_{\geq  \ell}+2(1+2n)a  \frac{r}{r^2+a^2} \Phi \phi^{(n)}_{\geq  \ell}\\
&+ \slashed{\Delta}_{\s^2}\phi^{(n)}_{\geq  \ell}+[n(n+1)+O_{\infty}(r^{-1})] \phi^{(n)}_{\geq  \ell}-[4nr^{-1}+O_{\infty}(r^{-2})]\frac{(r^2+a^2)^2}{\Delta} L \phi^{(n)}_{\geq  \ell}\\
&+a \sum_{m=0}^{n-1}O_{\infty}(r^0)\Phi\phi^{(m)}_{\geq  \ell}+n\sum_{k=0}^{n-1}O_{\infty}(r^0)\phi^{(m)}_{\geq  \ell}\Bigg]\\
=:&\:L(\mathcal{F}^{r^pL}_L)+\Lbar(\mathcal{F}^{r^pL}_{\Lbar})+\slashed{\rm div}_{\s^2} \mathcal{F}_{\snabla}^{r^pL}+\Phi(\mathcal{F}^{r^pL}_{\Phi})+\mathcal{J}^{r^p_L}.
\end{split}
\end{equation*}
We will determine $\mathcal{F}^{r^pL}_L$, $\mathcal{F}^{r^pL}_{\underline{L}}$ and $\mathcal{J}^{r^p_L}$ below so that the second equality above holds.\\
\\
\underline{$2r^{p-2}\frac{(r^2+a^2)^2}{\Delta}L\phi^{(n)}_{\geq  \ell}\cdot \underline{L}L\phi^{(n)}_{\geq  \ell}$:}\\
\\
We can write
\begin{equation*}
2r^{p-2}\frac{(r^2+a^2)^2}{\Delta}L\phi^{(n)}_{\geq  \ell}\cdot \underline{L}L\phi^{(n)}_{\geq  \ell}=\underline{L}\left(r^{p-2}\frac{(r^2+a^2)^2}{\Delta}(L\phi^{(n)}_{\geq  \ell})^2\right)+\left(\frac{p}{2}+O_{\infty}(r^{-1})\right)r^{p-1}(L\phi^{(n)}_{\geq  \ell})^2.
\end{equation*}
\\
-\underline{$\frac{1}{2}r^{p-2}L\phi^{(n)}_{\geq  \ell}\cdot ( \slashed{\Delta}_{\s^2}+n(n+1))\phi^{(n)}_{\geq  \ell}$:}\\
\\
We can write
\begin{equation*}
\begin{split}
-\frac{1}{2}r^{p-2}L\phi^{(n)}_{\geq  \ell}\cdot ( \slashed{\Delta}_{\s^2}+n(n+1))\phi^{(n)}_{\geq  \ell}=&\:\frac{1}{4}L\left(r^{p-2}|\snabla_{\s^2}\phi^{(n)}_{\geq  \ell}|^2-n(n+1)r^{p-2}(\phi_{\geq  \ell}^{(n)})^2\right)\\
&+\frac{1}{8}(2-p)r^{p-3}\frac{\Delta }{r^2+a^2}\left[|\snabla_{\s^2}\phi^{(n)}_{\geq  \ell}|^2-n(n+1)(\phi_{\geq  \ell}^{(n)})^2\right].
\end{split}
\end{equation*}
\\
\underline{$-\frac{1}{2}r^{p-2}L\phi^{(n)}_{\geq  \ell}\cdot a^2 T^2 \pi_{\geq  \ell}(\sin^2\theta \phi^{(n)})$:}\\
\\
We only rewrite this term if $n=0$ and $\ell=0$, as $\pi_{\geq 0}=\textnormal{id}$ (for $n\neq 0$ or $\ell\neq 0$ we directly group it with $J^{r^PL}$. We can write
\begin{equation*}
\begin{split}
-\frac{1}{2}r^{p-2}L\phi \cdot a^2  \sin^2\theta T^2 \phi=&\:-T\left(\frac{1}{2}r^{p-2}L\phi \cdot a^2 \sin^2\theta T  \phi\right)+\frac{1}{4}r^{p-2} a^2 \sin^2\theta L((T\phi)^2)\\
=&\:-(L+\underline{L})\left(\frac{1}{2}r^{p-2}L\phi \cdot a^2 \sin^2\theta T  \phi\right)+L\left(\frac{1}{4}r^{p-2} a^2 \sin^2\theta (T\phi)^2\right)\\
&+\Phi\left(\frac{a^2}{2(r^2+a^2)}r^{p-2}L\phi \cdot a^2 \sin^2\theta T  \phi\right)+\frac{1}{8}(2-p)\frac{\Delta}{r^2+a^2}r^{p-3} a^2 \sin^2\theta(T\phi)^2.
\end{split}
\end{equation*}
\\
\underline{$-\frac{1}{2}r^{p-2}L\phi^{(n)}_{\geq  \ell}\cdot 2a T\Phi\phi^{(n)}_{\geq \ell}$:}\\
\\
Again, we only rewrite this term if $n=0$ and $\ell=0$:
\begin{equation*}
\begin{split}
-a r^{p-2}L \phi\cdot T\Phi \phi=&\:-a r^{p-2}L \phi\cdot (L+\Lbar -\frac{a}{r^2+a^2}\Phi)\Phi \phi=-\Phi\left(\frac{a}{2}r^{p-2} (L\phi)^2+\frac{a^2}{r^2+a^2} L\phi \Phi \phi\right)\\
&+r^{p-2}\frac{a^2}{2(r^2+a^2)}L((\Phi\phi)^2)-\Lbar(ar^{p-2} L\phi \Phi \phi)+\Lbar(ar^{p-2} L\phi)\Phi \phi\\
=&\:-\Phi\left(\frac{a}{2}r^{p-2} (L\phi)^2+\frac{a^2}{r^2+a^2} L\phi \Phi \phi\right)+L\left(r^{p-2}\frac{a^2}{2(r^2+a^2)}(\Phi\phi)^2\right)+ar^{p-2}\Phi \phi \Lbar L\phi\\
&+\left[\frac{1}{2}(4-p)a^2r^{p-5}+O_{\infty}(r^{p-6})\right](\Phi\phi)^2+O_{\infty}(r^{p-3}) L\phi \Phi\phi.
\end{split}
\end{equation*}
By using \eqref{eq:phieq1}, we can further write:
\begin{equation*}
\begin{split}
ar^{p-2}\Phi \phi \underline{L}L\phi=&\:\frac{1}{4}a^3r^{p-2}\sin^2\theta \frac{\Delta}{(r^2+a^2)^2}\Phi \phi T^2\phi+\frac{1}{2}a^2  \frac{\Delta}{(r^2+a^2)^2}r^{p-2}\Phi \phi T\Phi\phi+  \frac{a}{4}\frac{\Delta}{(r^2+a^2)^2}r^{p-2}\Phi \phi \slashed{\Delta}_{\s^2}\phi\\
&+\frac{1}{2}a^2  \frac{r^{p-1}\Delta}{(r^2+a^2)^3} (\Phi \phi)^2+ O_{\infty}(r^{p-6})\phi \Phi \phi.
\end{split}
\end{equation*}
We have that
\begin{equation*}
\begin{split}
\frac{1}{4}a^3r^{p-2}\sin^2\theta& \frac{\Delta}{(r^2+a^2)^2}\Phi \phi T^2\phi+\frac{1}{2}a^2  \frac{\Delta}{(r^2+a^2)^2}r^{p-2}\Phi \phi T\Phi\phi= T\left(\frac{1}{4}a^3r^{p-2}\sin^2\theta \frac{\Delta}{(r^2+a^2)^2}\Phi \phi T\phi \right)\\
&+T\left(\frac{1}{4}a^2  \frac{\Delta}{(r^2+a^2)^2}r^{p-2}(\Phi \phi)^2\right)+\Phi\left(\frac{1}{8}a^3r^{p-2}\sin^2\theta \frac{\Delta}{(r^2+a^2)^2}(T\phi)^2\right)
\end{split}
\end{equation*}
and
\begin{equation*}
 \frac{a}{4}\frac{\Delta}{(r^2+a^2)^2}r^{p-2}\Phi \phi \slashed{\Delta}_{\s^2}\phi=\slashed{\rm div}_{\s^2}\left( \frac{a}{4}\frac{\Delta}{(r^2+a^2)^2}r^{p-2}\Phi \phi \snabla_{\s^2}\phi\right)+\Phi\left( \frac{a}{8}\frac{\Delta}{(r^2+a^2)^2}r^{p-2}| \snabla_{\s^2}\phi|^2\right).
\end{equation*}

We combine the above identities to obtain:
\begin{align*}
\mathcal{F}^{r^pL}_L=&\:\frac{1}{4}r^{p-2}(|\snabla_{\s^2}\phi^{(n)}_{\geq  \ell}|^2-n(n+1)(\phi_{\geq  \ell}^{(n)})^2)+\delta_{n0}\delta_{\ell0}\Bigg[\frac{1}{4}a^2 \sin^2\theta r^{p-2}  (T\phi)^2-\frac{1}{2}r^{p-2} a^2 \sin^2\theta L\phi  T  \phi \\
&+r^{p-2}\frac{3a^2}{4(r^2+a^2)}(\Phi\phi)^2+\frac{1}{4}a^3r^{p-2}\sin^2\theta \frac{\Delta}{(r^2+a^2)^2}\Phi \phi T\phi\Bigg],\\
\mathcal{F}^{r^pL}_{\Lbar}=&\:r^{p-2}\frac{(r^2+a^2)^2}{\Delta}(L\phi^{(n)}_{\geq  \ell})^2+\delta_{n0}\delta_{\ell 0}\Bigg[-\frac{1}{2}r^{p-2}a^2 \sin^2\theta L\phi  T  \phi+\frac{1}{4}a^2  \frac{\Delta}{(r^2+a^2)^2}r^{p-2}(\Phi \phi)^2\\
&+\frac{1}{4}a^3r^{p-2}\sin^2\theta \frac{\Delta}{(r^2+a^2)^2}\Phi \phi T\phi\Bigg],\\
\mathcal{J}^{r^pL}=&\:\left(\frac{p+4n}{2}+O_{\infty}(r^{-1})\right)r^{p-1}(L\phi^{(n)}_{\geq  \ell})^2+\frac{1}{8}(2-p)r^{p-3}\frac{\Delta }{r^2+a^2}[|\snabla_{\s^2}\phi^{(n)}_{\geq  \ell}|^2-n(n+1)(\phi_{\geq  \ell}^{(n)})^2]\\
&+O_{\infty}(r^{p-3})\phi^{(n)}_{\geq  \ell} L\phi^{(n)}_{\geq \ell}+O_{\infty}(r^{p-3})\Phi\phi^{(n)}_{\geq  \ell} L\phi^{(n)}_{\geq \ell}\\
&+a (1-\delta_{n0} \delta_{\ell 0})\sum_{m=0}^{n-1}O_{\infty}(r^{p-2})\Phi\phi^{(m)}_{\geq \ell} L\phi^{(n)}_{\geq  \ell}+(1-\delta_{n0}\delta_{\ell 0})\sum_{m=0}^{n-1}O_{\infty}(r^{p-2})\phi^{(m)}_{\geq  \ell} L\phi^{(n)}_{\geq \ell}\\
&+\delta_{n0}\delta_{\ell 0}\left[\frac{1}{8}(2-p)\frac{\Delta}{r^2+a^2}r^{p-3} a^2 \sin^2\theta(T\phi)^2+\left(\frac{1}{2}(5-p)r^{p-5}+O_{\infty}(r^{p-6})\right)(\Phi \phi)^2+ O_{\infty}(r^{p-6})\phi \Phi \phi\right]\\
&+(1-\delta_{n0}\delta_{\ell 0})\left[-\frac{1}{2}a^2r^{p-2}L\phi^{(n)}_{\geq  \ell}  T^2 \pi_{\geq \ell}(\sin^2\theta \phi^{(n)})-ar^{p-2}L\phi^{(n)}_{\geq  \ell} T\Phi\phi^{(n)}_{\geq  \ell}\right].
\end{align*}

We will now apply Lemma \ref{lm:intbyparts}. We deal with the boundary term at $r=R$ by a standard averaging argument: we multiply the $\mathcal{F}^{r^pL}_{\square}$ terms with a suitable smooth cut-off function $\chi: [r_+,\infty)\to 0$, with $\chi(\uprho)=0$ for $r\leq R-M$ and $\chi(\uprho)=1$ for $r\geq R$, see for example the proof of Proposition 6.5 of \cite{paper4}. Note that when $n\geq 1$, the flux terms on $N_{\tau}$ and $I_{s}$ coming from the terms in $\mathcal{F}^{r^pL}_L$ and $\mathcal{F}^{r^pL}_{\Lbar}$ are non-negative definite for $R>0$ suitably large and $-4n<p\leq 2$. When $n=0$, we can easily estimate the flux terms without a good sign by applying \eqref{eq:eboundradfield}.

The terms without a sign in $\mathcal{J}^{r^pL}$ are estimated by applying a weighted Young's inequality. In the $n=\ell=0$ case we additionally apply \eqref{eq:morawetz1} and \eqref{eq:hardyL}. For example: we estimate
\begin{equation*}
r^{p-3}|\Phi\phi||L\phi^{(n)}|+r^{p-3}|\phi||L\phi^{(n)}|\leq \epsilon r^{p-1}(L\phi)^2+ C\epsilon^{-1} r^{p-5}(\Phi\phi)^2+ C\epsilon^{-1}r^{p-5}\phi^2
\end{equation*}
and apply \eqref{eq:morawetz1} to estimate the $(\Phi\phi)^2$ term for $p\leq 2$, whereas we estimate the $\phi$ term using \eqref{eq:hardyL} (after multiplying with the cut-off $\chi^2$):
\begin{equation*}
\begin{split}
\int_{\tau_1}^{\tau_2}&\int_{\Sigma_{\tau}\cap\{r\geq R-M\}} \epsilon^{-1}r^{p-5}\chi^2\phi^2\,d\omega d\uprho \leq C\int_{\tau_1}^{\tau_2}\int_{\Sigma_{\tau}\cap\{r\geq R-M\}} \epsilon^{-1}r^{p-3}\chi^2 (L\phi)^2+r^{p-7}\chi^2 (T\phi)^2\\
&+r^{p-7}\chi^2 (\Phi \phi)^2+r^{p-3}(\chi')^2 \phi^2\,d\omega d\uprho,
\end{split}
\end{equation*}
where we absorb the $(L\phi)^2$ term into the left-hand side and control the remaining terms with \eqref{eq:eboundradfield2} and \eqref{eq:morawetz1}.

In the $n\geq 1$ case, we instead apply the higher-order integrated local energy decay estimate \eqref{eq:morawetzell3} and we moreover estimate via \eqref{eq:hardyL}:
\begin{equation*}
\int_{\tau_1}^{\tau_2} \int_{N_{\tau}} r^{p-5}(\phi^{(n)})^2\,d\omega d\uprho d\tau\leq C\int_{\tau_1}^{\tau_2} \int_{N_{\tau}} r^{p-3}(L\phi^{(n)}_{\geq \ell})^2+r^{p-3}(LT\phi^{(n-1)}_{\geq \ell})^2+r^{p-3}(L\Phi \phi^{(n-1)}_{\geq \ell})^2\,d\omega d\uprho d\tau+\ldots,
\end{equation*}
where $\ldots$ denotes boundary terms at $r=R$ that can easily be absorbed via an averaging argument and \eqref{eq:morawetzell3}, as above.

The estimate \eqref{eq:rpest} then follows. Similarly, we obtain the alternate version of \eqref{eq:rpest}, with the right-hand side \eqref{eq:rpestvar} by applying \eqref{eq:morawetz2} instead of \eqref{eq:morawetzell3}.
\end{proof}
\subsection{Hierarchies of $r^p$-weighted energy estimates for $P_0$, $P_1$ and $P_2$}
\label{sec:rpestPi}
When restricting to $\psi_{\ell}$, with $\ell=0,1,2$, we can obtain weighted hierarchies of energy estimates with larger weights in $r$ compared to those derived in Section \ref{sec:rpestmain} by considering the quantities $P_i$, $i=0,1,2$ defined in Section \ref{sec:defNPconstants}.
\begin{proposition}
\label{prop:rpestNpquant}
Let $0<p<4$. Then the following estimates hold: for $R>r_+$ suitably large, there exists a constant $C(M,a,R,p)>0$ such that:
\begin{itemize}
\item[for $\ell=0$:]
\begin{equation*}
\begin{split}
\int_{N_{\tau_2}}& r^{p} (P_0)^2\,d\omega d\uprho+\int_{\tau_1}^{\tau_2}\left[ \int_{N_{\tau}} r^{p-1} (P_0^2+(L\phi_0)^2)\,d\omega d\uprho\right]d\tau\leq C\int_{N_{\tau_1}} r^{p} (P_0)^2\,d\omega d\uprho\\
&+C\int_{\tau_1}^{\tau_2}\left[ \int_{N_{\tau}} a^2 r^{p-3}(LT \phi)^2+r^{p-5}(T\phi)^2\,d\omega d\uprho\right]d\tau+C\int_{\Sigma_{\tau_1}}\mathbf{J}^T[\psi]\cdot \mathbf{n}_{\tau_1}\,r^2d\omega d\uprho,
\end{split}
\end{equation*}
\item[for $\ell=1$:]
\begin{equation*}
\begin{split}
\int_{N_{\tau_2}}& r^{p} (P_1)^2\,d\omega d\uprho+\int_{\tau_1}^{\tau_2}\left[ \int_{N_{\tau}} r^{p-1} (P_1^2+(L\check{\phi}_1^{(1)})^2)\,d\omega d\uprho\right]d\tau\leq C\int_{N_{\tau_1}} r^{p} (P_1)^2\,d\omega d\uprho\\
&+Ca^2 \sum_{i=0}^1\int_{\tau_1}^{\tau_2}\left[ \int_{N_{\tau}} r^{p-3}(LT \phi^{(i)}_{\geq 3})^2+r^{p-5}(T \phi^{(i)}_{\geq 3})^2\,d\omega d\uprho\right]d\tau\\
&+C\sum_{j=0}^1\int_{\tau_1}^{\tau_2}\left[ \int_{N_{\tau}} r^{p-1}(L(rL)^jT\phi_1)^2+r^{p-5}(T\phi_1)^2\,d\omega d\uprho\right]d\tau+C\sum_{m=0}^1\int_{\Sigma_{\tau_1}}\mathbf{J}^T[T^m\psi]\cdot \mathbf{n}_{\tau_1}\,r^2d\omega d\uprho,
\end{split}
\end{equation*}
\item[for $\ell=2$:]
\begin{equation*}
\begin{split}
\int_{N_{\tau_2}}& r^{p} (P_2)^2\,d\omega d\uprho+\int_{\tau_1}^{\tau_2}\left[ \int_{N_{\tau}} r^{p-1} (P_2^2+(L\check{\phi}_2^{(2)})^2)\,d\omega d\uprho\right]d\tau\leq C\int_{N_{\tau_1}} r^{p} (P_2)^2\,d\omega d\uprho\\
&+Ca^2\sum_{i=0}^2\int_{\tau_1}^{\tau_2}\left[ \int_{N_{\tau}} r^{p-3}(LT \phi^{(i)}_{\geq 4})^2+r^{p-5}(T\phi^{(i)}_{\geq 4})^2\,d\omega d\uprho\right]d\tau\\
&+C\sum_{i=0}^1\sum_{j=0}^1\int_{\tau_1}^{\tau_2}\left[ \int_{N_{\tau}} r^{p-1}(L(rL)^jT\phi^{(i)}_2)^2+r^{p-5}(T\phi_2)^2\,d\omega d\uprho\right]d\tau\\
&+Ca^2\sum_{j=0}^2\int_{\tau_1}^{\tau_2}\left[ \int_{N_{\tau}} r^{p+1} (L(rL)^jT\phi_0)^2+r^{p-5}(T\phi_0)^2\,d\omega d\uprho\right]d\tau+C\sum_{m=0}^2\int_{\Sigma_{\tau_1}}\mathbf{J}^T[T^m\psi]\cdot \mathbf{n}_{\tau_1}\,r^2d\omega d\uprho.
\end{split}
\end{equation*}
\end{itemize}
\end{proposition}
\begin{proof}
Note first of all that for any function $f\in L^2(\s^2)$, we can decompose $f=\sum_{m\in \Z}f_m$, with $\Phi f_m=im f_m$ such that
\begin{equation*}
\int_{\s^2} |f|^2\,d\omega=\sum_{m\in \Z} \int_{\s^2} |f_m|^2\,d\omega.
\end{equation*}

We take the azimuthal modes $\psi_m$ to satisfy the above equality and multiply both sides of \eqref{eq:maineqP0}--\eqref{eq:maineqP2} with $-\frac{1}{2}P_j$, $j=0,1,2$, sum over $m=-j,\ldots,j$, respectively, and apply Lemma \ref{eq:lprojsin} to obtain the following (schematic) identities:
\begin{align*}
0=&\:\underline{L}(r^p P_0^2)+\left[\frac{p}{2}+O_{\infty}(r^{-1})\right]r^{p-1}P_0^2+O_{\infty}(r^{p-3})P_0\left[\phi_0+a^2rLT\phi_{0}+a^2T\phi_{0}+a^2rLT\phi_{2}+a^2T\phi_{2}\right],\\
0=&\:\underline{L}(r^p P_1^2)+\left[\frac{p+4}{2}+O_{\infty}(r^{-1})\right]r^{p-1}P_1^2+O_{\infty}(r^{p-3})P_1\left[\check{\phi}^{(1)}_1+\phi_1+a^2rLT\check{\phi}^{(1)}_{1,3}+a^2rLT{\phi}_{1,3}+a^2T\check{\phi}^{(1)}_{1,3}+a^2T\phi_{1,3}\right],\\
0=&\:\underline{L}(r^p P_2^2)+\left[\frac{p+8}{2}+O_{\infty}(r^{-1})\right]r^{p-1}P_2^2+O_{\infty}(r^{p-3})P_2\sum_{i=1}^2\Big[\check{\phi}^{(i)}_{2}+\phi_{2}+a^2rLT\check{\phi}^{(2)}_{0,2,4}+a^2rLT{\phi}_{0,2,4}\\
&+a^2T\check{\phi}^{(i)}_{0,2,4}+a^2T\phi_{0,2,4}\Big].
\end{align*}

Note that we can estimate (suppressing $\Phi$ derivatives):
\begin{align*}
(L\phi_0)^2\leq&\: CP_0^2+ Cr^{-4}(T\phi_{2})^2,\\
(L\check{\phi}^{(1)}_1)^2\leq&\: CP_1^2+ C\sum_{i=0}^1 r^{-4}(\phi_1^{(i)})^2+r^{-4}(T\phi_{3}^{(i)})^2+r^{-4}(T\phi_{1}^{(i)})^2\\
\leq &\: CP_1^2+ C\sum_{i=0}^1 r^{-4}(\phi_1^{(i)})^2+r^{-4}(T\phi_{3}^{(i)})^2+r^{-4}(T\phi_{1})^2+(LT\phi_1)^2,\\
(L\check{\phi}^{(2)}_2)^2\leq&\: CP_2^2+  C\sum_{i=0}^2 r^{-4}( \phi_{2}^{(i)})^2+r^{-4}(T\phi_{0}^{(i)})^2+ r^{-4}(T\phi_{2}^{(i)})^2+r^{-4}(T\phi_{4}^{(i)})^2\\
\leq &\: CP_2^2+  C[r^{-4}(\check{\phi}_{2}^{(2)})^2+r^{-4}(T\phi_{4}^{(2)})^2]\\
&+C\sum_{i=0}^1 r^{-4}(\phi_{2}^{(i)})^2+r^{-4}(T\phi_{2,4}^{(i)})^2+(LT\phi_2^{(1)})^2+[r^4 (L^2T\phi_0)^2+r^2(LT\phi_0)^2+r^{-4}(T\phi_0)^2].
\end{align*}

Furthermore, we apply \eqref{eq:hardyL}, to estimate for suitably large $R>0$
\begin{align*}
\int_{N_{\tau}}r^{p-5}\phi_0^2\,d\omega d\uprho\leq&\: C\int_{N_{\tau}}r^{p-3}(L\phi_0)^2+r^{p-7}(T\phi_0)^2\,d\omega d\uprho+\ldots,\\
\int_{N_{\tau}}r^{p-5}[(\check{\phi}_1^{(1)})^2+\phi_1^2]\,d\omega d\uprho\leq&\: C\int_{N_{\tau}}r^{p-3}(L\check{\phi}_1^{(1)})^2+r^{p-3}(LT\phi_1)^2+r^{p-7}(T\phi_1)^2\,d\omega d\uprho+\ldots,\\
\int_{N_{\tau}}r^{p-5}[(\check{\phi}_2^{(2)})^2+(\phi_2^{(1)})^2+\phi_2^2]\,d\omega d\uprho\leq&\: C\int_{N_{\tau}}r^{p-3}(L\check{\phi}_2^{(2)})^2+r^{p-3}(LT\phi^{(1)}_2)^2+r^{p-3}(LT\phi_2)^2+r^{p-7}(T\phi_2)^2\,d\omega d\uprho+\ldots,
\end{align*}
where $\ldots$ denotes omitted terms arising from averaging the boundary terms at $r=R$ via a cut-off function.

We conclude the proof by applying Lemma \ref{lm:intbyparts} together with the above identities and estimates.
\end{proof}
\begin{remark}
Observe that the integrals of $r^{p-5}(T\phi)^2$ on the right-hand sides of of the estimates in Proposition \ref{prop:rpestNpquant} can be further estimated for $p<4$ using the energy boundedness property in \eqref{eq:eboundradfield2}.
\end{remark}

\subsection{Additional hierarchies for higher-order derivatives}
In this section we extend the hierarchies of $r^p$-weighted estimates established in Sections \ref{sec:rpestmain} and \ref{sec:rpestPi} to higher-order quantities of the form $\snabla_{\s^2}^{\alpha}(rL)^kT^m\phi^{(n)}$ and $(rL)^kT^mP_j$ by commuting the wave operator $\square_g$ with the vector fields $rL$, $T$ and the angular derivative operator $\snabla_{\s^2}$. 

Note that $T$ is a Killing vector field, so $T^k\psi$ is a solution to \eqref{eq:waveeq} for any $k\in \N_0$ and the estimates derived in Sections \ref{sec:rpestmain} and \ref{sec:rpestPi} immediately apply when $\phi$ is replaced by $T^k\phi$.

The operator $\square_g$, however,  does \emph{not} commute with $rL$ and $\snabla_{\s^2}$, so the aim of the this section is to show that nevertheless analogues of the $r^p$-weighted energy estimates derived in Sections \ref{sec:rpestmain} and \ref{sec:rpestPi} still hold for $\snabla_{\s^2}^{\alpha}(rL)^k\phi^{(n)}$.

The higher-order hierarchies of $r^p$-weighted estimates derived in this section are essential for dealing with the coupling of spherical harmonic modes when proving improved energy decay estimates for the projections $\psi_{\geq \ell}$, with $\ell\geq 1$.

\subsubsection{Commuting with $(rL)^k$}
\label{sec:commrL}
We will show in this section that $r^p$-weighted energy estimates that are analogous to those derived in Sections \ref{sec:rpestmain} and \ref{sec:rpestPi} still hold when $\phi$ is replaced by the higher-order quantities $(rL)^k\phi$, with $k\in \N$.
\begin{lemma}
\label{lm:commeqho}
Let $\psi\in C^{\infty}(\mathcal{R})$ be a solution to \eqref{eq:waveeq}. Let $n,k\in \N_0$. Then $(rL)^k\phi^{(n)}$ satisfies the following equation
\begin{equation}
\label{eq:commeqho}
\begin{split}
4\underline{L}L(rL)^{k}\phi^{(n)}= &\:a^2\sin^2\theta \frac{\Delta}{(r^2+a^2)^2}T^2 (rL)^k\phi^{(n)}+2a  \frac{\Delta}{(r^2+a^2)^2}T\Phi (rL)^k\phi^{(n)}+\Delta (r^2+a^2)^{-2} \slashed{\Delta}_{\s^2}(rL)^k\phi^{(n)}\\
&+ (-kr^{-2}+O_{\infty}(r^{-3}))[a^2\sin^2\theta T^2 (rL)^{k-1}\phi^{(n)}+2a  T\Phi (rL)^{k-1}\phi^{(n)}+\slashed{\Delta}_{\s^2}(rL)^{k-1}\phi^{(n)}]\\
&-\left[\left(4n+\frac{1}{2}k\right)r^{-1}+O_{\infty}(r^{-2})\right] L(rL)^k \phi^{(n)}+r^{-2}\left[n(n+1)+\left(2n+\frac{1}{2}\right)k+O_{\infty}(r^{-1})\right] (rL)^k\phi^{(n)}\\
&+\sum_{j=1}^{k-1}O_{\infty}(r^{-2})(rL)^{j}\phi^{(n)}+\sum_{j=0}^kO_{\infty}(r^{-3})\Phi (rL)^{j}\phi^{(n)}\\
&+k(k-1)\sum_{j=0}^{k-2} O_{\infty}(r^{-2})[a^2\sin^2\theta T^2 (rL)^{j}\phi^{(n)}+2a  T\Phi (rL)^{j}\phi^{(n)}+\slashed{\Delta}_{\s^2}(rL)^{j}\phi^{(n)}]\\
&+a \sum_{m=0}^{n-1}\sum_{j=0}^kO_{\infty}(r^{-2})\Phi (rL)^j\phi^{(m)}+n\sum_{m=0}^{n-1}\sum_{j=0}^kO_{\infty}(r^{-2})(rL)^{j}\phi^{(m)}.
\end{split}
\end{equation}
Furthermore, for $\psi$ supported on a fixed azimuthal mode, we have that, schematically,
\begin{align}
\label{eq:maineqP0ho}
4\underline{L}(rL)^kP_0=&\:-k\left[\frac{1}{2}r^{-1}+O_{\infty}(r^{-2})\right](rL)^kP_0+k\sum_{j=0}^{k-1}O_{\infty}(r^{-1})(rL)^{j}P_0+\sum_{j=0}^kO_{\infty}(r^{-3})(rL)^{j}\phi_0\\ \nonumber
&+\sum_{j=0}^{k+1}O_{\infty}(r^{-3})T(rL)^j\pi_0(\sin^2\theta \phi ),\\
\label{eq:maineqP1ho}
4\underline{L}(rL)^kP_1=&\:\left[-\left(4+\frac{k}{2}\right)r^{-1}+O_{\infty}(r^{-2})\right](rL)^kP_1+k\sum_{j=0}^{k-1} O_{\infty}(r^{-1})(rL)^{j}P_1+\sum_{j=0}^kO_{\infty}(r^{-3})[(rL)^j\chphi^{(1)}_1+(rL)^j\phi_1]\\ \nonumber
&+ \sum_{j=0}^{k+1}O_{\infty}(r^{-3})[(rL)^jT\pi_1(\sin^2\theta \chphi^{(1)}) +(rL)^jT\pi_1(\sin^2\theta \phi )],\\
\label{eq:maineqP2ho}
4\underline{L}(rL)^kP_2=&\:\left[-\left(8+\frac{k}{2}\right)r^{-1}+O_{\infty}(r^{-2})\right](rL)^kP_2+k\sum_{j=0}^{k-1} O_{\infty}(r^{-1})(rL)^{j}P_2\\ \nonumber
&+\sum_{j=0}^kO_{\infty}(r^{-3})[(rL)^j\chphi^{(2)}_2+(rL)^j\chphi^{(1)}_2+(rL)^j\phi_2]\\ \nonumber
&+ \sum_{j=0}^{k+1}O_{\infty}(r^{-3})[(rL)^jT\pi_2(\sin^2\theta \chphi^{(2)}) +(rL)^jT\pi_2(\sin^2\theta \chphi^{(2)}) +(rL)^jT\pi_2(\sin^2\theta \phi )]. 
\end{align}
\end{lemma}
\begin{proof}
We have that
\begin{equation*}
\begin{split}
\underline{L}L(rL)^{k}\phi^{(n)}=&\:L(r\underline{L}L(rL)^{k-1}\phi^{(n)})-\frac{1}{2}\frac{\Delta}{r^2+a^2} r^{-1} L(rL)^k\phi-L\left(\frac{1}{2}\frac{\Delta}{r^2+a^2} r^{-1}\right)(rL)^k\phi+[\Lbar,L](rL)^k\phi\\
=&\:L(r\underline{L}L(rL)^{k-1}\phi^{(n)})-\frac{1}{2}\frac{\Delta}{r^2+a^2} r^{-1} L(rL)^k\phi+\left[\frac{1}{4}kr^{-2}+O_{\infty}(r^{-3})\right](rL)^k\phi+O_{\infty}(r^{-3})(rL)^k\Phi\phi.
\end{split}
\end{equation*}
Furthermore,
\begin{equation*}
\begin{split}
\underline{L}(rL)^kP_j=&\:rL(\underline{L}(rL)^{k-1}P_j)-\frac{1}{2}\frac{\Delta}{r^2+a^2}r^{-1} (rL)^kP_j+r[\underline{L},L] (rL)^{k-1}P_j\\
=&\:rL(\underline{L}(rL)^{k-1}P_j)-\frac{1}{2}\frac{\Delta}{r^2+a^2}r^{-1} (rL)^kP_j+O_{\infty}(r^{-2})(rL)^{k-1}\Phi P_j
\end{split}
\end{equation*}
The equations in the proposition then follow by an easy induction argument, where we suppress the $\Phi$ derivatives when acting on functions supported on a fixed azimuthal mode.
\end{proof}

\begin{proposition}
\label{prop:commrL}
Let $n,\ell\in \N_0$, with $\ell\geq n$, $K\in \N_0$ and $-4n<p\leq 2$. For $R>r_+$ suitably large, there exists a constant $C=C(M,a,n,\ell,K,R,p)>0$, such that 
\begin{equation}
\begin{split}
\label{eq:rpestLk}
\int_{N_{\tau_2}}&r^p(L\phi^{(n)}_{\geq \ell})^2+r^{p-4}\left(|\snabla_{\s^2} \phi^{(n)}_{\geq \ell}|^2-n(n+1) (\phi^{(n)}_{\geq \ell})^2\right)\,d\omega d\uprho+K\sum_{k=1}^K\int_{N_{\tau_2}} r^p(L(rL)^k\phi^{(n)}_{\geq \ell})^2+r^{p-4}|\snabla_{\s^2} (rL)^k\phi^{(n)}_{\geq \ell}|^2\,d\omega d\uprho\\
&+\int_{\tau_1}^{\tau_2}\left[\int_{N_{\tau}} r^{p-1}(L\phi^{(n)}_{\geq \ell})^2+(2-p)r^{p-3}\left(|\snabla_{\s^2} \phi^{(n)}_{\geq \ell}|^2-n(n+1) (\phi^{(n)}_{\geq \ell})^2\right)\,d\omega d\uprho\right]d\tau\\
&+K\sum_{k=1}^K\int_{\tau_1}^{\tau_2}\left[\int_{N_{\tau}} r^{p-1}(L(rL)^k\phi^{(n)}_{\geq \ell})^2+(2-p)r^{p-3}|\snabla_{\s^2}(rL)^k \phi^{(n)}_{\geq \ell}|^2\,d\omega d\uprho\right]d\tau\\
\leq \: &C\int_{N_{\tau_1}}r^p(L\phi^{(n)}_{\geq \ell})^2+r^{p-4}\left(|\snabla_{\s^2} \phi^{(n)}_{\geq \ell}|^2-n(n+1) (\phi^{(n)}_{\geq \ell})^2\right)\,d\omega d\uprho\\
&+CK\sum_{k=1}^K\int_{N_{\tau_1}} r^p(L(rL)^k\phi^{(n)}_{\geq \ell})^2+r^{p-4}|\snabla_{\s^2} (rL)^k\phi^{(n)}_{\geq \ell}|^2\,d\omega d\uprho\\
&+C(n+\ell)\int_{\tau_1}^{\tau_2}\left[\int_{N_{\tau}} r^{p-3}a^4((rL)^KT^2\phi^{(n)}_{\geq \max \{ \ell-2,0\}})^2+r^{p-3}a^2((rL)^KT\Phi \phi^{(n)}_{\geq  \ell})^2\,d\omega d\uprho\right]d\tau\\
&+CK \sum_{k=0}^{K-1}\int_{\tau_1}^{\tau_2}\left[\int_{N_{\tau}} r^{p-3}a^4((rL)^kT^2\phi^{(n)}_{\geq \max \{ \ell-2,0\}})^2+r^{p-3}a^2((rL)^kT\Phi \phi^{(n)}_{\geq  \ell})^2\,d\omega d\uprho\right]d\tau\\
&+Cn \int_{\tau_1}^{\tau_2}\left[\int_{N_{\tau}} r^{p-1}(L\Phi\phi_{\geq \ell}^{(n-1)})^2+ r^{p-1}(LT \phi_{\geq \ell}^{(n-1)})^2\,d\omega d\uprho\right]d\tau\\
&+Cn \sum_{k=0}^{K}\sum_{m=0}^{n-1}\int_{\tau_1}^{\tau_2}\left[\int_{N_{\tau}} a^2r^{p-3}((rL)^k\Phi\phi_{\geq  \ell}^{(m)})^2+r^{p-3}((rL)^k\phi_{\geq  \ell}^{(m)})^2\,d\omega d\uprho\right]d\tau\\
&+C\delta_{\ell 0 }\delta_{n 0}K \sum_{i=1}^2 \sum_{k=0}^{K-1}\int_{N_{\tau_i}} a^2r^{p-2}(L(rL)^kT\phi)^2\,d\omega d\uprho \\
&+ C\sum_{m=0}^{n+K}  \left[\int_{\Sigma_{\tau_1}}J^N[T^{m}\psi_{\geq \ell}]\cdot \mathbf{n}_{\tau_1}\, r^2d\omega d \uprho+a^2(1-\delta_{\ell 0} \delta_{n 0})\sum_{\substack{1\leq l\leq \lceil \ell/2\rceil \\ 0\leq l_1+l_2\leq l}} \int_{\Sigma_{\tau_1}}J^N[T^{m+2l}T^{l_1}\Phi^{l_2}\psi_{\geq \max\{\ell-2l,0\}}\cdot \mathbf{n}_{\tau_1}\, r^2d\omega d \uprho\right].
\end{split}
\end{equation}
Furthermore, the estimate \eqref{eq:rpestLk} also holds with the last line replaced by
\begin{equation*}
C\sum_{m=0}^{n+K} \int_{\Sigma_{\tau_1}}J^N[T^{m}\psi]\cdot \mathbf{n}_{\tau_1}\, r^2d\omega d \uprho.
\end{equation*}
\end{proposition}
\begin{proof}
We will prove the proposition via induction. The $K=0$ case follows from Proposition \ref{prop:rpest}. Now suppose \eqref{eq:rpestLk} holds for some $K\in \N$. 

We then proceed analogously to the proof of Proposition \ref{prop:rpest} by considering the multiplier $-\frac{1}{2}r^{p-2}\frac{(r^2+a^2)^2}{\Delta}L(rL)^{K+1}\phi^{(n)}_{\geq \ell}$ together with \eqref{eq:commeqho} with $k=K+1$:
\begin{equation*}
\begin{split}
0=&-\frac{1}{2}r^{p-2}L(rL)^{K+1}\phi^{(n)}_{\geq  \ell}\Bigg[-4\frac{(r^2+a^2)^2}{\Delta}\underline{L}L(rL)^{K+1}\phi^{(n)}_{\geq  \ell}+a^2\sin^2\theta \frac{\Delta}{(r^2+a^2)^2}T^2 (rL)^{K+1}\phi^{(n)}\\
&+2a  \frac{\Delta}{(r^2+a^2)^2}T\Phi (rL)^{K+1}\phi^{(n)}+\Delta (r^2+a^2)^{-2} \slashed{\Delta}_{\s^2}(rL)^{K+1}\phi^{(n)}\\
&+ (-(K+1)r^{-2}+O_{\infty}(r^{-3}))[a^2\sin^2\theta T^2 (rL)^{K}\phi^{(n)}+2a  T\Phi (rL)^{K}\phi^{(n)}+\slashed{\Delta}_{\s^2}(rL)^{K}\phi^{(n)}]\\
&-\left[\left(4n+\frac{1}{2}(K+1)\right)r^{-1}+O_{\infty}(r^{-2})\right] L(rL)^{K+1} \phi^{(n)}+r^{-2}\left[n(n+1)+\left(2n+\frac{1}{2}\right)(K+1)+O_{\infty}(r^{-1})\right] (rL)^{K+1}\phi^{(n)}\\
&+\sum_{j=1}^{K}O_{\infty}(r^{-2})(rL)^{j}\phi^{(n)}+\sum_{j=0}^{K+1}O_{\infty}(r^{-3})\Phi (rL)^{j}\phi^{(n)}\\
&+K(K+1)\sum_{j=0}^{k-2} O_{\infty}(r^{-2})[a^2\sin^2\theta T^2 (rL)^{j}\phi^{(n)}+2a  T\Phi (rL)^{j}\phi^{(n)}+\slashed{\Delta}_{\s^2}(rL)^{j}\phi^{(n)}]\\
&+a \sum_{m=0}^{n-1}\sum_{j=0}^{K+1}O_{\infty}(r^{-2})\Phi (rL)^j\phi^{(m)}+n\sum_{m=0}^{n-1}\sum_{j=0}^{K+1}O_{\infty}(r^{-2})(rL)^{j}\phi^{(m)}\Bigg]\\
=:&\:L(\mathcal{F}^{r^pL}_L)+\Lbar(\mathcal{F}^{r^pL}_{\Lbar})+\slashed{\rm div}_{\s^2} \mathcal{F}_{\snabla}^{r^pL}+\Phi(\mathcal{F}_{\Phi}^{r^pL})+\mathcal{J}^{r^p_L}.
\end{split}
\end{equation*}
The rest of the proof proceeds analogously to the proof of Proposition \ref{prop:rpest}, applying in addition the induction hypothesis \eqref{eq:rpestLk} where needed, with the only term that requires an additional, new, argument being of the form:
\begin{equation*}
\left[\frac{1}{2}(K+1)r^{p-2}+O_{\infty}(r^{p-3})\right]\slashed{\Delta}_{\s^2}(rL)^{K}\phi^{(n)}_{\geq \ell}\cdot L(rL)^{K+1}\phi^{(n)}_{\geq \ell}.
\end{equation*}
We start by splitting: $\slashed{\Delta}_{\s^2}(rL)^{K}\phi^{(n)}_{\geq \ell}=-\ell(\ell+1)(rL)^{K}\phi^{(n)}_{\ell}+\slashed{\Delta}_{\s^2}(rL)^{K}\phi^{(n)}_{\geq \ell+1}$. Then we apply a Young's inequality to obtain:
\begin{equation*}
r^{p-2}|(rL)^{K}\phi^{(n)}_{\ell}||L(rL)^{K+1}\phi^{(n)}_{\geq \ell}|\leq \epsilon r^{p-1}(L (rL)^{K+1}\phi^{(n)}_{\geq \ell}))^2+C_{\epsilon}r^{p-3}((rL)^{K}\phi^{(n)}_{\ell})^2.
\end{equation*}
We can immediately absorb the first term on the right-hand side. We can also control the second term, using the induction hypothesis \eqref{eq:rpestLk} if $K\geq 1$. If $K=0$, we instead apply \eqref{eq:hardyL} to estimate
\begin{equation*}
\int_{\mathcal{N}_{\tau}} r^{p-3}(\phi^{(n)}_{ \ell})^2\,d\omega d\uprho\leq C \int_{\mathcal{N}_{\tau_i}} r^{p-1}(L\phi^{(n)}_{\geq \ell})^2+r^{p-5}(T\phi^{(n)}_{\geq \ell})^2+r^{p-5}(\Phi\phi^{(n)}_{\geq \ell})^2\,d\omega d\uprho+(\ldots),
\end{equation*}
with $(\ldots)$ here denoting boundary terms that can be estimated via the use of a smooth cut-off function.

We can similarly split $(rL)^{K+1}\phi^{(n)}_{\geq \ell}=(rL)^{K}\phi^{(n)}_{\ell}+(rL)^{K}\phi^{(n)}_{\geq \ell+1}$ and repeat the argument above.

We apply a Leibniz rule in the $L$ and angular directions to estimate the remaining terms:
\begin{equation*}
\begin{split}
&\left[\frac{1}{2}(K+1)r^{p-2}+O_{\infty}(r^{p-3})\right]\slashed{\Delta}_{\s^2}(rL)^{K}\phi^{(n)}_{\geq \ell+1}\cdot L(rL)^{K+1}\phi^{(n)}_{\geq \ell+1}\\
=&\: L\left[-\left(\frac{1}{2}(K+1)r^{p-2}+O_{\infty}(r^{p-3})\right)\snabla_{\s^2}(rL)^{K}\phi^{(n)}_{\geq \ell+1} \cdot \snabla_{\s^2}(rL)^{K+1}\phi^{(n)}_{\geq \ell+1}\right]\\
&+\left[\frac{1}{2}(K+1)r^{p-3}+O_{\infty}(r^{p-4})\right]|\snabla_{\s^2}(rL)^{k+1}\phi^{(n)}_{\geq \ell+1}|^2\\
&-\frac{1}{2}(2-p)[r^{p-3}+O_{\infty}(r^{p-4})]\snabla_{\s^2}(rL)^{K}\phi^{(n)}_{\geq \ell+1} \cdot \snabla_{\s^2}(rL)^{K+1}\phi^{(n)}_{\geq \ell+1}+\slashed{\rm div}_{\s^2}(\ldots).
\end{split}
\end{equation*}
Note moreover that the second term on the right-hand side above has a good sign that we use to control the third term after applying Young's inequality. The first term on the right-hand side above contributes to the term $\mathcal{F}^{r^pL}_L$. Note that
\begin{equation*}
\begin{split}
&\left|\left(\frac{1}{2}(K+1)r^{p-2}+O_{\infty}(r^{p-3})\right)\snabla_{\s^2}(rL)^{K}\phi^{(n)}_{\geq \ell+1} \cdot \snabla_{\s^2}(rL)^{K+1}\phi^{(n)}_{\geq \ell+1}\right|\\
\leq&\: \epsilon r^{p-2}|\snabla_{\s^2}(rL)^{K+1}\phi^{(n)}_{\geq \ell+1}|^2+C_{\epsilon}r^{p-2}|\snabla_{\s^2}(rL)^{K}\phi^{(n)}_{\geq \ell+1}|^2.
\end{split}
\end{equation*}
We can absorb the terms in the right-hand side above into the remaining terms in $\mathcal{F}^{r^pL}_L$ together with the terms on the left-hand side of the induction hypothesis \eqref{eq:rpestLk}.

Finally, note that when $n=\ell=0$, we cannot no longer energy boundedness \eqref{eq:eboundradfield} to control the following integrals along $\mathcal{N}_{\tau_i}$ that show up when repeating the arguments of the proof of Proposition \ref{prop:rpest}:
\begin{equation*}
\sum_{i=1}^2\int_{N_{\tau_i}} a^2r^{p-4}((rL)^{K+1}T\phi)^2\,d\omega d\uprho.
\end{equation*}
\end{proof}
\begin{proposition}
\label{prop:hohierPj}
Let $K\in \N$ and $0<p<4$. Then the following estimates hold: for $R>r_+$ suitably large, there exists a constant $C(M,a,R,p)>0$ such that:
\begin{itemize}
\item[for $\ell=0$:]
\begin{equation*}
\begin{split}
\sum_{k=0}^K&\int_{N_{\tau_2}}r^{p} ((rL)^kP_0)^2\,d\omega d\uprho+\int_{\tau_1}^{\tau_2}\left[ \int_{N_{\tau}} r^{p-1} ((rL)^kP_0)^2+r^{p-1}(L(rL)^k\phi_0)^2\,d\omega d\uprho\right]d\tau\\
\leq&\: C\sum_{k=0}^K\int_{N_{\tau_1}} r^{p}  ((rL)^kP_0)^2\,d\omega d\uprho+C\sum_{k=0}^{K+1}\int_{\tau_1}^{\tau_2}\left[ \int_{N_{\tau}} r^{p-5}((rL)^{k}T \phi)^2\,d\omega d\uprho\right]d\tau\\
&+C\sum_{m=0}^{K}\int_{\Sigma_{\tau_1}}\mathbf{J}^T[T^m\psi]\cdot \mathbf{n}_{\tau_1}\,r^2d\omega d\uprho,
\end{split}
\end{equation*}
\item[for $\ell=1$:]
\begin{equation*}
\begin{split}
\sum_{k=0}^K&\int_{N_{\tau_2}}r^{p} ((rL)^kP_1)^2\,d\omega d\uprho+\int_{\tau_1}^{\tau_2}\left[ \int_{N_{\tau}} r^{p-1} ((rL)^kP_1)^2+r^{p-1}(L(rL)^k\check{\phi}^{(1)}_1)^2\,d\omega d\uprho\right]d\tau\\
\leq&\: C\sum_{k=0}^K\int_{N_{\tau_1}} r^{p}  ((rL)^kP_1)^2\,d\omega d\uprho\\
&+C\sum_{k=0}^{K+1}\int_{\tau_1}^{\tau_2}\left[ \int_{N_{\tau}} r^{p-5}((rL)^{k}T \phi_{\geq 1})^2+ r^{p-5}((rL)^{k}T \phi^{(1)}_{\geq 3})^2+r^{p-1}(L(rL)^{k}T \phi_{\geq 1})^2\,d\omega d\uprho\right]d\tau\\
&+C\sum_{m=0}^{K+1}\int_{\Sigma_{\tau_1}}\mathbf{J}^T[T^m\psi]\cdot \mathbf{n}_{\tau_1}\,r^2d\omega d\uprho,
\end{split}
\end{equation*}

\item[for $\ell=2$:]

\begin{equation*}
\begin{split}
\sum_{k=0}^K&\int_{N_{\tau_2}}r^{p} ((rL)^kP_2)^2\,d\omega d\uprho+\int_{\tau_1}^{\tau_2}\left[ \int_{N_{\tau}} r^{p-1} ((rL)^kP_2)^2+r^{p-1}(L(rL)^k\check{\phi}^{(2)}_2)^2\,d\omega d\uprho\right]d\tau\\
\leq&\: C\sum_{k=0}^K\int_{N_{\tau_1}} r^{p}  ((rL)^kP_2)^2\,d\omega d\uprho\\
&+Ca^2\sum_{k=0}^{K+1}\int_{\tau_1}^{\tau_2}\left[ \int_{N_{\tau}} r^{p+1}(L(rL)^{k}T \phi_{0})^2+r^{p-5}((rL)^{k}T \phi_{0})^2,d\omega d\uprho\right]d\tau\\
&+C\sum_{k=0}^{K+1}\int_{\tau_1}^{\tau_2}\left[ \int_{N_{\tau}} r^{p-5}((rL)^{k}T \phi_{\geq 2})^2+r^{p-5}((rL)^{k}T \phi^{(1)}_{\geq 2})^2+r^{p-5}((rL)^{k}T \phi^{(2)}_{\geq 4})^2+r^{p-1}(L(rL)^{k}T \phi_{\geq 2}^{(1)})^2\,d\omega d\uprho\right]d\tau\\
&+C\sum_{m=0}^{K+2}\int_{\Sigma_{\tau_1}}\mathbf{J}^T[T^m\psi]\cdot \mathbf{n}_{\tau_1}\,r^2d\omega d\uprho.
\end{split}
\end{equation*}
\end{itemize}
\end{proposition}
\begin{proof}
We repeat the proof of Proposition \ref{prop:rpestNpquant}, using equations \eqref{eq:maineqP0ho}--\eqref{eq:maineqP2ho} instead of \eqref{eq:maineqP0}--\eqref{eq:maineqP2}, with the role of $P_i$, $i=0,1,2$, replaced by $(rL)^kP_i$, with $0\leq k\leq K$.
\end{proof}

\subsubsection{Commuting with $rL$, $\snabla_{\s^2}$ and $T$}
In this section we show that we can also obtain $r^p$-weighted energy estimates when $\phi$ is replaced with $\snabla_{\s^2}^{\alpha}\phi$, $\alpha\in \N$, by making use of the estimates derived in Section \ref{sec:commrL}.

For the sake of convenience, will consider separately $\phi^{(n)}$ with $n=0$ (Proposition \ref{prop:rpestcommang1}) and $n\geq 1$ (Proposition \ref{prop:rpestcommang2}).
\begin{proposition}
\label{prop:rpestcommang1}
Let $K,J\in \N$ and $0<p< 2$. For $R>r_+$ suitably large, there exists a constant $C=C(M,a,K,J,R,p)>0$, such that 
\begin{equation}
\begin{split}
\label{eq:rpestcommang1}
\sum_{j+\alpha\leq J}&\sum_{k=0}^K\int_{N_{\tau_2}} r^p|\snabla_{\s^2}^{\alpha}T^jL(rL)^k\phi|^2+r^{p-4}|\snabla_{\s^2}^{\alpha+1}T^j (rL)^k\phi|^2\,d\omega d\uprho\\
&+\int_{\tau_1}^{\tau_2}\left[\int_{N_{\tau}} r^{p-1}|\snabla_{\s^2}^{\alpha}T^jL(rL)^k\phi|^2+(2-p)r^{p-3}\left(|\snabla_{\s^2}^{\alpha+1}(rL)^k T^j\phi|^2+a^2\sin^2\theta |\snabla_{\s^2}^{\alpha}T^{j+1}\phi|^2\right) \,d\omega d\uprho\right]d\tau\\
\leq &\:C\sum_{j+\alpha\leq J} \Bigg\{ \sum_{k=0}^K \int_{N_{\tau_1}} r^p|\snabla_{\s^2}^{\alpha}T^jL(rL)^k\phi|^2+r^{p-4}|\snabla_{\s^2}^{\alpha+1}T^j (rL)^k\phi|^2\,d\omega d\uprho\\
&+K \sum_{k=0}^{K-1} \int_{\tau_1}^{\tau_2}\left[\int_{N_{\tau}} r^{p-3}a^4|\snabla_{\s^2}^{\alpha}(rL)^kT^{2+j}\phi|^2+r^{p-3}a^2|\snabla_{\s^2}^{\alpha}(rL)^kT^{1+j}\Phi \phi|^2\,d\omega d\uprho\right]d\tau\\
&+K\sum_{i=1}^2 \sum_{k=0}^{K-1}\int_{N_{\tau_i}} a^2r^{p-2}|\snabla_{\s^2}^{\alpha} L(rL)^kT^{j+1}\phi)^2\,d\omega d\uprho\Bigg\} +C\sum_{m=0}^{n+K+J} \int_{\Sigma_{\tau_1}}J^N[T^{m}\psi]\cdot \mathbf{n}_{\tau_1}\, r^2d\omega d \uprho.
\end{split}
\end{equation}
\end{proposition}
\begin{proof}
Note that the $J=0$ case follows from Proposition \ref{prop:commrL}. We then carry out an inductive argument in $J$ by proceeding analogously to the proof of Proposition \ref{prop:commrL} with $n=0$, but considering multipliers of the form $-(-1)^{\alpha}\frac{1}{2}r^{p-2}\frac{(r^2+a^2)^2}{\Delta}\slashed{\Delta}_{\s^2}^{\alpha}L(rL)^{k}T^j\phi$. Note that the estimates in the proof of Proposition \ref{prop:commrL} automatically hold when $\phi$ is replaced by $T^j\phi$, since $T$ commutes with the differential operators in \eqref{eq:commeqho} by the Killing property of $T^j$. Furthermore, we integrate by parts an additional $\alpha$ times by parts in the angular direction in order to deal with the $\slashed{\Delta}_{\s^2}^{\alpha}$ derivative. For this step, we moreover use that $\slashed{\Delta}_{\s^2}$ commutes with most of the terms in \eqref{eq:commeqho}.

We will see, for example, the term:
\begin{equation*}
2(-1)^{\alpha}r^{p-2}\frac{(r^2+a^2)^2}{\Delta}\slashed{\Delta}_{\s^2}^{\alpha}L(rL)^k\phi\cdot \underline{L}L(rL)^k\phi
\end{equation*}
If $\alpha\in 2\N_0$, we integrate by parts on $\s^2$ to obtain:
\begin{equation*}
\begin{split}
2(-1)^{\alpha}r^{p-2}\frac{(r^2+a^2)^2}{\Delta}\slashed{\Delta}_{\s^2}^{\alpha}L(rL)^k\phi\cdot \underline{L}L(rL)^k\phi=&\:\underline{L}\left[r^{p-2}\frac{(r^2+a^2)^2}{\Delta}(\slashed{\Delta}_{\s^2}^{\frac{\alpha}{2}}L(rL)^k\phi)^2\right]\\
+\left(\frac{p}{2}+O_{\infty}(r^{-1})\right)r^{p-1}(\slashed{\Delta}_{\s^2}^{\frac{\alpha}{2}}L(rL)^k\phi)^2+\ldots,
\end{split}
\end{equation*}
where $\ldots$ denotes terms that vanish after integration over $\s^2$.

If $\alpha\in 2\N_0+1$, we instead obtain:
\begin{equation*}
\begin{split}
2(-1)^{\alpha}r^{p-2}\frac{(r^2+a^2)^2}{\Delta}\slashed{\Delta}_{\s^2}^{\alpha}L(rL)^k\phi\cdot \underline{L}L(rL)^k\phi=&\:\underline{L}\left[r^{p-2}\frac{(r^2+a^2)^2}{\Delta}|\slashed{\nabla}_{\s^2}\slashed{\Delta}_{\s^2}^{\frac{\alpha}{2}-1}L(rL)^k\phi|^2\right]\\
+\left(\frac{p}{2}+O_{\infty}(r^{-1})\right)r^{p-1}|\snabla_{\s^2}\slashed{\Delta}_{\s^2}^{\frac{\alpha}{2}-1}L(rL)^k\phi|^2+\ldots.
\end{split}
\end{equation*}

In both cases, control over the desired covariant derivatives on $\s^2$ follows from \eqref{eq:angdercontrol}. Other terms in \eqref{eq:commeqho} that commute with $\slashed{\Delta}_{\s^2}$ can be treated similarly.

Hence, the only term that requires an additional argument when $K=0$ is the following:
\begin{equation*}
-(-1)^{\alpha}\frac{1}{2}r^{p-2}(r^2+a^2)\slashed{\Delta}_{\s^2}^{\alpha}LT^j\phi \cdot a^2\sin^2\theta T^{2+j} \phi,
\end{equation*}
because $[\slashed{\Delta}_{\s^2},\sin^2\theta]\neq 0$.

After integrating by parts in the angular directions, we obtain 
\begin{equation*}
-(-1)^{\alpha}\frac{1}{2}r^{p-2}(r^2+a^2)\slashed{\Delta}_{\s^2}^{\alpha}LT^j\phi \cdot a^2\sin^2\theta T^{2+j} \phi=-\frac{1}{2}r^{p-2}(r^2+a^2)\snabla_{\s^2}^s\slashed{\Delta}_{\s^2}^{\frac{\alpha-s}{2}}LT^j\phi \cdot a^2\snabla_{\s^2}^s\slashed{\Delta}_{\s^2}^{\frac{\alpha-s}{2}}(\sin^2\theta T^{2+j} \phi)+\ldots,
\end{equation*}
with $s=0$ if $\alpha$ is even and $s=1$ if $\alpha$ is odd.

By using the following estimates:
\begin{align*}
|\snabla_{\s^2}(\sin^2\theta)|^2\leq &\: N_1 \sin^2\theta,\\
|\snabla_{\s^2}^s\slashed{\Delta}_{\s^2}^k(\sin^2\theta)|^2\leq &\: N_k,
\end{align*}
for all $k\in \N$ and $s=0,1$, with $N_1, N_k$ numerical constants together with \eqref{eq:angdercontrol}, we can estimate
\begin{equation*}
\begin{split}
\int_{\s^2}-\frac{1}{2}& r^{p-2}(r^2+a^2)\snabla_{\s^2}^s\slashed{\Delta}_{\s^2}^{\frac{\alpha-s}{2}}LT^j\phi \cdot a^2\snabla_{\s^2}^s\slashed{\Delta}_{\s^2}^{\frac{\alpha-s}{2}} (\sin^2\theta T^{2+j} \phi)\,d\omega\\
 \geq &\int_{\s^2}-\frac{1}{2}r^{p-2}(r^2+a^2)\snabla_{\s^2}^s\slashed{\Delta}_{\s^2}^{\frac{\alpha-s}{2}}LT^j\phi \cdot a^2\sin^2\theta\snabla_{\s^2}^s\slashed{\Delta}_{\s^2}^{\frac{\alpha-s}{2}} T^{2+j} \phi\\
&-\epsilon r^{p-1} |\snabla_{\s^2}^s\slashed{\Delta}_{\s^2}^{\frac{\alpha-s}{2}}LT^j\phi|^2- C_{\epsilon}r^{p-3} \sin^2\theta (\slashed{\Delta}_{\s^2}^{\frac{\alpha-s}{2}}T^{2+j}\phi)^2- C_{\epsilon}(\alpha-1)r^{p-3} \sin^2\theta \sum_{\beta=0}^{\alpha-2}|\snabla_{\s^2}^{\beta}T^{2+j}\phi|^2\,d\omega,
\end{split}
\end{equation*}
with $\epsilon>0$ suitably small so that we can absorb the corresponding term. When $\alpha=1$, we can estimate the integral of the term on the right-hand side above with the factor $C_{\epsilon}$ by applying \eqref{eq:eboundradfield2}, using that $p<2$. If $\alpha\geq 2$, we instead appeal to the induction step to estimate both terms with the factor $C_{\epsilon}$, using again that $p<2$. 

We then estimate the integral of the remaining term:
\begin{equation*}
 -\frac{1}{2}r^{p-2}(r^2+a^2)\snabla_{\s^2}^s\slashed{\Delta}_{\s^2}^{\frac{\alpha-s}{2}}LT^j\phi \cdot a^2\sin^2\theta\snabla_{\s^2}^s\slashed{\Delta}_{\s^2}^{\frac{\alpha-s}{2}} T^{2+j} \phi
\end{equation*}
as in the proof of Proposition \ref{prop:rpest}.
\end{proof}

\begin{remark}
In contrast with Proposition \ref{prop:commrL}, the Proposition \ref{prop:rpestcommang1} requires $p<2$ rather than $p\leq 2$. An alternative method for controlling higher-order angular derivatives, which would include the case $p=2$, would be to use that $\square_g$ commutes with the Carter operator $Q=\slashed{\Delta}_{\s^2}+a^2\sin^2\theta T^2-\Phi^2$, which is an additional symmetry property of the Kerr metric. Since the $p=2$ estimate is not relevant for the energy decay estimates in this paper, we have commuted instead with $\slashed{\Delta}_{\s^2}$ in the proof of Proposition \ref{prop:rpestcommang1}, as this is less reliant on the symmetry properties of the background metric. 
\end{remark}

\begin{proposition}
\label{prop:rpestcommang2}
Let $n\geq 1$ and $\ell\geq n$, or $n= 0$ and $\ell \geq 1$, $K,J\in \N_0$ and $0<p< 2$. For $R>r_+$ suitably large, there exists a constant $C=C(M,a,K,J,n,\ell,R,p)>0$, such that 
\begin{equation}
\begin{split}
\label{eq:rpestcommang2}
&\sum_{\alpha+j\leq J}\Bigg\{\int_{N_{\tau_2}}r^p|\snabla^{\alpha}LT^j\phi^{(n)}_{\geq \ell}|^2+r^{p-4}\left(|\snabla_{\s^2}^{\alpha+1} T^j\phi^{(n)}_{\geq \ell}|^2-n(n+1) |\snabla^{\alpha}T^j\phi^{(n)}_{\geq \ell}|^2\right)\,d\omega d\uprho\\
&+K\sum_{k=1}^K\int_{N_{\tau_2}} r^p|\snabla^{\alpha} L(rL)^kT^j\phi^{(n)}_{\geq \ell}|^2+r^{p-4}|\snabla_{\s^2}^{\alpha+1} (rL)^kT^j\phi^{(n)}_{\geq \ell}|^2\,d\omega d\uprho\\
&+\int_{\tau_1}^{\tau_2}\left[\int_{N_{\tau}} r^{p-1}|\snabla^{\alpha}LT^j\phi^{(n)}_{\geq \ell}|^2+(2-p)r^{p-3}\left(|\snabla_{\s^2}^{\alpha+1} T^j \phi^{(n)}_{\geq \ell}|^2-n(n+1) |\snabla^{\alpha}T^j\phi^{(n)}_{\geq \ell}|^2\right)\,d\omega d\uprho\right]d\tau\\
&+K\sum_{k=1}^K\int_{\tau_1}^{\tau_2}\left[\int_{N_{\tau}} r^{p-1}|\snabla^{\alpha} L(rL)^kT^j\phi^{(n)}_{\geq \ell})^2+(2-p)r^{p-3}[|\snabla_{\s^2}^{\alpha+1}(rL)^k T^j\phi^{(n)}_{\geq \ell}|^2\,d\omega d\uprho\right]d\tau \Bigg\}\\
\leq &\: C\sum_{\alpha+j\leq J}\Bigg\{\int_{N_{\tau_1}}r^p|\snabla^{\alpha}LT^j\phi^{(n)}_{\geq \ell}|^2+r^{p-4}\left(|\snabla_{\s^2}^{\alpha+1} T^j\phi^{(n)}_{\geq \ell}|^2-n(n+1) |\snabla^{\alpha}T^j\phi^{(n)}_{\geq \ell}|^2\right)\,d\omega d\uprho\\
&+K \sum_{k=1}^K\int_{N_{\tau_1}} r^p|\snabla^{\alpha} L(rL)^kT^j\phi^{(n)}_{\geq \ell}|^2+r^{p-4}|\snabla_{\s^2}^{\alpha+1} (rL)^kT^j\phi^{(n)}_{\geq \ell}|^2\,d\omega d\uprho\\
&+\int_{\tau_1}^{\tau_2}\left[\int_{N_{\tau}} r^{p-3}a^4|\snabla^{\alpha} (rL)^KT^{2+j}\phi^{(n)}_{\geq \max \{ \ell-2,0\}}|^2+r^{p-3}a^2|\snabla^{\alpha} (rL)^KT^{1+j}\Phi \phi^{(n)}_{\geq  \ell}|^2\,d\omega d\uprho\right]d\tau\\
&+K\sum_{k=0}^{K-1}\int_{\tau_1}^{\tau_2}\left[\int_{N_{\tau}} r^{p-3}a^4|\snabla^{\alpha} (rL)^kT^{2+j}\phi^{(n)}_{\geq \max \{ \ell-2,0\}}|^2+r^{p-3}a^2|\snabla^{\alpha} (rL)^kT^{1+j}\Phi \phi^{(n)}_{\geq  \ell}|^2\,d\omega d\uprho\right]d\tau\\
&+\int_{\tau_1}^{\tau_2}\left[\int_{N_{\tau}} a^2r^{p-1}|\snabla^{\alpha} L\Phi T^j\phi_{\geq \ell}^{(n-1)}|^2+r^{p-5}|\snabla^{\alpha} L T^{1+j} \phi_{\geq \ell}^{(n-1)}|^2\,d\omega d\uprho\right]d\tau\\
&+\sum_{k=0}^{K}\sum_{m=0}^{n-1}\int_{\tau_1}^{\tau_2}\left[\int_{N_{\tau}} a^2r^{p-3}|\snabla^{\alpha} (rL)^k\Phi T^j\phi_{\geq  \ell}^{(m)}|^2+r^{p-3}|\snabla^{\alpha} (rL)^kT^j\phi_{\geq  \ell}^{(m)}|^2\,d\omega d\uprho\right]d\tau\Bigg\}\\
&+ C\sum_{m=0}^{n+K+J}  \left[\int_{\Sigma_{\tau_1}}J^N[T^{m}\psi_{\geq \ell}]\cdot \mathbf{n}_{\tau_1}\, r^2d\omega d \uprho+a^2\sum_{\substack{1\leq l\leq \lceil \ell/2\rceil \\ 0\leq l_1+l_2\leq l}} \int_{\Sigma_{\tau_1}}J^N[T^{m+2l}T^{l_1}\Phi^{l_2}\psi_{\geq \max\{\ell-2l,0\}}\cdot \mathbf{n}_{\tau_1}\, r^2d\omega d \uprho\right].
\end{split}
\end{equation}
Furthermore, the estimate \eqref{eq:rpestLk} also holds with the last line replaced by
\begin{equation*}
C\sum_{m=0}^{n+K+J} \int_{\Sigma_{\tau_1}}J^N[T^{m}\psi]\cdot \mathbf{n}_{\tau_1}\, r^2d\omega d \uprho.
\end{equation*}
\end{proposition}
\begin{proof}
We then carry out an inductive argument in $J$ by repeating the arguments in the proof of Proposition \ref{prop:commrL}, but considering multipliers of the form $-(-1)^{\alpha}\frac{1}{2}r^{p-2}\frac{(r^2+a^2)^2}{\Delta}\slashed{\Delta}_{\s^2}^{\alpha}L(rL)^{k}T^j\phi^{(n)}_{\leq \ell}$ and additionally integrating by parts in the angular directions. In contrast with the proof of Proposition \ref{prop:rpestcommang1}, the fact that $[\snabla_{\s^2},\sin^2\theta]\neq 0$ does not affect the argument, as all terms involving factors of $\sin^2\theta$ are estimated using a straightforward Young's inequality.
\end{proof}
\subsubsection{Additional hierarchies for $T$-derivatives}
We can obtain control over \emph{additional} $r^p$-weighted hierarchies when considering $T\phi$ rather than $\phi$, by using the favourable commutation properties of $rL$ and $\snabla_{\s^2}$ that follow from Propositions \ref{prop:rpestcommang1}, \ref{prop:rpestcommang2} and \ref{prop:hohierPj}. These will be important for proving \emph{better} energy decay rates for \emph{additional} $T$-derivatives.

The relevant key lemma is the following:
\begin{lemma}
\label{lm:addrpestTder}
Let $n,K,J\in \N_0$ and $\ell\geq n$. Then there exists a constant $C(M,a,R,p,n,K,J,\ell)>0$ such that for all $p\in \R$:
\begin{equation}
\label{eq:Tderimprov}
\begin{split}
\sum_{k=0}^K\sum_{\alpha=0}^J&\int_{\s^2}r^{p+1}|\snabla_{\s^2}^{\alpha}L(rL)^{k}T\phi^{(n)}_{\geq \ell}|^2\,d\omega\\
\leq&\: C K\sum_{\alpha=0}^J\sum_{k=0}^{K}\sum_{\gamma+\beta\leq 1}\int_{\s^2}  r^{p-1}|\snabla_{\s^2}^{\alpha+\gamma}L (rL)^{k+\beta}\phi^{(n)}_{\geq \ell}|^2+ r^{p-3}|\snabla_{\s^2}^{\gamma+\alpha+1} (rL)^{k+\beta}\phi^{(n)}_{\geq \ell}|^2+  r^{p-3}|\snabla_{\s^2}^{\gamma+\alpha} (rL)^{k+\beta}\phi^{(n)}_{\geq \ell}|^2\,d\omega\\
&+C\sum_{\alpha=0}^J\sum_{k=0}^{K}\int_{\s^2} a^4r^{p-3}|\snabla_{\s^2}^{\alpha} (rL)^kT^2\phi^{(n)}_{\geq \max \{ \ell-2,0\}}|^2+ a^2r^{p-3}|\snabla_{\s^2}^{\alpha}(rL)^kT \Phi\phi^{(n)}_{\geq  \ell }|^2\,d\omega\\
&+C\sum_{\alpha=0}^J\sum_{k=0}^{K}\sum_{m=0}^{n-1}\int_{\s^2} a^2r^{p-5}|\snabla_{\s^2}^{\alpha}(rL)^k\Phi\phi_{\geq  \ell}^{(m)}|^2+r^{p-5}|\snabla_{\s^2}^{\alpha}(rL)^k\phi_{\geq  \ell}^{(m)}|^2\,d\omega.
\end{split}
\end{equation}
Furthermore,
\begin{equation}
\label{eq:TderimprovPj}
\begin{split}
\sum_{k=0}^K\int_{\s^2}r^{p+1}((rL)^{k}TP_j)^2\,d\omega\leq&\: C \sum_{k=0}^{K+1}\int_{\s^2}  r^{p-1}((rL)^{k}P_j)^2+\sum_{i=0}^jr^{p-5}((rL)^k\check{\phi}^{(i)}_j)^2\,d\omega\\
&+C\sum_{k=0}^{K+1}\sum_{i=0}^j\int_{\s^2} a^4r^{p-5}[((rL)^kT{\phi}^{(i)}_{\geq j})^2+((rL)^kT{\phi}^{(i)}_{\leq j-2})^2]\,d\omega.
\end{split}
\end{equation}
\end{lemma}
\begin{proof}
We first split:
\begin{align*}
T L (rL)^k \phi^{(n)}=&\:\underline{L} L (rL)^k \phi^{(n)}+L L (rL)^k \phi^{(n)}-\frac{a}{r^2+a^2}\Phi L (rL)^k \phi^{(n)},\\
T (rL)^k P_j=&\: \underline{L}(rL)^k P_j+ r^{-1} (rL)^{k+1}P_j-\frac{a}{r^2+a^2}\Phi(rL)^k P_j
\end{align*}
and then we square both sides of the equations together with Lemma \ref{lm:commeqho} to estimate $(\underline{L} L (rL)^k \phi^{(n)})^2$ and $(\underline{L}(rL)^k P_j)^2$ and obtain \eqref{eq:Tderimprov} with $J=0$.

Now take $J\geq 1$. Then we repeat the above procedure, but rather than squaring both sides of Lemma \ref{lm:commeqho}, we multiply the right-hand side of Lemma \ref{lm:commeqho} with $(-1)^{\alpha}\slashed{\Delta}_{\s^2}^{\alpha}$ acting on the right-hand side of Lemma \ref{lm:commeqho} and integrate by parts in the angular directions, using moreover \eqref{eq:angdercontrol}.
\end{proof}

\section{Energy decay estimates}
\label{sec:edecay}
In this section, we convert the hierarchy of $r^p$-weighted estimates established in Propositions \ref{prop:rpestcommang1}, \ref{prop:rpestcommang2} and \ref{prop:hohierPj} into energy decay estimates. We outline below the strategy for obtaining both sharp energy decay rates for $\psi_{\ell}$ with $\ell=0,1,2$ and suitably strong energy decay estimates for $\psi_{\geq 3}$:
\begin{enumerate}[A)]
\item We first obtain a preliminary energy decay estimate for the full solution $\psi$ in Proposition \ref{prop:edecaystep1} (\textbf{Proposition \ref{prop:edecay}}).
\item We then consider the restriction $\psi_{\geq 1}$ and obtain a corresponding improved energy decay estimate using the preliminary energy decay established in Proposition \ref{prop:edecaystep1} to deal with the contributions of the mode $\psi_0$ that play a role in the estimate. We then consider the further restrictions $\psi_{\geq 2}$ and $\psi_{\geq 3}$ and keep improving the energy decay estimates successively. (\textbf{Proposition \ref{prop:edecay}}).
\item Subsequently, we restrict to the single spherical harmonic modes $\psi_0$, $\psi_1$ and $\psi_2$ and improve the energy decay estimate even further. The corresponding equations are coupled with 
the remaining spherical harmonic modes and the terms that appear in the estimates due to this coupling are controlled using the energy decay estimates already established in Proposition \ref{prop:edecay}. We arrive at energy estimates that are sharp when considering initial data with \textbf{non-vanishing} Newman--Penrose charges $I_{\ell}$, $\ell=0,1,2$. (\textbf{Proposition \ref{prop:almostsharpedecay}})
\end{enumerate}
By applying the time-inversion theory from Section \ref{sec:timeinv} below, we will perform an additional step:
\begin{enumerate}[D)]
\item Applying the energy decay estimates Proposition \ref{prop:almostsharpedecay} to the time-integral $T^{-1}\psi$ (Proposition \ref{prop:mainpropTinv}) allows us to obtain sharp energy decay estimates when considering initial data with \textbf{vanishing} Newman--Penrose charges, in particular, compactly supported initial data (\textbf{Corollary \ref{cor:sharpdecayvanishingNP}}).
\end{enumerate}
\begin{proposition}
\label{prop:edecaystep1}
Let $K\in \N_0$ and $\delta>0$ be arbitrarily small. Then there exists a constant $C=C(M,a,K,R,\delta)>0$,
\begin{equation}
\label{eq:edecaystep1}
\begin{split}
\int_{\Sigma_{\tau}}&J^N[T^K\psi]\cdot \mathbf{n}_{\tau}\,r^2d\omega d\uprho+\sum_{m=1}^2\sum_{k_1+k_2+\alpha\leq K}(1+\tau)^{-m-2(K-k_2)}\int_{N_{\tau}}r^{m-\delta}|\snabla_{\s^2}^{\alpha} L(rL)^{k_1}T^{k_2}\phi|^2d\omega d\uprho\\
\leq&\: C(1+\tau)^{-2+\delta-2K}E_{\geq 0,K,\delta}[\psi]
\end{split}
\end{equation}
with
\begin{equation*}
\begin{split}
E_{\geq 0,K,\delta}[\psi]:=&\:\sum_{0\leq j=j_1+j_2\leq 1}\sum_{k_1+k_2+\alpha\leq K} \int_{N_{0}} r^{2-\delta-j}|\snabla_{\s^2}^{\alpha} L(rL)^{k_1}T^{j_1+2k_2}\Phi^{j_2}\phi|^2+r^{-2-\delta-j}|\snabla_{\s^2}^{1+\alpha} (rL)^{k_1}T^{j_1+2k_2}\Phi^{j_2}\phi|^2\,d\omega d\uprho\\
&+\sum_{0\leq k_1+k_2\leq 2+2K} \int_{\Sigma_{0}}J^N[T^{k_1}\Phi^{k_2} \psi]\cdot \mathbf{n}_{0}\,r^2d\omega d\uprho.
\end{split}
\end{equation*}
\end{proposition}
\begin{proof}
We consider first the case $K=0$. We obtain energy decay by applying Proposition \ref{prop:rpest} with $n=\ell=0$ in combination with the following ingredients:
\begin{enumerate}[1)]
\item the mean value theorem along dyadic time intervals (the ``pigeonhole principle''),\\
\item an \emph{interpolation inequality} (Lemma \ref{lm:interpol}) ,\\
\item \emph{energy boundedness estimate} in the form \eqref{eq:ebound}, \eqref{eq:eboundradfield} and \eqref{eq:eboundradfield2},\\
\item an integrated local energy decay estimate \eqref{eq:morawetz2}. 
\end{enumerate}
Let $\{\tau_i\}$ be a dyadic sequence, then Proposition \ref{prop:rpest} with $p=1$ together with \eqref{eq:eboundradfield2} implies that:
\begin{equation}
\label{eq:rpfirststep}
\int_{\tau_i}^{\tau_{i+1}}\int_{N_{\tau}} (L\phi)^2+r^{-2}(\underline{L}\phi)^2+r^{-2}|\snabla_{\s^2}\phi|^2\,d\omega d\uprho d\tau \leq C \int_{N_{\tau_i}} r (L\phi)^2+r^{-3}|\snabla_{\s^2}\phi|^2\,d\omega d\uprho+C\int_{\Sigma_{\tau_i}} J^N[\psi]\cdot \mathbf{n}_{\tau_i}\,r^2 d\omega d\uprho.
\end{equation}
By combining \eqref{eq:rpfirststep} with \eqref{eq:morawetz1}, applying the mean value theorem and then applying  \eqref{eq:ebound}, we obtain $\tau^{-1}$ decay for the $N$-energy, with a loss of $T$ and $\Phi$ derivatives on the right-hand side. We can improve this energy decay rate by considering the spacetime integral of the right-hand side of \eqref{eq:rpfirststep}, applying Proposition \ref{prop:rpest} with $p=2-\delta$ together with \eqref{eq:eboundradfield2} and repeating the arguments above together with the interpolation inequality in Lemma \ref{lm:interpol} to obtain \eqref{eq:edecaystep1} with $K=0$. Note that the decay rate corresponds to the total length of the hierarchy of $r^p$-weighted estimates applied minus $\delta$. 

Now consider $K=1$. The above arguments hold automatically for $\psi$ replaced by $T\psi$, but we can improve the energy decay even further by considering $p=3-\delta$ in the following way: we apply Lemma \ref{lm:addrpestTder} with $p=1-\delta$ and $J=K=0$, and instead of Proposition \ref{prop:rpest}, we appeal to the $r^p$-estimates in Proposition \ref{prop:rpestcommang1} with $p=1$, $J=1$, and $K=1$. When then obtain $\tau^{-3+\delta}$ decay. Subsequently, we consider Proposition \ref{prop:rpestcommang1} with $p=2-\delta$. As the length of our hierarchy is now $4$, we are left  $\tau^{-4+\delta}$ decay of the $N$-energy. In this process, we automatically obtain control over additional $r$-weighted quantities with slower decay rates.

We can treat the general $K$ case inductively, by repeatedly applying Lemma \ref{lm:addrpestTder} with $n=\ell=0$ and $p=1-\delta$ and suitably high values of $J$ and $K$, together with Proposition \ref{prop:rpestcommang1} with $p=1$ and $p=2-\delta$ and suitably high values of $J$ and $K$. We obtain in this way a hierarchy of length $2+2K$ and hence, $\tau^{-2-2K}$ decay for the $N$-energy.
\end{proof}

\begin{proposition}
\label{prop:edecay}
Let $K\in \N_0$ and $\delta>0$ be arbitrarily small. Then there exists a constant $C=C(M,a,K,R,\delta)>0$,
\begin{align}
\label{eq:decayl1}
\int_{\Sigma_{\tau}}&J^N[T^K\psi_{\geq 1}]\cdot \mathbf{n}_{\tau}\,r^2d\omega d\uprho+\sum_{m=1}^2\sum_{k_1+k_2+\alpha\leq K}(1+\tau)^{-m-2(K-k_2)+\delta}\int_{N_{\tau}}r^{m-\delta}|\snabla_{\s^2}^{\alpha} L(rL)^{k_1}T^{k_2}\phi_{\geq 1}|^2d\omega d\uprho\\ \nonumber
&+\sum_{m=0}^2\sum_{k_1+k_2+\alpha\leq K}(1+\tau)^{-2-m-2(K-k_2)+\delta}\int_{N_{\tau}}r^{m-\delta}|\snabla_{\s^2}^{\alpha} L(rL)^{k_1}T^{k_2}\phi^{(1)}_{\geq 1}|^2d\omega d\uprho\\ \nonumber
\leq&\: C(1+\tau)^{-4+\delta-2K}E_{\geq 1,K,\delta}[\psi],\\
\label{eq:decayl2}
\int_{\Sigma_{\tau}}&J^N[T^K\psi_{\geq 2}]\cdot \mathbf{n}_{\tau}\,r^2d\omega d\uprho+\sum_{m=1}^2\sum_{k_1+k_2+\alpha\leq K}(1+\tau)^{-m-2(K-k_2)+\delta}\int_{N_{\tau}}r^{m-\delta}|\snabla_{\s^2}^{\alpha} L(rL)^{k_1}T^{k_2}\phi_{\geq 2}|^2d\omega d\uprho\\ \nonumber
&+\sum_{m=0}^2\sum_{j=1}^2\sum_{k_1+k_2+\alpha\leq K}(1+\tau)^{-2j-m-2(K-k_2)+\delta}\int_{N_{\tau}}r^{m-\delta}|\snabla_{\s^2}^{\alpha} L(rL)^{k_1}T^{k_2}\phi^{(j)}_{\geq 2}|^2d\omega d\uprho\\ \nonumber
\leq&\: C(1+\tau)^{-6+\delta-2K}E_{\geq 2,K,\delta}[\psi],\\
\label{eq:decayl3}
\int_{\Sigma_{\tau}}&J^N[T^K\psi_{\geq 3}]\cdot \mathbf{n}_{\tau}\,r^2d\omega d\uprho+\sum_{m=1}^2\sum_{k_1+k_2+\alpha\leq K}(1+\tau)^{-m-2(K-k_2)+\delta}\int_{N_{\tau}}r^{m-\delta}|\snabla_{\s^2}^{\alpha} L(rL)^{k_1}T^{k_2}\phi_{\geq 3}|^2d\omega d\uprho\\ \nonumber
&+\sum_{m=0}^2\sum_{j=1}^3\sum_{k_1+k_2+\alpha\leq K}(1+\tau)^{-2j-m-2(K-k_2)+\delta}\int_{N_{\tau}}r^{m-\delta}|\snabla_{\s^2}^{\alpha} L(rL)^{k_1}T^{k_2}\phi^{(j)}_{\geq 3}|^2d\omega d\uprho\\ \nonumber
\leq&\: C(1+\tau)^{-8+\delta-2K}E_{\geq 3,K,\delta}[\psi],
\end{align}
with
\begin{align*}
E_{\geq \ell,K,\delta}[\psi]:=&\sum_{m_1\leq M_{T,\ell}, m_2\leq M_{\Phi,\ell}}\Bigg\{\sum_{j=0}^{\ell} \sum_{k+\alpha\leq K+3-j} \int_{N_{0}} r^{2-\delta}|\snabla_{\s^2}^{\alpha} L(rL)^{k}T^{m_1}\Phi^{m_2}\phi^{(j)}_{\geq j}|^2\\
&+r^{-2-\delta}|\snabla_{\s^2}^{1+\alpha} (rL)^{k_1}T^{m_1}\Phi^{m_2}\phi^{(j)}_{\geq j}|^2\,d\omega d\uprho+\int_{\Sigma_{0}}J^T[ T^{m_1}\Phi^{m_2}\psi]\cdot \mathbf{n}_{0}\,r^2d\omega d\uprho\Bigg\},
\end{align*}
where $\ell=1,2,3$ and $M_{T,\ell}$ and $M_{\Phi,\ell}$, are suitably large integers.
\end{proposition}
\begin{proof}
We consider first \eqref{eq:decayl1}. Let $K=0$.  Note that the estimates of Proposition \ref{prop:edecaystep1} still hold when we replace $\phi$ in the left-hand side by $\phi_{\geq 1}$, by orthogonality of the spherical harmonic modes. When considering $\phi_{\geq 1}$, we can further estimate via \eqref{eq:hardyL}:
\begin{equation*}
\begin{split}
\int_{\tau_i}^{\tau_{i+1}}&\int_{N_{\tau}} r^{2-\delta}[(L\phi_{\geq 1})^2+r^{-4}|\snabla_{\s^2} \phi_{\geq 1}|^2] \,d\omega d\uprho d\tau\\
 \leq &\:C\sum_{m_1+m_2\leq 1}\int_{\tau_i}^{\tau_{i+1}}\int_{N_{\tau}} r^{-\delta}\left[(L\phi_{\geq 1}^{(1)})^2+ (LT^{m_1}\Phi^{m_2}\phi_{\geq 1})^2+r^{-2}|\snabla_{\s^2 }T^{m_1}\Phi^{m_2}\phi_{\geq 1}|^2\right]\,d\omega d\uprho d\tau.
 \end{split}
\end{equation*}
We then apply Proposition \ref{prop:rpest} with $n=\ell=1$ and $p=1-\delta$ to obtain:
\begin{equation}
\label{eq:rpn1prelim}
\begin{split}
\int_{\tau_{i}}^{\tau_{i+1}}&\int_{N_{\tau}} r^{2-\delta}[(L\phi_{\geq 1})^2+r^{-4}|\snabla_{\s^2} \phi_{\geq 1}|^2] \,d\omega d\uprho d\tau\leq C\int_{N_{\tau_i}} r^{1-\delta}(L\phi_{\geq 1}^{(1)} )^2+r^{-1-\delta}[|\snabla_{\s^2}\phi_{\geq 1}^{(1)}|^2-2(\phi_{\geq 1}^{(1)})^2]  \,d\omega d\uprho\\
&+C\sum_{m_1+m_2\leq 1}\int_{\tau_{i}}^{\tau_{i+1}}\int_{N_{\tau}} r^{2-\delta} (LT T^{m_1}\Phi^{m_2}\phi)^2 \,d\omega d\uprho d\tau+C\int_{\tau_{i}}^{\tau_{i+1}}\int_{N_{\tau}} r^{-2-\delta} [(\Phi \phi_{\geq 1})^2 +(\phi_{\geq 1})^2] \,d\omega d\uprho d\tau\\
&+C\sum_{m_1+m_2\leq 1}E_{0,0,\delta}[T^{m_1}\Phi^{m_2}\psi]+C\sum_{m=0}^1\int_{\Sigma_{\tau_i}}J^N[T^m\psi]\cdot \mathbf{n}_{i}\,r^2d\omega d\uprho.
\end{split}
\end{equation}
We apply \eqref{eq:hardyL} to estimate further the third integral on the right-hand side of \eqref{eq:rpn1prelim}:
\begin{equation*}
\begin{split}
\int_{\tau_{i}}^{\tau_{i+1}}&\int_{N_{\tau}} r^{-2-\delta} [(\Phi \phi_{\geq 1})^2 +(\phi_{\geq 1})^2] \,d\omega d\uprho d\tau\leq C\int_{\tau_{i}}^{\tau_{i+1}}\int_{N_{\tau}} r^{-\delta} [(L\Phi \phi_{\geq 1})^2 +(L\phi_{\geq 1})^2] \,d\omega d\uprho d\tau\\
&+C\sum_{m_1+m_2=1}\int_{\tau_{i}}^{\tau_{i+1}}\int_{N_{\tau}} r^{-4-\delta} [(\Phi T^{m_1}\Phi^{m_2}\Phi \phi_{\geq 1})^2 +(T^{m_1}\Phi^{m_2}\phi_{\geq 1})^2] \,d\omega d\uprho d\tau.
\end{split}
\end{equation*}
Both terms can be controlled by applying the estimates in the proof of Proposition \ref{prop:edecaystep1}.

Note that the second integral on the right-hand side of  \eqref{eq:rpn1prelim} does not involve $\phi_{\geq 1}$, but rather the full function $\phi$, so we cannot appeal to Proposition \ref{prop:rpest} with $n=1$. However, as it involves $T\phi$ rather than $\phi$, we can apply Lemma \ref{lm:addrpestTder} with $p=1$ and $J=K=0$ to control it using Proposition \ref{prop:edecaystep1} with $K=1$. We are then left with $\tau^{-3+\delta}$ energy decay for the $N$-energy of $\psi_{\geq 1}$.

In order to obtain $\tau^{-4+\delta}$ decay, we apply Proposition \ref{prop:rpest} with $n=\ell=1$ and $p=2-\delta$ to control the right-hand side of \eqref{eq:rpn1prelim}. Note that the energy boundedness estimate \eqref{eq:ebound} and local integrated energy decay estimates \eqref{eq:morawetz1} and \eqref{eq:morawetz2} have $N$-energies of $\psi$ on the right-hand side, for which we have merely established $\tau^{-2+\delta}$ decay, so after applying the mean-value theorem along dyadic intervals, we are limited to  $\tau^{-3+\delta}$ energy decay for $\psi_{\geq 1}$. 

We therefore appeal instead to the \emph{refined} energy boundedness estimate \eqref{eq:morawetzell1} and the refined integrated local integrated energy decay estimates \eqref{eq:morawetzell2} and \eqref{eq:morawetzell3} for $\psi_{\geq 1}$, together with Proposition \ref{prop:rpest} with $n=0$ and $\ell=1$, so that we do not see the $N$-energy of the full solution $\psi$, but rather the $N$-energy of $T$-derivatives of $\psi$, which decay faster in time.

With the above observations, taking into account the additional loss of $T$ and $\Phi$ derivatives due to the use of the refined energy boundedness and local integrated energy decay estimates, we obtain \eqref{eq:decayl1} with $K=0$. Note that we have extended the total hierarchy of $r^p$-weighted estimates by $2$, which results in $\tau^{-4+\delta}$ decay for the $N$-energy of $\psi_{\geq 1}$, compared to $\tau^{-2+\delta}$-decay for $\psi_{\geq 0}$.

Proposition \ref{prop:rpest} with $n=\ell=1$ and $p=2-\delta$, $p=1-\delta$ and $p=-\delta$ provides moreover boundedness of the higher-order weighted energies:
\begin{equation*}
	\int_{\mathcal{N}_{\tau}} (1+\tau)^{2-m} r^{m-\delta} (L\phi^{(1)})^2\,d\omega d\uprho
\end{equation*}
with $m=0,1,2$.

We deal with the $K\geq 1$ cases via an inductive argument, applying repeatedly Lemma \ref{lm:addrpestTder} with $n=\ell=1$ and $p=1-\delta$ and suitably high values of $J$ and $K$, together with Proposition \ref{prop:rpestcommang2} with $p=1-\delta$ and $p=2-\delta$ and suitably high values of $J$ and $K$.

Now consider \eqref{eq:decayl2}. The estimates for $\psi_{\geq 1 }$ still apply to $\psi_{\geq 2}$. As in the $\ell \geq 1$ case, we apply \eqref{eq:hardyL} to estimate
\begin{equation*}
\begin{split}
\int_{\tau_i}^{\tau_{i+1}}&\int_{N_{\tau}} r^{2-\delta}[(L\phi_{\geq 2}^{(1)})^2+r^{-4}|\snabla_{\s^2} \phi_{\geq 2}^{(1)}|^2] \,d\omega d\uprho d\tau\\
 \leq &\:C\sum_{m_1+m_2\leq 1}\int_{\tau_i}^{\tau_{i+1}}\int_{N_{\tau}} r^{-\delta}\left[(L\phi_{\geq 2}^{(2)})^2+ (LT^{m_1}\Phi^{m_2}\phi_{\geq 2}^{(1)})^2+r^{-2}|\snabla_{\s^2 }T^{m_1}\Phi^{m_2}\phi_{\geq 2}^{(1)}|^2\right]\,d\omega d\uprho d\tau.
 \end{split}
\end{equation*}
We then apply Proposition \ref{prop:rpest} with $n=\ell=2$ and $p=1-\delta$ and $p=2-\delta$ to extend the hierarchy of $r^p$-weighted estimates by 2 (in the $K=0$ case). Note that we see in particular on the right-hand side the term:
\begin{equation*}
\int_{\tau_i}^{\tau_{i+1}}\int_{N_{\tau}} r^{p-3}a^4(T^2\phi^{(2)})^2\,d\omega d\uprho\leq C\sum_{k=0}^1\int_{\tau_i}^{\tau_{i+1}}\int_{N_{\tau}} r^{p+3}a^4(L(rL)^kT^2\phi)^2\,d\omega d\uprho
\end{equation*}
By applying Lemma \ref{lm:addrpestTder} twice, it follows that we can estimate the resulting integral with the left-hand side of \eqref{eq:edecaystep1} with $K=2$ and additional commutations with $T$ and $\Phi$.

We avoid $N$-energies of the full solution $\psi$ appearing on the right-hand side of Proposition \ref{prop:rpest} with $n=0,1$ and $\ell=2$, at the expense of losing additional $T$ and $\Phi$ derivatives on the right-hand side. Similarly, we apply \eqref{eq:morawetzell1}, \eqref{eq:morawetzell2} and \eqref{eq:morawetzell3} to ensure we only see $N$-energies of sufficiently many $T$-derivatives of the full solution $\psi$ on the right-hand side, which decay suitably fast. The rest of the argument proceeds in a similar manner.

Finally, we prove \eqref{eq:decayl3} by naturally following the strategy outlined above. Note that we exploit here the fact that the equation for $\psi_{\geq 3}$ is only coupled with $\psi_{\geq 1}$, and the terms involving $T^2\phi_{\geq 1}$ when $K=0$ can be estimated via \eqref{eq:decayl1} with $K=2$. The estimates for higher $K$ then follow inductively as above. 
\end{proof}
\begin{remark}
In principle, one can keep precise track of the number of additional $T$ and $\Phi$ derivatives on the right-hand sides of the inequalities in Proposition \ref{prop:edecay}, but we do not expect this number to be sharp. In order to arrive at a sharp number of $T$ and $\Phi$ derivatives, one can instead decompose $\psi$ into a bounded frequency part with respect to a time frequency $\omega$ and azimuthal mode $m$, for which there will not be a loss of $T$ and $\Phi$ derivatives, and a high-frequency part (supported on suitably large frequencies $\omega$ and $m$) that is decoupled with lower $\ell$ modes, as these are supported on lower values of $m$. These two parts can be treated separately. We do not carry out this procedure in the present paper.
\end{remark}

\begin{proposition}
\label{prop:almostsharpedecay}
Let $K\in \N_0$ and $\delta>0$ be arbitrarily small. Then there exists a constant\\ $C=C(M,a,K,R,\delta)>0$, such that
\begin{align}
\label{eq:almostsharpedecay0}
\int_{\Sigma_{\tau}}&J^N[T^K\psi_{0}]\cdot \mathbf{n}_{\tau}\,r^2d\omega d\uprho+\sum_{m=1}^2\sum_{k_1+k_2\leq K}(1+\tau)^{-m-2(K-k_2)+\delta}\int_{N_{\tau}}r^{m-\delta}(L(rL)^{k_1}T^{k_2}\phi_{0})^2d\omega d\uprho\\ \nonumber
\leq&\: C(1+\tau)^{-3+\delta-2K}E_{0,K,\delta}[\psi],\\
\label{eq:almostsharpedecay1}
\int_{\Sigma_{\tau}}&J^N[T^K\psi_{1}]\cdot \mathbf{n}_{\tau}\,r^2d\omega d\uprho+\sum_{m=1}^2\sum_{k_1+k_2\leq K}(1+\tau)^{-m-2(K-k_2)+\delta}\int_{N_{\tau}}r^{m-\delta}(L(rL)^{k_1}T^{k_2}\phi_1)^2d\omega d\uprho\\ \nonumber
&+\sum_{m=0}^2\sum_{k_1+k_2\leq K}(1+\tau)^{-2-m-2(K-k_2)+\delta}\int_{N_{\tau}}r^{m-\delta}(L(rL)^{k_1}T^{k_2}\phi_1^{(1)})^2d\omega d\uprho,\\ \nonumber
\leq&\: C(1+\tau)^{-5+\delta-2K}E_{1,K,\delta}[\psi],\\
\label{eq:almostsharpedecay2}
\int_{\Sigma_{\tau}}&J^N[T^K\psi_{2}]\cdot \mathbf{n}_{\tau}\,r^2d\omega d\uprho+\sum_{m=1}^2\sum_{k_1+k_2\leq K}(1+\tau)^{-m-2(K-k_2)+\delta}\int_{N_{\tau}}r^{m-\delta}(L(rL)^{k_1}T^{k_2}\phi_2)^2d\omega d\uprho\\ \nonumber
&+\sum_{j=0}^1\sum_{m=0}^2\sum_{k_1+k_2+\alpha\leq K}(1+\tau)^{-2j-m-2(K-k_2)+\delta}\int_{N_{\tau}}r^{m-\delta}( L(rL)^{k_1}T^{k_2}\phi^{(j)}_{2})^2d\omega d\uprho\\ \nonumber
\leq&\: C(1+\tau)^{-7+\delta-2K}E_{2,K,\delta }[\psi].
\end{align}
with
\begin{align*}
E_{\ell,K}[\psi]:=&\:\sum_{m\leq M_{\ell,T}}\Bigg\{\sum_{j=0}^{\ell}\sum_{k+\alpha\leq K+\ell-j} \int_{N_{0}} r^{2-\delta}|\snabla_{\s^2}^{\alpha} L(rL)^{k}T^{m} \phi^{(j)}_{\geq j}|^2+r^{-2-\delta}|\snabla_{\s^2}^{1+\alpha} (rL)^{k}T^{m}\phi^{(j)}_{\geq j}|^2\\
&+ r^{3-\delta}|\snabla_{\s^2}^{\alpha} (rL)^{k} T^{m}P_{\ell}|^2\,d\omega d\uprho+\int_{\Sigma_{0}}J^N[T^m\psi]\cdot \mathbf{n}_{0}\,r^2d\omega d\uprho\Bigg\}.
\end{align*}
where $M_{\ell,T}$ are suitably large integers.
\end{proposition}
\begin{proof}
We consider first $\psi_0$ with $K=0$. Then, we can extend the length of the hierarchy of $r^p$-estimates in the proof of Proposition \ref{prop:edecaystep1} by 2, via a consideration of the quantity $P_0$ and an application of the estimates in Proposition \ref{prop:rpestNpquant} for $\ell=0$ with $p=3-\delta$ and $p=4-\delta$.

Note that the full $\phi$ appears on the right-hand side of the relevant estimate of Proposition \ref{prop:rpestNpquant}, but it can be easily estimated using Proposition \ref{prop:rpestNpquant}. The lengths of the hierarchies in the proof of Proposition \ref{prop:edecaystep1} that we appeal to when $K\geq 1$ can similarly be extended by 2 by appealing to Lemma \ref{lm:addrpestTder} together with Proposition \ref{prop:hohierPj}.

The above argument for extending the length of the relevant hierarchies of $r^p$-estimates applies also when considering $\psi_1$ and $\psi_2$. We consider $P_j$ with $j=1,2$ and appeal again to Propositions \ref{prop:rpestNpquant} and \ref{prop:hohierPj}. In order to estimate the integrals on the right-hand side of the estimates in Propositions \ref{prop:rpestNpquant} and \ref{prop:hohierPj} involving the full $\phi$, we apply Proposition \ref{prop:edecay}.
\end{proof}

\begin{remark}
Note that the estimates for each $\psi_{\ell}$ in Proposition \ref{prop:almostsharpedecay} can be derived independently of each other, i.e. they are not coupled. The mode coupling that occurs in these estimates can instead be dealt with by applying the estimates that we already established in Propositions \ref{prop:edecaystep1} and \ref{prop:edecay}.
\end{remark}

\subsection{Additional energy decay estimates}
It will be useful to establish also energy decay estimates for the commuted quantities of the form $(rX)^j \psi_{\ell}$, with $\ell=0,1,2$.

\begin{corollary}
\label{cor:edecaycommrX}
Let $\ell=0,1,2$. Let $\delta>0$ be arbitrarily small and $K,J\in \N_0$. Then there exists a constant $C=C(M,a,R,K,J,\ell,\delta)>0$, such that
\begin{equation}
\label{eq:edecaycommrX}
\sum_{j=0}^J\int_{\Sigma_{\tau}}J^N[(rX)^{j}T^K\psi_{\ell}]\cdot \mathbf{n}_{\tau}\,r^2d\omega d\uprho\leq C(1+\tau)^{-3-2\ell+\delta-2K}\left(E_{\ell,J+K,\delta}[\psi]+\sum_{j=0}^J E_{\ell,K,\delta}[N^j\psi]\right)\\
\end{equation}
Furthermore, for $\ell=3,4$, there exists a constant $C=C(M,a,R,K,J,\ell,\delta)>0$, such that
\begin{equation}
\label{eq:edecaycommrX34}
\sum_{j=0}^J\int_{\Sigma_{\tau}}J^N[(rX)^{j}T^K\psi_{\ell}]\cdot \mathbf{n}_{\tau}\,r^2d\omega d\uprho\leq C(1+\tau)^{-8+\delta-2K}\left(E_{\geq 3,J+K,\delta}[\psi]+\sum_{j=0}^J E_{\geq 3,K,\delta}[N^j\psi]\right).\\
\end{equation}
\end{corollary}
\begin{proof}
Let $\ell=0,1,2$ and consider the case $K=0$. We will show that \eqref{eq:edecaycommrX} and \eqref{eq:edecaycommrX} hold by induction in $J$. Then $J=0$ case follows from Proposition \ref{prop:almostsharpedecay} for the $\ell=0,1,2$ cases and Proposition \ref{prop:edecay} for $\ell=3,4$.

Let us take as our induction step the estimate
\begin{equation}
\label{eq:commrLind}
\sum_{0\leq j_1+j_2\leq J} \int_{N_{\tau}} r^{-\delta} (L(rL)^{j_1}T^{j_2} \phi_{\ell})^2\,d\omega d\uprho\leq CE_{\ell,J,\delta}[\psi](1+\tau)^{-3-2\ell+\delta}.
\end{equation}
Furthermore, by applying the arguments in the proof of Lemma \ref{lm:addrpestTder}, we obtain:
\begin{equation*}
\begin{split}
\int_{N_{\tau}} r^{-\delta} (L(rL)^{J+1} \phi_{\ell})^2\,d\omega d\uprho\leq&\: C\sum_{j=0}^{J} \int_{N_{\tau}} r^{2-\delta} (L(rL)^j T\phi_{\ell})^2\,d\omega d\uprho+ C\sum_{j=0}^{J} \int_{N_{\tau}} r^{-\delta} (L(rL)^j \phi_{\ell})^2\,d\omega d\uprho\\
&+C\sum_{j=0}^{J} \int_{N_{\tau}} r^{-2-\delta}  ((rL)^j T^2\phi)^2+r^{-4-\delta}((rL)^j \phi_{\ell})^2\,d\omega d\uprho.
\end{split}
\end{equation*}
By Proposition \ref{prop:almostsharpedecay}, we have that
\begin{align}
\label{eq:commrL1}
\sum_{j=0}^J \int_{N_{\tau}} r^{1-\delta} (L(rL)^j \phi_{\ell})^2\,d\omega d\uprho\leq CE_{\ell,J,\delta}[\psi](1+\tau)^{-1-2\ell},\\
\label{eq:commrL2}
\sum_{j=0}^J \int_{N_{\tau}} r^{2-\delta} (L(rL)^j \phi_{\ell})^2\,d\omega d\uprho\leq CE_{\ell,J,\delta}[\psi](1+\tau)^{-2\ell}.
\end{align}

Using \eqref{eq:commrL2} together with the fact that \eqref{eq:commrLind} also holds if we replace $\psi$ with $T\psi$, we therefore obtain \eqref{eq:commrLind} with $J$ replaced by $J+1$.

In order to remover the $r^{-\delta}$ degeneracy in \eqref{eq:commrLind}, we interpolate between \eqref{eq:commrLind} and \eqref{eq:commrL1} using Lemma \ref{lm:interpol} to obtain
\begin{equation*}
\sum_{j=0}^J\int_{N_{\tau}} J^N[(rL)^j\psi_{\ell}]\cdot \mathbf{n}_{\tau}\,r^2d\omega d\uprho \leq C E_{\ell,J,\delta}[\psi](1+\tau)^{-3-2\ell+2\delta}.
\end{equation*}
The $K\geq 1$ case proceeds analogously. By expressing $X$ in terms of $L$ and $T$ and using that \eqref{eq:waveeq} commutes with $T$, we obtain the desired estimates on $N_{\tau}=\Sigma_{\tau}\cap\{r\geq R\}$. 

We obtain the desired estimate on $\Sigma_{\tau}\cap \{r_+< r_0\leq r\leq R\}$, with $r_0$ arbitrarily close to $r_+$, via standard elliptic estimates and commutation of \eqref{eq:waveeq} with $T$ and $\Phi$: in particular, in the $K=0$ we have that
\begin{equation}
\label{eq:interiorcomm}
\sum_{j=0}^J\int_{\Sigma_{\tau}\cap \{r_+< r_0\leq r\leq R\}} J^N[(rX)^j\psi_{\ell}]\cdot \mathbf{n}_{\tau}\,r^2d\omega d\uprho \leq C E_{\ell,J,\delta}[\psi](1+\tau)^{-3-2\ell+2\delta}.
\end{equation}
Finally, we consider the region $\Sigma_{\tau}\cap \{r_+\leq r \leq r_0\}$, with $r_0$ suitably small. There we can apply the commuted redshift estimates from Section 7 of \cite{lecturesMD} to obtain:
\begin{equation*}
\sum_{j=0}^J\int_{\Sigma_{\tau}\cap \{r_+\leq r\leq R\}} J^N[(rX)^j\psi_{\ell}]\cdot \mathbf{n}_{\tau}\,r^2d\omega d\uprho \leq C \left(\sum_{j=0}^JE_{\ell,0,\delta}[N^j\psi]+E_{\ell,J,\delta}[\psi]\right)(1+\tau)^{-3-2\ell+2\delta}.
\end{equation*}
We combine the above $N$-energy decay estimates to conclude that \eqref{eq:edecaycommrX} holds.
\end{proof}
\section{Hierarchies of $r^{-2k}$-weighted elliptic estimates}
\label{sec:elliptic}
In this section, we consider the following inhomogeneous equation;
\begin{equation}
\label{eq:inhomelliptic}
\mathcal{L}\psi=F,
\end{equation}
with $\mathcal{L}$ defined as follows:
\begin{equation*}
\mathcal{L}\psi=X(\Delta X\psi)+ 2a X\Phi \psi+\slashed{\Delta}_{\s^2}\psi.
\end{equation*}
Note that if $\psi$ is a solution to \eqref{eq:waveeq}, then it is a solution to \eqref{eq:inhomelliptic}, with
\begin{equation}
\label{eq:inhomo}
\begin{split}
F=&\: 2[h\Delta-(r^2+a^2)] (r^2+a^2)^{-1/2} XT\phi+[(\Delta h)' -2h\Delta r(r^2+a^2)^{-1}]T\psi+[2h(r^2+a^2)-h^2\Delta-a^2\sin^2\theta] T^2\psi\\
&+2a(h-1) T\Phi\psi.
\end{split}
\end{equation}

One may easily verify that the operator $\mathcal{L}$ is elliptic \emph{if and only if} we restrict its domain to lie outside of the ergoregion $\{\Delta-a^2\sin^2\theta>0\}$. The reason for this is that $T$ is timelike if and only if we restrict to the subset of the spacetime that lies outside the ergoregion. In the special case $a=0$, ellipticity only breaks down at the boundary $r=r_+$.

In order to deal with the presence of the ergoregion when deriving appropriately weighted $L^2$ estimates for $\psi$ in terms of $F$, we appeal to elliptic versions of the redshift estimates that were introduced in the hyperbolic setting in \cite{lecturesMD}.\footnote{The estimates use strict positivity of $\frac{d\Delta}{dr}$, which near the event horizon $\mathcal{H}^+$ is a manifestation of the redshift effect and therefore motivates the nomenclature. In contrast with the \emph{hyperbolic} redshift estimates, the elliptic redshift estimates make use of the fact that $\frac{d\Delta}{dr}>0$ \underline{globally} in $r$, rather than just locally near $r=r_+$.} 

\subsection{Spacelike redshift estimates}
\label{sec:ellipticredshift}
In contrast with standard elliptic estimates for uniformly elliptic operators, e.g. the Laplacian in Euclidean space, the spacelike redshift estimate below establishes control over weighted $H^1$ norms of $\psi$, rather than (weighted) $H^2$ norms of $\psi$.
\begin{proposition}
\label{prop:ellipticredshift}
There exists a constant $C=C(M,a)>0$, such that for solutions $\psi$ to \eqref{eq:inhomelliptic} that are suitably regular and decaying in $\uprho$:
\begin{equation}
\label{eq:ellipticredshift1}
\int_{\Sigma_{\tau}} r (X\psi)^2\,d\omega d\uprho\leq C \int_{\Sigma_{\tau}} r^{-1} F^2\,d\omega d\uprho. 
\end{equation}
If $\psi$ is a solution to \eqref{eq:waveeq}, then we moreover have:
\begin{align}
\label{eq:ellipticredshift2}
\int_{\Sigma_{\tau}} r (X\psi)^2\,d\omega d\uprho\leq&\: C \int_{\Sigma_{\tau}}J^N[T\psi]\cdot \mathbf{n}_{\tau}\,d\mu_{\tau}+C \int_{N_{\tau}}r  (LT \phi)^2\,d\omega d\uprho,\\
\label{eq:ellipticredshiftlgeq3}
\int_{\Sigma_{\tau}} r (X\psi_{\geq 3})^2\,d\omega d\uprho\leq &\: C \int_{\Sigma_{\tau}}J^N[T\psi_{\geq 3}]\cdot \mathbf{n}_{\tau}\,d\mu_{\tau}+ C\int_{\Sigma_{\tau}} r^{-1}(T^2\psi_1)^2+r^{-1}(T^2\psi_2)^2\,d\omega d\uprho+C \int_{N_{\tau}}r  (LT \phi_{\geq 3})^2\,d\omega d\uprho.
\end{align}
\end{proposition}
\begin{proof}
We multiply both sides of \eqref{eq:inhomelliptic} by $X\psi$, to obtain
\begin{equation*}
\begin{split}
X\psi F=&\: \Delta X\psi X^2\psi+\Delta' (X\psi)^2+2a X\psi X\Phi\psi+X\psi\slashed{\Delta}_{\s^2}\psi \\
=&\: X \left(\frac{1}{2}\Delta (X\psi)^2-\frac{1}{2}|\snabla_{\s^2}\psi|^2\right)+\frac{1}{2}\Delta' (X\psi)^2+\Phi(a (X\psi)^2)+\slashed{\rm div}_{\s^2}(X\psi \slashed{\nabla}_{\s^2}\psi).
\end{split}
\end{equation*}
Integrating along $\Sigma_{\tau}$ and using that $ \lim_{\uprho\to \infty} \snabla_{\s^2}\psi(\tau,\uprho,\theta,\varphi)=0$, we obtain \eqref{eq:ellipticredshift1}. 

In order to obtain \eqref{eq:ellipticredshift2}, we first apply Young's inequality to the terms in $r^{-1}F^2$ to obtain
\begin{equation*}
r^{-1}F^2\lesssim r  (XT \phi)^2+r^{-3}(T^2\phi)^2+r^{-3}(\Phi T\phi)^2+ r^{-3}(T\phi)^2,
\end{equation*}
absorbing the integral of $r^{-3}(T\phi)^2$ into the remaining terms, using \eqref{eq:hardyX}, and then relating with $J^T[T\psi]\cdot \mathbf{n}_{\tau}$ by applying \eqref{eq:hardyX} once more.
\end{proof}

In the proposition below, we obtain \underline{local}, higher-order versions of the estimate \eqref{eq:ellipticredshift1} by projecting to fixed spherical harmonic modes. These provide control inside the ergoregion and will be combined with the elliptic estimates outside the ergoregion derived in Section \ref{sec:ellipticoutsideergo}.

\begin{proposition}
\label{prop:redshiftestfixedell}
Let $\psi$ be a solution to \eqref{eq:inhomelliptic} that is suitably regular, then:
\begin{enumerate}[\rm (i)] 
\item For arbitrary $r_0>r_+$, there exists a constant $C=C(M,a,r_0,\ell)>0$ such that
\begin{equation}
\label{eq:redshiftestfixedell}
\int_{\Sigma_{\tau} \cap \{r\leq r_0\}}(X(\Delta X^{\ell+1}\psi_{\ell}))^2+(X^{\ell+1}\psi_{\ell})^2\,d\omega d\uprho \leq C\int_{\Sigma_{\tau} \cap \{r\leq r_0\}}(\pi_{\ell}X^{\ell}F)^2\,d\omega d\uprho.
\end{equation}
\item For all $r\geq r_+$, there exists a constant $C=C(M,a)>0$, such that furthermore:
\begin{equation}
\label{eq:estforprop}
||X^{\ell+1}\psi_{\ell}||_{L^2(S^2_{\tau,r'})}\leq C\sup_{r_+\leq r\leq r'}||r^{-1}\pi_{\ell}X^{\ell}F||_{L^2(S^2_{\tau,r})}.
\end{equation}
\item There exists a constant $C=C(M,a,\ell)>0$ such that
\begin{equation}
\label{eq:fasterdecatXlpsi}
\begin{split}
||X^{\ell+1}\psi_{\ell}||_{L^2(S^2_{\tau,r'})}\leq&\: C\sum_{k=0}^{\ell+1}\sup_{r_+\leq r \leq r'}||r^{-\ell}(rX)^{k}T\psi_{\ell} ||_{L^2(S^2_{\tau,r})}+C\sum_{k=0}^{\ell}\sup_{r_+\leq r \leq r'}||r^{-2-\ell}(rX)^{k}T^2\psi_{\ell} ||_{ L^2(S^2_{\tau,r})}\\
&+Ca^2\sup_{r_+\leq r \leq r'}||r^{-2}X^{\ell}T^2\psi_{\ell-2} ||_{L^2(S^2_{\tau,r})}+Ca^2\sup_{r_+\leq r \leq r'}||r^{-2}X^{\ell}T^2\psi_{\ell+2} ||_{L^2(S^2_{\tau,r})}.
\end{split}
\end{equation}
\item There exists a constant $C=C(M,a,\ell)>0$ such that for $\ell\geq 1$:
\begin{equation}
\label{eq:fasterdecatXlpsiv2}
\begin{split}
||r^2X^{\ell+1}\psi_{\ell}||_{L^2(S^2_{\tau,r'})}\leq&\: C\sum_{k=0}^{\ell+1}\sup_{r_+\leq r \leq r'}||r^{2-\ell}(rX)^{k}T\psi_{\ell} ||_{L^2(S^2_{\tau,r})}+C\sum_{k=0}^{\ell}\sup_{r_+\leq r \leq r'}||r^{-\ell}(rX)^{k}T^2\psi_{\ell} ||_{ L^2(S^2_{\tau,r})}\\
&+Ca^2\sup_{r_+\leq r \leq r'}||X^{\ell}T^2\psi_{\ell-2} ||_{L^2(S^2_{\tau,r})}+Ca^2\sup_{r_+\leq r \leq r'}||X^{\ell}T^2\psi_{\ell+2} ||_{L^2(S^2_{\tau,r})}.
\end{split}
\end{equation}

\end{enumerate}
\end{proposition}
\begin{proof}
It follows easily by induction that for all $n\geq 0$:
\begin{equation*}
X^{n+1}(\Delta X\psi)=X(\Delta X^{n+1}\psi)+n \Delta' X^{n+1}\psi+n(n+1)X^n\psi.
\end{equation*}
Hence,
\begin{equation}
\label{eq:commidn}
X^nF=X(\Delta X^{n+1}\psi)+n\Delta' X^{n+1}\psi+n(n+1)X^n\psi+ 2a {X}\Phi X^n\psi+\slashed{\Delta}_{\s^2}(X^n\psi).
\end{equation}
A cancellation occurs in \eqref{eq:commidn} when $n=\ell$ after projecting onto the $\ell$-th spherical harmonic mode:
\begin{equation}
\label{eq:commidell}
X^{\ell}{\pi}_{\ell}F=X(\Delta X^{\ell+1}\psi_{\ell})+\ell \Delta'X^{\ell+1}\psi_{\ell}+2a \Phi X^{\ell+1}\psi_{\ell}
\end{equation}
Now consider the multiplier $X^{\ell+1}\psi_{\ell}$ to obtain:
\begin{equation*}
\begin{split}
X^{\ell+1}\psi_{\ell}\cdot X^{\ell}{\pi}_{\ell}F=& \Delta X^{\ell+1}\psi_{\ell}X^{\ell+2}\psi_{\ell}+(\ell+1)\Delta'(X^{\ell+1}\psi_{\ell})^2 +2a X^{\ell+1}\psi_{\ell} X^{\ell+1}\Phi \psi_{\ell}\\
=&\: X\left(\frac{1}{2}  \Delta (X^{\ell+1}\psi_{\ell})^2\right)+\Phi(a (X^{\ell+1} \psi_{\ell})^2)+\left(\ell+\frac{1}{2}\right) \Delta' (X^{\ell+1}\psi_{\ell})^2.
\end{split}
\end{equation*}
Hence, we can integrate the above inequality to obtain for all $r_0\geq r_+$:
\begin{equation}
\label{eq:redshiftest}
\begin{split}
\int_{\s^2} &\frac{1}{2}  \Delta (X^{\ell+1}\psi_{\ell})|_{r=r_0,\tau=\tau'}^2\,d\omega+\left(\ell+\frac{1}{2}\right)\int_{\Sigma_{\tau'}\cap\{r\leq r_0\}} \Delta'(X^{\ell+1}\psi_{\ell})^2\,d\omega d\uprho\leq  \int_{\Sigma_{\tau'}\cap\{ r\leq r_0\}} |X^{\ell+1}\psi_{\ell}\cdot X^{\ell}{\pi}_{\ell}F|\,d\omega d\uprho\\
\leq&\:   \int_{\Sigma_{\tau'}\cap\{ r\leq r_0\}} \epsilon \Delta' (X^{\ell+1}\psi_{\ell})^2+\frac{1}{2\epsilon} \frac{1}{\Delta'}(X^{\ell}{\pi}_{\ell}F)^2\,d\omega d\uprho.
\end{split}
\end{equation}
Then, \eqref{eq:redshiftestfixedell} immediately follows after absorbing the $(X^{\ell+1}\psi_{\ell})^2$ term on the very right-hand side of \eqref{eq:redshiftest} into the left-hand side (taking $\epsilon>0$ suitably small) and including a $(X(\Delta X^{\ell+1}\psi_{\ell})^2$ term by using:
\begin{equation*}
(X(\Delta X^{\ell+1}\psi_{\ell})^2\leq C  (\Delta'X^{\ell+1}\psi_{\ell})^2+C a^2(\Phi X^{\ell+1}\psi_{\ell})^2+(X^{\ell}{\pi}_{\ell}F)^2.
\end{equation*}
It moreover follows that for $r'\in [r_+,r_0]$:
\begin{equation*}
\begin{split}
\int_{\s^2}  \Delta(r') (X^{\ell+1}\psi_{\ell})|_{r=r',\tau=\tau'}^2\,d\omega\leq&\:  C\int_{\Sigma_{\tau'}\cap\{ r\leq r'\}} r^{-1}(X^{\ell}{\pi}_{\ell}F)^2\,d\omega d\rho\\
\leq &\: C \Delta(r')r'^{-1} \int_{\s^2}\sup_{r_+\leq r \leq r'} r^{-1}(X^{\ell}{\pi}_{\ell}F)^2|_{\tau=\tau'}\,d\omega\\
\end{split}
\end{equation*}
and hence we obtain \eqref{eq:estforprop} for the range $r'\in [r_+,r_0]$, with $r_0>r_+$ arbitrarily large, with a constant $C$ depending moreover on $r_0$.

In order to prove \eqref{eq:estforprop} for $r' \in (r_0,\infty)$, we appeal again to \eqref{eq:redshiftest} and Young's inequality to estimate:
\begin{equation*}
\begin{split}
\int_{\s^2} &\frac{1}{2} \sup_{r_0\leq \tilde{r}\leq r'} \Delta (X^{\ell+1}\psi_{\ell})|_{r=\tilde{r},\tau=\tau'}^2\,d\omega \leq \int_{\s^2} \frac{1}{2}  \Delta (X^{\ell+1}\psi_{\ell})|_{r=r_0}^2\,d\omega+ \int_{\Sigma_{\tau'}\cap\{ r_0\leq \tilde{r}\leq r'\}} |X^{\ell+1}\psi_{\ell}\cdot X^{\ell}{\pi}_{\ell}F|\,d\omega d\uprho\\
\leq&\:   \int_{\s^2} \frac{1}{2}  \Delta (X^{\ell+1}\psi_{\ell})|_{r=r_0,\tau=\tau'}^2\,d\omega+\frac{1}{4(r'-r_0)} \int_{\Sigma_{\tau'}\cap\{ r_0\leq \tilde{r}\leq r'\}} \Delta (X^{\ell+1}\psi_{\ell})^2\,d\omega d\uprho\\
&+ (r'-r_0)\int_{\Sigma_{\tau'}\cap\{ r_0\leq \tilde{r}\leq r'\}} \Delta^{-1}(X^{\ell}{\pi}_{\ell}F)^2\,d\omega d\uprho\\
\leq&\:  \int_{\s^2} \frac{1}{2}  \Delta (X^{\ell+1}\psi_{\ell})|_{r=r_0,\tau=\tau'}^2\,d\omega+ \frac{1}{4}\int_{\s^2} \sup_{r_0\leq \tilde{r}\leq r'} \Delta (X^{\ell+1}\psi_{\ell})|_{r=\tilde{r},\tau=\tau'}^2\,d\omega\\
&+ (r'-r_0)^2\int_{\s^2}  \sup_{r_0\leq \tilde{r}\leq r'} \Delta^{-1}(X^{\ell}{\pi}_{\ell}F)^2|_{r=r',\tau=\tau'}\,d\omega.
\end{split}
\end{equation*}
Hence, for all  $r' \in [r_0,\infty)$
\begin{equation*}
\begin{split}
\int_{\s^2}& \Delta (X^{\ell+1}\psi_{\ell})|_{r=r',\tau=\tau'}^2\,d\omega\leq 4(r'-r_0)^2\int_{\s^2}  \sup_{r_0\leq \tilde{r}\leq r'} \Delta^{-1}(X^{\ell}{\pi}_{\ell}F)^2|_{r=r',\tau=\tau'}\,d\omega\\
&+ C \Delta(r_0) \int_{\s^2}\sup_{r_+\leq r \leq r_0} r^{-2}(X^{\ell}{\pi}_{\ell}F)^2|_{\tau=\tau'}\,d\omega.
\end{split}
\end{equation*}
from which we can conclude \eqref{eq:estforprop}.

We obtain \eqref{eq:fasterdecatXlpsi} by applying \eqref{eq:estforprop} and inserting for $F$ the right-hand side of \eqref{eq:inhomo}.

Finally, we observe that we can rearrange the terms in \eqref{eq:commidell} and multiply by $\Delta^{\ell}$ as follows:
\begin{equation*}
\Delta^{\ell}X^{\ell}{\pi}_{\ell}F=X(\Delta^{\ell+1} X^{\ell+1}\psi_{\ell})+2a \Delta^{\ell} \Phi X^{\ell+1}\psi_{\ell}.
\end{equation*}
Suppose now that $\ell\geq 1$. Then we integrate to obtain for $\uprho\geq r_0$:
\begin{equation*}
|\Delta^{\ell+1} X^{\ell+1}\psi_{\ell}|(\tau,\uprho,\theta,\varphi_*)\leq C r^{2\ell+1} ||X^{\ell+1}\psi_{\ell}||_{L^{\infty}(\Sigma_{\tau})}+C\Delta^{\ell}|| rX^{\ell}{\pi}_{\ell}F||_{L^{\infty}(\Sigma_{\tau})},
\end{equation*}
and therefore
\begin{equation}
\label{eq:ellipticestrXellpsi}
||r X^{\ell+1}\psi_{\ell}||_{L^{\infty}(\Sigma_{\tau})}\leq C||X^{\ell+1}\psi_{\ell}||_{L^{\infty}(\Sigma_{\tau})}+|| rX^{\ell}{\pi}_{\ell}F||_{L^{\infty}(\Sigma_{\tau})}.
\end{equation}
Repeating the above step, using now \eqref{eq:ellipticestrXellpsi}, we obtain for $\ell\geq 1$:
\begin{equation*}
|\Delta^{\ell+1} X^{\ell+1}\psi_{\ell}|(\tau,\uprho,\theta,\varphi_*)\leq C r^{2\ell} \Delta^{\ell}  ||rX^{\ell+1}\psi_{\ell}||_{L^{\infty}(\Sigma_{\tau})}+\Delta^{\ell} || rX^{\ell}{\pi}_{\ell}F||_{L^{\infty}(\Sigma_{\tau})}
\end{equation*}
and \eqref{eq:fasterdecatXlpsiv2} follows by applying a standard Sobolev inequality on $\s^2$.
\end{proof}

\subsection{Elliptic estimates outside ergoregion}
\label{sec:ellipticoutsideergo}
In this section, we derive elliptic estimates outside the ergoregion, making use of the uniform ellipticity of $\mathcal{L}$ when restricted to a region where $r$ is sufficiently large. In order to couple these estimates with the spacelike redshift estimates from Section \ref{sec:ellipticredshift}, we restrict to fixed spherical harmonic modes. 
\begin{proposition}
\label{prop:auxellipticell}
Let $\ell\in \N_0$ and $k\in \R$, with $k<\ell+\frac{1}{2}$. Let $\psi$ be a solution to \eqref{eq:inhomelliptic}, such that $r\psi\in C^{\infty}(\widehat{\Sigma}_{\tau})$. Then, for $R_0>2M\geq r_+$ suitably large, there exists a constant $C=C(M,a,R_0,\ell,k)>0$, such that
\begin{equation}
\begin{split}
\label{eq:auxellipticell}
\int_{\Sigma_{\tau}\cap \{r\geq R_0\}} r^{2k+2}(X^{\ell+1}\psi_{\ell})^2\,d\omega d\uprho+ r^{2k}(X(\Delta X^{\ell+1}\psi_{\ell}))^2\,d\omega d\uprho\leq &\: C\int_{\Sigma_{\tau}}r^{2k}(X^{\ell}{\pi}_{\ell}F)^2 \,d\omega d\uprho.
\end{split}
\end{equation}
\end{proposition}
\begin{proof}
Denote
\begin{equation*}
G_n:=X^nF- 2a \Phi X^{n+1}\psi.
\end{equation*}
Then we can write
\begin{equation*}
G_n=X(\Delta X^{n+1}\psi)+n\Delta' X^{n+1}\psi+n(n+1)X^n\psi+\slashed{\Delta}_{\s^2}(X^n\psi),
\end{equation*}
so that
\begin{equation*}
X(\Delta X^{\ell+1}\psi_{\ell})+\ell \Delta' X^{\ell+1}\psi=\pi_{\ell}G_{\ell}.
\end{equation*}
Let $\chi: [r_+,\infty)$ be a cut-off function such that $\chi(r)=1$ for $r\geq R_0$ and $\chi(r)=0$ for $r\leq R_0-M$, with $R_0>r_++M$ to be chosen suitably large. Then for any $k\in \N$
\begin{equation*}
\chi^2 r^{2k}(X(\Delta X^{\ell+1}\psi_{\ell}))^2+\chi^2 \ell^2(\Delta' )^2 r^{2k}(X^{\ell+1}\psi_{\ell})^2+\chi^2 2\ell \Delta' r^{2k}X^{\ell+1}\psi X(\Delta X^{\ell+1}\psi_{\ell})=\chi r^{2k}(\pi_{\ell}G_{\ell})^2.
\end{equation*}
We can further estimate
\begin{equation*}
\begin{split}
\chi^2 2\ell \Delta' r^{2k}X^{\ell+1}\psi_{\ell} X(\Delta X^{\ell+1}\psi_{\ell})=&\:2\ell \chi^2  (\Delta')^2 r^{2k}  (X^{\ell+1}\psi_{\ell})^2+\ell \chi^2 \Delta'\Delta r^{2k} X ((X^{\ell+1}\psi_{\ell})^2)\\
=&\: \ell\left[2(\Delta')^2r^{2k}- (\Delta' \Delta r^{2k})' \right]\chi^2  (X^{\ell+1}\psi_{\ell})^2-2\ell \chi' \chi \Delta' \Delta(X^{\ell+1}\psi_{\ell})^2\\
&+X(\ell \chi^2\Delta'\Delta r^{2k} (X^{\ell+1}\psi_{\ell})^2)\\
=&\: -\ell[4k-2+O(r^{-1})] r^{2k+2}(X^{\ell+1}\psi_{\ell})^2-2\ell \chi' \chi\Delta' \Delta(X^{\ell+1}\psi_{\ell})^2\\
&+X(\ell \chi \Delta' r^{2k} (X^{\ell+1}\psi_{\ell})^2).
\end{split}
\end{equation*}
Note that
\begin{equation*}
\int_{\Sigma_{\tau}}X(\ell \chi^2 \Delta' \Delta r^{2k} (X^{\ell+1}\psi_{\ell})^2) \,d\omega d\uprho=0
\end{equation*}
if $\ell=0$ or $k<\ell+\frac{1}{2}$.

By \eqref{eq:hardyX} we can moreover estimate for $k<\ell+\frac{1}{2}$ and $\epsilon>0$ arbitrarily small:
\begin{equation*}
\begin{split}
\frac{(2k-1)^2}{4}(1-2\epsilon)&\int_{\Sigma_{\tau}} r^{2k-2} \Delta^2 \chi^2(X^{\ell+1}\psi_{\ell})^2\,d\omega d\uprho \leq (1-\epsilon)\int_{\Sigma_{\tau}}\chi^2 r^{2k}(X(\Delta X^{\ell+1}\psi_{\ell}))^2\,d\omega d\uprho\\
&+C_{\epsilon} \int_{\Sigma_{\tau}}  (\chi')^2\Delta^2(X^{\ell+1}\psi_{\ell})^2 \,d\omega d\uprho.
\end{split}
\end{equation*}

Finally, observe that
\begin{equation*}
r^{2k}(\pi_{\ell}G_{\ell})^2\leq C r^{2k}(X^{\ell}{\pi}_{\ell}F)^2+ a^2 r^{2k}(\Phi X^{\ell+1}\psi_{\ell})^2.
\end{equation*}

We then combine the above estimates to obtain the following integral inequality:
\begin{equation*}
\begin{split}
\int_{\Sigma_{\tau}}& \chi \epsilon r^{2k}(X(\Delta X^{\ell+1}\psi_{\ell}))^2+\left[4\ell^2+(1-2k)2\ell+\frac{1}{4}(2k-1)^2-2\epsilon\frac{(2k-1)^2}{4} +O_{\infty}(r^{-1})\right]\chi^2r^{2k+2}(X^{\ell+1}\psi_{\ell})^2\,d\omega d\uprho\\
\leq &\: C_{\epsilon,R_0,\ell}\int_{\Sigma_{\tau}\cap \{R_0-M\leq r\leq R_0\}}  (X^{\ell+1}\psi_{\ell})^2 \,d\omega d\uprho+C\int_{\Sigma_{\tau}}r^{2k}(X^{\ell}{\pi}_{\ell}F)^2 \,d\omega d\uprho.
\end{split}
\end{equation*}

We can choose $R_0$ suitably large and $\epsilon>0$ suitably small in order to make the left-hand side non-negative definite, provided
\begin{equation*}
0<(2\ell)^2+(1-2k)2\ell+\frac{1}{4}(2k-1)^2=\left(k-2\ell-\frac{1}{2}\right)^2,
\end{equation*}
which follows in particular from the condition $k<\ell+\frac{1}{2}$.

We conclude \eqref{eq:auxellipticell} by applying additionally \eqref{eq:redshiftestfixedell}.
\end{proof}

In the proposition below, we obtain additional control over arbitrarily many $X$-derivatives of $\psi_{\ell}$, starting from the estimates established in Proposition \ref{prop:auxellipticell}.

\begin{proposition}
\label{prop:mainellipticell}
Let $J,\ell\in \N_0$. Restrict $-\frac{1}{2}<k<\ell+\frac{1}{2}$ when $\ell\geq 1$ and $k>-\frac{1}{2}$ when $\ell=0$. Let $\psi$ be a solution to \eqref{eq:inhomelliptic}, such that $r\psi\in C^{\infty}(\widehat{\Sigma}_{\tau})$. Then, for $R_0>2M\geq r_+$ suitably large, there exists a constant $C=C(M,a,R_0,\ell,k,J)>0$, such that
\begin{equation}
\begin{split}
\label{eq:mainellipticell1}
\sum_{n=0}^{\ell+J}& \int_{\Sigma_{\tau}}r^{-2k+2} (X(rX)^n \psi_{\ell})^2+\ell^2r^{-2k} ((rX)^n \psi_{\ell})^2\,d\omega d\uprho\leq C \sum_{n=0}^{\ell+J}\int_{\Sigma_{\tau}} r^{-2k}( (rX)^{n}({\pi}_{\ell}F))^2\,d\omega d\uprho.
\end{split}
\end{equation}
If $\psi$ is a solution to \eqref{eq:waveeq}, then
\begin{equation}
\begin{split}
\label{eq:mainellipticell2}
\sum_{n=0}^{\ell+J}& \int_{\Sigma_{\tau}}r^{-2k+2} (X(rX)^n \psi_{\ell})^2+\ell^2r^{-2k} ((rX)^n \psi_{\ell})^2\,d\omega d\uprho\leq C \sum_{n=0}^{\ell+J}\int_{\Sigma_{\tau}} r^{-2k+4} (X T (rX)^n \psi_{\ell})^2\\
&+r^{-2k+2} (T^2 (rX)^n \psi_{\ell})^2+ a^2r^{-2k} (T^2 (rX)^n \psi_{\ell-2})^2+a^2r^{-2k} (T^2 (rX)^n \psi_{ \ell+2})^2\,d\omega d\uprho.
\end{split}
\end{equation}
\end{proposition}
\begin{proof}
Let $n\in \N_0$. By \eqref{eq:commidn} it follows that
\begin{equation}
\label{eq:mainloelliptic}
\begin{split}
r^{m}X^{n+1}\psi_{\ell} X^n{\pi}_{\ell}F=&\:\Delta r^{m} X^{n+1}\psi_{\ell} X^{n+2}\psi_{\ell}+(n+1)\Delta' r^{m}(X^{n+1}\psi_{\ell})^2-\big[\ell(\ell+1)\\
&-n(n+1)\big]r^{m}X^{n+1}\psi_{\ell}X^n\psi_{\ell}+ 2a r^{m} X^{n+1}\psi_{\ell} \Phi X^{n+1}\psi_{\ell}\\
=&\: \Delta r^{m} X^{n+1}\psi_{\ell} X^{n+2}\psi_{\ell}+(n+1)\Delta' r^{m}(X^{n+1}\psi_{\ell})^2\\
&-X\left(\frac{1}{2}[\ell(\ell+1)-n(n+1)]r^{m}(X^n\psi_{\ell})^2\right)\\
&+\frac{1}{2}m[\ell(\ell+1)-n(n+1)]r^{m-1}(X^n\psi_{\ell})^2+ \Phi(a r^{m} (X^{n+1}\psi_{\ell})^2).
\end{split}
\end{equation}

We integrate over $\Sigma_{\tau }\cap \{r\geq R_0\}$, where $R_0\geq r_+$ will be chosen appropriately large:
\begin{equation*}
-\int_{\Sigma_{\tau}\cap \{r\geq R_0\}}X\left(\frac{1}{2}(\ell(\ell+1)-n(n+1))r^{m}(X^n\psi_{\ell})^2\right)\,d\omega d\uprho=\int_{S^2_{\tau,R_0}}\frac{1}{2}(\ell(\ell+1)-n(n+1))r^{m}(X^n\psi_{\ell})^2\,d\omega
\end{equation*}
if $m<2n+2$.\\
\\

\underline{$J=0$:}\\
\\
We consider first the case $J=0$. If $\ell=0$, then \eqref{eq:mainellipticell1} follows immediately from  \eqref{eq:auxellipticell}. Suppose $\ell\geq 1$ and $n\leq \ell-1$. Then we estimate in $\{r\geq R_0\}$:
\begin{align*}
|\Delta  r^{m} X^{n+1}\psi_{\ell} X^{n+2}\psi_{\ell}|\leq&\: \frac{1}{2}\epsilon \Delta r^{m-1}(X^{n+1}\psi_{\ell})^2+\frac{1}{2\epsilon} \Delta r^{m+1}( X^{n+2}\psi_{\ell})^2,\\
|r^{m}X^{n+1}\psi_{\ell} X^n{\pi}_{\ell}F|\leq &\: \frac{1}{2}\epsilon r^{m+1}(X^{n+1}\psi_{\ell})^2+\frac{1}{2\epsilon} r^{m-1}(X^n{\pi}_{\ell}F)^2.
\end{align*}

In the case when $m\leq 0$ and $n=\ell-1$, we need to additionally apply \eqref{eq:commidn} in $\{r\geq R_0\}$ to further estimate:
\begin{equation*}
\begin{split}
-\frac{1}{2}m&[\ell(\ell+1)-(\ell-1)\ell]r^{m-1}(X^{\ell-1}\psi_{\ell})^2=- \ell mr^{m-1}(X^{\ell-1}\psi_{\ell})^2\\
\leq &\:-m(1+\epsilon) \ell[r^2+O(r)] r^{m-1}(X^{\ell}\psi_{\ell})^2+C_{\epsilon}r^{2m+2}(X^{\ell+1}\psi_{\ell})^2+C_{\epsilon} r^{m-1}(X^{\ell-1}{\pi}_{\ell}F)^2\\
\leq &\:-\frac{1}{2}m(1+\epsilon) \ell[1+O(r^{-1})] \Delta' r^{m}(X^{\ell}\psi_{\ell})^2+C_{\epsilon}r^{m+1}(X^{\ell+1}\psi_{\ell})^2+C_{\epsilon} r^{m-1}(X^{\ell-1}{\pi}_{\ell}F)^2.
\end{split}
\end{equation*}
Hence, for $m>-2$, $R_0$ appropriately large and $\epsilon>0$ appropriately small, we can absorb the first term on the very right-hand side above into the $(X^{\ell}\psi_{\ell})^2=(X^{n+1}\psi_{\ell})^2$ term on the very right-hand side of \eqref{eq:mainloelliptic} with $n=\ell-1$.

Combining the above estimates and restricting $-2<m<2(\ell-1)+2$, we therefore obtain:
\begin{equation*}
\begin{split}
\int_{\Sigma_{\tau}\cap\{r\geq R_0\}} r^{m+1}(X^{\ell}\psi_{\ell})^2+ r^{m-1}(X^{\ell-1}\psi_{\ell})^2\,d\omega d\uprho \leq&\: C\int_{\Sigma_{\tau}\cap \{r\geq R_0\}}r^{m+3}(X^{\ell+1}\psi_{\ell})^2\,d\omega d\uprho\\
&+C\int_{\Sigma_{\tau}\cap \{r \geq R_0\}} r^{m-1}(X^{\ell-1}{\pi}_{\ell}F)^2\,d\omega d\uprho.
\end{split}
\end{equation*}

We can now apply \eqref{eq:auxellipticell} to estimate the right-hand side further and obtain
\begin{equation*}
\begin{split}
\int_{\Sigma_{\tau}\cap \{r\geq R_0\}} r^{m+1}(X^{\ell}\psi_{\ell})^2+ r^{m-1}(X^{\ell-1}\psi_{\ell})^2\,d\omega d\uprho \leq C\int_{\Sigma_{\tau}} r^{m-1}(X^{\ell-1}{\pi}_{\ell}F)^2+r^{m+1}(X^{\ell}{\pi}_{\ell}F)^2\,d\omega d\uprho.
\end{split}
\end{equation*}
We can easily remove the restriction to $r\geq R_0$ above by applying additionally \eqref{eq:hardyX} together with \eqref{eq:redshiftestfixedell}.

Suppose $\ell\geq 2$. By \eqref{eq:commidn}, we can moreover estimate, for all $n\leq \ell-2$:
\begin{equation*}
(X^{n}\psi_{\ell})^2\leq Cr^2 (X^{n+1}\psi_{\ell})^2+Cr^4(X^{n+2}\psi_{\ell})^2+ (X^{n}{\pi}_{\ell}F)^2,
\end{equation*}
so we can in fact control all lower-order derivatives:
\begin{equation*}
\begin{split}
\sum_{l=0}^{\ell}\int_{\Sigma_{\tau}} r^{m+1-2l}(X^{\ell-l}\psi_{\ell})^2\,d\omega d\uprho \leq C\sum_{l=0}^{\ell}\int_{\Sigma_{\tau}} r^{m+1-2l}(X^{\ell-l}{\pi}_{\ell}F)^2\,d\omega d\uprho,
\end{split}
\end{equation*}
for $-2<m<2\ell$. By rearranging terms, taking $k:=\ell-\frac{m+1}{2}$ and expanding the terms in $(rX)^n{\pi}_{\ell}F$, we conclude that \eqref{eq:mainellipticell1} must hold for $J=0$.
\\
\\

\underline{$J\geq 1$:}\\
\\
We consider the case $J\geq 1$. We will carry out an induction argument in $J$.  First of all, we have obtained above the $J=0$ case. Suppose \eqref{eq:mainellipticell1} holds for some $J\in \N_0$. We will show that it also holds for $J$ replaced with $J+1$. Consider \eqref{eq:mainloelliptic} with $n= \ell+J+1$. Then we write
\begin{equation*}
\begin{split}
 \Delta r^{m}& X^{n+1}\psi_{\ell} X^{n+2}\psi_{\ell}+(n+1)\Delta' r^{m}(X^{n+1}\psi_{\ell})^2\\
 =&\:\left[\left(\ell+J+\frac{3}{2}\right)r\Delta'  -\frac{m}{2}\Delta\right]r^{m-1}(X^{\ell+J+2}\psi_{\ell})^2+X\left(\frac{1}{2}\Delta r^{m} (X^{\ell+J+2}\psi_{\ell})^2\right).
 \end{split}
\end{equation*}
One can easily verify that the factor in front of the first term on the right-hand side is certainly strictly positive when $m<2(\ell+J+1)+2$. We can therefore integrate the right-hand side of \eqref{eq:mainloelliptic} over $\Sigma_{\tau}$ to obtain for $m<2(\ell+J+1)+2$:
\begin{equation*}
\begin{split}
\int_{\Sigma_{\tau}} r^{m+1} (X^{\ell+J+2}\psi_{\ell})^2\,d\omega d\uprho\leq &\:C\int_{\Sigma_{\tau}} r^{m-1} (X^{\ell+J+1}\psi_{\ell})^2\,d\omega d\uprho+ C\int_{S^2_{\tau,r_+}} (X^{\ell+J+1}\psi_{\ell})^2\,d\omega \\
&+C\int_{\Sigma_{\tau}}r^{m-1}(X^{\ell+J+1}{\pi}_{\ell}F)^2\,d\omega d\uprho.
\end{split}
\end{equation*}
We then arrive at \eqref{eq:mainellipticell1} with $J$ replaced by $J+1$ by taking $k=\ell-\frac{m+1}{2}$ and applying the induction assumption.

We finally obtain \eqref{eq:mainellipticell2} by expanding out the terms in $(rX)^n{\pi}_{\ell}F$.
\end{proof}

\subsection{Application of the hierarchy of elliptic estimates}
\label{sec:applelliptichier}
We establish below improved energy decay for energies restricted to $\psi_0$, $\psi_1$ and $\psi_2$ containing additional $r$-weights with negative powers, via an application of an $r^{-2k}$-weighted hierarchy of elliptic estimates from Proposition \ref{prop:mainellipticell}. These are important for establishing almost-sharp pointwise decay in regions of bounded $r$.
\begin{proposition}
\label{prop:edecayelliptic}
Let $\delta>0$ be arbitrarily small and $K,J\in \N_0$. Then there exists a constant $C=C(M,a,R,K,J,\ell,\delta)>0$, such that
\begin{align}
\label{eq:addedecayl0}
\sum_{n=0}^{J}& \int_{\Sigma_{\tau}} r^{-2} J^N[(rX)^nT^K\psi_{0}] \cdot \mathbf{n}_{\tau}\,r^2d\omega d\uprho\\ \nonumber
\leq &\: C \sum_{\ell\in \{0,2\}}(1+\tau)^{-5+\delta-2K}\left[ E_{\ell,J+K+1,\delta}[\psi]+\sum_{j=0}^JE_{\ell,K+1,\delta}[N^j\psi]\right],\\ 
\label{eq:addedecayl1}
\sum_{n=0}^{1+J}& \int_{\Sigma_{\tau}} r^{-3+\eta} J^N[(rX)^nT^K\psi_{1}] \cdot \mathbf{n}_{\tau}\,r^2d\omega d\uprho\\ \nonumber
\leq &\: C (1+\tau)^{-8+\eta+\delta-2K}\left(E_{1,J+K+3,\delta}[\psi]+\sum_{j=0}^{1+J} E_{1,K+2,\delta}[N^j\psi]\right),\\ 
\label{eq:addedecayl2}
\sum_{n=0}^{2+J}& \int_{\Sigma_{\tau}} r^{-1-\eta} J^N[(rX)^nT^K\psi_{2}] \cdot \mathbf{n}_{\tau}\,r^2d\omega d\uprho\\ \nonumber
\leq&\: C (1+\tau)^{-8+\eta+\delta-2K}\left(E_{2,J+K+5,\delta}[\psi]+E_{0,J+K+5,\delta}[\psi]+\sum_{j=0}^{2+J} E_{2,K+3,\delta}[N^j\psi]+\sum_{j=0}^{2+J} E_{0,K+3,\delta}[N^j\psi]\right),\\
\label{eq:addedecayl3}
\int_{\Sigma_{\tau}}& r^{-1} J^N[T^K\psi_{\geq 3}] \cdot \mathbf{n}_{\tau}\,r^2d\omega d\uprho\leq C (1+\tau)^{-9+\eta+\delta-2K}\left(E_{\geq 3, K+1,\delta}[\psi]+E_{1,K+2,\delta}[\psi]+E_{2,K+2,\delta}[\psi]\right).
\end{align}
\end{proposition}
\begin{proof}
We start by considering $\psi_0$. By \eqref{eq:mainellipticell2} with $k=1$, together with \eqref{eq:hardyX} we have that
\begin{equation}
\begin{split}
	\sum_{n=0}^{J} \int_{\Sigma_{\tau}} r^{-2} J^N[(rX)^n\psi_{0}] \cdot \mathbf{n}_{\tau}\,r^2d\omega d\uprho \leq&\:  C\sum_{n=0}^{J} \int_{\Sigma_{\tau}} (X(rX)^n\psi_0)^2+ r^{-2} (T(rX)^n\psi_0)^2 d\omega d\uprho\\
	\leq &\: C\sum_{n=0}^{J} \int_{\Sigma_{\tau}} r^2(XT(rX)^n\psi_0)^2+ (T(rX)^n\psi_0)^2+ r^{-2}(T^2(rX)^n\psi_2)^2 d\omega d\uprho\\
	\leq &\: C\sum_{n=0}^{J} \int_{\Sigma_{\tau}}(J^N[(rX)^nT\psi_0]+J^N[(rX)^nT\psi_2])\cdot \mathbf{n}_{\tau}\,r^2 d\omega d\uprho.    
	\end{split}
\end{equation}
We conclude that \eqref{eq:addedecayl0} must hold by applying \eqref{eq:edecaycommrX}
  for $\ell=0$ and $\ell=2$.
  
We apply \eqref{eq:mainellipticell2} with $k=\frac{3}{2}-\eta$ and then $k=\frac{1}{2}-\eta$, with $\eta>0$ appropriately small, for $\ell=1$ to obtain:
\begin{equation*}
\begin{split}
\sum_{n=0}^{1+J}& \int_{\Sigma_{\tau}} r^{-3+\eta} J^N[(rX)^n\psi_{1}] \cdot \mathbf{n}_{\tau}\,r^2d\omega d\uprho \\
\leq&\: C\sum_{n=0}^{1+J} \int_{\Sigma_{\tau}} r^{-1+\eta} (X(rX)^n\psi_1)^2+ r^{-3+\eta} ((rX)^n\psi_1)^2+r^{-3+\eta} (T(rX)^n\psi_1)^2 d\omega d\uprho\\
\leq&\: C \sum_{n=0}^{1+J}  \int_{\Sigma_{\tau}} r^{1+\eta} (X(rX)^nT\psi_1)^2+ r^{-1+\eta} ((rX)^nT\psi_1)^2+r^{-3+\eta} (T^2(rX)^n\psi_{1})^2\\
&+r^{-3+\eta} (T^2(rX)^n\psi_{3})^2 d\omega d\uprho\\
\leq&\: C \sum_{n=0}^{1+J}  \int_{\Sigma_{\tau}} r^{3+\eta} (X(rX)^nT^2\psi_1)^2+ r^{1+\eta} ((rX)^nT^2\psi_1)^2+r^{-3+\eta} (T^2(rX)^n\psi_{ 3})^2 d\omega d\uprho\\
\leq&\: C\sum_{n=0}^{1+J} \int_{\Sigma_{\tau}} \left(r^{1+\eta} J^N[(rX)^nT^2\psi_{1}]+J^N[(rX)^nT\psi_{ 3}]\right) \cdot \mathbf{n}_{\tau}\,r^2d\omega d\uprho.
\end{split}
\end{equation*}
By combining the energy decay estimates with $r^{2-\delta}$ weights in \eqref{eq:decayl3} and \eqref{eq:almostsharpedecay1} with the energy decay estimates from Corollary \ref{cor:edecaycommrX} via the interpolation inequality in Lemma \ref{lm:interpol}, we therefore obtain:
\begin{equation*}
\sum_{n=0}^{1+J} \int_{\Sigma_{\tau}} r^{-3+\eta} J^N[(rX)^nT^K\psi_{1}] \cdot \mathbf{n}_{\tau}\,r^2d\omega d\uprho\leq C (1+\tau)^{-8-2\ell+\delta+\eta-2K}\left(E_{1,J+K+3,\delta}[\psi]+\sum_{j=0}^{1+J} E_{1,K+2,\delta}[N^j\psi]\right).
\end{equation*}

We can similarly apply \eqref{eq:mainellipticell2} with $k=\frac{1}{2}+\delta$ for $\ell=2$ to obtain:
\begin{equation*}
\begin{split}
\sum_{n=0}^{2+J}& \int_{\Sigma_{\tau}} r^{-1-\delta} J^N[(rX)^n\psi_{2}] \cdot \mathbf{n}_{\tau}\,r^2d\omega d\uprho \\
\leq&\: C \sum_{n=0}^{2+J}  \int_{\Sigma_{\tau}} r^{1-\delta} (X(rX)^nT\psi_2)^2+r^{-1-\delta} ((rX)^nT\psi_2)^2\\
&+r^{-1-\delta} (T^2(rX)^n\psi_{2})^2+r^{-1-\eta} (T^2(rX)^n\psi_{0})^2+r^{-1-\delta} (T^2(rX)^n\psi_{4})^2 \,\omega d\uprho\\
\leq&\:\sum_{n=0}^{2+J} \int_{\Sigma_{\tau}} r^{1-\delta} J^N[(rX)^nT\psi_{2}] \cdot \mathbf{n}_{\tau}+ J^N[(rX)^nT\psi_{4}] \cdot \mathbf{n}_{\tau}\,r^2d\omega d\uprho\\
&+\sum_{n=0}^{2+J} \int_{\Sigma_{\tau}} r^{-1-\delta} (T^2(rX)^n\psi_{0})^2\,d\omega d\uprho.
\end{split}
\end{equation*}
Note that by \eqref{eq:pointwapp3} with $h=\psi_0$, we can further estimate:
\begin{equation*}
\begin{split}
	\sum_{n=0}^{2+J} \int_{\Sigma_{\tau}} r^{-1-\delta} (T^2(rX)^n\psi_{0})^2\,d\omega d\uprho\leq &\: C\sum_{n=0}^{2+J} ||T^2(rX)^n\psi_{0}||^2_{L^{\infty}(\Sigma_{\tau})}\\
	\leq &\: C \sqrt{\int_{\Sigma_{\tau}} r^{-2}J^N[T^2\psi_0]\cdot \mathbf{n}_{\tau}\,r^2d\omega d\uprho }\sqrt{\int_{\Sigma_{\tau}} J^N[T^2\psi_0]\cdot \mathbf{n}_{\tau}\,r^2d\omega d\uprho }.
	\end{split}
\end{equation*}

Hence, after applying \eqref{eq:edecaycommrX}, \eqref{eq:edecaycommrX34} and \eqref{eq:addedecayl0}, we conclude that
\begin{equation*}
\begin{split}
\sum_{n=0}^{2+J}& \int_{\Sigma_{\tau}} r^{-1-\delta} J^N[(rX)^nT^K\psi_{2}] \cdot \mathbf{n}_{\tau}\,r^2d\omega d\uprho\\
\leq&\: C (1+\tau)^{-8-2\ell-2K}\left(E_{2,J+K+5,\delta}[\psi]+E_{0,J+K+5,\delta}[\psi]+\sum_{j=0}^{2+J} E_{2,K+3,\delta}[N^j\psi]+\sum_{j=0}^{2+J} E_{0,K+3,\delta}[N^j\psi]\right).
\end{split}
\end{equation*}

Finally, we note that by \eqref{eq:ellipticredshiftlgeq3} with $\psi$ replaced by $T^K\psi$, we can estimate
\begin{equation*}
\begin{split}
\int_{\Sigma_{\tau}} r^{-1} J^N[T^K\psi_{\geq 3}] \cdot \mathbf{n}_{\tau}\,r^2d\omega d\uprho\leq &\: C\int_{\Sigma_{\tau}} J^N[T^{K+1}\psi_{\geq 3}] \cdot \mathbf{n}_{\tau}\,r^2d\omega d\uprho+C\int_{N_{\tau}}r (L T^{K+1}\phi_{\geq 3})^2 \cdot \mathbf{n}_{\tau}\,r^2d\omega d\uprho\\
&+C\int_{\Sigma_{\tau}} (J^N[T^{K+2}\psi_1]+J^N[T^{K+2}\psi_2])\cdot \mathbf{n}_{\tau}\,r^2d\omega d\uprho.
\end{split}
\end{equation*}
Hence, \eqref{eq:addedecayl3} follows by combining the above equation with the energy decay estimates \eqref{eq:decayl3}, \eqref{eq:almostsharpedecay1} and \eqref{eq:almostsharpedecay2}.
\end{proof}
\begin{remark}
In the proof of Proposition \ref{prop:edecayelliptic}, we applied  \eqref{eq:mainellipticell2} with $k=\frac{1}{2}+\delta$ in the $\ell=2$ case. Note that \eqref{eq:mainellipticell2} in fact applies with $\frac{1}{2}-\eta\leq k\leq \frac{5}{2}-\eta$ when $\ell=2$. However, as we already applied an almost-sharp decay estimate for $T^2\psi_0$, which does not further improve by considering additional weights in $r^{-1}$, we cannot exploit the full hierarchy of elliptic estimates to improve the above decay rate when considering $r^{-5+\eta} J^N[(rX)^n\psi_{2}]$ instead of $r^{-1-\eta} J^N[(rX)^n\psi_{2}]$. \textbf{This is a manifestation of angular mode coupling limiting the sharp decay rate of the $\ell=2$ mode.}
\end{remark}
\section{Pointwise decay estimates}
\label{sec:poinwdecay}
We apply the energy decay estimates of Section \ref{sec:edecay} and \ref{sec:applelliptichier} to obtain $L^{\infty}$ estimates for $\psi$, $r\psi$ and various higher-order quantities of the form $(rX)^J T^K\psi$. These pointwise estimates will be used in subsequent sections to determine the \underline{precise} late-time behaviour of $\psi_{\geq \ell}$ with $\ell=0,1,2$.
\begin{proposition}
\label{prop:pointwisedecayl0}
Let $\delta>0$ be arbitrarily small and $K,J\in \N_0$. Then there exists a constant $C=C(M,a,R,\delta,K,J)>0$, such that
\begin{align}
\label{eq:pointwl01}
||T^K\phi_0||_{L^{\infty}(\Sigma_{\tau})}\leq&\: C (1+\tau)^{-1-K+2\delta}\sqrt{E_{0,K,\delta}[\psi]},\\
\label{eq:pointwl01b}
|| (rX)^JT^K\phi_0||_{L^{\infty}(\Sigma_{\tau})}\leq&\: C (1+\tau)^{-1-K+2\delta}\sqrt{\sum_{\ell\in\{0,2\}} E_{\ell,K+J,\delta}[\psi]+\sum_{j=0}^{J}E_{\ell,K,\delta}[N^j\psi]},\\
\label{eq:pointwl02}
||\sqrt{r} T^K\psi_0||_{L^{\infty}(\Sigma_{\tau})}\leq&\: C (1+\tau)^{-\frac{3}{2}-K+2\delta}\sqrt{E_{0,K,\delta}[\psi]},\\
\label{eq:pointwl03}
|| T^K\psi_0||_{L^{\infty}(\Sigma_{\tau})}\leq &\: C (1+\tau)^{-2-K+2\delta}\sqrt{\sum_{\ell\in\{0,2\}}E_{\ell,K+1,\delta}[\psi]},\\
\label{eq:pointwl04}
|| (rX)T^K\psi_0||_{L^{\infty}(\Sigma_{\tau})}\leq&\: C (1+\tau)^{-2-K+2\delta}\sqrt{\sum_{\ell\in\{0,2\}} E_{\ell,K+2,\delta}[\psi]+\sum_{j=0}^{1}E_{\ell,K+1,\delta}[N^j\psi]},
\end{align}
\end{proposition}
\begin{proof}
We obtain \eqref{eq:pointwl01} and \eqref{eq:pointwl02} by applying \eqref{eq:pointwapp1} and \eqref{eq:pointwapp2} with $f=\phi_0$ and $h=\psi_0$ together with the energy decay estimate \eqref{eq:almostsharpedecay0}. 

The estimate \eqref{eq:pointwl03} follows by applying \eqref{eq:pointwapp3} with $h=\psi_0$ in combination with \eqref{eq:almostsharpedecay0} and \eqref{eq:addedecayl0}. In order to obtain \eqref{eq:pointwl01b} and \eqref{eq:pointwl04}, we  instead apply the higher-order energy decay estimate \eqref{eq:edecaycommrX} in combination with \eqref{eq:pointwapp3}.
\end{proof}

\begin{proposition}
\label{prop:pointwisedecayl1}
Let $\delta>0$ be arbitrarily small and $K,J\in \N_0$. Then there exists a constant $C=C(M,a,R,K,J,\delta)>0$, such that for $m=0,1$:
\begin{align}
\label{eq:pointwl11}
||T^K\phi_1||_{L^{\infty}(\Sigma_{\tau})}\leq&\: C (1+\tau)^{-2-K+2\delta}\sqrt{E_{1,K,\delta}[\psi]},\\
\label{eq:pointwl12}
||\sqrt{r} T^K\psi_1||_{L^{\infty}(\Sigma_{\tau})}\leq&\: C (1+\tau)^{-\frac{5}{2}-K+2\delta}\sqrt{E_{1,K,\delta}[\psi]},\\
\label{eq:pointwl1b}
||T^K\phi^{(1)}_1||_{L^{\infty}(N_{\tau})}&+(1+\tau)^{\frac{1}{2}(1+\delta)}||r^{-\frac{1}{2}(1+\delta)}T^K\phi^{(1)}_1||_{L^{\infty}(N_{\tau})}\\ \nonumber
\leq&\: C (1+\tau)^{-1-K+2\delta}\sqrt{E_{1,K,\delta}[\psi]},\\
\label{eq:pointwl1c}
||(rX)^J T^K\phi^{(m)}_1||_{L^{\infty}(N_{\tau})}&+(1+\tau)^{\frac{1}{2}(1+\delta)}||r^{-\frac{1}{2}(1+\delta)}(rX)^JT^K\phi^{(m)}_1||_{L^{\infty}(N_{\tau})}\\ \nonumber
\leq&\: C (1+\tau)^{-2+j-K+2\delta}\sqrt{E_{1,K+J,\delta}[\psi]+\sum_{j=0}^{J} E_{1,K,\delta}[N^j\psi]},\\
\label{eq:pointwl13}
r^{-1}|| T^K\psi_1||_{L^{\infty}(\Sigma_{\tau})}&+|| XT^K\psi_1||_{L^{\infty}(\Sigma_{\tau})}\\ \nonumber
\leq&\: C (1+\tau)^{-4-K+2\delta}\sqrt{E_{1,K+3,\delta}[\psi]+\sum_{j=0}^{1} E_{1,K+2,\delta}[N^j\psi]},\\
\label{eq:pointwl14}
|| rX^2T^K\psi_1||_{L^{\infty}(\Sigma_{\tau})}\leq&\: C (1+\tau)^{-4-K+2\delta} \sqrt{E_{1,K+4,\delta}[\psi]+\sum_{j=0}^2E_{1,K+2,\delta}[N^j\psi]}.
\end{align}
\end{proposition}
\begin{proof}
We note first of all that by a standard Sobolev inequality on $\s^2$ together with \eqref{eq:poincare3}, the $L^{\infty}(\s^2)$ norm of $\psi_{1}$ and can be uniformly bounded by its $L^2(\s^2)$ norm.

	We obtain \eqref{eq:pointwl11} and \eqref{eq:pointwl12} by applying \eqref{eq:pointwapp1} and \eqref{eq:pointwapp2} with $f=\phi_1$ and $h=\psi_1$ together with the energy decay estimate \eqref{eq:almostsharpedecay1}. Similarly, \eqref{eq:pointwl1b} follows from \eqref{eq:pointwapp1} and \eqref{eq:pointwapp1b} with $f=\phi^{(1)}$ and \eqref{eq:almostsharpedecay1}.
	
	The estimate \eqref{eq:pointwl13} follows by applying \eqref{eq:pointwapp4} with $k=1$, $h=\psi_1$ and $h=rX\psi_1$, together with the energy decay estimates
	 \eqref{eq:almostsharpedecay1} and \eqref{eq:addedecayl1} and Lemma \ref{lm:interpol}. The estimate \eqref{eq:pointwl14} follows by additionally taking $h=(rX)^2\psi_1$.
	 
	 Finally, \eqref{eq:pointwl1c} follows as \eqref{eq:pointwl1b} but via the higher-order energy decay estimates \eqref{eq:addedecayl1}.
	\end{proof}

\begin{proposition}
\label{prop:pointwisedecayl2}
Let $\delta>0$ be arbitrarily small and $K,J\in \N_0$. Then there exists a constant $C=C(M,a,R,K,J,\delta)>0$, such that for $m=0,1,2$:
\begin{align}
\label{eq:pointwl21}
||r T^K\psi_2||_{L^{\infty}(\Sigma_{\tau})}\leq&\: C (1+\tau)^{-3-K+2\delta}\sqrt{E_{2,K,\delta}[\psi]},\\
\label{eq:pointwl22}
||\sqrt{r} T^K\psi_2||_{L^{\infty}(\Sigma_{\tau})}\leq&\: C (1+\tau)^{-\frac{7}{2}-K+2\delta}\sqrt{E_{2,K,\delta}[\psi]},\\
\label{eq:pointwl2b}
||T^K\phi^{(m)}_2||_{L^{\infty}(N_{\tau})}+&(1+\tau)^{\frac{1}{2}+\delta}||r^{-\frac{1}{2}(1+\delta)}T^K\phi^{(m)}_2||_{L^{\infty}(N_{\tau})}\leq C (1+\tau)^{-1-(2-m)-K+2\delta}\sqrt{E_{2,K,\delta}[\psi]},\\
\label{eq:pointwl2c}
||(rX)^JT^K\phi^{(m)}_2||_{L^{\infty}(N_{\tau})}+&(1+\tau)^{\frac{1}{2}+\delta}||r^{-\frac{1}{2}(1+\delta)}(rX)^J T^K\phi^{(m)}_2||_{L^{\infty}(N_{\tau})}\\ \nonumber
\leq&\: C (1+\tau)^{-1-(2-m)-K+2\delta}\sqrt{E_{2,K+J,\delta}[\psi]+\sum_{j=0}^J E_{2,K,\delta}[N^j\psi]},\\
\label{eq:pointwl23}
\sum_{n=0}^{J+2}|| (rX)^nT^K\psi_2||_{L^{\infty}(\Sigma_{\tau})}\leq&\: C (1+\tau)^{-4-K+2\delta}\Bigg(\sqrt{E_{2,K+J+5,\delta}[\psi]+\sum_{j=0}^{2} E_{2,K+J+3,\delta}[N^j\psi]}\\ \nonumber
&+\sqrt{E_{0,K+J+5,\delta}[\psi]+\sum_{j=0}^{2} E_{0,K+J+3,\delta}[N^j\psi]\Bigg)}.
\end{align}
\end{proposition}
\begin{proof}
Note that as in the $\ell=1$ case the $L^{\infty}(\s^2)$ norm of $\psi_{2}$ and can be uniformly bounded by its $L^2(\s^2)$ norm.

	We obtain \eqref{eq:pointwl21} and \eqref{eq:pointwl22} by applying \eqref{eq:pointwapp1} and \eqref{eq:pointwapp2} with $f=\phi_2$ and $h=\psi_2$ together with the energy decay estimate \eqref{eq:almostsharpedecay2}. Similarly, \eqref{eq:pointwl2b} follows from \eqref{eq:pointwapp1} and \eqref{eq:pointwapp1b} with $f=\phi^{(m)}$, $m=1,2$, and \eqref{eq:almostsharpedecay1}.
	
	By applying \eqref{eq:pointwapp4} for $k=0$, together with \eqref{eq:addedecayl2}, \eqref{eq:edecaycommrX} and Lemma \ref{lm:interpol}, we obtain \eqref{eq:pointwl23}.
	
	Finally, \eqref{eq:pointwl2c} follows as \eqref{eq:pointwl2b}, but via the higher-order energy decay estimates \eqref{eq:addedecayl1}.
\end{proof}

\begin{proposition}
\label{prop:pointwisedecayl3}
Let $\delta>0$ be arbitrarily small and $K,J\in \N_0$. Let $Q$ denote the Carter operator:
\begin{equation*}
Q=\slashed{\Delta}_{\s^2}+a^2\sin^2\theta T^2-\Phi^2.
\end{equation*}

Then there exists a constant $C=C(M,a,R,K,J,\delta)>0$, such that for $m=0,1,2,3$:
\begin{align}
\label{eq:pointwl31}
||T^K\phi_{\geq 3}||_{L^{\infty}(\Sigma_{\tau})}\leq&\: C (1+\tau)^{-\frac{7}{2}-K+2\delta}\sqrt{\sum_{j_1+j_2+j_3\leq 1}E_{\geq 3,K,\delta}[Q^{j_1}T^{2j_2}\Phi^{2j_3}\psi]},\\
\label{eq:pointwl32}
||\sqrt{r} T^K\psi_{\geq 3}||_{L^{\infty}(\Sigma_{\tau})}\leq&\: C (1+\tau)^{-4-K+2\delta}\sqrt{\sum_{j_1+j_2+j_3\leq 1}E_{\geq 3,K,\delta}[Q^{j_1}T^{2j_2}\Phi^{2j_3}\psi]},\\
\label{eq:pointwl33}
||T^K\psi_{\geq 3}||_{L^{\infty}(\Sigma_{\tau})}\leq&\: C (1+\tau)^{-\frac{9}{2}-K+2\delta}\Bigg[\sum_{j_1+j_2+j_3\leq 1}E_{\geq 3,K+1,\delta}[Q^{j_1}T^{2j_2}\Phi^{2j_3}\psi]+E_{1,K+2,\delta}[Q^{j_1}T^{2j_2}\Phi^{2j_3}\psi]\\ \nonumber
&+E_{2,K+2,\delta}[Q^{j_1}T^{2j_2}\Phi^{2j_3}\psi]\Bigg].
\end{align}
We also have the following additional estimates for $\ell=3,4$:
\begin{align}
\label{eq:pointwl3c}
||(rX)^JT^K\phi^{(m)}_{\ell}||_{L^{\infty}(N_{\tau})}+&(1+\tau)^{\frac{1}{2}+\delta}||r^{-\frac{1}{2}(1+\delta)}(rX)^J T^K\phi^{(m)}_{\ell}||_{L^{\infty}(N_{\tau})}\\ \nonumber
\leq&\: C (1+\tau)^{-\frac{7}{2}+m-K+2\delta}\sqrt{E_{\geq 3,K+J,\delta}[\psi]+\sum_{j=0}^J E_{\geq 3,K,\delta}[N^j\psi]},\\
\label{eq:pointwl34}
||r^2 XT^K\psi_{\ell}||_{L^{\infty}(\Sigma_{\tau})}\leq &\: C (1+\tau)^{-\frac{7}{2}-K+2\delta}\Bigg[\sqrt{E_{\geq 3,K+1,\delta}[\psi]}+\sum_{j=0}^1\sqrt{E_{\geq 3,K,\delta}[\psi]}\Bigg].
\end{align}
\end{proposition}
\begin{proof}
In order to bound the $L^{\infty}(\s^2)$ norm of $\psi_{\geq 3}$ we consider the Carter operator $Q$ and note that:
\begin{equation*}
\begin{split}
||\psi_{\geq 3}||^2_{L^{\infty}(\s^2)}\leq &\: C \int_{\s^2} \psi_{\geq 3}^2+|\snabla_{\s^2}\psi_{\geq 3}|^2+|\snabla_{\s^2}^2\psi_{\geq 3}|^2\,d\omega\\
\leq &\: C \int_{\s^2} \psi_{\geq 3}^2+(\slashed{\Delta}_{\s^2}\psi_{\geq 3})^2\,d\omega\\
\leq &\: C \int_{\s^2} \psi_{\geq 3}^2+(Q\psi_{\geq 3})^2+(T^2\psi_{\geq 3})^2+(\Phi^2\psi_{\geq 3})^2\,d\omega.
\end{split}
\end{equation*}
Furthermore, since $[T,\square_g]=[\Phi,\square_g]=[Q,\square_g]=0$, all estimates derived for $\psi$ automatically hold for $\psi$ replaced by $Q\psi$, $T^2\psi$ or $\Phi^2\psi$. It therefore remains only to establish control of the $L^2(\s^2)$-norm of $\psi_{\geq 3}$.

We obtain \eqref{eq:pointwl31} and \eqref{eq:pointwl32} by applying \eqref{eq:pointwapp1} and \eqref{eq:pointwapp2} with $f=\phi_{\geq 3}$ and $h=\psi_{\geq 3}$ together with the energy decay estimate \eqref{eq:decayl3}. The estimate \eqref{eq:pointwl34} follows similarly by taking $h=rX \psi_{3}$ or $h=rX \psi_{4}$ and applying the energy decay estimate \eqref{eq:edecaycommrX} with $\ell=3,4$, respectively.

Finally, \eqref{eq:pointwl33} and \eqref{eq:pointwl3c} follow from \eqref{eq:pointwapp4} with $k=0$ and $h=\psi_{\geq 3}$,  and \eqref{eq:pointwapp1}, \eqref{eq:pointwapp1b}, respectively,  combined with \eqref{eq:addedecayl3}.
\end{proof}

In the following proposition we establish \underline{improved} decay rates for certain $X$-derivatives of $\psi_{\ell}$ compared to the decay of $\psi_{\ell}$ itself. These estimates will be crucial for propagating asymptotics of $\psi_{\ell}$ from a region near infinity to the rest of the spacetime.

\begin{proposition}
\label{prop:fasterdecayforpropag}
Let $\delta>0$ be arbitrarily small and $K\in \N_0$. Then there exists a constant $C=C(M,a,R,K,\delta)>0$, such that:
\begin{align}
\label{eq:fasterdecayl0}
||XT^K\psi_0||_{L^{\infty}(\Sigma_{\tau})}\leq&\: C (1+\tau)^{-3-K+2\delta}\left[\sum_{\ell\in\{0,2\}} \sqrt{E_{\ell,K+2,\delta}[\psi]}+\sum_{j=0}^1\sqrt{E_{\ell,K+1,\delta}[N^j\psi]}\right],\\
\label{eq:fasterdecayl1}
||X^2T^K\psi_1||_{L^{\infty}(\Sigma_{\tau})}\leq&\:  C (1+\tau)^{-5-K+2\delta}\Bigg[\sqrt{E_{1,K+5,\delta}[\psi]}+\sum_{j=0}^2\sqrt{E_{1,K+3,\delta}[N^j\psi]}+\sqrt{E_{\geq 3,K+2,\delta}[\psi]}\\ \nonumber
&+\sum_{j=0}^1\sqrt{E_{\geq 3,K+2,\delta}[N^j\psi]}\Bigg],\\
\label{eq:fasterdecayl2}
||r^2X^3T^K\psi_2||_{L^{\infty}(\Sigma_{\tau})}+&||rX^2T^K\psi_2||_{L^{\infty}(\Sigma_{\tau})}\\ \nonumber
\leq &\:   C(1+\tau)^{-5-K+2\delta}\Bigg[\sqrt{E_{2,K+7,\delta}[\psi]}+\sum_{j=0}^{3} \sqrt{E_{2,K+4,\delta}[N^j\psi]}\\ \nonumber
&+\sqrt{E_{0,K+7,\delta}[\psi]}+\sum_{j=0}^{3} \sqrt{E_{0,K+4,\delta}[N^j\psi]}\\ \nonumber
&+\sqrt{E_{\geq 3,K+2,\delta}[\psi]}+\sum_{j=0}^1\sqrt{E_{\geq 3,K+2,\delta}[N^j\psi]} \Bigg].
\end{align}
\end{proposition}
\begin{proof}
Recall from \eqref{eq:fasterdecatXlpsi} that
\begin{equation*}
\begin{split}
||X^{\ell+1}\psi_{\ell}||_{L^2(S^2_{\tau,r'})}\leq&\: C\sum_{k=0}^{\ell+1}\sup_{r_+\leq r \leq r'}||r^{-\ell}(rX)^{k}T\psi_{\ell} ||_{L^2(S^2_{\tau,r})}+C\sum_{k=0}^{\ell}\sup_{r_+\leq r \leq r'}||r^{-2-\ell}(rX)^{k}T^2\psi_{\ell} ||_{ L^2(S^2_{\tau,r})}\\
&+Ca^2\sup_{r_+\leq r \leq r'}||r^{-2}X^{\ell}T^2\psi_{\ell-2} ||_{L^2(S^2_{\tau,r})}+Ca^2\sup_{r_+\leq r \leq r'}||r^{-2}X^{\ell}T^2\psi_{\ell+2} ||_{L^2(S^2_{\tau,r})}.
\end{split}
\end{equation*}
Hence,
\begin{equation*}
\begin{split}
||XT^K\psi_0||_{L^{\infty}(\Sigma_{\tau})}\leq&\: C\sum_{k=0}^{1}||(rX)^{k}T\psi_{0} ||_{L^{\infty}(\Sigma_{\tau})}+C||r^{-2}T^2\psi_{0} ||_{L^{\infty}(\Sigma_{\tau})}+Ca^2||r^{-2}T^2\psi_{2} ||_{L^{\infty}(\Sigma_{\tau})}\\
\leq &\:  C (1+\tau)^{-3-K+2\delta}\sqrt{\sum_{\ell\in\{0,2\}} E_{\ell,K+2,\delta}[\psi]+\sum_{j=0}^1E_{\ell,K+1,\delta}[N^j\psi]}.
\end{split}
\end{equation*}
Furthermore, by \eqref{eq:pointwl13} and \eqref{eq:pointwl14}.
\begin{equation*}
\begin{split}
||X^2T^K\psi_1||_{L^{\infty}(\Sigma_{\tau})}\leq&\: C\sum_{k=0}^{2}||r^{-1}(rX)^{k}T\psi_{1} ||_{L^{\infty}(\Sigma_{\tau})}+C\sum_{k=0}^1||r^{-3}(rX)^{k}T^2\psi_{1} ||_{L^{\infty}(\Sigma_{\tau})}+Ca^2||r^{-2}XT^2\psi_{3} ||_{L^{\infty}(\Sigma_{\tau})}\\
\leq &\:  C (1+\tau)^{-5-K+2\delta}\Bigg[\sqrt{E_{1,K+5,\delta}[\psi]}+\sum_{j=0}^2\sqrt{E_{1,K+3,\delta}[N^j\psi]}+\sqrt{E_{\geq 3,K+2,\delta}[\psi]}\\
&+\sum_{j=0}^1\sqrt{E_{\geq 3,K+2,\delta}[N^j\psi]}\Bigg].
\end{split}
\end{equation*}
We apply \eqref{eq:fasterdecatXlpsiv2} and \eqref{eq:pointwl23} to obtain:
\begin{equation*}
\begin{split}
||r^2X^3T^K\psi_2||_{L^{\infty}(\Sigma_{\tau})}\leq&\: C\sum_{k=0}^{3}||(rX)^{k}T\psi_{2} ||_{L^{\infty}(\Sigma_{\tau})}+C\sum_{k=0}^2||r^{-2}(rX)^{k}T^2\psi_{2} ||_{L^{\infty}(\Sigma_{\tau})}\\
&+Ca^2||XT^2\psi_{0} ||_{L^{\infty}(\Sigma_{\tau})}+Ca^2||XT^2\psi_{4} ||_{L^{\infty}(\Sigma_{\tau})}\\
\leq &\:   C(1+\tau)^{-5-K+2\delta}\Bigg[\sqrt{E_{2,K+7,\delta}[\psi]}+\sum_{j=0}^{3} \sqrt{E_{2,K+4,\delta}[N^j\psi]}\\ 
&+\sqrt{E_{0,K+7,\delta}[\psi]}+\sum_{j=0}^{3} \sqrt{E_{0,K+4,\delta}[N^j\psi]}\\
&+\sqrt{E_{\geq 3,K+2,\delta}[\psi]}+\sum_{j=0}^1\sqrt{E_{\geq 3,K+2,\delta}[N^j\psi]} \Bigg].
\end{split}
\end{equation*}

We obtain an estimate for $||r X^2T^K\psi_2||_{L^{\infty}(\Sigma_{\tau})}$ by integrating from $\rho=\infty$.
\end{proof}

\section{Elliptic theory of time inversion}
\label{sec:timeinv}
The aim of this section is to construct a solution $\widetilde{\psi}$ to \eqref{eq:waveeq}, such that
\begin{equation*}
T\widetilde{\psi}=\psi,
\end{equation*}
with $\psi$ another solution to \eqref{eq:waveeq} arising from smooth, compactly supported initial data on $\Sigma_0$. We will denote $T^{-1}\psi=\widetilde{\psi}$. This propositions in this section are self-contained and independent from the estimates in Sections \ref{sec:rpest}--\ref{sec:poinwdecay}.

The construction of $\widetilde{\psi}$ relies on the invertibility of the differential operator $\mathcal{L}$:
\begin{equation*}
\mathcal{L}f=X(\Delta Xf)+2aX\Phi f+\slashed{\Delta}_{\s^2}f,
\end{equation*}
which was introduced in Section \ref{sec:elliptic},  acting on suitable Hilbert spaces.

Let $\mathbf{H}_k$ denote the completion of the space
\begin{equation*}
C^{\infty}_{\rm rad}(\Sigma_0):=\{f\in C^{\infty}(\Sigma_0)\,|\,(r^2+a^2)^{\frac{1}{2}}f\in C^{\infty}(\widehat{\Sigma}_0)\},
\end{equation*}
with respect to the norm $||\cdot||_k$, defined as follows:
\begin{equation*}
||f||_k^2:=\sum_{k_1+k_2\leq k} \int_{\Sigma_0} |\snabla_{\s^2}^{k_1}(rX)^{k_2}f|^2\,d\omega d\uprho.
\end{equation*}

Let $D_k(\mathcal{L})$ denote the closure of $C^{\infty}_{\rm rad}(\Sigma_0)$ under the norm:
\begin{equation*}
||f||_{k}+||\mathcal{L}f||_{k}.
\end{equation*}
Then
\begin{equation*}
\mathcal{L}: D_k(\mathcal{L})\to \mathbf{H}_k
\end{equation*}
is a densely defined, closed, linear operator.

\subsection{Inverting $\mathcal{L}$}
In this section, we will establish invertibility and regularity properties of $\mathcal{L}$. The strategy for obtaining invertibility of $\mathcal{L}$ can be compared to the strategy of obtaining invertibility of the resolvent operators considered in Section 4 of \cite{warn15}. In particular, as in \cite{warn15}, the use of (elliptic) redshift estimates will be important.

 The following proposition contains the key estimates that are relevant for invertibility:
\begin{proposition}
\label{prop:keyestimateTinv}
Let $(r^2+a^2)^{\frac{1}{2}}f\in C^{\infty}(\widehat{\Sigma}_0)$. Then
\begin{equation}
\label{eq:keytimeinvest}
\int_{\Sigma_0}f^2+r^2 (Xf)^2+|\snabla_{\s^2}f|^2 \,d\omega d\uprho\leq C \int_{\Sigma_0} (\mathcal{L}f)^2\,d\omega d\uprho.
\end{equation}
More generally,
\begin{equation}
\label{eq:keytimeinvestho}
||f||_{k}+||\snabla_{\s^2}f||_{k}+||r Xf||_{k}\leq C||\mathcal{L}f||_{k}.
\end{equation}
Furthermore, $\mathbf{H}_{k+1}\subset D_{k}(\mathcal{L})$ and the equation \eqref{eq:keytimeinvestho} holds also for all $f\in D_{k}(\mathcal{L})$.
\end{proposition}
\begin{proof}
Let $q\in \R$. After applying the Leibniz rule, we obtain the following identity:
\begin{equation}
\label{eq:keytimeinvid}
\begin{split}
r^qXf \mathcal{L}f=&\:\frac{1}{2} X( r^q\Delta(Xf)^2)+\frac{1}{2}\left(\Delta'-q\Delta r^{-1}\right) r^q(Xf)^2+a\Phi(r^q(Xf)^2)+\textnormal{div}_{\s^2}(r^qXf \snabla_{\s^2}f)\\
&-\frac{1}{2}X(r^q|\snabla_{\s^2}f|^2)+\frac{1}{2}q r^{q-1}|\snabla_{\s^2}f|^2.
\end{split}
\end{equation}
Integrating the above identity with $q=1$ over $\Sigma_0$ and using that $(r^2+a^2)^{\frac{1}{2}}f\in C^{\infty}(\widehat{\Sigma}_0)$, we therefore obtain:
\begin{equation*}
\int_{\Sigma_0} \frac{1}{2}(r^2-a^2) (Xf)^2+\frac{1}{2}|\snabla_{\s^2}f|^2 \,d\omega d\uprho+\frac{1}{2}\int_{\Sigma_0\cap\{\uprho=r_+\}}r|\snabla_{\s^2}f|^2\,d\omega\leq \int_{\Sigma_0} r |Xf| |\mathcal{L}f|\,d\omega d\uprho
\end{equation*}
By applying a weighted Young's inequality to right-hand side and absorbing the resulting $(Xf)^2$ term into the left-hand side, we are left with:
\begin{equation}
\label{eq:keytimeinvest0}
\int_{\Sigma_0}r^2 (Xf)^2+|\snabla_{\s^2}f|^2 \,d\omega d\uprho+\int_{\Sigma\cap\{\uprho=r_+\}}|\snabla_{\s^2}f|^2\,d\omega\leq C \int_{\Sigma_0}  (\mathcal{L}f)^2\,d\omega d\uprho.
\end{equation}
By applying \eqref{eq:hardyX} we obtain \eqref{eq:keytimeinvest}.

In order to obtain \eqref{eq:keytimeinvestho}, we first observe that:
\begin{equation*}
X^{k_2} (\mathcal{L}f-k_2(k_2+1)f)=X(\Delta X^{k_2+1}f)+k_2\Delta' X^{k_2+1}f+ 2a \Phi X^{k_2+1}f+\slashed{\Delta}_{\s^2} X^{k_2}f.
\end{equation*}
We then multiply both sides by $(-1)^{k_1}\slashed{\Delta}_{\s^2}^{k_1}r^{2k_2+1}X^{k_2+1} f$ and apply the Leibniz rule multiple times to obtain:
\begin{equation*}
\begin{split}
r^{2k_2+1}&\snabla_{\s^2}^s \slashed{\Delta}_{\s^2}^{\frac{k_1-s}{2}}X^{k_2+1} f \cdot \snabla_{\s^2}^s \slashed{\Delta}_{\s^2}^{\frac{k_1-s}{2}}X^{k_2}( \mathcal{L}f-k_2(k_2+1)f)=\left(k_2+\frac{1}{2}\right)[\Delta'-\Delta r^{-1}] r^{2k_2+1}|\snabla_{\s^2}^s \slashed{\Delta}_{\s^2}^{\frac{k_1-s}{2}}X^{k_2+1}f|^2\\
&+\left(k_2+\frac{1}{2}\right)r^{2k_2}|\snabla_{\s^2}^{1-s} \slashed{\Delta}_{\s^2}^{\frac{k_1+s}{2}}X^{k_2}f|^2-\frac{1}{2}X(r^{2k_2+1}|\snabla_{\s^2}^{1-s} \slashed{\Delta}_{\s^2}^{\frac{k_1+s}{2}}X^{k_2}f|^2)+\ldots
\end{split}
\end{equation*}
where $s=0$ if $k_1$ is even and $s=1$ if $k_1$ is odd, and the terms in $\ldots$ on the right-hand side are total derivatives that vanish after integrating over $\Sigma_0$. We therefore obtain:
\begin{equation}
\label{eq:esthoder}
\begin{split}
\int_{\Sigma_0}& r^{2k_2+2} |\snabla_{\s^2}^s \slashed{\Delta}_{\s^2}^{\frac{k_1-s}{2}}X^{k_2+1}f|^2+ r^{2k_2} |\snabla_{\s^2}^{1-s} \slashed{\Delta}_{\s^2}^{\frac{k_1+s}{2}}X^{k_2}f|^2\,d\omega d\uprho\\
\leq&\:  C \int_{\Sigma_0}r^{2k_2} |\snabla_{\s^2}^s \slashed{\Delta}_{\s^2}^{\frac{k_1-s}{2}}X^{k_2}(\mathcal{L}f-k_2(k_2+1)f)|^2\,d\omega d\uprho.
\end{split}
\end{equation}
We then obtain \eqref{eq:keytimeinvestho} via \eqref{eq:angdercontrol} and a straightforward induction argument in $k_2$.

Now suppose $f\in D_k(\mathcal{L})$. By definition, there exists a sequence $\{f_j\}$ such that $(r^1+a^2)^{\frac{1}{2}}f_k\in C^{\infty}(\widehat{\Sigma}_0)$ and $f_j \to f$, $\mathcal{L}f_k\to \mathcal{L}f$ with respect to $||\cdot||_k$. Applying the estimates above to the difference $f_j-f_i$, $j,i\in \N_0$, we can conclude that $\{f_j\}$ is Cauchy with respect to $||\cdot||_{k+1}$ and hence the limit $f\in \mathbf{H}_{k+1}$. It follows moreover that \eqref{eq:keytimeinvestho} holds for $f\in D_k(\mathcal{L})$.
\end{proof}

Let  $\mathcal{L}^*: D_0(\mathcal{L}^*) \to \mathbf{H}_0$ denote the Hilbert space adjoint operator of $\mathcal{L}$ with respect to the standard $L^2$ norm on $[r_+,\infty)\times \s^2$, i.e. $v\in D_0(\mathcal{L}^*)$ if and only if: there exists a $f\in \mathbf{H}_0$ such that for all $u \in C^{\infty}_c({\Sigma}_0)$
\begin{equation*}
\la \mathcal{L}u, v\ra_{L^2([r_+,\infty)\times \s^2)}=\la u, f\ra_{L^2([r_+,\infty)\times \s^2)}
\end{equation*}
and $\mathcal{L}^*v:=f$.

We moreover define the operator $\mathcal{L}^{\dag}: D_0(\mathcal{L}^{\dag}) \to \mathbf{H}_k$, with $\mathcal{L}^{\dag}u=\mathcal{L}u$ and where $D_0(\mathcal{L}^{\dag})$ is the closure of 
\begin{equation*}
C^{\infty}_{\rm rad,*}(\Sigma_0)=\{f\in C^{\infty}_{\rm rad}(\Sigma_0)\,|\, f|_{r=r_+}=0\}
\end{equation*}
under the norm $||u||_0+||\mathcal{L}^{\dag}u||_0$.
\begin{lemma}
\label{lm:adjoints}
The following identity holds: $\mathcal{L}^{\dag}=\mathcal{L}^*$, with $D_0(\mathcal{L}^{\dag})=D_0(\mathcal{L}^*)$.
\end{lemma}
\begin{proof}
We will establish the equivalent statement: $(\mathcal{L}^{\dag})^*=\mathcal{L}$ with $D_0((\mathcal{L}^{\dag})^*)=D_0(\mathcal{L})$.

By definition of the adjoint $(\mathcal{L}^{\dag})^*$, we have that for all $u\in C^{\infty}_{\rm rad,*}(\Sigma_0)$ and $v\in D_0((\mathcal{L}^{\dag})^*)$:
\begin{equation}
\label{eq:adjid1}
\la \mathcal{L}^{\dag}u, v\ra_{L^2}=\la u, (\mathcal{L}^{\dag})^* v\ra_{L^2}.
\end{equation}
Now, let $u\in C^{\infty}_{\rm rad,*}(\Sigma_0)$ and $v\in C^{\infty}_{\rm rad}(\Sigma_0)$. Then we can integrate by parts and use that $u|_{r=r_+}=0$ and $u,v$ are compactly supported to obtain:
\begin{equation}
\label{eq:adjid2}
\la \mathcal{L}^{\dag}u, v\ra_{L^2}=\la u, \mathcal{L} v\ra_{L^2}.
\end{equation}

Since $\mathcal{L}$ and $\mathcal{L}^{\dag}$ are closed operators, by construction, we have that \eqref{eq:adjid2} must also hold for $u\in D_0(\mathcal{L}^{\dag})$ and $v\in D_0(\mathcal{L})$. Applying both \eqref{eq:adjid1} and \eqref{eq:adjid2}, we can therefore infer that:
\begin{equation*}
D_0(\mathcal{L})\subseteq D_0((\mathcal{L}^{\dag})^*)
\end{equation*}
and $(\mathcal{L}^{\dag})^*|_{D_0(\mathcal{L})}=\mathcal{L}$.

In order to conclude that $D_0(\mathcal{L})= D_0((\mathcal{L}^{\dag})^*)$, it remains to show that for any element $v\in D_0((\mathcal{L}^{\dag})^*)$, there exists a sequence $(v_k)$, with $v_k\in C_{\rm rad}^{\infty}(\Sigma_0)$ such that
\begin{equation}
\label{eq:convadj}
||v-v_k||_{L^2([r_+,\infty)\times \s^2)}+||\mathcal{L}(v-v_k)||_{L^2([r_+,\infty)\times \s^2)}\to 0 \quad \textnormal{as $k\to \infty$.}
\end{equation}
First, observe that if $v\in D_0((\mathcal{L}^{\dag})^*)$, then by \eqref{eq:adjid1} combined with \eqref{eq:adjid2}, $\mathcal{L}v\in \mathbf{H}_0$ exists in a weak sense.  Then, we can take $v_k$ to be a convolution of $v$ with a suitable mollifier to conclude that \eqref{eq:convadj} must hold.
\end{proof}

\begin{proposition}
\label{prop:Linvenest}
The operator $\mathcal{L}: D_k(\mathcal{L})\to\mathbf{H}_k$ is invertible and the inverse $\mathcal{L}^{-1}:\mathbf{H}_k\to  D_k(\mathcal{L})$ satisfies
\begin{equation}
\label{eq:Linvenest}
||\mathcal{L}^{-1}(F)||_{ k}+||rX\mathcal{L}^{-1}(F)||_{k}+||\slashed{\nabla}_{\s^2}\mathcal{L}^{-1}(F)||_{k}\leq C ||F||_{ k},
\end{equation}
for any $F\in \mathbf{H}_k$.
\end{proposition}
\begin{proof}
We first consider $k=0$. By Proposition \ref{prop:keyestimateTinv}, we have that $\mathcal{L}$ is injective, i.e. $\ker \mathcal{L}=\{0\}$. We will conclude that $\mathcal{L}$ is bijective by showing that $\textnormal{Ran}\, \mathcal{L}=\mathbf{H}_0=L^2([r_+,\infty)\times \s^2)$.

Let $v\in (\textnormal{Ran}\, \mathcal{L})^{\perp}$ and $u\in C_{\rm rad}^{\infty}(\Sigma_0)$. Then
\begin{equation*}
0=\la v, \mathcal{L}u \ra_{L^2}= \la \mathcal{L}^* v,u \ra_{L^2},
\end{equation*}
so $v\in D_0(\mathcal{L}^*)$ and moreover $v\in \ker \mathcal{L}^*$. Hence, $(\textnormal{Ran}\, \mathcal{L})^{\perp} \subseteq  \ker \mathcal{L}^*$. Furthermore, if $v\in  \ker \mathcal{L}^*$, then we can similarly conclude that $v\in (\textnormal{Ran}\, \mathcal{L})^{\perp}$ and therefore  $(\textnormal{Ran}\, \mathcal{L})^{\perp} =  \ker \mathcal{L}^*$.

By Lemma \ref{lm:adjoints}, we have that $\mathcal{L}^*=\mathcal{L}^{\dag}$. By definition of $\mathcal{L}^{\dag}$, we moreover have that $\mathcal{L}^{\dag}|_{C_{\rm rad, *}^{\infty} (\Sigma_0) }=\mathcal{L}|_{C_{\rm rad, *}^{\infty} (\Sigma_0) }$, so we can apply Proposition \ref{prop:keyestimateTinv} again to obtain $\ker \mathcal{L}^*=\{0\}$.

Note, by Proposition \ref{prop:keyestimateTinv} together with the fact that $\mathcal{L}$ is closed, it follows that $\textnormal{Ran}\, \mathcal{L}$ is closed and hence
\begin{equation*}
\textnormal{Ran}\, \mathcal{L}=(\textnormal{Ran}\, \mathcal{L})^{\perp \perp}=( \ker \mathcal{L}^*)^{\perp}=\mathbf{H}_0.
\end{equation*}
We conclude that $\mathcal{L}: D_0(\mathcal{L})\to \mathbf{H_0}$ is bijective and we denote its inverse by $\mathcal{L}^{-1}$. The estimate \eqref{eq:Linvenest} with $k=0$ follows immediately by taking $f=\mathcal{L}^{-1}(F)$ in \eqref{eq:keytimeinvestho}.

We can obtain the $k>0$ case by an inductive argument. Suppose we have established $\mathcal{L}^{-1}(F)\in \mathbf{H}_K$ and \eqref{eq:Linvenest} with $k=K$ for some $K\in \N_0$. Then we can repeat the argument above where in the proof of Lemma \ref{lm:adjoints} we replace $\mathcal{L}$ by the operator $\snabla_{\s^2}^{k_1}X^{k_2}(\mathcal{L}-k_2(k_2+1))$, with $k_1+k_2=1$, applying \eqref{eq:esthoder} instead of \eqref{eq:keytimeinvest} to conclude that in fact $\snabla_{\s^2}^{k_1}(rX)^{k_2}\mathcal{L}^{-1}(F)$, with $k_1+k_2=1$ lies in $\mathbf{H}_K$ and \eqref{eq:Linvenest} holds for $\mathcal{L}^{-1}(F)$ replaced by  $\snabla_{\s^2}^{k_1}X^{k_2}\mathcal{L}^{-1}(F)$. Hence,  $\mathcal{L}^{-1}(F)\in \mathbf{H}_{K+1}$ and \eqref{eq:Linvenest} holds for $k=K+1$.
\end{proof}

By Proposition \ref{prop:Linvenest} together with a standard Sobolev inequality, we immediately obtain:

\begin{corollary}
\label{cor:regularitytimeinv}
Let $F\in C^{\infty}_{\rm rad}(\Sigma_0)$. Then $\mathcal{L}^{-1}(F)\in \bigcap_{k\in \N_0} \mathbf{H}_k$ and in particular $\mathcal{L}^{-1}(F)\in C^{\infty}(\Sigma_0)$.
\end{corollary}
\subsection{Decay towards $\mathcal{I}^+$}
In this section, we will establish additional $r$-decay estimates of $\mathcal{L}^{-1}(F)$ (with suitably decaying $F$) towards $\mathcal{I}^+$. We will derive elliptic analogues of the $r^p$-weighted estimates from Section \ref{sec:rpest}. 

We first define the following higher-order quantities:
\begin{align*}
f^{(n)}:=&\:((r^2+a^2)X)^n f,\\
\check{f}^{(0)}:=&\:f,\\
\check{f}^{(1)}:=&\:[1+(\alpha+\alpha_{\Phi}\Phi)r^{-1}+(\beta+\beta_{\Phi^2}\Phi^2)r^{-2}](r^2+a^2)Xf,\\
\check{f}^{(2)}:=&\:[1+\gamma r^{-1}](r^2+a^2)X\check{f}^{(1)},
\end{align*}
where the coefficients $\alpha,\alpha_{\Phi},\beta,\beta_{\Phi},\gamma,\gamma_{\Phi}$ are chosen as in Section \ref{sec:defNPconstants}.

We moreover observe that
\begin{equation*}
[\mathcal{L},\pi_{\ell}]=0,
\end{equation*}
so we can independently derive additional decay estimates for $\mathcal{L}^{-1}(F_{\ell})=(\mathcal{L}^{-1}(F))_{\ell}$.

We define the following auxiliary operators:
\begin{align*}
\widehat{\mathcal{L}}:=&\: (r^2+a^2)^{\frac{1}{2}} \mathcal{L}(r^2+a^2)^{-\frac{1}{2}} ,\\
\widehat{\mathcal{L}}^{(n)}:=&\: (r^2+a^2) (X(r^2+a^2))^n (r^2+a^2)^{-1}  \widehat{\mathcal{L}},\\
\check{\mathcal{L}}^{(0)}:=&\:  (r^2+a^2) \mathcal{L}(r^2+a^2)^{-\frac{1}{2}},\\
\check{\mathcal{L}}^{(1)}:=&\: (r^2+a^2)^{\frac{3}{2}}  X(r^2+a^2)\left[1+(\alpha+\alpha_{\Phi}\Phi)r^{-1}+(\beta+\beta_{\Phi^2}\Phi^2)r^{-2}\right] (r^2+a^2)^{-1} \widehat{\mathcal{L}},\\
\check{\mathcal{L}}^{(2)}:=&\: (r^2+a^2)^{\frac{3}{2}}X(r^2+a^2)[1+\gamma r^{-1}](r^2+a^2)^{-1}\check{\mathcal{L}}^{(1)}.
\end{align*}

\begin{lemma}
\label{lm:hoeqtin}
Let $n\in \N_0$, then
\begin{equation}
\label{eq:tinweightho}
\begin{split}
(r^2+a^2)^{-1} \widehat{\mathcal{L}}^{(n)}f:=&\:X \left( \Delta (r^2+a^2)^{-1}Xf^{(n)} \right)-2n[1+O_{\infty}(r^{-1})]Xf^{(n)}+2a (r^2+a^2)^{-1}X \Phi  f^{(n)}\\
&+(r^2+a^2)^{-1}[\slashed{\Delta}_{\s^2}-n(n+1)]f^{(n)}+O_{\infty}(r^{-3})[f^{(n)}+\Phi f^{(n)}]\\
&+n \sum_{m=0}^{n-1}O_{\infty}(r^{-2})[f^{(m)}+\Phi f^{(m)}].
\end{split}
\end{equation}
Furthermore, for $n=0,1,2$:
\begin{equation}
\label{eq:tinweightho2}
\begin{split}
(r^2+a^2)^{-\frac{3}{2}}\check{\mathcal{L}}^{(n)}f_n:=&\:X \left( \Delta (r^2+a^2)^{-1}Xf^{(n)}_n \right)-2n[1+O_{\infty}(r^{-1})]Xf^{(n)}_n+2a (r^2+a^2)^{-1}X \Phi  f^{(n)}_n\\
&+\sum_{m=0}^{n}O_{\infty}(r^{-3})[f^{(m)}_n+\Phi f^{(m)}_n].
\end{split}
\end{equation}
\end{lemma}
\begin{proof}
We can write
\begin{equation}
\label{eq:Lpsiviaphi}
\begin{split}
\mathcal{L}((r^2+a^2)^{-\frac{1}{2}} f)=&\:X\left(\Delta X((r^2+a^2)^{-1/2} f)\right)+2a X(r^2+a^2)^{-1/2} \Phi f)+(r^2+a^2)^{-1/2}\slashed{\Delta}_{\s^2}f\\
=&\:X\left((r^2+a^2)^{1/2} \Delta (r^2+a^2)^{-1}X f \right)-X\left( \Delta r (r^2+a^2)^{-3/2} f \right)+2a (r^2+a^2)^{-1/2}X \Phi  f\\
&+(r^2+a^2)^{-1/2}\slashed{\Delta}_{\s^2}f-2a r (r^2+a^2)^{-3/2} \Phi f\\
=&\: (r^2+a^2)^{1/2}X \left( \Delta (r^2+a^2)^{-1}Xf \right)+2a (r^2+a^2)^{-1/2}X \Phi  f-2a r (r^2+a^2)^{-3/2} \Phi f\\
&+(r^2+a^2)^{-1/2}\slashed{\Delta}_{\s^2}f-\left( \Delta r (r^2+a^2)^{-3/2} \right)'f.
\end{split}
\end{equation}
Hence, \eqref{eq:tinweightho} follows for $n=0$. The $n>0$ case then follows by an induction argument that is analogous to the proof of Proposition \ref{eq:ho}. Equation \eqref{eq:tinweightho2} follows by observing that cancellations occur when considering the modified quantities $\check{f}^{(n)}$ as in the proof of Proposition \ref{prop:maineqchph}.
\end{proof}
The following proposition contains the elliptic analogues of $r^p$-weighted energy estimates:
\begin{proposition}
\label{prop:ellipticrp}
Let $n\in\{0,1,2\}$, then for $f\in C^{\infty}(\widehat{\Sigma}_0)$ and $\delta\in (0,1)$, there exists a constant $C_{\delta}(M,a)>0$, such that
\begin{equation}
\label{eq:ellipticrp}
\begin{split}
\int_{\Sigma_0}& r^{1-\delta}(Xf^{(n)}_{\geq n+1})^2+r^{-1-\delta}|\snabla_{\s^2}f^{(n)}_{\geq n+1}|^2+r^{-1-\delta}(f^{(n)}_{\geq n+1})^2\,d\omega d\uprho\\
\leq&\: C_{\delta}\int_{\Sigma_0} r^{-1-\delta}\left(\widehat{\mathcal{L}}^{(n)} f_{\geq n+1}-n\sum_{m=0}^{n-1}O_{\infty}(r^0)[f^{(m)}_{\geq n+1}+\Phi f^{(m)}_{\geq n+1}]\right)^2\,d\omega d\uprho\\
&+C_{\delta}\sum_{m=0}^{n}\int_{\Sigma_0}((rX)^m\mathcal{L}((r^2+a^2)^{-\frac{1}{2}} f_{\geq n+1}))^2\,d\omega d\rho
\end{split}
\end{equation}
and
\begin{equation}
\label{eq:ellipticrpmod}
\begin{split}
\int_{\Sigma_0} r^{3-\delta}  (X\check{f}_n^{(n)})^2+r^{-1-\delta}(\check{f}_n^{(n)})^2 \,d\omega d\uprho \leq&\: C \int_{\Sigma_0} r^{-1-\delta}\left(\check{\mathcal{L}}^{(n)} f-n \sum_{m=0}^{n-1}O_{\infty}(r^0)[f_n^{(m)}+\Phi f_n^{(m)}]\right)^2\,d\omega\\
&+ C\sum_{m=0}^{n}\int_{\Sigma_0}((rX)^m\mathcal{L}((r^2+a^2)^{-\frac{1}{2}} f_n))^2\,d\omega d\rho.
\end{split}
\end{equation}
\end{proposition}
\begin{proof}
Let $\chi$ be a smooth cut-off function such that $\chi(r)=1$ for $r\geq R$ and $\chi(r)=0$ for $r\leq R-M$, where $R>r_++M$ will be chosen suitably large. Consider the multiplier $- (r^2+a^2)^{\frac{p}{2}}  X f$. Then by \eqref{eq:tinweightho} we obtain
\begin{equation*}
\begin{split}
-(r^2+a^2)^{\frac{p}{2}-1} X f  \cdot \widehat{\mathcal{L}} f=&-\frac{1}{2} \Delta^{-1}(r^2+a^2)^{\frac{p+2}{2}} X\left( \Delta^2 (r^2+a^2)^{-2}(Xf)^2 \right)-a  (r^2+a^2)^{\frac{p-4}{2}}\Phi ((Xf)^2)\\
&+\slashed{\rm div}_{\s^2 }(\ldots)-\frac{1}{2} (r^2+a^2)^{\frac{p-2}{2}}X(|\snabla_{\s^2 }f|^2)+O(r^{p-3})\chi \Phi f X f+O(r^{p-3}) f X f\\
=&\:\left[\frac{p}{2}r^{p-1}+O(r^{p-2})\right] (X f)^2 +\frac{1}{2}(p-2)[r^{p-3}+O(r^{p-4})] |\snabla_{\s^2 }f|^2\\
&+O(r^{p-3}) \Phi fX f+O(r^{p-3}) fX f+\frac{1}{2} (r_+^2+a^2)^{\frac{p-2}{2}}|\snabla_{\s^2 }f|^2|_{\mathcal{H}^+}\\
&-\frac{1}{2}X\left(  \Delta (r^2+a^2)^{\frac{p-2}{2}} (X f)^2\right)-a (r^2+a^2)^{\frac{p-2}{2}}\Phi ((X f)^2)+\slashed{\rm div}_{\s^2 }(\ldots).
\end{split}
\end{equation*}

We integrate over $\Sigma_0$, applying \eqref{eq:hardyX} and \eqref{eq:keytimeinvest} to control terms in bounded $r$ regions, to obtain: for $1\leq p\leq 2$
\begin{equation}
\label{eq:zerothordrptimeinv}
\int_{\Sigma_0}r^{p-1}  (Xf)^2+(2-p)r^{p-3} |\snabla_{\s^2}f|^2+r^{-2}f^2 \,d\omega d\uprho\leq C \int_{\Sigma_0} r^{p-3}( \widehat{\mathcal{L}} f)^2+(\mathcal{L}((r^2+a^2)^{-\frac{1}{2}} f)^2\,d\omega d\uprho.
\end{equation}
Note that when replacing $f$ with $f_0=\pi_0 f$, we can in fact take $p=4-\delta$ to obtain:
\begin{equation}
\label{eq:zerothordrptimeinvl0}
\int_{\Sigma_0}r^{3-\delta}  (Xf_0)^2+r^{-1-\delta}f_0^2 \,d\omega d\uprho\leq C \int_{\Sigma_0} r^{-1-\delta}( \check{\mathcal{L}}f_0)^2+(\mathcal{L}((r^2+a^2)^{-\frac{1}{2}} f_0)^2\,d\omega d\uprho.
\end{equation}

For $n>0$ we consider the multiplier $- (r^2+a^2)^{\frac{p}{2}} X f^{(n)}$ together with $(r^2+a^2)^{-1} \widehat{\mathcal{L}}^{(n)}f$, and we proceed analogously to above to obtain:
\begin{equation}
\label{eq:zerothordrptimeinv2}
\begin{split}
\int_{\Sigma_0}& r^{p-1} (Xf^{(n)})^2+(2-p)r^{p-3} [|\snabla_{\s^2}f^{(n)}|^2-n(n+1)(f^{(n)})^2] \,d\omega d\uprho \\
\leq&\: C \int_{\Sigma_0} r^{p-3}\left(\widehat{\mathcal{L}}^{(n)} f-\sum_{m=0}^{n-1}O_{\infty}(r^0)[f^{(m)}+\Phi f^{(m)}]\right)^2+\sum_{m=0}^n((rX)^m\mathcal{L}((r^2+a^2)^{-\frac{1}{2}} f))^2\,d\omega d\uprho,
\end{split}
\end{equation}
for $0<p<2$.

Suppose $n=1$. Then for $f_{\geq 2}$, the left-hand side of \eqref{eq:zerothordrptimeinv2} is non-negative definite by \eqref{eq:poincare1}. Furthermore,  if we consider \eqref{eq:tinweightho2} with the multiplier $- (r^2+a^2)^{\frac{p}{2}} \chi  X \check{f}_1^{(1)}$, we can take $p=4-\delta$ to obtain:
\begin{equation}
\label{eq:zerothordrptimeinv3}
\begin{split}
\int_{\Sigma_0} r^{3-\delta}  (X\check{f}_1^{(1)})^2+r^{-1-\delta}(\check{f}_1^{(1)})^2 \,d\omega d\uprho \leq&\: C \int_{\Sigma_0} r^{-1-\delta}\left(\check{\mathcal{L}}^{(1)} f-O_{\infty}(r^0)[f_1+\Phi f_1]\right)^2\,d\omega\\
&+ C\sum_{m=0}^{1}\int_{\Sigma_0}((rX)^m\mathcal{L}((r^2+a^2)^{-\frac{1}{2}} f_1))^2\,d\omega d\rho. 
\end{split}
\end{equation}

Similarly, for $n=2$ and $f_{\geq 3}$, \eqref{eq:zerothordrptimeinv2} has a non-negative definite left-hand side. For $f_2$, we consider \eqref{eq:tinweightho2} with the multiplier $- (r^2+a^2)^{\frac{p}{2}} \chi  X \check{f}_2^{(2)}$ with $p=4-\delta$ to obtain:
\begin{equation}
\label{eq:zerothordrptimeinv4}
\begin{split}
\int_{\Sigma_0} r^{3-\delta}  (X\check{f}_2^{(2)})^2+r^{-1-\delta}(\check{f}_2^{(2)})^2 \,d\omega d\uprho \leq&\: C \int_{\Sigma_0} r^{-1-\delta}\left(\check{\mathcal{L}}^{(2)} f-O_{\infty}(r^0)\sum_{m=0}^1[f_2^{(m)}+\Phi f_2^{(m)}]\right)^2\,d\omega\\
&+ C\sum_{m=0}^{2}\int_{\Sigma_0}((rX)^m\mathcal{L}((r^2+a^2)^{-\frac{1}{2}} f_2))^2\,d\omega d\rho.
\end{split}
\end{equation}
\end{proof}

We can moreover obtain $r^p$-weighted energy estimates for higher-order derivatives with respect to $\snabla_{\s^2}$ and $rX$:

\begin{corollary}
\label{cor:ellipticrpcommrX}
Let $n\in\{0,1,2\}$ and $k\in \N_0$, then for $f\in C^{\infty}(\widehat{\Sigma}_0)$ and $\delta\in (0,1)$, there exists a constant $C_{\delta}(M,a,k)>0$, such that
\begin{equation}
\label{eq:ellipticrpcommrX}
\begin{split}
\sum_{k\leq k_1+k_2\leq k+1}&\int_{\Sigma_0} r^{-1-\delta}|\snabla_{\s^2}^{k_1}(rX)^{k_2}f^{(n)}_{\geq n+1}|^2\,d\omega d\uprho\\
\leq&\: C_{\delta}\sum_{k_1+k_2=k}\Bigg\{\int_{\Sigma_0} r^{-1-\delta}\left|\snabla_{\s^2}^{k_1}(rX)^{k_2}\left[\widehat{\mathcal{L}}^{(n)} f_{\geq n+1}-\sum_{m=0}^{n-1}O_{\infty}(r^0)(f^{(m)}_{\geq n+1}+\Phi f^{(m)}_{\geq n+1})\right]\right|^2\,d\omega\\
&+C_{\delta}\sum_{m=0}^{n}\int_{\Sigma_0}|\snabla_{\s^2}^{k_2}(rX)^{m+k_2}\mathcal{L}((r^2+a^2)^{-\frac{1}{2}} f_{\geq n+1})|^2\,d\omega d\rho\Bigg\}.
\end{split}
\end{equation}
and
\begin{equation}
\label{eq:ellipticrpmodcommrX}
\begin{split}
\int_{\Sigma_0} &r^{3-\delta}  (X(rX)^{k}\check{f}_n^{(n)})^2+r^{-1-\delta}((rX)^{k}\check{f}_n^{(n)})^2 \,d\omega d\uprho \\
\leq&\: C \int_{\Sigma_0} r^{-1-\delta}\left((rX)^k\check{\mathcal{L}}^{(n)} f-n O_{\infty}(r^0)\sum_{m=0}^{n-1}[(rX)^kf_n^{(m)}+\Phi (rX)^kf_n^{(m)}]\right)^2\,d\omega\\
&+ C\sum_{m=0}^{n}\int_{\Sigma_0}((rX)^{m+k}\mathcal{L}((r^2+a^2)^{-\frac{1}{2}} f_n))^2\,d\omega d\rho.
\end{split}
\end{equation}

\end{corollary}
\begin{proof}
We consider $(rX)^{k_2}\widehat{\mathcal{L}}^{(n)}f$ with the multiplier $-(-1)^{k_1} (r^2+a^2)^{\frac{p}{2}}  X \slashed{\Delta}_{\s^2}^{k_1}(rX)^{k_2}f$, and $(rX)^{k_2}\check{\mathcal{L}}^{(n)}f_n$ with the multiplier $- (r^2+a^2)^{\frac{p}{2}}  X (rX)^{k_2}f_n$ and proceed to integrate by parts over $\s^2$, as in the proof of Proposition \ref{prop:keyestimateTinv}, and apply the arguments of Proposition \ref{prop:ellipticrp}.
\end{proof}

In order to establish stronger decay properties of $\mathcal{L}^{-1}(F)$, we introduce the following norms: let $0<\delta<1$, then
\begin{align*}
||f||_{k,0,0, \delta}^2:=&\:\sum_{k_2\leq k} \int_{\Sigma_0} r^{-1-\delta}((rX)^{k_2}f_0)^2\,d\omega d\uprho,\\
||f||_{k,1, m, \delta}^2:=&\:\sum_{k_2\leq k} \sum_{n=0}^m\int_{\Sigma_0} r^{-1-\delta}((rX)^{k_2}f^{(n)}_{1})^2\,d\omega d\uprho\quad 0\leq m\leq 1,\\
||f||_{k,2, m,\delta}^2:=&\:\sum_{k_2\leq k} \sum_{n=0}^m\int_{\Sigma_0} r^{-1-\delta}((rX)^{k_2}f^{(n)}_{2})^2\,d\omega d\uprho \quad 0\leq m\leq 2,\\
||f||_{k,3,m, \delta}^2:=&\:\sum_{k_1+k_2\leq k} \sum_{n=0}^m\int_{\Sigma_0} r^{-1-\delta}|\snabla_{\s^2}^{k_1}(rX)^{k_2}f^{(n)}_{\geq 3}|^2\,d\omega d\uprho\quad 0\leq m\leq 3.
\end{align*}
We denote with $\mathbf{H}_{k,n,m,\delta}$ the completions of:  $\pi_n C^{\infty}(\widehat{\Sigma}_0)$ for $n\in \{0,1,2\}$ and $\pi_{\geq 3} C^{\infty}(\widehat{\Sigma})$ for $n=3$ under the norms $||\cdot ||_{k,3, m,\delta}$ defined above.

Let $D_{k,n,m,\delta}$, with $n\in \{0,1,2\}$, denote the completion of $\pi_n C^{\infty}(\widehat{\Sigma}_0)$ under the graph norm $||\cdot ||_{k,n,m,\delta}$ with respect to the operator $A$, where $A=\check{\mathcal{L}}^{(0)}$ for $n=0$, and when $n\in\{1,2\}$: $A=\widehat{\mathcal{L}}^{(m)}$ for $m\leq n-1$ and $A=\check{\mathcal{L}}^{(n)}$, for $m=n$.

Similarly, let $D_{k,3,m,\delta}$, denote the completion of $\pi_{\geq 3} C^{\infty}(\widehat{\Sigma}_0)$ under the graph norm $||\cdot ||_{k,3,m,\delta}$, with respect to the operator $A$, where $A=\widehat{\mathcal{L}}^{(m)}$.

Then
\begin{equation*}
A: D_{k,n,m,\delta}\to \mathbf{H}_{k,n,m,\delta}
\end{equation*}
is a densely defined, closed, linear operator.

\begin{proposition}
\label{prop:rdecayTinv}
Let $K\in\N_0$ and let $(r^2+a^2)^{-\frac{1}{2}}F\in \mathbf{H}_K$. Then there exists a unique $f\in (r^2+a^2)^{\frac{1}{2}}\mathbf{H}_K$, such that
\begin{equation*}
\widehat{\mathcal{L}}f=F.
\end{equation*}
If moreover,
\begin{align*}
(r^2+a^2)^{\frac{1}{2}}\check{F}_0\in \mathbf{H}_{K,0,0,\delta},\\
F_1\in \mathbf{H}_{K,1,0,\delta}\quad \textnormal{and}\quad (r^2+a^2)^{\frac{1}{2}}\check{F}_1^{(1)}\in \mathbf{H}_{K,1,0,\delta},\\
F_2\in \mathbf{H}_{K,2,1,\delta}\quad \textnormal{and}\quad (r^2+a^2)^{\frac{1}{2}}\check{F}_2^{(2)}\in \mathbf{H}_{K,2,0,\delta},\\
F_{\geq 3}\in \mathbf{H}_{K,3,3,\delta},
\end{align*}
then we can estimate
\begin{align}
\label{eq:rdecayTinv0}
||f_0||_{K,0,0,\delta}+||r^2Xf_0||_{K,0,0,\delta}\leq&\: C||(r^2+a^2)^{\frac{1}{2}}\check{F}_0||_{K,0,0,\delta},\\
\label{eq:rdecayTinv1}
||f_1||_{K,1,1,\delta}+||r^2Xf_1||_{K,1,1,\delta}\leq&\: C\left[ ||F_1||_{K,1,0,\delta}+||(r^2+a^2)^{\frac{1}{2}}\check{F}_1||_{K,1,0,\delta}\right],\\
\label{eq:rdecayTinv2}
||f_2||_{K,2,2,\delta}+||r^2Xf_2||_{K,2,2,\delta}\leq&\: C\left[ ||F_2||_{K,2,1,\delta}+||(r^2+a^2)^{\frac{1}{2}}\check{F}_2||_{K,2,0,\delta}\right],\\
\label{eq:rdecayTinv3}
||f_{\geq 3}||_{K+1,3,3,\delta}\leq&\: C||F_{\geq 3}||_{K,3,3,\delta}.
\end{align}
\end{proposition}
\begin{proof}
The existence and uniqueness of $f\in (r^2+a^2)^{\frac{1}{2}}\mathbf{H}_K$ follows directly from Proposition \ref{prop:Linvenest}. The improved estimates \eqref{eq:rdecayTinv0}--\eqref{eq:rdecayTinv3} would follow from Corollary \ref{cor:ellipticrpcommrX}, if we knew \emph{a priori} that $f$ and its appropriately weighted derivatives decayed sufficiently fast towards $\mathcal{I}^+$. Since this decay does not follow from the fact that $f\in (r^2+a^2)^{\frac{1}{2}}\mathbf{H}_K$, we proceed by deriving invertibility of the operators $A: D_{k,m,\delta}\to \mathbf{H}_{k,m,\delta}$, as defined above.

We consider first \eqref{eq:rdecayTinv0} with $K=0$. Consider the restricted inverse operator:
\begin{equation*}
\widehat{\mathcal{L}}^{-1}: \mathbf{H}_{0,0,0,\delta}\to \mathcal{L}^{-1}(\mathbf{H}_{0,0,0,\delta})\subseteq (r^2+a^2)^{\frac{1}{2}}\mathbf{H}_0,
\end{equation*}
Let $F_0\in \mathbf{H}_{0,0,0,\delta}$ and write $f_0:=\widehat{\mathcal{L}}^{-1}(F_0)$. Let $\{(f_0)_j\}$ be a sequence in $C^{\infty}(\widehat{\Sigma_0})$ defined as follows:
\begin{equation*}
(f_0)_j=f_0* \eta_{\frac{1}{j}},
\end{equation*}
with $\eta_{\frac{1}{j}}$ a standard mollifier defined on an extension of $\widehat{\Sigma}_0$ (where we view $\widehat{\Sigma}_0$ as a compact subset of $\R\times \s^2$). Note that $\{(r^2+a^2)^{-\frac{1}{2}} (f_0)_j\}$ converges to $(r^2+a^2)^{-\frac{1}{2}} f_0$ with respect to $||\cdot||_0$.

By linearity of $\mathcal{L}$, we also have that $(F_0)_j:=\mathcal{L}(f_0)_j=F_0*\eta_{\frac{1}{j}}$, so by $F_0\in \mathbf{H}_{0,0,0,\delta}$, it is straightforward to show that:
\begin{equation*}
||F_0-(F_0)_j||_{0,0,0,\delta}\to 0\:\textnormal{ as $j\to\infty$}.
\end{equation*}
Hence, $\{(F_0)_j\}$ is Cauchy with respect to $||\cdot||_{0,0,0,\delta}$ and we can apply \eqref{eq:ellipticrpmod} to the difference $(f_0)_j-(f_0)_i$, $i,j\in \N_0$, to obtain (in particular) the Cauchy property of $\{(f_0)_j\}$ with respect to $||\cdot||_{0,0,0,\delta}$. It follows that $f_0\in \mathbf{H}_{0,0,0,\delta}$ and we can conclude that \eqref{eq:rdecayTinv0} holds with $K=0$.

In order to derive \eqref{eq:rdecayTinv1} with $K=0$, we first establish invertibility of $\widehat{\mathcal{L}}: D_{1,0,0,\delta}\to \mathbf{H}_{1,0,0,\delta}$ as outlined above and then consider the operator $\check{\mathcal{L}}^{(1)}: D_{1,0,1,\delta}\to \mathbf{H}_{1,0,1,\delta}$. Similarly, we obtain \eqref{eq:rdecayTinv2} by first considering $\widehat{\mathcal{L}}: D_{2,0,0,\delta}\to \mathbf{H}_{2,0,0,\delta}$ , then  $\widehat{\mathcal{L}}^{(1)}: D_{2,0,1,\delta}\to \mathbf{H}_{2,0,1,\delta}$, and finally $\check{\mathcal{L}}^{(2)}: D_{2,0,2,\delta}\to \mathbf{H}_{2,0,2,\delta}$. In order to derive \eqref{eq:rdecayTinv3}, we analogously establish successively invertibility of $\widehat{\mathcal{L}}^{(n)}$ for $n=0,1,2,3$.

Finally, the $K>0$ cases follow via an inductive argument as in the proof of Proposition \ref{prop:Linvenest}.
\end{proof}

\subsection{Construction of time integral data}
If $\psi$ is a solution to \eqref{eq:waveeq}, then it follows from \eqref{eq:waveeq2} that the restriction $\psi|_{\Sigma_0}$ satisfies the inhomogeneous equation
\begin{equation*}
\widehat{\mathcal{L}}\psi|_{\Sigma_0}=F[T\psi|_{\Sigma_0}],
\end{equation*}
with
\begin{equation*}
(r^2+a^2)^{-\frac{1}{2}}F[f]:=2[h\Delta-(r^2+a^2)] Xf+[(\Delta h)'-2r]f+[2h(r^2+a^2)-h^2\Delta-a^2\sin^2\theta] Tf+2a(h-1) \Phi f.
\end{equation*}
\begin{proposition}
\label{prop:mainpropTinv}
Consider initial data $(\psi|_{\Sigma_0},T\psi_{\Sigma_0})$ for \eqref{eq:waveeq}, with $(\phi|_{\Sigma_0},T\phi_{\Sigma_0}) \in (C^{\infty}(\widehat{\Sigma}))^2$.
\begin{enumerate}[\rm (i)]
\item
Then there exists a unique solution to \eqref{eq:waveeq}, denoted $T^{-1}\psi$, such that $T^{-1}\psi\in C^{\infty}(\mathcal{R})$, 
\begin{equation*}
T(T^{-1}\psi)=\psi
\end{equation*}
and $r^{-1}T^{-1}\psi|_{\Sigma_0}\in L^2(\Sigma_0)$ and $T^{-1}\psi|_{\Sigma_0}\in \dot{H}^1(\Sigma_0)$.
\item
If we moreover assume that
\begin{equation}
\label{eq:vanishnp}
(I_0[\psi],I_1[\psi],I_2[\psi])=(0,0,0),
\end{equation}
then the energy norms appearing on the right-hand sides of the estimates of Section \ref{sec:poinwdecay} are finite for all $K\in \N_0$ and $\delta>0$ if we replace $\psi$ with $T^{-1}\psi$.
\end{enumerate}
\end{proposition}
\begin{proof}
We can rewrite $F_0[\psi|_{\Sigma_0}]$ as follows in terms of $\phi|_{\Sigma_0}$:
\begin{equation}
\label{eq:F0}
\begin{split}
F_0[\psi|_{\Sigma_0}]=&\: 2[h\Delta-(r^2+a^2)] X\phi_0|_{\Sigma_0}+[(\Delta h)'-2(h\Delta)r (r^2+a^2)^{-1}] \phi_0|_{\Sigma_0}+[2h(r^2+a^2)-h^2\Delta] T\phi_0|_{\Sigma_0}\\
&-a^2\sin^2\theta\pi_0(T\phi|_{\Sigma_0})\\
=&\:2r^2 P_0+O_{\infty}(r^{-1})[\phi^{(1)}_0|_{\Sigma_0}+T\phi_0|_{\Sigma_0}+\phi_0|_{\Sigma_0}].
\end{split}
\end{equation}

By combining the equations in Proposition \ref{prop:maineqchph} with the expressions in Lemma \ref{lm:hoeqtin}, we obtain moreover:
\begin{align}
\label{eq:F1}
\check{F}^{(1)}_1[\psi|_{\Sigma_0}]=&\:2[r^2+O_{\infty}(r)] P_1|_{\Sigma_0}+O_{\infty}(r^{-1})\left[\sum_{m=0}^1\sum_{j=0}^1 \phi_1^{(2)}+T^j\phi^{(1)}_{2m+1}+T^j\phi_{2m+1}\right],\\
\label{eq:F2}
\check{F}^{(2)}_2[\psi|_{\Sigma_0}]=&\:2[r^2+O_{\infty}(r)] P_2|_{\Sigma_0}+O_{\infty}(r^{-1})\left[\sum_{m=0}^2\sum_{j=0}^1  \phi_2^{(3)}+T^j\phi^{(2)}_{2m}+T^j\phi^{(1)}_{2m}+\phi_{2m}\right],
\end{align}

Denote:
\begin{equation*}
T^{-1} \psi|_{\Sigma_0}:=\mathcal{L}^{-1}(F[\psi|_{\Sigma_0}]).
\end{equation*}
Then by the regularity established in Corollary \ref{cor:regularitytimeinv}, $T^{-1}\psi|_{\Sigma_0}\in C^{\infty}(\Sigma_0)$ and we denote with $T^{-1}\psi$ the solution to \eqref{eq:waveeq} with initial data $(T^{-1}\psi|_{\Sigma_0},\psi_{\Sigma_0})$. We then have that
\begin{equation*}
T(T^{-1}\psi)=\psi.
\end{equation*}
Furthermore, by the injectivity properties of $\mathcal{L}$ following from Proposition \ref{prop:Linvenest}, $T^{-1}\psi$ must the unique solution to \eqref{eq:waveeq} satisfying the conditions: $T(T^{-1}\psi)=\psi$, $r^{-1}T^{-1}\psi|_{\Sigma_0}\in L^2(\Sigma_0)$ and $T^{-1}\psi|_{\Sigma_0}\in \dot{H}^1(\Sigma_0)$.

Furthermore, by the above expressions $(r^2+a^2)^{\frac{1}{2}}F_0, (r^2+a^2)^{\frac{1}{2}}\check{F}_1, (r^2+a^2)^{\frac{1}{2}}\check{F}_2\in C^{\infty}(\widehat{\Sigma}_0)$, if \eqref{eq:vanishnp} holds, so we can apply Proposition \ref{prop:rdecayTinv} to conclude that for $T^{-1}\phi:=(r^2+a^2)^{\frac{1}{2}}T^{-1}\psi$:
\begin{align*}
(T^{-1}\phi)_0|_{\Sigma_0},\quad r^2X(T^{-1}\phi)_0|_{\Sigma_0}\in \mathbf{H}_{K,0,0,\delta},\\
(T^{-1}\phi)_1|_{\Sigma_0},\quad r^2X(T^{-1}\phi)_1|_{\Sigma_0}\in \mathbf{H}_{K,1,1,\delta},\\
(T^{-1}\phi)_2|_{\Sigma_0},\quad r^2X(T^{-1}\phi)_2|_{\Sigma_0}\in \mathbf{H}_{K,2,2,\delta},\\
(T^{-1}\phi)_{\geq 3}|_{\Sigma_0},\quad rX (T^{-1}\phi)_{\geq 3}|_{\Sigma_0}\in \mathbf{H}_{K,3,3,\delta}.
\end{align*}
for all $K\geq 0$ and $\delta>0$ and hence all the energies appearing in the estimates of Section \ref{sec:poinwdecay} are finite if we replace $\psi$ with $T^{-1}\psi$.
\end{proof}

\begin{corollary}
\label{cor:sharpdecayvanishingNP}
Consider initial data $(\psi|_{\Sigma_0},T\psi_{\Sigma_0})$ for \eqref{eq:waveeq}, with $(\phi|_{\Sigma_0},T\phi_{\Sigma_0}) \in (C^{\infty}(\widehat{\Sigma}))^2$, such that \eqref{eq:vanishnp} holds. 

Then the energy decay estimates in Section \ref{sec:edecay} and the pointwise decay estimates in Section \ref{sec:poinwdecay} hold when $\psi$ is replaced with $T^{-1}\psi$ and we can moreover express:
\begin{equation*}
T^{-1}\psi(\tau,\uprho,\theta,\varphi_*)=-\int_{\tau}^{\infty}\psi(\tau',\uprho,\theta,\varphi_*)\,,d\tau'.
\end{equation*}.
\end{corollary}

\subsection{Time-inverted Newman--Penrose charges}
\label{sec:timinvNP}
In this section, we will express the Newman--Penrose charges $I_0[T^{-1}\psi]$, $I_{1m}[T^{-1}\psi]$, with $|m|\leq 1$ and $I_{2m}[T^{-1}\psi]$, with $|m|\leq 2$, which are defined in Section \ref{sec:defNPconstants},  in terms of integrals over $F[\psi|_{\Sigma_0}]$, where
\begin{equation}
\label{eq:timeinvimportant}
\begin{split}
\mathcal{L}(T^{-1}\psi|_{\Sigma_0})=&\:(r^2+a^2)^{-\frac{1}{2}}F[\psi|_{\Sigma_0}]\\
=&\: 2[h\Delta-(r^2+a^2)] X\psi|_{\Sigma_0}+[(\Delta h)'-2r] \psi|_{\Sigma_0}+[2h(r^2+a^2)-h^2\Delta-a^2\sin^2\theta] T\psi|_{\Sigma_0}\\
&+2a(h-1) \Phi \psi|_{\Sigma_0}.
\end{split}
\end{equation}

\begin{definition}
We define the time-inverted Newman--Penrose charges $I^{(1)}_{\ell}[\psi]$, with $\ell=0,1,2$ as follows:
\begin{equation*}
I^{(1)}_{\ell}[\psi]:=I_{\ell}[T^{-1}\psi].
\end{equation*}
\end{definition}

\begin{proposition}
\label{prop:TinvNPconst}
Consider initial data $(\psi|_{\Sigma_0},T\psi|_{\Sigma_0})$ for \eqref{eq:waveeq}, with $(\phi|_{\Sigma_0},T\phi_{\Sigma_0}) \in (C^{\infty}(\widehat{\Sigma}))^2$, such that moreover $I_0[\psi]=0$. Then we can express:
\begin{equation*}
\begin{split}
I_0^{(1)}[\psi]=I_0[T^{-1}\psi]=&M\int_{r_+}^{\infty} (r^2+a^2)^{-\frac{1}{2}}\pi_0F(\uprho')\, d\uprho'-\frac{1}{2}\lim_{\uprho'\to \infty}\uprho' \pi_0 F(\uprho')-\lim_{\uprho'\to \infty}[\Delta h -2(r^2+a^2)]\phi_0(\uprho')\\
&-\frac{1}{2}a^2 \lim_{\uprho' \to \infty}\pi_0(\sin^2\theta  \phi)(\uprho').
\end{split}
\end{equation*}
In particular, if $(\phi|_{\Sigma_0},T\phi_{\Sigma_0}) \in (C_c^{\infty}({\Sigma}))^2$, then
\begin{equation*}
I_0^{(1)}[\psi]=M\int_{r_+}^{\infty} (r^2+a^2)^{-\frac{1}{2}} F_0[\psi|_{\Sigma_0}](\uprho')\, d\uprho.
\end{equation*}
Furthermore, if $(\phi|_{\Sigma_0},T\phi_{\Sigma_0}) \in (C_c^{\infty}({\Sigma}))^2$, then
\begin{align*}
I_{1m}^{(1)}[\psi]=I_{1m}[T^{-1}\psi]=&\frac{1}{2}M\int_{r_+}^{\infty}e^{-2i a m \int_r^{\infty} \Delta^{-1} dr'} \Delta X((r^2+a^2)^{-\frac{1}{2}} F_{1m}[\psi|_{\Sigma_0}])(\uprho')\,d\uprho',\\
I_{2m}^{(1)}[\psi]=I_{2m}[T^{-1}\psi]=&  \frac{1}{5}M\int_{r_+}^{\infty}e^{-2i a m \int_r^{\infty} \Delta^{-1} dr'} \Delta^2 X^2((r^2+a^2)^{-\frac{1}{2}} F_{2m}[\psi|_{\Sigma_0}])(\uprho')\,d\uprho'.
\end{align*}
\end{proposition}
\begin{proof}
We will suppress in the notation below the restriction $|_{\Sigma_0}$. We first consider the projection onto $\ell=0$. The quantity $(T^{-1}\psi)_0$ satisfies:
\begin{equation*}
X(\Delta X(T^{-1}\psi)_0)=(r^2+a^2)^{-\frac{1}{2}}F_0
\end{equation*}
and hence, by integrating from $\uprho=r_+$, we obtain
\begin{equation*}
X(T^{-1}\psi)_0(\uprho)=\Delta^{-1}\int_{r_+}^{\uprho} (r^2+a^2)^{-\frac{1}{2}}F_0(\uprho')\,d\uprho'.
\end{equation*}
Integrating the above expression again, starting from $\uprho=\infty$ and using that $\lim_{\uprho\to \infty}(T^{-1}\psi)_0(\uprho)=0$ by the regularity properties of $(T^{-1}\psi)_0$ following from Proposition \ref{prop:mainpropTinv}, we obtain
\begin{align*}
(T^{-1}\psi)_0(\uprho)=&\:-\int_{\uprho}^{\infty} \Delta^{-1}(\uprho_1)\int_{r_+}^{\uprho_1} (r^2+a^2)^{-\frac{1}{2}}F_0(\uprho_2)\,d\uprho_2 d\uprho_1.
\end{align*}
It will be convenient to denote 
\begin{align*}
Q_0=&\:\int_{r_+}^{\infty} (r^2+a^2)^{-\frac{1}{2}}F_0(\uprho')\, d\uprho',\\
G_0(\uprho)=&\:-\int_{\uprho}^{\infty} (r^2+a^2)^{-\frac{1}{2}}F_0(\uprho')\, d\uprho.
\end{align*}
It follows that
\begin{equation*}
G_0(\uprho)=-\uprho^{-1}\lim_{\uprho'\to \infty}\uprho' F_0(\uprho')+O_{\infty}(\uprho^{-2}),
\end{equation*}
with $\lim_{\uprho'\to \infty}\uprho' F_0(\uprho')$ well-defined by Proposition \ref{prop:mainpropTinv}.
Then we obtain for $\uprho\geq R$, with $R>r_+$:
\begin{align*}
(T^{-1}\phi)_0(\uprho)=&\:-(r^2+a^2)^{\frac{1}{2}}\int_{\uprho}^{\infty} [\uprho_1^{-2}+2M\uprho_1^{-3}+O_{\infty}(\uprho_1^{-4})] [Q_0-\uprho_1^{-1} \lim_{\uprho'\to \infty}\uprho' F_0(\uprho')+O_{\infty}(\uprho_1^{-2})] d\uprho_1\\
=&\: -Q_0(1+M\uprho^{-1})+\frac{1}{2}\lim_{\uprho'\to \infty}\uprho' F_0(\uprho')\uprho^{-1} +O_{\infty}(\uprho^{-2})\\
(r^2XT^{-1}\phi)_0(\uprho)=&\:M Q_0-\frac{1}{2}\lim_{\uprho'\to \infty}\uprho' F_0(\uprho')+O_{\infty}(\uprho^{-1}).
\end{align*}

In order to determine $I_0[T^{-1}\psi]$, we need the following relation between $X$, $L$ and $T$:
\begin{equation*}
X=2(r^2+a^2)\Delta^{-1} L+(h-2(r^2+a^2)\Delta^{-1})T-2a \Delta^{-1}\Phi.
\end{equation*}

We conclude that
\begin{equation*}
I_0[T^{-1}\psi]=M Q_0-\frac{1}{2}\lim_{\uprho'\to \infty}\uprho' F_0(\uprho')-\lim_{\uprho'\to \infty}(\Delta h -2(r^2+a^2))\phi_0(\uprho')-\frac{1}{2}a^2 \lim_{\uprho' \to \infty}\pi_0(\sin^2\theta  \phi)(\uprho').
\end{equation*}

By rearranging the terms in \eqref{eq:commidell}, we obtain the following expression for general $\ell$:
\begin{equation}
\label{eq:keypropeqfixedl}
X(e^{-2i a m \int_r^{\infty} \Delta^{-1} dr'}\Delta^{\ell+1} X^{\ell+1}(T^{-1}\psi)_{\ell m})=e^{-2i a m \int_r^{\infty} \Delta^{-1} dr'} \Delta^{\ell}X^{\ell}F_{\ell m}.
\end{equation}
Now consider the projection onto $\ell=1$ and $m\in \{-1,0,1\}$ and integrate \eqref{eq:keypropeqfixedl} to obtain:
\begin{equation*}
X^{2}(T^{-1}\psi)_{1 m}(\uprho)=\Delta^{-2}e^{2i a m \int_r^{\infty} \Delta^{-1} dr'}\int_{r_+}^{\uprho}e^{-2i a m \int_r^{\infty} \Delta^{-1} dr'} \Delta X ((r^2+a^2)^{-\frac{1}{2}}F_{1 m})(\uprho')\,d\uprho'.
\end{equation*}
Integrating the above expression again, starting from $\uprho=\infty$ and using that $\lim_{\uprho\to \infty}X^kT^{-1}\psi_1(\uprho)=0$ by the regularity properties of $T^{-1}\psi_1$ following from Proposition \ref{prop:mainpropTinv}, we obtain
\begin{align*}
X(T^{-1}\psi)_{1 m}(\uprho)=&-\int_{\rho}^{\infty} \Delta^{-2}e^{2i a m \int_r^{\infty} \Delta^{-1} dr'} \int_{r_+}^{\uprho_1}e^{-2i a m \int_r^{\infty} \Delta^{-1} dr'} \Delta X ((r^2+a^2)^{-\frac{1}{2}}F_{1 m})(\uprho')\,d\uprho_2 d\uprho_1,\\
(T^{-1}\psi)_{1 m}(\uprho)=&\int_{\rho}^{\infty} \int_{\rho_1}^{\infty}\Delta^{-2}e^{2i a m \int_r^{\infty} \Delta^{-1} dr'} \int_{r_+}^{\uprho_2}e^{-2i a m \int_r^{\infty} \Delta^{-1} dr'} \Delta X ((r^2+a^2)^{-\frac{1}{2}}F_{1 m})(\uprho')\,d \uprho_3 d\uprho_2 d\uprho_1
\end{align*}
By assumption of compact support of $\phi|_{\Sigma_0}$, the following expressions are well-defined
\begin{align*}
Q_{1m}=&\:\int_{r_+}^{\infty}e^{-2i a m \int_r^{\infty} \Delta^{-1} dr'} \Delta X((r^2+a^2)^{-\frac{1}{2}}F_{1 m})(\uprho')\,d\uprho'<\infty,\\
G_{1m}(\uprho)=&\:-\int_{\uprho}^{\infty} e^{-2i a m \int_r^{\infty} \Delta^{-1} dr'} \Delta X((r^2+a^2)^{-\frac{1}{2}}F_{1 m})(\uprho')\,d\uprho'
\end{align*}
and it follows that $G_{1m}$ has compact support, so in particular:
\begin{equation*}
G_{1m}(\uprho)=O_{\infty}(\uprho^{-2}).
\end{equation*}
Note moreover that
\begin{equation*}
\left(1-\frac{r_+}{r}\right)e^{-2i a m \int_r^{\infty} \Delta^{-1} dr'}=1+O_{\infty}(\uprho^{-2}).
\end{equation*}
We can therefore express for $\uprho\geq R$:
\begin{equation*}
(T^{-1}\psi)_{1 m}(\uprho)=\int_{\rho}^{\infty} \int_{\rho_1}^{\infty}[\uprho_2^{-4}+4M\uprho_2^{-5}+O_{\infty}(\uprho_2^{-6})]\left[Q_{1m}+O_{\infty}(\uprho_2^{-2})\right]\,d\uprho_2 d\uprho_1,
\end{equation*}
and hence
\begin{equation*}
\begin{split}
(T^{-1}\phi)_{1 m}(\uprho)=&(r^2+a^2)^{\frac{1}{2}}\int_{\rho}^{\infty} \int_{\rho_1}^{\infty}[Q_{1m}\uprho_2^{-4}+4MQ_{1m}\uprho_2^{-5}+O_{\infty}(\uprho_2^{-6})\,d\uprho_2 d\uprho_1\\
=&\frac{1}{6}Q_{1m} \uprho^{-1}+\frac{1}{3}MQ_{1m}\uprho^{-2}+O_{\infty}(\uprho^{-3}).
\end{split}
\end{equation*}
We then obtain:
\begin{align*}
r^2X(T^{-1}\phi)_{1m}(\uprho)=&-\frac{1}{6}Q_{1m}-\frac{2}{3}MQ_{1m} \uprho^{-1}+O_{\infty}(\uprho^{-2}),\\
r^2X(r^2X(T^{-1}\phi)_{1m})(\uprho)=&\: \frac{2}{3}MQ_{1m}+O_{\infty}(\uprho^{-1}).
\end{align*}
By using the compactness of the support of $(\phi,T\phi)$, together with \eqref{eq:keypropeqfixedl} and $\lim_{r\to \infty}(T^{-1}\phi)_1=0$, we obtain:
\begin{equation*}
\begin{split}
I_{1m}[T^{-1}\psi]=&\: \lim_{r\to \infty} r^2X(T^{-1}\check{\phi}^{(1)})_{1m}+i m a (T^{-1}\check{\phi}^{(1)})_{1m}\\
=&\:\lim_{r\to \infty}r^2X(r^2 X (T^{-1}\phi)_{1m})-(\alpha+im \alpha_{\Phi})r^2X(T^{-1}\phi)_{1m}+i m a r^2X(T^{-1}\phi)_{1m}\\
=&\: \frac{1}{2}MQ_{1m}.
\end{split}
\end{equation*}
Finally, we consider the projection onto $\ell=2$ and $m\in \{-,2,\ldots ,2\}$ and integrate \eqref{eq:keypropeqfixedl} to obtain:
\begin{equation*}
X^{3}(T^{-1}\psi)_{2 m}(\uprho)=\Delta^{-3}e^{2i a m \int_r^{\infty} \Delta^{-1} dr'}\int_{r_+}^{\uprho}e^{-2i a m \int_r^{\infty} \Delta^{-2} dr'} \Delta^2 X^2 ((r^2+a^2)^{-\frac{1}{2}}F_{2 m})(\uprho')\,d\uprho'.
\end{equation*}
Integrating the above expression multiple times, starting from $\uprho=\infty$ and using that $\lim_{\uprho\to \infty}X^k(T^{-1}\psi)_2(\uprho)=0$ by the regularity properties of $(T^{-1}\psi)_2$ following from Proposition \ref{prop:mainpropTinv}, we obtain
\begin{equation*}
(T^{-1}\psi)_{2 m}(\uprho)=-\int_{\rho}^{\infty} \int_{\rho_1}^{\infty}\int_{\rho_2}^{\infty}\Delta^{-3}e^{2i a m \int_r^{\infty} \Delta^{-1} dr'} \int_{r_+}^{\uprho_3}e^{-2i a m \int_r^{\infty} \Delta^{-1} dr'} \Delta^2 X^2 ((r^2+a^2)^{-\frac{1}{2}}F_{2 m})(\uprho')\,d \uprho_4 \ldots  d\uprho_1.
\end{equation*}
By assumption of compact support of $\phi|_{\Sigma_0}$, the following expressions are well-defined
\begin{align*}
Q_{2m}=&\:\int_{r_+}^{\infty}e^{-2i a m \int_r^{\infty} \Delta^{-1} dr'} \Delta^2 X^2((r^2+a^2)^{-\frac{1}{2}}F_{2 m})(\uprho')\,d\uprho'<\infty,\\
G_{2m}(\uprho)=&\:-\int_{\uprho}^{\infty} e^{-2i a m \int_r^{\infty} \Delta^{-1} dr'} \Delta^2 X^2((r^2+a^2)^{-\frac{1}{2}}F_{2 m})(\uprho')\,d\uprho'
\end{align*}
and it follows that $G_{2m}$ is compactly supported, so in particular:
\begin{equation*}
G_{2m}(\uprho)=O_{\infty}(\uprho^{-2}).
\end{equation*}
We can therefore express
\begin{equation*}
\begin{split}
(T^{-1}\phi)_{2m}(\uprho)=&-(r^2+a^2)^{\frac{1}{2}}\int_{\rho}^{\infty} \int_{\rho_1}^{\infty}\int_{\rho_2}^{\infty} Q_{2m}\uprho_3^{-6}+6MQ_{2m}\uprho_3^{-7}+O_{\infty}(\uprho_3^{-8})\,d\uprho_ 3d\uprho_2 d\uprho_1\\
=&-\frac{1}{60}Q_{2m} \uprho^{-2}-\frac{1}{20}MQ_{2m}\uprho^{-3}+O_{\infty}(\uprho^{-4}).
\end{split}
\end{equation*}
We then obtain:
\begin{align*}
r^2X(T^{-1}\phi)_{2m}(\uprho)=&\frac{1}{30}Q_{2m}\uprho^{-1}+\frac{3}{20}MQ_{2m} \uprho^{-2}+O_{\infty}(\uprho^{-3}),\\
(r^2X)^2(T^{-1}\phi)_{2m}(\uprho)=&\: -\frac{1}{30}Q_{2m}-\frac{3}{10}MQ_{2m}\uprho^{-1}+O_{\infty}(\uprho^{-2}),\\
(r^2X)^3(T^{-1}\phi)_{2m}(\uprho)=&\: \frac{3}{10}M Q_{2m}+O_{\infty}(\uprho^{-1}).
\end{align*}
We use that $\lim_{r\to \infty} (T^{-1}\phi)_2=\lim_{r\to \infty} r^2X(T^{-1}\phi)_2=0$ together with compactness of the support of $(\phi,T\phi)$ to express:
\begin{equation*}
\begin{split}
I_{2m}[T^{-1}\psi]=&\: \lim_{r\to \infty} r^2X(T^{-1}\check{\phi}^{(2)})_{2m}+i m a (T^{-1}\check{\phi}^{(2)})_{2m}\\
=&\:\lim_{r\to \infty}(r^2X)^3 T^{-1}\phi_{2m}-(\gamma+\alpha+im \alpha_{\Phi})(r^2X)^2 T^{-1}\phi_{2m}+i m a (r^2X)^2(T^{-1}\phi)_{2m}\\
=&\:\frac{1}{5}MQ_{2m}.
\end{split}
\end{equation*}
\end{proof}

\begin{remark}
One can generalize the argument in the proof of Proposition \ref{prop:TinvNPconst} in order to define $I_{\ell m}[T^{-1}\psi]$ for $\ell=1,2$ when the initial data $(\psi|_{\Sigma_0},T\psi|_{\Sigma_0})$ is not compactly supported, but satisfies  $(\phi|_{\Sigma_0},T\phi_{\Sigma_0}) \in (C^{N}(\widehat{\Sigma}))^2$, for some suitably large $N$, together with the conditions: $I_{\ell}[\psi]=0$, as in the $\ell=0$ case. In this setting, the expressions for $I_{\ell m}[T^{-1}\psi]$ will be considerably more complicated than in the case of compactly supported initial data, so we do not pursue this generalization here.
\end{remark}

In the corollary below, we obtain several simplified expressions for the time-inverted Newman--Penrose charges $I_{\ell m}^{(1)}[\psi]$ with $\ell=0,1,2$.
\begin{corollary}
Let $(\phi|_{\Sigma_0},T\phi_{\Sigma_0}) \in (C_c^{\infty}({\Sigma}))^2$.
\begin{enumerate}[\rm (i)]
\item
We can express:
\begin{align*}
I_0[T^{-1}\psi]=&\:   =\frac{1}{4\pi}M(r_+^2+a^2)\int_{\Sigma_0\cap\mathcal{H}^+}\psi\,d\omega+\frac{1}{4\pi}M\int_{\Sigma_0} \mathbf{n}_0(\psi)\,d\mu_0.
\end{align*}
\item
We alternatively express $I_i[T^{-1}\psi]$, with $i=1,2,3$ as integrals along $\mathcal{I}^+$:
\begin{align*}
I_0^{(1)}[\psi]=&\:   M\int_{0}^{\infty} \phi_0|_{\mathcal{I}^+}(\tau)\,d\tau\\
 I_1^{(1)}[\psi](\theta,\varphi_*)=&\:3M \int_{0}^{\infty} r^2X\phi_1|_{\mathcal{I}^+}(\tau,\theta,\varphi_*)\,d\tau, \\
I_2^{(1)}[\psi](\theta,\varphi_*)=&\:6M\int_{0}^{\infty} (r^2X)^2\phi_2|_{\mathcal{I}^+}(\tau,\theta,\varphi_*)\,d\tau.
\end{align*}
\end{enumerate}
\end{corollary}
\begin{proof}
In order to obtain (i), we first observe that
\begin{equation*}
\rho^2d\tau^{\sharp}=(r^2+a^2-h\Delta)X+[a^2\sin^2\theta +h^2\Delta-2h(r^2+a^2)]T+a(1-h)\Phi.
\end{equation*}
We can therefore write
\begin{equation*}
(r^2+a^2)^{-\frac{1}{2}}F[\psi|_{\Sigma_0}]=-\rho^2d\tau^{\sharp}(\psi)|_{\Sigma_0}+X\left[(h\Delta-(r^2+a^2))\psi\right]-a(1-h)\Phi(\psi).
\end{equation*}

Recall from \eqref{eq:normalder} that
\begin{equation*}
-\rho^2\sin \theta d\tau^{\sharp}(\psi)|_{\Sigma_0}=\mathbf{n}_{0}(\psi)\sqrt{\det g|_{\Sigma_{0}}}.
\end{equation*}
The expression in (i) now immediately follows.

We consider now (ii). Note that the quantities $Q_0$ and $Q_{im}$, with $i=1,2$, which are defined in the proof of Proposition \ref{prop:TinvNPconst}, can be expressed as limits of projected radiation fields of $T^{-1}\psi$:
\begin{align*}
(T^{-1}\phi)_0|_{\Sigma_0}(\infty)=&-Q_0,\\
r^2X(T^{-1}\phi)_{1m}|_{\Sigma_0}(\infty)=& -\frac{1}{6}Q_{1m},\\
(r^2X)^2(T^{-1}\phi)_{2m}|_{\Sigma_0}(\infty)=&-\frac{1}{30} Q_{2m} .
\end{align*}
Using that 
\begin{equation*}
(T^{-1}\phi)_0|_{\mathcal{I}^+},r^2X(T^{-1}\phi)_{1m}|_{\mathcal{I}^+},(r^2X)^2(T^{-1}\phi)_{2m}|_{\mathcal{I}^+}\to 0
\end{equation*}
as $\tau\to \infty$, by the decay estimates established in Propositions \ref{prop:pointwisedecayl0}--\ref{prop:pointwisedecayl3} applied to $T^{-1}\psi$, we can integrate $\phi_0|_{\mathcal{I}^+},r^2X\phi_{1m}|_{\mathcal{I}^+},(r^2X)^2\phi_{2m}|_{\mathcal{I}^+}$ to obtain (ii).
\end{proof}

\section{Late-time polynomial tails: the $\ell=0$ projection}
\label{sec:latetimeasympl0}
We derive in this section the precise leading-order behaviour in time of $\psi_{\geq \ell}$, with $\ell=0,1,2$, which will take the form of inverse polynomial tails.
 
\label{sec:asympl0nonzeronp}

We introduce the following additional initial data quantities: for $0<\beta\leq 1$ and $R_0>r_+$ arbitrarily large, we define
\begin{equation*}
D_{0,\beta,K}[\psi]:= \max_{1\leq k\leq K}\left|\left|v^{2+\beta+k}L^k\left(P_0-2I_0[\psi]v^{-2}\right)\right|\right|_{L^{\infty}(\Sigma_0\cap\{r\geq R_0\})}.
\end{equation*}
We moreover introduce the following auxiliary norm on $\Sigma_{\tau}$: for $0<\delta<1$, we define
\begin{equation*}
\begin{split}
S_{0,\delta,K}[\psi]:=&\:\sup_{\tau\geq 0 }\Bigg[(1+\tau)^{1-\delta}\sum_{k_1+k_2\leq K}(1+\tau)^{k_2}|| (rL)^{k_1}T^{k_2}\phi_0||_{L^{\infty}(\Sigma_{\tau})}\\
&+a^2(1+\tau)^{2-\delta}\sum_{\ell\in\{0,2\}}\sum_{k_1+k_2\leq K+1, k_2\leq K}(1+\tau)^{k_2}|| (rL)^{k_1}T^{k_2+1}\phi_{\ell}||_{L^{\infty}(\Sigma_{\tau})}\Bigg].
\end{split}
\end{equation*}

The lemma below establishes boundedness of $S_{0,\delta,K}[\psi]$ in terms of the weighted initial data energy norms defined in Section \ref{sec:edecay}.
\begin{lemma}
\label{lm:Sl0est}
There exists a constant $C=C(M,a,\delta,K)>0$ such that:
\begin{equation*}
S_{0,\delta,K}[\psi]\leq C \sqrt{\sum_{k+j\leq K+2}E_{0,k,\frac{\delta}{2}}[N^j\psi]+E_{2,k,\frac{\delta}{2}}[N^j\psi]}.
\end{equation*}
\end{lemma}
\begin{proof}
The estimate follows immediately by applying the estimates in Proposition \ref{prop:pointwisedecayl0} and \ref{prop:pointwisedecayl2}.
\end{proof}

We will consider in this section the timelike hypersurfaces $\gamma^{\alpha}=\{v-u=v^{\alpha}\}$, with $\alpha<1$ and the following subsets of $\{r\geq R\}$, with $R>r_+$ suitably large:
\begin{equation*}
\mathcal{A}_{\gamma^{\alpha}}=\{(u,v,\theta,\varphi_*)\in \mathcal{R}\cap\{r\geq R\}\,|\, u\leq v-v^{\alpha}\}.
\end{equation*}
Let us moreover introduce the notation $(u,v_{\gamma^{\alpha}}(u),\theta,\varphi_*)$ and $(u_{\gamma^{\alpha}}(v),v,\theta,\varphi_*)$ for points on the curve $\gamma^{\alpha}$.
\begin{figure}[H]
	\begin{center}
\includegraphics[scale=0.7]{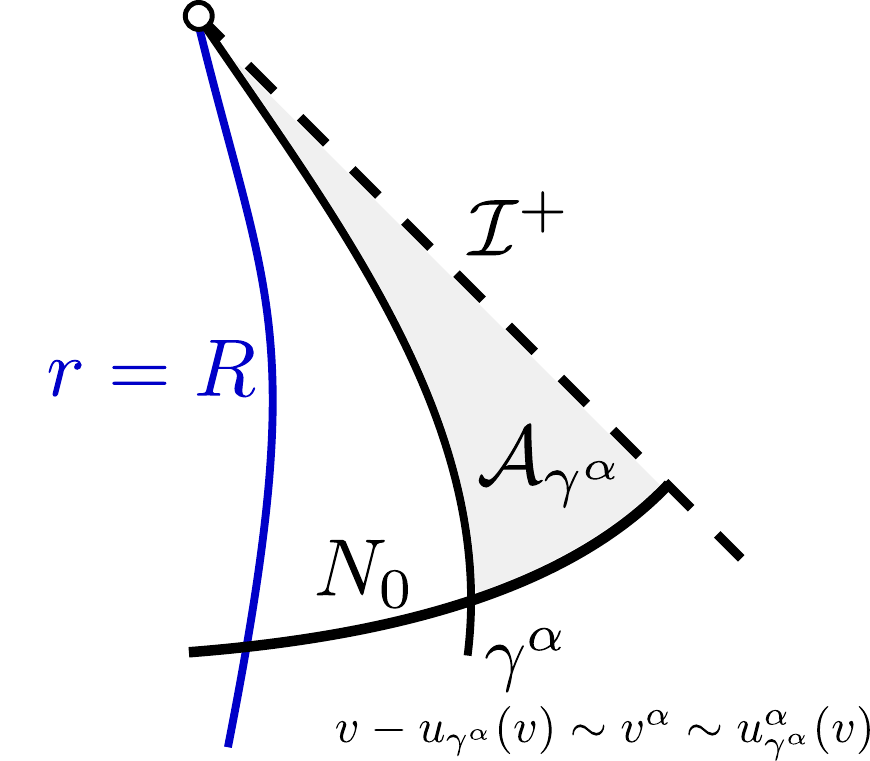}
\end{center}
\vspace{-0.2cm}
\caption{A 2-dimensional representation of the spacetime region $\mathcal{A}_{\gamma^{\alpha}}$.}
	\label{fig:contour}
\end{figure}

\subsection{Asymptotics in $\mathcal{A}_{\gamma^{\alpha}}$ }
\label{sec:asympl01}

The following relation between $r$ and the variables $u$ and $v$ will be important:
\begin{lemma}
Let $R_0>r_+$. Then there exists constants $c=c(M,a,R_0)$ and $C=C(M,a,R_0)$, such that in the region $r\geq R_0$:
\begin{equation}
\label{eq:ruv}
c M\log \left(1+\frac{|v-u|}{M}\right) \leq \left| r -\frac{v-u}{2}\right|\leq C M \log \left(1+\frac{|v-u|}{M}\right).
\end{equation}
\end{lemma}
\begin{proof}
We recall that $v-u=2r_*$ and
\begin{equation*}
\frac{dr_*}{dr}=\frac{r^2+a^2}{\Delta}.
\end{equation*}
Solutions $r_*$ have the following form:
\begin{equation*}
r_*(r)=r+2M\log r +O_{\infty}(r^0).
\end{equation*}
The estimates in \eqref{eq:ruv} then follows immediately.
\end{proof}

\begin{proposition}
\label{prop:asymplkP0}
Let $\alpha_K>\frac{3+K}{2+K}$. Then there exists $\delta(\alpha_K)>0$ suitably small, such that
\begin{equation}
\label{eq:asymplkP0}
\left|L^{K}P_0(u,v)-2I_0[\psi]L^k(v^{-2})\right|\leq CS_{0,\delta,K}[\psi] v^{-(3+K-2\delta)\alpha_K}+CD_{0,\beta,K}[\psi]v^{-2-K-\beta}.
\end{equation}
\end{proposition}
\begin{proof}
Let $K\in \N_0$. By acting with $L^K$ on both sides of \eqref{eq:maineqP0}, we obtain:
\begin{equation}
\label{eq:eqP0asymp}
4\underline{L}L^KP_0=\sum_{k=0}^KO_{\infty}(r^{-3-K})(rL)^k\phi_0+\sum_{k=0}^{K+1}O_{\infty}(r^{-3-K})[(rL)^kT\phi_{2,0}+(rL)^kT\phi_{0}].
\end{equation}
We can estimate
\begin{equation}
\label{eq:Lbarl0}
||\underline{L}L^KP_0||_{L^{\infty}(\Sigma_{\tau)}}\leq C (1+\tau)^{-1+\delta} r^{-3-K}S_{0,\delta,K}[\psi].
\end{equation}
Integrating from $\Sigma_0$ in the $\underline{L}$ direction, applying \eqref{eq:ruv} together with \eqref{eq:Lbarl0}, we obtain
\begin{equation*}
\begin{split}
\left|L^KP_0(u,v)-L^KP_0(u_{\Sigma_0}(v),v)\right|\leq &\:CS_{0,\delta,K}[\psi] \int_{u_{\Sigma_0}(v)}^u (1+\tau)^{-1+\delta}(u',v) r^{-3-K}(u',v)\,du'\\
\leq &\:CS_{0,\delta,K}[\psi] \int_{u_{\Sigma_0}(v)}^u u'^{-1+\delta} (v-u')^{-3-K}\,du'\\
\leq &\:CS_{0,\delta,K}[\psi] v^{-(3+K-2\delta)\alpha_K}\int_{u_{\Sigma_0}(v)}^u u'^{-1-\delta}\,du'\\
\leq &\:CS_{0,\delta,K}[\psi] v^{-(3+K-2\delta)\alpha_K}.
\end{split}
\end{equation*}
Furthermore, by definition of $D_{0,\beta,K}[\psi]$:
\begin{equation*}
\left|L^KP_0(u_{\Sigma_0}(v),v)-2I_0[\psi]L^K(v^{-2})\right|\leq C v^{-2-K-\beta}D_{0,\beta,K}[\psi].
\end{equation*}
By combining the above, we obtain \eqref{eq:Lbarl0}.
\end{proof}

\begin{corollary}
\label{cor:asympLTkphi0}
Let $\alpha_K>\frac{2+K}{3+K}$. Then there exists $\delta(\alpha_K)>0$ suitably small, such that
\begin{equation}
\begin{split}
\label{eq:asympLTkphi0}
&\left|L T^K \phi_0(u,v)-2 I_0[\psi]L^K(v^{-2})\right|\leq CS_{0,\delta,K}[\psi] \Bigg(v^{-(3+K-2\delta)\alpha_K}+u^{-K+\delta}(v-u)^{-3}+u^{-1+\delta} (v-u)^{-3-K}\\
&+u^{-2-K+\delta}(v-u)^{-2}\Bigg)+CD_{0,\beta,K}[\psi]v^{-2-K-\beta}.
\end{split}
\end{equation}
\end{corollary}
\begin{proof}
It follows from a straightforward inductive argument that we can express
\begin{equation*}
T^KP_0=L^KP_0+\sum_{k=0}^{K-1} \underline{L} L^{k} T^{K-1-k}P_0.
\end{equation*}
Furthermore, we have that
\begin{equation*}
\begin{split}
\sum_{k=0}^{K-1}|\underline{L} L^{k} T^{K-1-k}P_0|(u,v)\leq&\: C \sum_{k=0}^{K-1}u^{-(K-k)+\delta} r^{-3-k}S_{0,\delta,K}[\psi]\\
\leq &\: S_{0,\delta,K}[\psi][u^{-K+\delta} r^{-3}+u^{-1+\delta}r^{-3-K}].
\end{split}
\end{equation*}
Now, recall that
\begin{equation*}
P_0=L\phi_0-\frac{1}{4}a^2\frac{\Delta}{(r^2+a^2)^2}\pi_0(\sin^2\theta T\phi),
\end{equation*}
so we can estimate
\begin{equation*}
\begin{split}
|LT^K\phi_0(u,v)-T^KP_0(u,v)|\leq &\:CS_{0,\delta,K}[\psi] (v-u)^{-2} u^{-2-K+\delta}.
\end{split}
\end{equation*}
The estimate \eqref{eq:asympLTkphi0} follows by combining the above and applying \eqref{eq:ruv}.
\end{proof}

\begin{corollary}
\label{cor:asympl0nearI}
For $\alpha_K'$ suitably close to 1, $\delta>0$ suitably small and $0<\beta\leq 1$ suitably large, there exists $\nu>0$ and a constant $C=C(M,a,\alpha_K',\nu,\delta)>0$ such that
\begin{equation}
\label{eq:globalasympl0}
\begin{split}
\left|T^K \psi_0(u,v)-4I_0[\psi]T^K\left(\frac{1}{u v}\right)\right|\leq&\: C\left(\sqrt{E_{0,K,\delta}[\psi]} +S_{0,\delta,K}[\psi] +D_{0,\beta,K}[\psi] +|I_0[\psi]|\right)v^{-1}u^{-1-K-\nu},
\end{split}
\end{equation}
for all $(u,v,\theta,\varphi)\in \mathcal{A}_{\gamma^{\alpha_K'}}$.

In particular,
\begin{align}
\label{eq:l0asympinfty}
\left|T^K \phi_0|_{\mathcal{I}^+}-2I_0[\psi]T^K\left(u^{-1}\right)\right|\leq&\: C\left(S_{0,\delta,K}[\psi] +D_{0,\beta,K}[\psi]+|I_0[\psi]|+\sqrt{E_{0,K,\delta}[\psi]}\right)u^{-1-K-\nu},\\
\label{eq:l0asympgamma}
\left|T^K \psi_0|_{\gamma^{\alpha_K'}}-4I_0[\psi]T^K((1+\tau)^{-2})\right|\leq&\: C\left(S_{0,\delta,K}[\psi] +D_{0,\beta,K}[\psi]+|I_0[\psi]|+\sqrt{E_{0,K,\delta}[\psi]}\right)(1+\tau)^{-1-K-\nu}.
\end{align}
\end{corollary}
\begin{proof}
Observe first of all that for $K\geq 1$:
\begin{equation*}
\int_{v_{\gamma^{\alpha_K}(u)}}^v L^{K}(v^{-2})\,dv'=L^{K-1}(v^{-2})-L^{K-1}(v'^{-2})|_{v'=v_{\gamma^{\alpha_K}(u)}}
\end{equation*}
with
\begin{equation*}
|v_{\gamma^{\alpha_K}(u)}^{-1-K}-u^{-1-K}|\leq C u^{-2-K+\alpha_k}
\end{equation*}
and
\begin{equation*}
L^{K-1}(v^{-2})-L^{K-1}(v'^{-2})|_{v'=u}=T^{K}\left(\frac{1}{u}-\frac{1}{v}\right)=(v-u)T^K\left(\frac{1}{uv}\right).
\end{equation*}
Using the above together with \eqref{eq:asympLTkphi0} and taking $\alpha_K$ suitably close to 1, we can integrate $L T^K\phi_0$ to obtain
\begin{equation*}
\begin{split}
&\left|T^K \phi_0(u,v)-T^K \phi_0(u,v_{\gamma^{\alpha_K}(u)})-2I_0[\psi](v-u)T^K\left(\frac{1}{u v}\right)\right|\\
\leq&\: C\left(S_{0,\delta,K}[\psi] +D_{0,\beta,K}[\psi]+|I_0[\psi]|\right)u^{-1-K-\nu},
\end{split}
\end{equation*}
for some $\nu>0$.

By \eqref{eq:pointwl02}, the following holds: given $\delta'>0$ arbitrarily small, there exists a constant $C>0$ such that
\begin{equation*}
\begin{split}
|T^K\phi_0|(u,v_{\gamma^{\alpha_K}(u)})\leq&\: C r^{\frac{1}{2}}(u,v_{\gamma^{\alpha_K}(u)})u^{-\frac{3}{2}-K+2\delta'}\sqrt{E_{0,K,\delta'}[\psi]}\\
\leq&\: C\sqrt{E_{0,K}[\psi]}u^{-\frac{3}{2}+\frac{1}{2}\alpha_k-K+\delta'},
\end{split}
\end{equation*}
so that
\begin{equation*}
\begin{split}
\left|T^K \phi_0(u,v)-2I_0[\psi]T^K\left(\frac{v-u}{u v}\right)\right|\leq&\: C\left(S_{0,\delta,K}[\psi] +D_{0,\beta,K}[\psi]+|I_0[\psi]|+\sqrt{E_{0,K,\delta}[\psi]}\right)u^{-1-K-\nu},
\end{split}
\end{equation*}
for $\delta>0$ suitably small.

By dividing the above equation by $r$, we moreover obtain the following estimate in the \emph{smaller} region $\mathcal{A}_{\gamma^{\alpha_K'}}$, with $\alpha_K'>\alpha_K$ suitably close to 1, then there exists a constant $\nu>0$ such that
\begin{equation*}
\begin{split}
\left|T^K \psi_0(u,v)-4I_0[\psi]T^K\left(\frac{1}{u v}\right)\right|\leq&\: C\left(S_{0,\delta,K}[\psi] +D_{0,\beta,K}[\psi]+|I_0[\psi]|+\sqrt{E_{0,K,\delta}[\psi]}\right)v^{-1}u^{-1-K-\nu}.
\end{split}
\end{equation*}
\end{proof}

\subsection{Asymptotics in $\mathcal{R}\setminus \mathcal{A}_{\gamma^{\alpha}}$ }
\label{sec:asympl02}
In this section, we \emph{propagate} the late-time asymptotics derived in Corollary \ref{cor:asympl0nearI} to the rest of the spacetime.
\begin{proposition}
\label{prop:asympl0boundr}
Let $K\in \N_0$ and let $0<\alpha<1$ be sufficiently close to $1$. Then there exists a constant $C>0$ such that
\begin{equation}
\label{eq:asympinteriorl0}
\begin{split}
&\left|T^K \psi_0-4I_0[\psi]T^K((1+\tau)^{-2})\right|(\tau,\uprho)\\
\leq&\: C(1+\tau)^{-2-K-\nu}\left(S_{0,\delta,K}[\psi] +D_{0,\beta,K}[\psi]+|I_0[\psi]|+\sum_{\ell\in\{0,2\}} \sqrt{E_{\ell,K+2,\delta}[\psi]}+\sum_{j=0}^1\sqrt{E_{\ell,K+1,\delta}[N^j\psi]} \right)
\end{split}
\end{equation}
for all $(\tau,\uprho,\theta,\varphi_*)\in \mathcal{R}\setminus \mathcal{A}_{\gamma^{\alpha}}$.
\end{proposition}
\begin{proof}
We apply the fundamental theorem of calculus, integrating $XT^K\psi_0$ between $\uprho=\uprho'$ and $\uprho=\uprho_{\gamma^{\alpha}(\tau)}$. We use \eqref{eq:l0asympgamma} to estimate the boundary term on $\uprho=\uprho_{\gamma^{\alpha}(\tau)}$ and \eqref{eq:fasterdecayl0} to estimate the contribution of the integral.
\end{proof}

\begin{remark}
One can insert the expression \eqref{eq:globalasympl0} when integrating both sides of \eqref{eq:eqP0asymp} to refine the asymptotics of $\phi_0$ along $\mathcal{I}^+$ and obtain the next-to-leading order logarithmic asymptotics; see Remark \ref{rm:log} and also \cite{logasymptotics} for an application of this argument in the case $a=0$. We do not pursue this refinement in the present paper.
\end{remark}

\subsection{Asymptotics with vanishing Newman--Penrose charges}

The estimates in Section \ref{sec:asympl01} and \ref{sec:asympl02} provide the late-time asymptotics for $\psi$ arising from initial data with $I_0[\psi]\neq 0$. Using the time-integral construction from Section \ref{sec:timeinv}, we can moreover apply these estimates to the setting when $I_0[\psi]=0$, and in particular, to the setting where the initial data of $\psi$ is smooth and compactly supported.

\begin{proposition}
\label{prop:asympl0NP0boundr}
Consider initial data $(\psi|_{\Sigma_0},T\psi|_{\Sigma_0})$ for \eqref{eq:waveeq}, with $(\phi|_{\Sigma_0},T\phi_{\Sigma_0}) \in (C^{\infty}(\widehat{\Sigma}))^2$, such that moreover $I_0[\psi]=0$.  Let $r_0>r_+$. Then there exists a $\nu>0$ and a constant $C=C(M,a,\Sigma_0,r_0,\nu)>0$, such that
\begin{equation*}
\begin{split}
&\left|T^K \psi_0(u,v)-4I_0[T^{-1}\psi]T^{K+1}\left(\frac{1}{u v}\right)\right|\\
\leq&\: Cv^{-1}u^{-2-K-\nu}\left(\sqrt{E_{0,K+1,\delta}[T^{-1}\psi]} +S_{0,\delta,K+1}[T^{-1}\psi] +D_{0,\beta,K+1}[T^{-1}\psi] +|I_0[T^{-1}\psi]|\right)
\end{split}
\end{equation*}
in $\{r\geq r_0\}$, whereas in $\{r\leq r_0\}$, we can express:
\begin{equation*}
\begin{split}
&\left|T^K \psi_0+8I_0[T^{-1}\psi]T^K((1+\tau)^{-3})\right|(\tau,\uprho)\\
\leq&\: C(1+\tau)^{-3-K-\nu}\Bigg(S_{0,\delta,K+1}[T^{-1}\psi] +D_{0,\beta,K+1}[T^{-1}\psi]+|I_0[T^{-1}\psi]|\\
&+\sum_{\ell\in\{0,2\}} \sqrt{E_{\ell,K+3,\delta}[T^{-1}\psi]}+\sum_{j=0}^1\sqrt{E_{\ell,K+2,\delta}[N^jT^{-1}\psi]} \Bigg).
\end{split}
\end{equation*}
\end{proposition}
\begin{proof}
We apply Propositions \ref{prop:mainpropTinv} and \ref{prop:TinvNPconst} to conclude that $T^{-1}\phi$ has sufficiently high regularity to conclude that all the relevant initial data energy norms for $T^{-1}\psi$ are finite and $I_0[T^{-1}\psi]$ is well-defined. From this it follows that Corollary \ref{cor:asympl0nearI} and Proposition \ref{prop:asympl0boundr} immediately apply when $\psi$ is replaced by $T^{-1}\psi$.
\end{proof}
\section{Late-time polynomial tails: the $\ell=1$ projection}
\label{sec:latetimeasympl1}
We derive in this section the leading-order late-time asymptotics of the restriction $\psi_1$ of $\psi$.

It will be useful for the estimates in this section to define the following quantities: let $\beta>0$ and $0<\delta<1$, then
\begin{align*}
D_{1,\beta,K}[\psi]:=&\:\max_{0\leq k\leq K} ||v^{4+\beta+k}L^k\left(r^{-2}P_1-8I_1[\psi]v^{-4}\right)||_{L^{\infty}(\Sigma_0\cap\{r\geq R_0\})},\\
S_{1,\delta,K}[\psi]:=&\:\sup_{\tau\geq 0}\: (1+\tau)^{1-\delta}\sum_{j=0}^1\sum_{\substack{0\leq k_1+k_2\leq K+1\\ k_2\leq K}}(1+\tau)^{k_2}||(rL)^{k_1}T^{k_2}\check{\phi}_1^{(j)}||_{L^{\infty}(\Sigma_{\tau})}\\
&+\sup_{\tau\geq 0} (1+\tau)^{2-\delta}\sum_{\ell\in\{1,3\}}\sum_{j=0}^1\sum_{\substack{0\leq k_1+k_2\leq K+1\\ k_2\leq K}}(1+\tau)^{k_2} ||(rL)^{k_1}T^{k_2+1}\check{\phi}_{\ell}^{(j)}||_{L^{\infty}(\Sigma_{\tau})}.
\end{align*}

\begin{lemma}
There exists a constant $C=C(M,a,\delta,K)>0$ such that:
\begin{equation*}
S_{1,\delta,K}[\psi]\leq C \sqrt{\sum_{k+j\leq K+2}E_{1,k,\frac{\delta}{2}}[N^j\psi]+E_{\geq 3,k,\frac{\delta}{2}}[N^j\psi]}.
\end{equation*}
\end{lemma}
\begin{proof}
The estimate follows immediately by applying the estimates in Proposition \ref{prop:pointwisedecayl1} and \ref{prop:pointwisedecayl3}.
\end{proof}

\subsection{Asymptotics in $\mathcal{A}_{\gamma^{\alpha}}$ }
\label{sec:asympl11}

\begin{proposition}
Let $\alpha_K>\frac{4+K}{5+K}$. Then there exists $\delta(\alpha_K)>0$ suitably small, such that
\begin{equation}
\label{eq:asymplkP1}
\left|L^{K}(r^{-2}P_1)(u,v,\theta,\varphi_*)-8I_1[\psi](\theta,\varphi_*) L^K(v^{-4})\right|\leq CS_{1,\delta,K}[\psi] v^{-(5+K-2\delta)\alpha_K}+C(D_{1,\beta,K}[\psi]+I_1[\psi])v^{-(4+K+\beta)}.
\end{equation}
\end{proposition}
\begin{proof}
Let $K\in \N_0$. By acting with $r^{-2}$ and then with  $L^K$ on both sides of \eqref{eq:maineqP1}, we obtain:
\begin{equation}
\label{eq:LbareqP1asymp}
4\underline{L}L^K(r^{-2}P_1)=\sum_{k=0}^KO_{\infty}(r^{-4-K})(rL)^kP_1+\sum_{j=0}^1\sum_{k=0}^KO_{\infty}(r^{-5-K})(rL)^k\phi^{(j)}_1+\sum_{j=0}^1\sum_{k=0}^{K+1}O_{\infty}(r^{-5-K})(rL)^kT\pi_1(\sin^2\theta \phi^{(j)}).
\end{equation}

We can further estimate for all $k\in \N_0$:
\begin{equation}
\label{eq:estP1Lphi1}
|(rL)^k P_1-r^{-1}(rL)^{k+1}\check{\phi}^{(1)}_1| (u,v,\theta,\varphi_*)\leq Cr^{-2}(u,v)\sum_{j=0}^1a^2||(rL)^kT\check{\phi}_{3}^{(j)}||_{L^{\infty}(\Sigma_{\tau})}+a ||(rL)^k\check{\phi}^{(j)}_1||_{L^{\infty}(\Sigma_{\tau})}
\end{equation}
and hence obtain:
\begin{equation}
\label{eq:Lbarl1}
|\underline{L}L^K(r^{-2}P_1)|(u,v,\theta,\varphi_*)\leq C u^{-1+\delta} r^{-5-K}S_{1,\delta,K}[\psi].
\end{equation}
Integrating from $\Sigma_0$ in the $\underline{L}$ direction, applying \eqref{eq:ruv} together with \eqref{eq:Lbarl0}, we obtain
\begin{equation*}
\begin{split}
\left|L^K(r^{-2}P_1)(u,v,\theta,\varphi_*)-L^K(r^{-2}P_1)(u_{\Sigma_0}(v),v,\theta,\varphi_*)\right|\leq &\:CS_{1,\delta,K}[\psi] \int_{u_{\Sigma_0}(v)}^u (1+\tau)^{-1+\delta}(u',v) r^{-5-K}(u',v)\,du'\\
\leq &\:CS_{1,\delta,K}[\psi] \int_{u_{\Sigma_0}(v)}^u u'^{-1+\delta} (v-u')^{-5-K}\,du'\\
\leq &\:CS_{1,\delta,K}[\psi] v^{-(5+K-2\delta)\alpha_K}\int_{u_{\Sigma_0}(v)}^u u'^{-1-\delta}\,du'\\
\leq &\:CS_{1,\delta,K}[\psi] v^{-(5+K-2\delta)\alpha_K}.
\end{split}
\end{equation*}
Furthermore, we have that
\begin{equation*}
\left|L^K(r^{-2}P_1)(u_{\Sigma_0}(v),v,\theta,\varphi_*)-8I_1[\psi]L^K(v^{-4})\right|\leq C v^{-4-K-\beta}D_{1,\beta,K}[\psi].
\end{equation*}
so we can conclude that
\begin{equation*}
\left|L^K(r^{-2}P_1)(u,v,\theta,\varphi_*)-8I_1[\psi]L^K(v^{-4})\right|\leq CS_{1,\delta,K}[\psi] v^{-(5+K-2\delta)\alpha_K}+CD_{1,\beta,K}[\psi]v^{-4-K-\beta}.
\end{equation*}
\end{proof}

\begin{corollary}
Let $\alpha_K>\frac{4+K}{5+K}$. Then there exists $\delta(\alpha_K)>0$ suitably small, such that
\begin{equation}
\label{eq:asympTKLphi1}
\begin{split}
&\left|LT^K\check{\phi}_1^{(1)}(u,v,\theta,\varphi_*)-2I_1[\psi](v-u)^2L^K(v^{-4})\right|\\
\leq&\: CS_{1,\delta,K}[\psi] [a (v-u)^{-2} u^{-1-K+\delta}+(v-u)^2v^{-(5+K-2\delta)\alpha_K}+u^{-K+\delta} (v-u)^{-3}+u^{-1+\delta}(v-u)^{-3-K}]\\
&+CD_{1,\beta,K}[\psi](v-u)^2v^{-4-K-\beta}.
\end{split}
\end{equation}
\end{corollary}
\begin{proof}
By induction, it follows that we can express
\begin{equation*}
r^{-2}T^KP_1=L^K(r^{-2}P_1)+\sum_{k=0}^{K-1} \underline{L} L^{k} (r^{-2}T^{K-1-k}P_1).
\end{equation*}
Furthermore, by applying \eqref{eq:LbareqP1asymp} with $P_1$ replaced by $T^{K-1-k}P_1$, we obtain:
\begin{equation*}
\begin{split}
\sum_{k=0}^{K-1}|\underline{L} L^{k} T^{K-1-k}P_1|(u,v)\leq&\: C \sum_{k=0}^{K-1}u^{-(K-k)+\delta} r^{-5-k}S_{1,\delta,K}[\psi]\\
\leq &\: S_{1,\delta,K}[\psi][u^{-K+\delta} r^{-5}+u^{-1+\delta}r^{-5-K}].
\end{split}
\end{equation*}
Hence, it follows from \eqref{eq:asymplkP1} that
\begin{equation*}
\begin{split}
\left|T^KP_1(u,v,\theta,\varphi_*)-8I_1[\psi]r^2L^K(v^{-4})\right|\leq&\: CS_{1,\delta,K}[\psi] [r^2v^{-(5+K-2\delta)\alpha_K}+u^{-K+\delta} r^{-3}+u^{-1+\delta}r^{-3-K}]\\
&+CD_{1,\beta,K}[\psi]r^2v^{-4-K-\beta}.
\end{split}
\end{equation*}
By \eqref{eq:estP1Lphi1} with $k=0$ and $\phi$ replaced by $T^K\phi$, we obtain:
\begin{equation*}
\begin{split}
\left|LT^K\check{\phi}_1^{(1)}(u,v,\theta,\varphi_*)-T^KP_1\right|\leq&\: Ca S_{1,\delta,K}[\psi]r^{-2} u^{-1-K+\delta}.
\end{split}
\end{equation*}
By combining the above and applying \eqref{eq:ruv}, we conclude that \eqref{eq:asympTKLphi1} holds.
\end{proof}

The lemma below provides some integral identities that will be useful for integrating (multiple times) in the $L$-direction.
\begin{lemma}
\label{lm:basicintid}
Let $K\in \N_0$ and $m\in \R$, with $m>1$. Then:
\begin{equation}
\label{eq:usefulvint}
\int_{v_1}^{v_2} (v-u)^{m-2} L^K(v^{-m})\,dv=(m-1)^{-1}\left[(v_2-u)^{m-1} T^K\left(\frac{1}{u v^{m-1}}\right)|_{v=v_2}-(v_1-u)^{m-1} T^K\left(\frac{1}{u v^{m-1}}\right)|_{v=v_1}\right].
\end{equation}
\end{lemma}
\begin{proof}
Note that
\begin{equation*}
\begin{split}
\int_{v_1}^{v_2} (v-u)^{m-2} L^K(v^{-m})\,dv=&\:\int_{v_1}^{v_2} T^K((v-u)^{m-2} v^{-m})\,dv\\
=&\: \int_{v_1}^{v_2} LT^K\left((m-1)^{-1}\frac{(v-u)^{m-1}}{u v^{m-1}}\right)\,dv\\
=&(m-1)^{-1}\left[(v_2-u)^{m-1} T^K\left(\frac{1}{u v^{m-1}}\right)|_{v=v_2}-(v_1-u)^{m-1} T^K\left(\frac{1}{u v^{m-1}}\right)|_{v=v_1}\right].
\end{split}
\end{equation*}
\end{proof}

\begin{proposition}
\label{prop:asympl1nearI}
For $\alpha_K''$ suitably close to 1, $\delta>0$ suitably small and $0<\beta\leq 1$ suitably large, there exists $\nu>0$ and a constant $C=C(M,a,\alpha_K'',\nu,\delta)>0$ such that
\begin{align}
\label{eq:asympcheckphi1}
\Big|T^K\check{\phi}^{(1)}_1(u,v,\theta,\varphi_*)&-\frac{2}{3}I_1[\psi](\theta,\varphi_*)(v-u)^3T^K(u^{-1}v^{-3})\Big|\\
\leq&\: C\left(\sqrt{E_{1,K,\delta}[\psi]} +S_{1,\delta,K}[\psi] +D_{1,\beta,K}[\psi] +\sum_{m=-1}^1|I_{1m}[\psi]|\right)u^{-1-K-\nu},\\
\label{eq:asymppsi1}
\Big|r^{-1}T^K\psi_1(u,v,\theta,\varphi_*)&-\frac{8}{3}I_1[\psi](\theta,\varphi_*) T^K(u^{-2}v^{-2})\Big|\\
\leq&\:  C\left(\sqrt{E_{1,K,\delta}[\psi]} +S_{1,\delta,K}[\psi] +D_{1,\beta,K}[\psi] +\sum_{m=-1}^1|I_{1m}[\psi]|\right)v^{-2}u^{-2-K-\nu},\\
\label{eq:asympLpsi1}
\Big|LT^K\psi_1(u,v,\theta,\varphi_*)&-\frac{4}{3}I_1[\psi](\theta,\varphi_*) T^K(u^{-2}v^{-2})\Big|\\
\leq&\:  C\left(\sqrt{E_{1,K,\delta}[\psi]} +S_{1,\delta,K}[\psi] +D_{1,\beta,K}[\psi] +\sum_{m=-1}^1|I_{1m}[\psi]|\right)v^{-2}u^{-2-K-\nu}
\end{align}
in $\mathcal{A}_{\gamma^{\alpha_K''}}$.

In particular, we have that
\begin{align}
\label{eq:asympcheckphiinfty1}
\left|T^K{\check{\phi}^{(1)}_1}|_{\mathcal{I}^+}-\frac{2}{3}I_1[\psi]T^K(u^{-1})\right|\leq&\: C\left(\sqrt{E_{1,K,\delta}[\psi]} +S_{1,\delta,K}[\psi] +D_{1,\beta,K}[\psi] +\sum_{m=-1}^1|I_{1m}[\psi]|\right)u^{-1-K-\nu},\\
\label{eq:asympcheckphiinfty}
\left|T^K\phi_1|_{\mathcal{I}^+}-\frac{2}{3}I_1[\psi] T^K(u^{-2})\right|\leq&\:  C\left(\sqrt{E_{1,K,\delta}[\psi]} +S_{1,\delta,K}[\psi] +D_{1,\beta,K}[\psi] +\sum_{m=-1}^1|I_{1m}[\psi]|\right)u^{-2-K-\nu}
\end{align}
and
\begin{align}
\label{eq:asymppsi1gamma}
\left|r^{-1}T^K\psi_1|_{\gamma^{\alpha_K'}}-\frac{8}{3}I_1[\psi] T^K(\tau^{-4})\right|\leq &\: C\left(\sqrt{E_{1,K,\delta}[\psi]} +S_{1,\delta,K}[\psi] +D_{1,\beta,K}[\psi] +\sum_{m=-1}^1|I_{1m}[\psi]|\right)\tau^{-4-K-\nu},\\
\label{eq:asympLpsi1gamma}
\left|XT^K\psi_1|_{\gamma^{\alpha_K'}}-\frac{8}{3}I_1[\psi] T^K(\tau^{-4})\right|\leq&\:   C\left(\sqrt{E_{1,K,\delta}[\psi]} +S_{1,\delta,K}[\psi] +D_{1,\beta,K}[\psi] +\sum_{m=-1}^1|I_{1m}[\psi]|\right) \tau^{-4-K-\nu}.
\end{align}
\end{proposition}
\begin{proof}
Let $(u,v,\theta,\varphi_*)$ be a point in $\mathcal{A}_{\gamma^{\alpha_K}}$. Let us denote with $\gamma_L$ the integral curve tangent to $L$ connecting the point $(u,v,\theta,\varphi_*)$ with the curve $\gamma^{\alpha_K}$, where $\dot \gamma_L=L$. Then $(u_{\gamma_L}(v'),\theta_{\gamma_L}(v'))=(u,\theta)$ for all $v_{\gamma^{\alpha_K}(u)}\leq v'\leq v$ and $L(\varphi_*)=\frac{a}{r^2+a^2}$ with $(\varphi_*)_{\gamma_L}(v)=\varphi_*$.

We can therefore obtain for all $v_{\gamma^{\alpha_K}(u)}\leq v'\leq v$:
\begin{equation}
\label{eq:varphidiff}
|\varphi_*-(\varphi_*)_{\gamma_L}(v')|\leq  \int_{v_{\gamma^{\alpha_K}}}^v |a| (r^2+a^2)^{-1}\,dv'\leq C u^{-\alpha_K}.
\end{equation}

Hence, for all $v_{\gamma^{\alpha_K}(u)}\leq v'\leq v$:
\begin{equation}
\label{eq:NPconstdiffangle}
|I_1(\theta,(\varphi_*)_{\gamma_L}(v'))-I_1(\theta,\varphi_*)|\leq C \sum_{m=-1}^1|I_{1m}|u^{-\alpha_K}.
\end{equation}

By combining the above estimate with Lemma \ref{lm:basicintid} for $m=4$, we have that:
\begin{equation*}
\left|\int_{v_{\gamma^{\alpha_K}(u)}}^v2I_1[\psi]|_{\gamma_L}(v'-u)^2L^K(v'^{-4})\,dv'-\frac{2}{3}I_1[\psi](\theta,\varphi_*)(v-u)^3T^K(u^{-1}v^{-3})\right|\leq C\sum_{m=-1}^1|I_{1m}| u^{-4-K+3\alpha_K}.
\end{equation*}
By integrating \eqref{eq:asympTKLphi1}, and applying \eqref{eq:pointwl1b}, we obtain
\begin{equation*}
\begin{split}
|T^K\check{\phi}_1^{(1)}|(u,v_{\gamma^{\alpha_K}(u)},\theta,\varphi_*)\leq&\: C r^{\frac{1}{2}(1+\delta)}(u,v_{\gamma^{\alpha_K}(u)})u^{-\frac{3}{2}-K+2\delta}\sqrt{E_{1,K,\delta}[\psi]}\\
\leq&\: C\sqrt{E_{1,K,\delta}[\psi]}u^{-\frac{3}{2}+\frac{1}{2}\alpha_K-K+2\delta},
\end{split}
\end{equation*}
we then obtain \eqref{eq:asympcheckphi1} and  \eqref{eq:asympcheckphiinfty1} for some $\nu>0$ suitably small. 

Furthermore, since $T^K \check{\phi}_1^{(1)}=2(r^2+O(r))L\phi_1+aKO(r^{-1})T^{K-1}\phi_1$, we have that:
\begin{equation}
\label{eq:asymplphi1}
\begin{split}
&\left|LT^K\phi_1(u,v,\theta,\varphi_*)-\frac{4}{3}I_1[\psi](\theta,\varphi_*)(v-u)T^K(u^{-1}v^{-3})\right|\\
\leq&\:  C\left(\sqrt{E_{1,K,\delta}[\psi]} +S_{1,\delta,K}[\psi] +D_{1,\beta,K}[\psi] +\sum_{m=-1}^1|I_{1m}|\right)(v-u)^{-2}u^{-1-K-\nu}.
\end{split}
\end{equation}

Let $\alpha_K'>\alpha_K$, then we can estimate in $\mathcal{A}_{\gamma^{\alpha_K'}}$ using \eqref{eq:pointwl12}:
\begin{equation*}
\begin{split}
|T^K\phi_1|(u,v_{\gamma^{\alpha_K'}(u)},\theta,\varphi)\leq&\:  C\sqrt{E_{1,K,\delta}[\psi]}u^{-\frac{5}{2}+\frac{1}{2}\alpha'_K-K+\delta},
\end{split}
\end{equation*}
so by integrating \eqref{eq:asymplphi1} and applying \eqref{eq:NPconstdiffangle} and Lemma \ref{lm:basicintid} again, taking $\alpha_K'>\alpha_K$ suitably large, we obtain the following estimate in $\mathcal{A}_{\gamma^{\alpha_K'}}$:
\begin{equation}
\label{eq:asympphi1}
\begin{split}
&\left|T^K\phi_1(u,v,\theta,\varphi_*)-\frac{2}{3}I_1[\psi](\theta,\varphi_*)(v-u)^2 T^K(u^{-2}v^{-2})\right|\\
\leq&\:  C\left(\sqrt{E_{1,K,\delta}[\psi]} +S_{1,\delta,K}[\psi] +D_{1,\beta,K}[\psi] +\sum_{m=-1}^1|I_{1m}|\right)u^{-2-K-\nu}.
\end{split}
\end{equation}
Hence, \eqref{eq:asympcheckphiinfty} follows.

We conclude that \eqref{eq:asymppsi1} must hold in the smaller region $\mathcal{A}_{\gamma^{\alpha_K''}}$, with $\alpha_K''>\alpha_K$ suitably large, by multiplying both sides of \eqref{eq:asympphi1} with $r^{-2}$.

Finally, we combine  \eqref{eq:asympphi1} with \eqref{eq:asymplphi1} and use that
\begin{equation*}
L\psi_1=-\Delta\frac{r^3}{2(r^2+a^2)^{\frac{5}{2}}}r^{-1}\psi_1+\frac{1}{\sqrt{r^2+a^2}}L\phi_1
\end{equation*}
to obtain \eqref{eq:asympLpsi1}. The estimates \eqref{eq:asymppsi1gamma} and \eqref{eq:asymppsi1gamma} then follows immediately.
\end{proof}

\subsection{Asymptotics in $\mathcal{R}\setminus \mathcal{A}_{\gamma^{\alpha}}$ }
In this section, we will extend the late-time asymptotics in Proposition \ref{prop:asympl1nearI} to the rest of the spacetime. In particular, we will demonstrate the presence of oscillations in the late-time behaviour along the null generators of the event horizon $\mathcal{H}^+$.\begin{proposition}
\label{prop:asympboundrl1}
Let $K\in \N_0$ and let $\alpha>0$ be arbitrarily large. Then there exists a constant $C>0$ such that
\begin{align}
\label{eq:asympinteriorl1}
&\left|T^K\psi_1(\tau,\uprho,\theta,\varphi_*)-\frac{8}{3}I_1[\psi](\theta,\varphi_*) \uprho T^K((1+\tau)^{-4})\right|(\tau,\uprho,\theta,\varphi_*)\\ \nonumber
\leq&\: C(1+\tau)^{-4-K-\nu}\Bigg(S_{1,\delta,K}[\psi] +D_{1,\beta,K}[\psi]+\sum_{m=-1}^1|I_{1m}[\psi]|+\sqrt{E_{1,K+5,\delta}[\psi]}+\sum_{j=0}^2\sqrt{E_{1,K+3,\delta}[N^j\psi]}\\ \nonumber
&+\sqrt{E_{\geq 3,K+2,\delta}[\psi]}+\sum_{j=0}^1\sqrt{E_{\geq 3,K+2,\delta}[N^j\psi]}\Bigg),\\
\label{eq:asympinteriorXpsil1}
&\left|XT^K\psi_1(\tau,\uprho,\theta,\varphi_*)-\frac{8}{3}I_1[\psi](
\theta,\varphi_*) T^K((1+\tau)^{-4})\right|(\tau,\uprho,\theta,\varphi_*)\\ \nonumber
\leq&\:  C(1+\tau)^{-4-K-\nu}\Bigg(S_{1,\delta,K}[\psi] +D_{1,\beta,K}[\psi]+\sum_{m=-1}^1|I_{1m}|+\sqrt{E_{1,K+5,\delta}[\psi]}+\sum_{j=0}^2\sqrt{E_{1,K+3,\delta}[N^j\psi]}\\ \nonumber
&+\sqrt{E_{\geq 3,K+2,\delta}[\psi]}+\sum_{j=0}^1\sqrt{E_{\geq 3,K+2,\delta}[N^j\psi]}\Bigg),
\end{align}
for all $(\tau,\uprho,\theta,\varphi_*)\in \mathcal{R} \setminus \mathcal{A}_{\gamma^{\alpha}}$.

In particular, the following oscillatory asymptotics hold along the null generators of $\mathcal{H}^+$:
\begin{equation}
\label{eq:horizonsci}
\begin{split}
&\left|T^K\psi_1|_{\mathcal{H}^+}(v,\theta,\varphi_{\mathcal{H}^+})-\frac{8}{3}r_+\sum_{m=-1}^1I_{1m}[\psi]Y_{1,m}(\theta,\varphi_{\mathcal{H}^+})e^{im \upomega_+ v} T^K((1+\tau)^{-4})\right|\\ 
\leq&\:C(1+\tau)^{-4-K-\nu}\Bigg(S_{1,\delta,K}[\psi] +D_{1,\beta,K}[\psi]+\sum_{m=-1}^1|I_{1m}|+\sqrt{E_{1,K+5,\delta}[\psi]}+\sum_{j=0}^2\sqrt{E_{1,K+3,\delta}[N^j\psi]}\\ 
&+\sqrt{E_{\geq 3,K+2,\delta}[\psi]}+\sum_{j=0}^1\sqrt{E_{\geq 3,K+2,\delta}[N^j\psi]}\Bigg),
\end{split}
\end{equation}
where we use $Y_{\ell,m}(\theta,\varphi_{\mathcal{H}^+})$ denote spherical harmonics with respect to the polar angle $\theta$ and the azimuthal angle $\varphi_{\mathcal{H}^+}$.
\end{proposition}
\begin{proof}
In order to obtain \eqref{eq:asympinteriorXpsil1}, we apply the fundamental theorem of calculus, integrating $X^2T^K\psi_1$ between $\uprho=\uprho'$ and $\uprho=\uprho_{\gamma^{\alpha}(\tau)}$. We use \eqref{eq:asympLpsi1gamma} to estimate the boundary term on $\uprho=\uprho_{\gamma^{\alpha}(\tau)}$ and \eqref{eq:fasterdecayl1} to estimate the contribution of the integral.

The estimate \eqref{eq:asympinteriorl1} then follows by applying the fundamental theorem of calculus again, integrating now $XT^K\psi_1$ between $\uprho=\uprho'$ and $\uprho=\uprho_{\gamma^{\alpha}(\tau)}$ and using the estimates \eqref{eq:asympinteriorXpsil1} and \eqref{eq:asymppsi1gamma}.

Finally, \eqref{eq:horizonsci} follows from \eqref{eq:asympinteriorl1} by using the relation \eqref{eq:phihor} between $\varphi_{\mathcal{H}^+}$ and $\varphi_*$ to conclude that
\begin{align*}
e^{im \varphi_*}=&\:e^{im \varphi_{\mathcal{H}^+}}e^{im \upomega_+ v},\quad \textnormal{and hence}\\
Y_{\ell,m}(\theta,\varphi_*)=&\:Y_{\ell,m}(\theta,\varphi_{\mathcal{H}^+})e^{im \upomega_+ v}.
\end{align*} 
\end{proof}

\subsection{Asymptotics with vanishing Newman--Penrose charges}
We can now apply the time-integral construction from Section \ref{sec:timeinv} to obtain late-time asymptotics for $\psi_1$ arising from initial data that are smooth and compactly supported.
\begin{proposition}
\label{prop:asympl1NP0boundr}
Consider initial data $(\psi|_{\Sigma_0},T\psi|_{\Sigma_0})$ for \eqref{eq:waveeq}, with $(\phi|_{\Sigma_0},T\phi_{\Sigma_0}) \in (C_c^{\infty}(\Sigma_0))^2$.  Let $r_0>r_+$. Then there exists a $\nu>0$ and a constant $C=C(M,a,\Sigma_0,r_0,\nu)>0$, such that
\begin{equation*}
\begin{split}
\Big|r^{-1}T^K\psi_1(u,v,\theta,\varphi_*)&-\frac{8}{3}I_1[T^{-1}\psi](\theta,\varphi_*) T^{K+1}(u^{-2}v^{-2})\Big|\\
\leq&\:  Cv^{-2}u^{-3-K-\nu}\Bigg(S_{1,\delta,K+1}[T^{-1}\psi] +D_{1,\beta,K+1}[\psi]+\sum_{m=-1}^1|I_{1m}[T^{-1}\psi]|\\ \nonumber
&+\sqrt{E_{1,K+6,\delta}[T^{-1}\psi]}+\sum_{j=0}^2\sqrt{E_{1,K+4,\delta}[N^jT^{-1}\psi]}+\sqrt{E_{\geq 3,K+3,\delta}[T^{-1}\psi]}\\
&+\sum_{j=0}^1\sqrt{E_{\geq 3,K+3,\delta}[N^jT^{-1}\psi]}\Bigg)
\end{split}
\end{equation*}
in $\{r\geq r_0\}$, whereas in $\{r\leq r_0\}$, we can express:
\begin{equation*}
\begin{split}
&\left|T^K\psi_1(\tau,\uprho,\theta,\varphi_*)+\frac{32}{3}I_1[T^{-1}\psi](\theta,\varphi_*) \uprho T^K((1+\tau)^{-5})\right|(\tau,\uprho,\theta,\varphi_*)\\
\leq&\: C(1+\tau)^{-5-K-\nu}\Bigg(S_{1,\delta,K+1}[T^{-1}\psi] +D_{1,\beta,K+1}[T^{-1}\psi]+\sum_{m=-1}^1|I_{1m}[T^{-1}\psi]|\\ \nonumber
&+\sqrt{E_{1,K+6,\delta}[T^{-1}\psi]}+\sum_{j=0}^2\sqrt{E_{1,K+4,\delta}[N^jT^{-1}\psi]}+\sqrt{E_{\geq 3,K+3,\delta}[T^{-1}\psi]}\\
&+\sum_{j=0}^1\sqrt{E_{\geq 3,K+3,\delta}[N^jT^{-1}\psi]}\Bigg).
\end{split}
\end{equation*}
In particular, the following oscillatory asymptotics hold along the null generators of $\mathcal{H}^+$:
\begin{equation*}
\begin{split}
&\left|T^K\psi_1|_{\mathcal{H}^+}(v,\theta,\varphi_{\mathcal{H}^+})+\frac{32}{3}r_+\sum_{m=-1}^1I_{1m}[T^{-1}\psi]Y_{1,m}(\theta,\varphi_{\mathcal{H}^+})e^{im \upomega_+ v} T^K((1+\tau)^{-5})\right|\\ \nonumber
\leq&\: C(1+\tau)^{-5-K-\nu}\Bigg(S_{1,\delta,K+1}[T^{-1}\psi] +D_{1,\beta,K+1}[T^{-1}\psi]+\sum_{m=-1}^1|I_{1m}[T^{-1}\psi]|\\ \nonumber
&+\sqrt{E_{1,K+6,\delta}[T^{-1}\psi]}+\sum_{j=0}^2\sqrt{E_{1,K+4,\delta}[N^jT^{-1}\psi]}+\sqrt{E_{\geq 3,K+3,\delta}[T^{-1}\psi]}\\
&+\sum_{j=0}^1\sqrt{E_{\geq 3,K+3,\delta}[N^jT^{-1}\psi]}\Bigg).
\end{split}
\end{equation*}
\end{proposition}
\begin{proof}
We apply Propositions \ref{prop:mainpropTinv} and \ref{prop:TinvNPconst} to conclude that $T^{-1}\phi$ has sufficiently high regularity to conclude that all the relevant energies for $T^{-1}\psi$ are finite and $I_1[T^{-1}\psi]$ is well-defined. Then the estimates follows immediately by applying the estimates in Propositions \ref{prop:asympl1nearI} and \ref{prop:asympboundrl1} to $T^{-1}\psi$ instead of $\psi$.
\end{proof}

\section{Late-time polynomial tails: the $\ell=2$ projection}
\label{sec:latetimeasympl2}
In this section we derive the late-time asymptotics for the projection $\psi_2$. In contrast with the $\ell=0,1$ cases, \textbf{the leading-order late-time behaviour of $\psi_2$ is coupled with the leading-order late-time behaviour of $\psi_0$} and will therefore not just depend on the Newman--Penrose charge $I_2[\psi]$, but also on $I_0[\psi]$. As will be shown below, this coupling mechanism will limit the decay rate of $\psi_2$ in regions of bounded $r$.

It will be useful for the estimates in this section to define the following quantities: let $\beta>0$ and $0<\delta<1$, then
\begin{align*}
D_{2,\beta,K}[\psi]:=&\:\max_{0\leq k\leq K} ||v^{6+\beta+k}L^k\left(r^{-4}P_2-32I_2[\psi]v^{-6}\right)||_{L^{\infty}(\Sigma_0\cap\{r\geq R_0\})},\\
S_{2,\delta,K}[\psi]:=&\:\sup_{\tau\geq 0}\: (1+\tau)^{1-\delta}\sum_{j=0}^2\sum_{\substack{0\leq k_1+k_2\leq K+1\\ k_2\leq K}}(1+\tau)^{k_2}||(rL)^{k_1}T^{k_2}\check{\phi}_2^{(j)}||_{L^{\infty}(\Sigma_{\tau})}\\
&+\sup_{\tau\geq 0} (1+\tau)^{1-\delta}\sum_{\
\ell\in\{2,4\}}\sum_{j=0}^2\sum_{\substack{0\leq k_1+k_2\leq K+1\\ k_2\leq K}}(1+\tau)^{k_2} ||(rL)^{k_1}T^{k_2+1}\check{\phi}_{\ell}^{(j)}||_{L^{\infty}(\Sigma_{\tau})}\\
&+\sup_{\tau\geq 0} (1+\tau)^{1-\delta}\sum_{j=0}^1\sum_{\substack{0\leq k_1+k_2\leq K+1\\ k_2\leq K}}(1+\tau)^{k_2} ||(rL)^{k_1}T^{k_2+1}\check{\phi}_0^{(j)}||_{L^{\infty}(\Sigma_{\tau})}\\
&+S_{0,\delta,3+K}[\psi]+|I_0[\psi]|+D_{0,1-\delta,3+K}[\psi].
\end{align*}

\begin{lemma}
There exists a constant $C=C(M,a,\delta,K)>0$ such that:
\begin{equation*}
S_{2,\delta,K}[\psi]\leq C \sum_{k+j\leq K+2}\sqrt{E_{2,k,\frac{\delta}{2}}[N^j\psi]}+\sqrt{E_{\geq 3,k,\frac{\delta}{2}}[N^j\psi]}+\sqrt{E_{0,k+3,\frac{\delta}{2}}[N^j\psi]}+C|I_0[\psi]|+CD_{0,1-\delta,3+K}[\psi].
\end{equation*}
\end{lemma}
\begin{proof}
The estimate follows immediately by applying the estimates in Proposition \ref{prop:pointwisedecayl2} and \ref{prop:pointwisedecayl2}, together with Lemma \ref{lm:Sl0est} and also \eqref{eq:auxestLTl0} below to estimate the terms involving $(rL)^{k_1} T^{k_2+1} \check{\phi}_0^{(1)}$.
\end{proof}

\subsection{Asymptotics in $\mathcal{A}_{\gamma^{\alpha}}$ }
\label{sec:asympl21}

We recall here \eqref{eq:maineqP2}, but with the dependence on $\pi_2(\sin^2\theta \check{\phi}^{(2)}_0)$ stated more precisely:
\begin{equation*}
\begin{split}
4\underline{L}P_2=&\:[-8r^{-1}+O_{\infty}(r^{-2})]P_2-[2r^{-1}+O_{\infty}(r^{-2})]\frac{a^2\Delta}{(r^2+a^2)}\pi_2(\sin^2\theta T\check{\phi}^{(2)}_0)+L(\frac{a^2\Delta}{(r^2+a^2)}\pi_2(\sin^2\theta T\check{\phi}^{(2)}_0))\\
&+O_{\infty}(r^{-3})[\chphi^{(2)}_2+\chphi^{(1)}_2+\phi_2]+O_{\infty}(r^{-3})[T\pi_2(\sin^2\theta \chphi_4^{(2)})+T\pi_2(\sin^2\theta \chphi^{(1)}) +T\pi_2(\sin^2\theta \phi )]\\ 
&+O_{\infty}(r^{-2})[LT\pi_2(\sin^2\theta \chphi_4^{(2)})+LT\pi_2(\sin^2\theta \chphi^{(1)}) +LT\pi_2(\sin^2\theta \phi )]
\end{split}
\end{equation*}
Due to the slower decay properties of $\check{\phi}^{(2)}_0$, it will be convenient to put some of the terms involving $\check{\phi}^{(2)}_0$ on the left-hand side of the above equation by introducing the modified quantity:
\begin{equation*}
\widetilde{P}_2=P_2-\frac{1}{4}r^{4}L\left(r^{-4}\frac{a^2\Delta}{(r^2+a^2)}\pi_2(\sin^2\theta T\check{\phi}^{(2)}_0)\right).
\end{equation*}
Then:
\begin{equation}
\label{eq:tildeP2}
\begin{split}
4\underline{L}(r^{-4}\widetilde{P}_2)=&\: L^2\left(r^{-4}\frac{a^2\Delta}{(r^2+a^2)}\pi_2(\sin^2\theta \check{\phi}^{(2)}_0)\right)+O_{\infty}(r^{-6})P_2\\
&+O_{\infty}(r^{-7})[\chphi^{(2)}_2+\chphi^{(1)}_2+\phi_2]+O_{\infty}(r^{-7})[T\pi_2(\sin^2\theta \chphi_4^{(2)})+T\pi_2(\sin^2\theta \chphi^{(1)}) +T\pi_2(\sin^2\theta \phi )]\\ 
&+O_{\infty}(r^{-7})[rLT\pi_2(\sin^2\theta \chphi_4^{(2)})+rLT\pi_2(\sin^2\theta \chphi^{(1)}) +rLT\pi_2(\sin^2\theta \phi )].
\end{split}
\end{equation}

In order to determine the late-time asymptotics of $\psi_2$, it will be necessary to appeal to the late-time asymptotics of $\psi_0$ obtained in Section \ref{sec:asympl0nonzeronp}. We derive in the lemma below the necessary estimates for $\psi_0$.
\begin{lemma}
\label{lm:estcouplingl0}
Let $K\in \N_0$ and $\alpha_K>\frac{6+K}{7+K}$. Then there exists $\delta(\alpha_K)>0$ suitably small, such that for all $j,m\in \N_0$ with $j+m=K$, we have that in $\mathcal{A}_{\gamma^{\alpha_K}}$:
\begin{align}
\label{eq:estcouplingl0a}
\Bigg|\frac{1}{4}L^{j+1}&\left[ r^{-4}\frac{a^2\Delta}{(r^2+a^2)^2} T^m\check{\phi}^{(2)}_0\right]-2a^2I_0[\psi]L^KT^m\Bigg(24(v-u)^{-2}v^{-4}-24(v-u)^{-4}v^{-2}\Bigg|\\ \nonumber
\leq&\:  C(r^{-6-j}u^{-2-m+\delta}+r^{-6-K}u^{-1+\delta}+v^{-(7-K+2\delta)\alpha_{K-1}}) S_{2,\delta,K-1}[\psi],\\ \nonumber
\label{eq:estcouplingl0b}
\Bigg|\frac{1}{4}L^{j+2}&\left[ r^{-4}\frac{a^2\Delta}{(r^2+a^2)^2}T^m \check{\phi}^{(2)}_0\right]-2a^2I_0[\psi]\underline{L}L^KT^m\Bigg(-96(v-u)^{-1}v^{-5}-24(v-u)^{-2}v^{-4}\\
&+16(v-u)^{-3}v^{-3}+24(v-u)^{-4}v^{-2}\Bigg)\Bigg|\\ \nonumber
\leq&\:C(r^{-7-j}u^{-2-m+\delta}+r^{-7-K}u^{-1+\delta}+v^{-(8-K+2\delta)\alpha_{K}}) S_{2,\delta,K}[\psi],\\
\label{eq:estcouplingl0c}
\Bigg|\frac{1}{4}\frac{a^2\Delta}{(r^2+a^2)^2}&  T^{K+1}\check{\phi}^{(2)}_0-2 a^2I_0[\psi]L^K\left(\frac{3}{2}(v-u)^2v^{-4}-(v-u) v^{-3}\right)\Bigg|\\ \nonumber
\leq&\: C(r^{-3-K}u^{-1+\delta}+r^{-2}u^{-2-K+\delta}+v^{-(4-K+2\delta)\alpha_{K-1}}) S_{2,\delta,K-1}[\psi].
\end{align}
Furthermore, for all $k\leq K$, we can estimate
\begin{equation}
\label{eq:estcouplingl02}
\begin{split}
r^{-6}|(rL)^k(r^{-2} T\check{\phi}^{(2)}_0)|\leq CS_{2,\delta,K}[\psi][(r^{-5}v^{-3}+r^{-7}u^{-2+\delta}+r^{-8}u^{-1+\delta})
\end{split}
\end{equation}
\end{lemma}
\begin{proof}
We can write
\begin{equation*}
\begin{split}
\frac{1}{4} L\left(r^{-4}\frac{a^2\Delta}{(r^2+a^2)^2}T^m\check{\phi}^{(2)}_0\right)=&\:\frac{1}{2} L\left(a^2\left[r^{-4}+O_{\infty}(r^{-5})\right]LT^m\check{\phi}^{(1)}_0\right)\\
=&\: L\left(a^2\left[r^{-2}+O_{\infty}(r^{-3})\right] L^2T^m{\phi}_0\right)+L\left(a^2\left[r^{-3}+O_{\infty}(r^{-4})\right] LT^m{\phi}_0\right)\\
=&\: a^2(r^{-2}+O_{\infty}(r^{-3})) L^3T^m{\phi}_0-\frac{3}{2}a^2(r^{-4}+O_{\infty}(r^{-5}))\sin^2\theta LT^m{\phi}_0.
\end{split}
\end{equation*}
And hence,
\begin{equation*}
\begin{split}
\frac{1}{4}&L^{K}\left[L\left(r^{-4}\frac{a^2\Delta}{(r^2+a^2)^2}\check{\phi}^{(2)}_0\right)\right]=a^2\sum_{k=0}^K\frac{K!}{k!(K-k)!} L^k(r^{-2}+O_{\infty}(r^{-3}))L^{3+K-k}T^m{\phi}_0\\
&-\frac{3}{2}a^2\sum_{k=0}^K \frac{K!}{k!(K-k)!} L^k(r^{-4}+O_{\infty}(r^{-5})) L^{1+K-k}T^m{\phi}_0
\end{split}
\end{equation*}
By \eqref{eq:asymplkP0} together with the arguments in the proof of Corollary \ref{cor:asympLTkphi0}, we moreover have that for all $j\geq 1$ and $\beta\geq 0$:
\begin{equation}
\label{eq:auxestLTl0}
\begin{split}
\left|L^{j}T^m\phi_0(u,v)+2I_0[\psi]L^{j+m}(v^{-1})\right|\leq&\: CS_{0,\delta,j+m-1}[\psi]\left[ r^{-1-j}u^{-2-m+\delta}+u^{-1+\delta} r^{-1-j-m}+v^{-(2+j+m-2\delta)\alpha_{j+m-1}}\right]\\
&+CD_{0,\beta,K}[\psi]v^{-1-j-m-\beta}
\end{split}
\end{equation}
in $\mathcal{A}_{\gamma^{\alpha_{j+m-1}}}$. By combining the above, we obtain
\begin{equation*}
\begin{split}
\Bigg|\frac{1}{4}&L^{j+1}\left[ r^{-4}\frac{a^2\Delta}{(r^2+a^2)^2}T^m\check{\phi}^{(2)}_0\right]-2a^2I_0[\psi]T^mL^j\Bigg(24(v-u)^{-2}v^{-4}-24(v-u)^{-4}v^{-2}\Bigg|\\
\leq&\: C(r^{-6-j}u^{-2-m+\delta}+r^{-6-K}u^{-1+\delta}+v^{-(7-K+2\delta)\alpha_{K-1}}) S_{2,\delta,K-1}[\psi].
\end{split}
\end{equation*}
and hence
\begin{equation*}
\begin{split}
\Bigg|\frac{1}{4}&L^{j+2}\left[ r^{-4}\frac{a^2\Delta}{(r^2+a^2)^2}T^m\check{\phi}^{(2)}_0\right]-2a^2I_0[\psi]L^jT^m\Bigg(24(v-u)^{-2}v^{-4}-24(v-u)^{-4}v^{-2}\Bigg|\\
\leq&\: C(r^{-7-j}u^{-2-m+\delta}+r^{-7-K}u^{-1+\delta}+v^{-(8-K+2\delta)\alpha_{K}}) S_{2,\delta,K}[\psi].
\end{split}
\end{equation*}

We obtain \eqref{eq:estcouplingl0c} and \eqref{eq:estcouplingl02} in an analogous manner, appealing again to the arguments in the proof of Corollary \ref{cor:asympLTkphi0}.
\end{proof}

Let $f$ be a suitably regular function on $\s^2$. Then recall that by \eqref{eq:l2projl0}, we can write:
\begin{equation*}
\pi_2(\sin^2\theta \pi_0(f))=-\frac{4}{3}\sqrt{\frac{\pi}{5}}\pi_0(f)Y_{2,0}(\theta).
\end{equation*}

\begin{proposition}
Let $\alpha_K>\frac{6+K}{7+K}$. Then there exists $\delta(\alpha_K)>0$ suitably small, such that in $\mathcal{A}_{\gamma^{\alpha_K}}$:
\begin{equation}
\label{eq:asymplkP2}
\begin{split}
&\Bigg|L^K(r^{-4}\widetilde{P}_2)-32I_2[\psi]L^K(v^{-6})+\frac{8}{3}\sqrt{\frac{\pi}{5}}a^2I_0[\psi]Y_{2,0}L^K(v^{-6})\\\
&+\frac{8}{3}\sqrt{\frac{\pi}{5}}a^2I_0[\psi]Y_{2,0}L^K\Bigg(-96(v-u)^{-1}v^{-5}-24(v-u)^{-2}v^{-4}+16(v-u)^{-3}v^{-3}+24(v-u)^{-4}v^{-2}\Bigg)\Bigg|\\
\leq&\: C v^{-(8-K+2\delta)\alpha_{K}+1}S_{2,\delta,K}[\psi]+CD_{2,\beta,K}[\psi]v^{-6-K-\beta}.
\end{split}
\end{equation}
\end{proposition}
\begin{proof}
Let $K\in \N_0$. By multiplying with $r^{-4}$ and then acting with  $L^K$ on both sides of \eqref{eq:maineqP2}, we obtain:
\begin{equation}
\label{eq:LbareqP2asymp}
\begin{split}
4\underline{L}L^K(r^{-4}\widetilde{P}_2)=&\:L^{2+K}\left(r^{-4}\frac{a^2\Delta}{(r^2+a^2)}\pi_2(\sin^2\theta \check{\phi}^{(2)}_0)\right)\\
&+\sum_{k=0}^KO_{\infty}(r^{-6-K})(rL)^kP_2+\sum_{j=0}^2\sum_{k=0}^KO_{\infty}(r^{-7-K})(rL)^k\check{\phi}^{(j)}_2+\sum_{j=0}^2\sum_{k=0}^{K+1}O_{\infty}(r^{-7-K})(rL)^kT\pi_2(\sin^2\theta \phi^{(j)}_4)\\
&+\sum_{j=0}^1\sum_{k=0}^{K+1}O_{\infty}(r^{-7-K})(rL)^kT\pi_2(\sin^2\theta \check{\phi}^{(j)}_0).
\end{split}
\end{equation}

We can further estimate for all $k\in \N_0$:
\begin{equation}
\label{eq:estP2Lphi2}
\begin{split}
|(rL)^k P_2-r^{-1}(rL)^{k+1}\check{\phi}^{(2)}_2| (u,v,\theta,\varphi_*)\leq&\: Ca^2|(rL)^k(r^{-2} T\check{\phi}^{(2)}_0)|+Car^{-2}(u,v)\sum_{j=0}^1(a||(rL)^kT\check{\phi}_{0}^{(j)}||_{L^{\infty}(\Sigma_{\tau)}}\\
&+||(rL)^k\check{\phi}^{(j)}_2||_{L^{\infty}(\Sigma_{\tau)}}+a||(rL)^kT\check{\phi}^{(j)}_4||_{L^{\infty}(\Sigma_{\tau)}}+a||(rL)^kT\check{\phi}^{(2)}_4||_{L^{\infty}(\Sigma_{\tau)}}).
\end{split}
\end{equation}
and hence obtain after applying \eqref{eq:estcouplingl0b} and \eqref{eq:estcouplingl02}:
\begin{equation}
\label{eq:Lbarl2}
\begin{split}
&\Bigg|\underline{L}\Bigg[L^K(r^{-4}\widetilde{P}_2)+\frac{8}{3}\sqrt{\frac{\pi}{5}}a^2I_0[\psi]Y_{2,0}L^K\Bigg(-96(v-u)^{-1}v^{-5}-24(v-u)^{-2}v^{-4}+16(v-u)^{-3}v^{-3}+24(v-u)^{-4}v^{-2}\Bigg)\Bigg]\Bigg|\\
\leq&\: C(r^{-7-K}u^{-1+\delta}+v^{-(8-K+2\delta)\alpha_{K}}) S_{2,\delta,K}[\psi].
\end{split}
\end{equation}
Integrating from $\Sigma_0$ in the $\underline{L}$ direction, applying \eqref{eq:ruv} together with \eqref{eq:Lbarl0}, we obtain
\begin{equation*}
\begin{split}
\int_{u_{\Sigma_0}(v)}^u u'^{-1+\delta}r^{-7-K}(u',v)\,du'\leq &\:C\int_{u_{\Sigma_0}(v)}^u u'^{-1+\delta} (v-u')^{-7-K}\,du'\\
\leq &\:C v^{-(7+K-2\delta)\alpha_K}\int_{u_{\Sigma_0}(v)}^u u'^{-1-\delta}\,du'\\
\leq &\:Cv^{-(7+K-2\delta)\alpha_K}.
\end{split}
\end{equation*}
Hence, \eqref{eq:asymplkP2} follows after integrating the above inequalities.
\end{proof}

\begin{corollary}
Let $K\in \N_0$ and $\alpha_K>\frac{6+K}{7+K}$. Then there exists $\delta(\alpha_K)>0$ suitably small, such that in $\mathcal{A}_{\gamma^{\alpha_K}}$:
\begin{equation}
\label{eq:asympLphi22}
\begin{split}
\Bigg|LT^K\check{\phi}^{(2)}_2&-2I_2[\psi](v-u)^4L^K(v^{-6})+\frac{40}{3}\sqrt{\frac{\pi}{5}}a^2I_0[\psi]Y_{2,0}(v-u)^4L^K(v^{-6})\\
&+\frac{8}{3}\sqrt{\frac{\pi}{5}} a^2Y_{2,0} I_0[\psi]T^K\left(-6(v-u)^3v^{-5}+\frac{3}{2}(v-u)^2v^{-4}\right)\Bigg|\\
\leq &\:C (u ^{-1+\delta} r^{-2-K}+u ^{-1-K+\delta} r^{-2}+r^4 v^{-(8K+2\delta)\alpha_{K}+1})S_{2,\delta,K}[\psi]+CD_{2,\beta,K}[\psi]r^4v^{-6-K-\beta}.
\end{split}
\end{equation}
\end{corollary}
\begin{proof}
In light of the estimate \eqref{eq:asymplkP2}, it is convenient to introduce a second modification of $P_2$:
\begin{equation*}
\widetilde{\widetilde{P}_2}:=\widetilde{P}_2+\frac{8}{3}\sqrt{\frac{\pi}{5}}a^2I_0[\psi]Y_{2,0}\Bigg(-96(v-u)^{-1}v^{-5}-24(v-u)^{-2}v^{-4}+16(v-u)^{-3}v^{-3}+24(v-u)^{-4}v^{-2}\Bigg)
\end{equation*}
By \eqref{eq:Lbarl2}, we have that
\begin{equation*}
|\underline{L}(r^{-4}L^K\widetilde{\widetilde{P}_2})|\leq C(r^{-7-K}u^{-1+\delta}+v^{-(8-K+2\delta)\alpha_{K}}) S_{2,\delta,K}[\psi]
\end{equation*}
and by induction, it follows that we can express
\begin{equation*}
r^{-4}T^K\widetilde{\widetilde{P}_2}=L^K(r^{-4}\widetilde{\widetilde{P}_2})+\sum_{k=0}^{K-1} \underline{L} L^{k} (r^{-4}T^{K-1-k}\widetilde{\widetilde{P}_2}).
\end{equation*}
Furthermore, by repeating the arguments leading to \eqref{eq:Lbarl2} for $P_2$ replaced by $T^{K-1-k}P_2$ and applying \eqref{eq:estcouplingl0b} and \eqref{eq:estcouplingl02}, we obtain
\begin{equation*}
\begin{split}
\sum_{k=0}^{K-1}|\underline{L} L^{k} (r^{-4}T^{K-1-k}\widetilde{P_2})|(u,v,\theta,\varphi_*)\leq&\: C (u ^{-1+\delta} r^{-6-K}+u ^{-1-K+\delta} r^{-6}+v^{-(7-K+2\delta)\alpha_{K}})S_{2,\delta,K}[\psi].
\end{split}
\end{equation*}

Furthermore,
\begin{equation*}
\begin{split}
LT^K\check{\phi}^{(2)}_2=&\:T^K\widetilde{\widetilde{P}_2}+\frac{8}{3}\sqrt{\frac{\pi}{5}}a^2I_0[\psi]Y_{2,0}T^K\Bigg(-6(v-u)^{3}v^{-5}-\frac{3}{2}(v-u)^{2}v^{-4}+(v-u)^{1}v^{-3}+\frac{3}{2}v^{-2}\Bigg)\\
&+\frac{1}{4}r^{4}T^KL\left(r^{-4}\frac{a^2\Delta}{(r^2+a^2)}\pi_2(\sin^2\theta \check{\phi}^{(2)}_0)\right)+\frac{1}{4}\frac{a^2\Delta}{(r^2+a^2)^2} \pi_2(\sin^2\theta T^{K+1}\check{\phi}^{(2)}_0),
\end{split}
\end{equation*}
so after applying Lemma \ref{lm:estcouplingl0} we conclude that
\begin{equation*}
\begin{split}
\Bigg|LT^K\check{\phi}^{(2)}_2&-2I_2[\psi](v-u)^4L^K(v^{-6})+\frac{40}{3}\sqrt{\frac{\pi}{5}}a^2I_0[\psi]Y_{2,0}(v-u)^4L^K(v^{-6})\\
&+\frac{8}{3}\sqrt{\frac{\pi}{5}} a^2Y_{2,0} I_0[\psi]T^K\left(-6(v-u)^3v^{-5}+\frac{3}{2}(v-u)^2v^{-4}\right)\Bigg|\\
\leq &\:C (u ^{-1+\delta} r^{-2-K}+u ^{-1-K+\delta} r^{-2}+r^4 v^{-(8K+2\delta)\alpha_{K}+1})S_{2,\delta,K}[\psi]+CD_{2,\beta,K}[\psi]r^4v^{-6-K-\beta}.
\end{split}
\end{equation*}
\end{proof}

\begin{proposition}
\label{prop:mainpropasympl2}
For $\alpha_K'''$ suitably close to 1, $\delta>0$ suitably small and $0<\beta\leq 1$ suitably large, there exists $\nu>0$ and a constant $C=C(M,a,\alpha_K''',\nu,\delta)>0$ such that
\begin{equation}
\label{eq:asympsil2}
\begin{split}
&\Bigg|T^K\psi_2(u,v,\theta,\varphi_*)-\frac{4}{15}I_2[\psi](\theta,\varphi_*)(v-u)^2 T^K(u^{-3}v^{-3})\\
&+\frac{64}{3}\sqrt{\frac{\pi}{5}}a^2 I_0[\psi]Y_{2,0}(\theta)T^K\Bigg[\frac{1}{12}u^{-3}v^{-3}(v-u)^2-\frac{1}{4} u^{-3}v^{-2}(v-u)+\frac{1}{4}u^{-3}v^{-1}\Bigg]\Bigg|\\
&\leq  C u ^{-3-K-\nu} v^{-1}(S_{2,\delta,K}[\psi]+D_{2,\beta,K}[\psi]+\sum_{m=-2}^2|I_{2m}[\psi]|)+ C\sqrt{E_{2,K,\delta}[\psi]}u^{-4-K-\nu}
\end{split}
\end{equation}
in $\mathcal{A}_{\gamma^{\alpha_K'''}}$ and along $\mathcal{I}^+$:
\begin{align}
\label{eq:asympl2inf1}
\Bigg|T^K\phi_2|_{\mathcal{I}^+}&-\frac{2}{15}I_2[\psi]T^K(u^{-3})+\frac{8}{9}\sqrt{\frac{\pi}{5}}a^2 I_0[\psi]Y_{2,0}(\theta)T^K(u^{-3})\Bigg|\\
\leq&\:  C u ^{-3-K-\nu} \Big(S_{2,\delta,K}[\psi]+D_{2,\beta,K}[\psi]+\sum_{m=-2}^2|I_{2m}[\psi]|+\sqrt{E_{2,K,\delta}[\psi]}\Big),\\
\label{eq:asympl2inf2}
\Bigg|T^K\check{\phi}^{(1)}_2|_{\mathcal{I}^+}&-\frac{1}{5}I_2[\psi]T^K(u^{-2})\Bigg|\\
\leq&\:   C u ^{-2-K-\nu} \Bigg(S_{2,\delta,K}[\psi]+D_{2,\beta,K}[\psi]+\sum_{m=-2}^2|I_{2m}[\psi]|+\sqrt{E_{2,K,\delta}[\psi]}\Bigg),\\
\label{eq:asympl2inf3}
\Bigg|T^K\check{\phi}^{(2)}_2|_{\mathcal{I}^+}&-\frac{2}{5}I_2[\psi]T^K(u^{-1})\Bigg|\\
\leq &\:  C u ^{-1-K-\nu} \Bigg(S_{2,\delta,K}[\psi]+D_{2,\beta,K}[\psi]+\sum_{m=-2}^2|I_{2m}[\psi]|+\sqrt{E_{2,K,\delta}[\psi]}\Bigg).
\end{align}
\end{proposition}
In particular,
\begin{equation}
\label{eq:asympl2gamma}
\begin{split}
\Bigg|T^K\psi_2|_{\gamma^{\alpha_K''}}&+\frac{16}{3}\sqrt{\frac{\pi}{5}}a^2 I_0[\psi]Y_{2,0}(\theta)T^K((1+\tau)^{-4})\Bigg|\\
\leq &C (1+\tau) ^{-4-K-\nu} \left(S_{2,\delta,K}[\psi]+D_{2,\beta,K}[\psi]+\sum_{m=-2}^2|I_{2m}[\psi]|+\sqrt{E_{2,K,\delta}[\psi]}\right).
\end{split}
\end{equation}
\begin{proof}
Note first of all that \eqref{eq:NPconstdiffangle} holds also with $I_1[\psi]$ replaced by $I_2[\psi]$.

We then integrate \eqref{eq:asympLphi22} in the $L$ direction, starting from $\gamma^{\alpha_K}$ and use \eqref{eq:pointwl2c} to obtain:
\begin{equation*}
\begin{split}
|T^K\check{\phi}^{(2)}_2|(u,v_{\gamma^{\alpha_K}(u)},\theta,\varphi_*)\leq&\: C r^{\frac{1}{2}(1+\delta')}(u,v_{\gamma^{\alpha_K}(u)})u^{-\frac{3}{2}-K+2\delta'}\sqrt{E_{2,K,\delta'}[\psi]}\\
\leq&\: C\sqrt{E_{2,K,\delta'}[\psi]}u^{-\frac{3}{2}+\frac{1}{2}\alpha_K-K+2\delta'},
\end{split}
\end{equation*}
for $\delta'>0$ arbitrarily small, together with \eqref{eq:usefulvint}, \eqref{eq:varphidiff} and $0<\alpha_K<1$ suitably large, to obtain:
\begin{equation*}
\begin{split}
&\Bigg|T^K\check{\phi}^{(2)}_2(u,v,\theta,\varphi_*)-\frac{2}{5}I_2[\psi](\theta,\varphi_*)(v-u)^5 T^K(u^{-1}v^{-5})\\
&+\frac{8}{3}\sqrt{\frac{\pi}{5}}a^2 I_0[\psi]Y_{2,0}(\theta)T^K\Bigg[u^{-1}v^{-5}(v-u)^5-\frac{3}{2}u^{-1}v^{-4}(v-u)^4 +\frac{1}{2}u^{-1}v^{-3}(v-u)^3\Bigg]\Bigg|\\
\leq &\: C u ^{-1-K-\nu} (v-u)^3v^{-3}(S_{2,\delta,K}[\psi]+D_{2,\beta,K}[\psi]+\sum_{m=-2}^2|I_{2m}[\psi]|)+ C\sqrt{E_{2,K,\delta}[\psi]}u^{-1-K-\nu},
\end{split}
\end{equation*}
for some $\nu>0$, with $\delta>0$ suitably small. In particular, \eqref{eq:asympl2inf3} follows.

Since $T^K\check{\phi}_2^{(2)}=2(r^2+O(r))LT^K\check{\phi}_2^{(1)}+K O(r^{-1})T^{K-1}\check{\phi}_2^{(1)}$, we can integrate the above equation once more, starting from $\gamma^{\alpha_K'}$, with $\alpha_K'>\alpha_K$ suitably large (depending on $\nu$ above), using \eqref{eq:pointwl2c} again to estimate:
\begin{equation*}
\begin{split}
|T^K\check{\phi}^{(1)}_2|(u,v_{\gamma^{\alpha_K}(u)},\theta,\varphi)\leq&\: C\sqrt{E_{2,K,\delta'}[\psi]}u^{-\frac{5}{2}+\frac{1}{2}\alpha_K-K+2\delta'},
\end{split}
\end{equation*}
for $\delta'>0$ and obtain:
\begin{equation*}
\begin{split}
&\Bigg|T^K\check{\phi}^{(1)}_2-\frac{1}{5}I_2[\psi](v-u)^4 T^K(u^{-2}v^{-4})\\
&+\frac{16}{3}\sqrt{\frac{\pi}{5}}a^2 I_0[\psi]Y_{2,0}T^K\Bigg[\frac{1}{4}u^{-2}v^{-4}(v-u)^4-\frac{1}{2}u^{-2}v^{-3}(v-u)^3+\frac{1}{4}u^{-2}v^{-2}(v-u)^{2}\Bigg]\Bigg|\\
&\:\leq C u ^{-2-K-\nu} (v-u)^2v^{-2}(S_{2,\delta,K}[\psi]+D_{2,\beta,K}[\psi]+\sum_{m=-2}^2|I_{2m}[\psi]|)+ C\sqrt{E_{2,K,\delta}[\psi]}u^{-2-K-\nu}
\end{split}
\end{equation*}
in $\mathcal{A}_{\gamma^{\alpha_K'}}$ with $\delta>0$ suitably small. In particular, \eqref{eq:asympl2inf2} follows.

Using that $T^K\check{\phi}_2^{(1)}=2(r^2+O(r))L{\phi}_2+KO(r^{-1})T^{K-1}\phi_2$ and for $\alpha_K'<\alpha_K''<1$ \eqref{eq:pointwl2c} implies that
\begin{equation*}
\begin{split}
|T^K\phi_2|(u,v_{\gamma^{\alpha_K''}(u)},\theta,\varphi)\leq&\: C\sqrt{E_{2,K,\delta'}[\psi]}u^{-\frac{7}{2}+\frac{1}{2}\alpha_K''-K+2\delta'},
\end{split}
\end{equation*}
we integrate a final time starting from $\gamma^{\alpha_K''}$, with $\alpha_K''>\alpha_K''$ suitably large and $\delta'>0$ correspondingly small, to obtain:
\begin{equation*}
\begin{split}
&\Bigg|T^K\phi_2-\frac{2}{15}I_2[\psi](v-u)^3 T^K(u^{-3}v^{-3})+\frac{32}{3}\sqrt{\frac{\pi}{5}}a^2 I_0[\psi]Y_{2,0}T^K\Bigg[\frac{1}{12}u^{-3}v^{-3}(v-u)^3-\frac{1}{4} u^{-3}v^{-2}(v-u)^2\\
&+\frac{1}{4}u^{-3}v^{-1}(v-u)\Bigg]\Bigg|\\
&\leq  C u ^{-3-K-\nu} (v-u)v^{-1}(S_{2,\delta,K}[\psi]+D_{2,\beta,K}[\psi]+\sum_{m=-2}^2|I_{2m}[\psi]|)+ C\sqrt{E_{2,K,\delta}[\psi]}u^{-3-K-\nu}
\end{split}
\end{equation*}
in $\mathcal{A}_{\gamma^{\alpha_K''}}$. We obtain in particular \eqref{eq:asympl2inf1}.

The estimate \eqref{eq:asympsil2} then follows by restricting to a smaller region $\mathcal{A}_{\gamma^{\alpha_K'''}}$, with $\alpha_K'''>\alpha_K''$ suitably large depending on $\nu$ and dividing the equation above by $r$. The estimate
\eqref{eq:asympl2gamma} follows directly.
\end{proof}

\subsection{Asymptotics in $\mathcal{R}\setminus \mathcal{A}_{\gamma^{\alpha}}$ }
We extend now the late-time asymptotics derived in the region $\mathcal{A}_{\gamma^{\alpha}}$ in Proposition \ref{prop:mainpropasympl2} to the rest of the spacetime.

\begin{proposition}
\label{prop:asympboundrl2}
Let $K\in \N_0$ and let $\alpha>0$ be arbitrarily large. Then there exists a constant $C=C(M,a,K,\delta)>0$ such that
\begin{equation}
\label{eq:asympinteriorl2}
\begin{split}
\Bigg|T^K\psi_2&-\frac{4}{3}\sqrt{\frac{\pi}{5}}a^2 I_0[\psi]Y_{2,0}(\theta)T^K((1+\tau)^{-4})\Bigg|(\tau,\uprho,\theta,\varphi_*)\\ 
\leq &\: C (1+\tau) ^{-4-K-\nu} \left(S_{2,\delta,K}[\psi]+D_{2,\beta,K}[\psi]+\sum_{m=-2}^2|I_{2m}[\psi]|\right)\\
&+C (1+\tau) ^{-4-K-\nu}\Bigg[\sqrt{E_{2,K+7,\delta}[\psi]}+\sum_{j=0}^{3} \sqrt{E_{2,K+4,\delta}[N^j\psi]}+\sqrt{E_{0,K+7,\delta}[\psi]}\\ 
&+\sum_{j=0}^{3} \sqrt{E_{0,K+4,\delta}[N^j\psi]}+\sqrt{E_{\geq 3,K+2,\delta}[\psi]}+\sum_{j=0}^1\sqrt{E_{\geq 3,K+2,\delta}[N^j\psi]} \Bigg].
\end{split}
\end{equation}
for all $(\tau,\uprho,\theta,\varphi_*)\in \mathcal{R} \setminus \mathcal{A}_{\gamma^{\alpha}}$ and $\delta>0$ suitably small.
\end{proposition}
\begin{proof}
First of all, by \eqref{eq:pointwl23}, we have that:
\begin{equation*}
\begin{split}
| XT^K\psi_2|_{\gamma^{\alpha}(\tau)}|(\tau,\theta,\varphi_*)\leq&\: C (1+\tau)^{-4-\alpha-K+2\delta}\Bigg(\sqrt{E_{2,K+5,\delta}[\psi]+\sum_{j=0}^{2} E_{2,K+3,\delta}[N^j\psi]}\\ \nonumber
&+\sqrt{E_{0,K+5,\delta}[\psi]+\sum_{j=0}^{2} E_{0,K+3,\delta}[N^j\psi]\Bigg)}.
\end{split}
\end{equation*}
By \eqref{eq:fasterdecayl2}, we moreover have that
\begin{equation*}
\begin{split}
 | r X^2T^K\psi_2|(\tau,\uprho,\theta,\varphi_*)\leq&\: C(1+\tau)^{-5-K+2\delta}\Bigg[\sqrt{E_{2,K+7,\delta}[\psi]}+\sum_{j=0}^{3} \sqrt{E_{2,K+4,\delta}[N^j\psi]}+\sqrt{E_{0,K+7,\delta}[\psi]}\\ 
 &+\sum_{j=0}^{3} \sqrt{E_{0,K+4,\delta}[N^j\psi]}+\sqrt{E_{\geq 3,K+2,\delta}[\psi]}+\sum_{j=0}^1\sqrt{E_{\geq 3,K+2,\delta}[N^j\psi]} \Bigg].
\end{split}
\end{equation*}
We apply the fundamental theorem of calculus, integrating $X^2T^K\psi_2$ between $\uprho=\uprho'$ and $\uprho=\uprho_{\gamma^{\alpha}(\tau)}$, together with the estimates above to conclude that there exists a $\nu>0$ such that
\begin{equation}
\label{eq:estXpsi2}
\begin{split}
|T^K X\psi_2|(\tau,\uprho,\theta,\varphi_*)\leq &\: C (1+\tau) ^{-4-K-\nu} \Bigg[\sqrt{E_{2,K+7,\delta}[\psi]}+\sum_{j=0}^{3} \sqrt{E_{2,K+4,\delta}[N^j\psi]}+\sqrt{E_{0,K+7,\delta}[\psi]}\\ 
&+\sum_{j=0}^{3} \sqrt{E_{0,K+4,\delta}[N^j\psi]}+\sqrt{E_{\geq 3,K+2,\delta}[\psi]}+\sum_{j=0}^1\sqrt{E_{\geq 3,K+2,\delta}[N^j\psi]} \Bigg].
\end{split}
\end{equation}
We conclude the proof by applying the fundamental theorem of calculus again, integrating now $XT^K\psi_2$ between $\uprho=\uprho'$ and $\uprho=\uprho_{\gamma^{\alpha}(\tau)}$. The corresponding boundary term at $\uprho=\uprho_{\gamma^{\alpha}(\tau)}$ can be estimated by \eqref{eq:asympl2gamma} and we estimate the integral term with \eqref{eq:estXpsi2}. Taking $\alpha$ suitably large, we arrive at \eqref{eq:asympinteriorl2}.
\end{proof}

\subsection{Asymptotics with vanishing Newman--Penrose charges}
We apply here the time-integral construction from Section \ref{sec:timeinv} to obtain the late-time asymptotics for $\psi_2$ arising from initial data that is smooth and compactly supported.
\begin{proposition}
\label{prop:asympl2NP0boundr}
Consider initial data $(\psi|_{\Sigma_0},T\psi|_{\Sigma_0})$ for \eqref{eq:waveeq}, with $(\phi|_{\Sigma_0},T\phi_{\Sigma_0}) \in (C_c^{\infty}(\Sigma))^2$.  Let $r_0>r_+$. Then there exists a $\nu>0$ and a constant $C=C(M,a,\Sigma_0,r_0,\nu)>0$, such that
\begin{equation*}
\begin{split}
&\Bigg|T^K\psi_2(u,v,\theta,\varphi_*)-\frac{4}{15}I_2[T^{-1}\psi](\theta,\varphi_*)(v-u)^2 T^{K+1}(u^{-3}v^{-3})\\
&+\frac{64}{3}\sqrt{\frac{\pi}{5}}a^2 I_0[T^{-1}\psi]Y_{2,0}(\theta)T^{K+1}\Bigg[\frac{1}{12}u^{-3}v^{-3}(v-u)^2-\frac{1}{4} u^{-3}v^{-2}(v-u)+\frac{1}{4}u^{-3}v^{-1}\Bigg]\Bigg|\\
\leq &\:  C u ^{-4-K-\nu} v^{-1}(S_{2,\delta,K+1}[T^{-1}\psi]+D_{2,\beta,K+1}[T^{-1}\psi]+\sum_{m=-2}^2|I_{2m}[T^{-1}\psi]|)\\
&+ Cu ^{-4-K-\nu} v^{-1}\Bigg[\sqrt{E_{2,K+8,\delta}[T^{-1}\psi]}+\sum_{j=0}^{3} \sqrt{E_{2,K+5,\delta}[N^jT^{-1}\psi]}+\sqrt{E_{0,K+8,\delta}[T^{-1}\psi]}\\ 
&+\sum_{j=0}^{3} \sqrt{E_{0,K+5,\delta}[N^j\psi]}+\sqrt{E_{\geq 3,K+3,\delta}[T^{-1}\psi]}+\sum_{j=0}^1\sqrt{E_{\geq 3,K+3,\delta}[N^jT^{-1}\psi]} \Bigg]
\end{split}
\end{equation*}
in $\{r\geq r_0\}$. 

In particular,
\begin{equation*}
\begin{split}
\Bigg|T^K\phi_2|_{\mathcal{I}^+}(u,\theta,\varphi_*)&+\frac{2}{5}I_2[T^{-1}\psi](\theta,\varphi_*)T^K(u^{-4})-\frac{8}{3}\sqrt{\frac{\pi}{5}}a^2 I_0[T^{-1}\psi]Y_{2,0}(\theta)T^K(u^{-4})\Bigg|\\
\leq&\:  C u ^{-4-K-\nu} \Big(S_{2,\delta,K+1}[T^{-1}\psi]+D_{2,\beta,K+1}[T^{-1}\psi]+\sum_{m=-2}^2|I_{2m}[T^{-1}\psi]|+\sqrt{E_{2,K+1,\delta}[T^{-1}\psi]}\Big).
\end{split}
\end{equation*}
In $\{r\leq r_0\}$, we can express:
\begin{equation*}
\begin{split}
&\Bigg|T^K\psi_2+\frac{16}{3}\sqrt{\frac{\pi}{5}}a^2 I_0[T^{-1}\psi]Y_{2,0}(\theta)T^K((1+\tau)^{-5})\Bigg|(\tau,\uprho,\theta,\varphi_*)\\ 
 \leq &\:  C (1+\tau) ^{-5-K-\nu}(S_{2,\delta,K+1}[T^{-1}\psi]+D_{2,\beta,K+1}[T^{-1}\psi]+\sum_{m=-2}^2|I_{2m}[T^{-1}\psi]|)\\
&+ C (1+\tau) ^{-5-K-\nu}\Bigg[\sqrt{E_{2,K+8,\delta}[T^{-1}\psi]}+\sum_{j=0}^{3} \sqrt{E_{2,K+5,\delta}[N^jT^{-1}\psi]}+\sqrt{E_{0,K+8,\delta}[T^{-1}\psi]}\\ 
&+\sum_{j=0}^{3} \sqrt{E_{0,K+5,\delta}[N^j\psi]}+\sqrt{E_{\geq 3,K+3,\delta}[T^{-1}\psi]}+\sum_{j=0}^1\sqrt{E_{\geq 3,K+3,\delta}[N^jT^{-1}\psi]} \Bigg].
\end{split}
\end{equation*}
\end{proposition}
\begin{proof}
We apply Propositions \ref{prop:mainpropTinv} and \ref{prop:TinvNPconst} to conclude that $T^{-1}\phi$ has sufficiently high regularity to conclude that all the relevant energies for $T^{-1}\psi$ are finite and $I_2[T^{-1}\psi]$ is well-defined. Then the estimates follows immediately by applying the estimates in Propositions \ref{prop:mainpropasympl2} and \ref{prop:asympboundrl2} to $T^{-1}\psi$ instead of $\psi$.
\end{proof}

\appendix

\section{Weighted pointwise estimates}
\label{sec:apppoint}
We derive in this section a lemma which is convenient for turning weighted energy estimates into pointwise estimates.
\begin{lemma}
\label{lm:pointw}
Let $h,f: \mathcal{R}\to \R$ be $C^1$ functions, such that $\lim_{\uprho \to \infty}h(\tau,\uprho,\theta,\varphi_*)=0$. Let $k\in \R_{>0}$ and $\delta>0$ arbitrarily small. Let $k\geq 0$, then there exists $C=C(M,a,R,k,\delta)>0$, such that
\begin{align}
\label{eq:pointwapp1}
	\int_{S^2_{\tau',\uprho'}} f^2\,d\omega\leq &\:C\int_{N_{\tau'}} r^{1+\delta}(Lf)^2+r^{-3+\delta}[(Tf)^2+(\Phi f)^2] \,d\omega d\uprho+\int_{\Sigma_{\tau'}\cap\{R-M\leq r\leq R\}} f^2\,d\omega d\uprho \:\: \textnormal{if  $\uprho'>R$},\\
	\label{eq:pointwapp1b}
	\int_{S^2_{\tau',\uprho'}} \uprho'^{-1-\delta}f^2\,d\omega\leq &\:C\int_{N_{\tau'}} r^{-\delta}(Lf)^2+r^{-4-\delta}[(Tf)^2+(\Phi f)^2] \,d\omega d\uprho+\int_{\Sigma_{\tau'}\cap\{R-M\leq r\leq R\}} f^2\,d\omega d\uprho \:\: \textnormal{if  $\uprho'>R$},\\
	\label{eq:pointwapp2}
		\int_{S^2_{\tau',\uprho'}}  \uprho' h^2\,d\omega\leq &\: C\int_{\Sigma_{\tau'}}J^N[h]\cdot \mathbf{n}_{\tau}
	\,r^2 d\omega d\uprho,\\
	\label{eq:pointwapp3}
		\int_{S^2_{\tau',\uprho'}}  h^2\,d\omega\leq &\: C\sqrt{\int_{\Sigma_{\tau'}} r^{-2}J^N[h]\cdot \mathbf{n}_{\tau}
	\,r^2 d\omega d\uprho}\cdot \sqrt{\int_{\Sigma_{\tau'}} J^N[h]\cdot \mathbf{n}_{\tau}
	\,r^2 d\omega d\uprho},\\
	\label{eq:pointwapp4}
		\int_{S^2_{\tau',\uprho'}}  {\uprho'}^{-2k}h^2\,d\omega\leq &\: C\int_{\Sigma_{\tau'}} \left(r^{-2k-1}J^N[h]\cdot \mathbf{n}_{\tau}
	+r^{-2k-3}h^2\right)\,r^2 d\omega d\uprho.
\end{align}
\end{lemma}
\begin{proof}
Let $R_0>r_++M$ and let $\chi: [r_+,\infty)\to \R$ be a smooth cut-off function, such that $\chi(r)=1$ for all $r\geq R_0$ and $\chi=0$ for $r\leq R_0-M$.

By applying the fundamental theorem of calculus, integrating from $\uprho=R_0-M$, together with Cauchy--Schwarz and \eqref{eq:hardyX}, we obtain for $\uprho'\geq R_0$:
	\begin{equation*}
	\begin{split}
	f^2(\uprho',\theta,\varphi_*,\tau')=&\:\int_{R_0-M}^{\uprho'} 2 \chi f X(\chi f)\,d\uprho\Big|_{\tau=\tau'}\\
	\leq &\: \sqrt{\int_{R_0-M}^{\uprho'} r^{-1-\delta}(\chi f)^2\,d\uprho}\cdot \sqrt{\int_{R_0-M}^{\uprho'} r^{1+\delta}(X(\chi f))^2\,d\uprho}\Big|_{\tau=\tau'}\\
	\leq &\: C\sqrt{\int_{R_0-M}^{\uprho'} r^{1+\delta}(Xf)^2\,d\uprho}\cdot  \sqrt{\int_{R_0-M}^{\uprho'} r^{1-\delta}(Xf)^2\,d\uprho}\Big|_{\tau=\tau'}\\
	\leq &\:C \int_{R_0-M}^{\uprho'} r^{1+\delta}(Lf)^2+r^{-3+\delta}[(Tf)^2+(\Phi f)^2]\,d\uprho+C\int_{R_0-M}^{R_0} f^2\,d\uprho \Big|_{\tau=\tau'}.
	\end{split}
	\end{equation*}
	The estimate \eqref{eq:pointwapp1} then follows by integrating over $\s^2$, choosing $R_0$ appropriately.
	
	We similarly obtain
	\begin{equation*}
	\begin{split}
	\uprho'^{-1-\delta}f^2(\uprho',\theta,\varphi_*,\tau')=&\:\int_{R_0-M}^{\uprho'} -(1+\delta)r^{-2-\delta} (\chi f)^2+ 2r^{-1-\delta}\chi f X(\chi f)\,d\uprho\Big|_{\tau=\tau'}\\
	\leq &\: \int_{R_0-M}^{\uprho'}-\left(\frac{1}{2}+\delta\right)r^{-2} (\chi f)^2+ 2r^{-\delta}(X(\chi f))^2\,d\uprho\Big|_{\tau=\tau'}\\
	\leq &\:C \int_{R_0-M}^{\uprho'}r^{-\delta}(Lf)^2+r^{-4-\delta}[(Tf)^2+(\Phi f)^2]\,d\uprho+C\int_{R_0-M}^{R_0} f^2\,d\uprho \Big|_{\tau=\tau'}.
	\end{split}
	\end{equation*}

	We apply the fundamental theorem of calculus again to obtain:
	\begin{equation*}
	\begin{split}
	h^2(\uprho',\theta,\varphi_*,\tau')	=&\left(0-\int_{\uprho'}^{\infty} X(h)\,d\uprho\right)^2\Big|_{\tau=\tau'}\leq \rho'^{-1} \int_{\rho'}^{\infty} r^2(Xh)^2\,d\uprho\Big|_{\tau=\tau'}.
	\end{split}
	\end{equation*}
	The estimate \eqref{eq:pointwapp2} then follows by integrating over $\s^2$ and applying Cauchy--Schwarz again on $\s^2$.

	By applying the fundamental theorem of calculus, integrating from $\uprho=\infty$, in combination with Cauchy--Schwarz and \eqref{eq:hardyX}, we obtain
	\begin{equation*}
		\begin{split}
			h^2(\uprho',\theta,\varphi_*,\tau')=&\:0-\int_{\uprho'}^{\infty} 2 h X h\,d\uprho'\\
			\leq &\: C\sqrt{\int_{\uprho'}^{\infty} h^2\,d\uprho}\sqrt{\int_{\uprho'}^{\infty} (Xh)^2\,d\uprho}\\
			\leq &\: C\sqrt{\int_{\uprho'}^{\infty} r^2(Xh)^2\,d\uprho}\sqrt{\int_{R_0-M}^{\uprho'}(Xh)^2\,d\uprho}.
		\end{split}
	\end{equation*}
	We obtain \eqref{eq:pointwapp3} by integrating over $\s^2$.
	 
	Finally, we repeat the above application of the fundamental theorem of calculus, integrating from $\uprho=\infty$ to obtain
	\begin{equation*}
		\begin{split}
			r^{-2k}h_{\geq 1}^2(\uprho,\theta,\varphi_*,\tau')=&\:0-\int_{\uprho'}^{\infty} X(r^{-2k} h_{\geq 1}^2)\,d\uprho\\
			\leq &\: C\int_{\uprho'}^{\infty} r^{-2k-1} h_{\geq 1}^2+r^{-2k+1}(Xh_{\geq 1})^2\,d\uprho
		\end{split}
	\end{equation*}
	We then integrate over $\s^2$ and apply \eqref{eq:poincare1} to obtain \eqref{eq:pointwapp3}.
\end{proof}

\section{A basic interpolation inequality}
The lemma below is useful for interpolating between $r$- and $\tau$-decay.
\begin{lemma}
\label{lm:interpol}
Let $f: \mathcal{R}\to \R$ be a continuous function, such that for $q_1,q_2\in \R$ with $0\leq q_1\leq q_2$ and $r_+<r_1<r_2\leq \infty$:
\begin{align}
\label{eq:interpol1}
\int_{r_1}^{r_2}\int_{\s^2} r^{q_1} f^2(\tau,\uprho,\theta,\varphi_*)\,d\omega d\uprho\leq D_1 (1+\tau)^{-p},\\
\label{eq:interpol2}
\int_{r_1}^{r_2}\int_{\s^2} r^{q_2} f^2(\tau,\uprho,\theta,\varphi_*)\,d\omega d\uprho\leq D_2 (1+\tau)^{-p+(q_2-q_1)}.
\end{align}
Then for all $q_1\leq q\leq q_2$:
\begin{equation}
\label{eq:interpol3}
\int_{r_1}^{r_2}\int_{\s^2} r^{q} f^2(\tau,\uprho,\theta,\varphi_*)\,d\omega d\uprho\leq (D_1+D_2) (1+\tau)^{-p+(q-q_1)}.
\end{equation}
\end{lemma}
\begin{proof}
We split $[r_1,r_2]=J_{\leq}+J_{>}$, with $J_{\leq}=[r_1,r_2]\times \s^2\cap \{r\leq (1+\tau)\}$ and $J_{>}=[r_1,r_2]\times \s^2\cap \{r>(1+\tau)\}$. Then by applying \eqref{eq:interpol1} and \eqref{eq:interpol2}, we obtain
\begin{equation*}
\begin{split}
\int_{r_1}^{r_2}\int_{\s^2} r^{q} f^2\,d\omega d\uprho=&\:\int_{J_{\leq}} r^{q} f^2 \,d\omega d\uprho+\int_{J_{>}} r^{q} f^2 \,d\omega d\uprho\\
\leq &\:(1+\tau)^{q-q_1}\int_{J_{\leq}} r^{q_1} f^2 \,d\omega d\uprho+(1+\tau)^{-(q_2-q)}\int_{J_{>}} r^{q_2} f^2 \,d\omega d\uprho\\
\leq &\:D_1(1+\tau)^{-p+(q-q_1)}+D_2(1+\tau)^{-(q_2-q)-p+(q_2-q_1)},
\end{split}
\end{equation*}
from which \eqref{eq:interpol3} immediately follows.
\end{proof}

\bibliographystyle{alpha}
%\bibliography{../bibliography}

\end{document}